\documentclass[10pt,a4paper]{article}
\usepackage[margin=1in]{geometry}
\usepackage[T1]{fontenc}
\usepackage{lmodern}
\usepackage{microtype}
\usepackage{amsthm,amsmath,amsfonts,amssymb}
\usepackage[authoryear]{natbib}
\usepackage[hidelinks]{hyperref}
\usepackage{graphicx}
\usepackage{mathtools}
\usepackage{mathrsfs}
\usepackage{enumitem}
\usepackage{tikz}
\usetikzlibrary{arrows.meta,positioning,fit,calc}

\numberwithin{equation}{section}
\mathtoolsset{showonlyrefs=true}
\DeclareMathOperator*{\argmax}{arg\,max}

\theoremstyle{plain}
\newtheorem{theorem}{Theorem}[section]
\newtheorem{corollary}[theorem]{Corollary}
\newtheorem{proposition}[theorem]{Proposition}
\newtheorem{lemma}[theorem]{Lemma}

\theoremstyle{definition}
\newtheorem{definition}{Definition}[section]
\newtheorem{assume}{Assumption}[section]

\newtheorem{example}{Example}[section]
\newtheorem{remark}{Remark}[section]

\newcommand{\USC}{\operatorname{USC}}
\newcommand{\LSC}{\operatorname{LSC}}
\DeclareMathOperator{\Glob}{Glob}
\DeclareMathOperator{\Loc}{Loc}

\tikzset{
  box/.style={draw,rounded corners=2pt,align=center,inner sep=6pt},
  arr/.style={-{Latex[length=2.2mm]},thick,shorten <=2pt,shorten >=2pt},
  lab/.style={fill=white,inner sep=1.2pt,outer sep=0pt,
              font=\scriptsize,align=center,text width=20mm}
}

\newcommand{\combinedtocfile}{maintoc}
\makeatletter
\let\combined@addtocontents\addtocontents
\renewcommand{\addtocontents}[2]{%
  \def\combined@extension{#1}%
  \def\combined@toc{toc}%
  \ifx\combined@extension\combined@toc
    \combined@addtocontents{\combinedtocfile}{#2}%
  \else
    \combined@addtocontents{#1}{#2}%
  \fi}
\renewcommand{\tableofcontents}{%
  \section*{\contentsname}%
  \@starttoc{\combinedtocfile}}
\AtBeginDocument{%
  \let\combined@maketitle\maketitle
  \let\combined@@maketitle\@maketitle
  \let\combined@title\title
  \let\combined@author\author
  \let\combined@date\date
  \let\combined@thanks\thanks
  \let\combined@and\and}
\newcommand{\startcombinedsupplement}{%
  \clearpage
  \renewcommand{\combinedtocfile}{supplementtoc}%
  \let\maketitle\combined@maketitle
  \let\@maketitle\combined@@maketitle
  \let\title\combined@title
  \let\author\combined@author
  \let\date\combined@date
  \let\thanks\combined@thanks
  \let\and\combined@and
  \setcounter{footnote}{0}}
\makeatother

\begin{document}

\title{Sublinear Markov Semigroups Beyond Uniqueness}

\author{
Jongjin Park\thanks{\raggedright Research Institute of Mathematics, Seoul National University, Seoul, Republic of Korea. Email: \texttt{pjj4230@snu.ac.kr}.}
\and
David Criens\thanks{\raggedright Department of Mathematical Stochastics, University of Freiburg, Freiburg, Germany. Email: \texttt{david.criens@stochastik.uni-freiburg.de}.}
\and
Hyungbin Park\thanks{\raggedright Department of Mathematical Sciences, Seoul National University, Seoul, Republic of Korea. Emails: \texttt{hyungbin@snu.ac.kr}, \texttt{hyungbin2015@gmail.com}.}
}
\date{}
\maketitle

\begin{abstract}
To what extent does a local operator on smooth test functions determine a sub-Markov semigroup and its underlying structure of laws? Many existing approaches work with  Feller semigroups on the space \(C_b\) of bounded continuous functions and rely on uniqueness, often established through a comparison principle for the associated evolution equation. However, even when the local operator is linear, uniqueness may fail and a \(C_b\)-Feller restriction may exclude the associated canonical semigroup. We discover that the space \(\operatorname{USC}_b\) of bounded upper semicontinuous functions offers a natural setting for this study. Within this framework, a complete description of the relationships among sublinear local operators, sublinear sub-Markov semigroups, and their underlying structures of laws is achieved without assuming uniqueness. Under suitable regularity and Lyapunov conditions, we construct the greatest semigroup compatible with a prescribed sublinear local operator and obtain a stochastic representation in terms of a Markov process under sublinear expectations. We characterize all sublinear sub-Markov semigroups dominated by this greatest semigroup in two ways: through a one-to-one correspondence with stable subfamilies of its representing laws, and through an order relation between their infinitesimal generators and the prescribed local operator. We further characterize the greatest semigroup via the greatest viscosity subsolution of the associated Hamilton--Jacobi--Bellman equation. These results clarify what information about global stochastic evolution can be recovered from local infinitesimal descriptions and what freedom remains in the absence of uniqueness. The resulting characterizations are new  even when the local operator is linear.
\end{abstract}

\par\smallskip

\par\medskip

\setcounter{secnumdepth}{3}
\setcounter{tocdepth}{2}

\tableofcontents

\paragraph{Supplementary Material}
The Supplementary Material contains Appendices D--F. Every cross-reference
to a result in the Supplementary Material is identified explicitly.

\section{Introduction}\label{sec:intro}

\subsection{Motivation}

To what extent does an infinitesimal description determine a sub-Markov semigroup and its underlying structure of laws? Understanding the relationship between semigroups and local operators is a classical problem. A local operator specified on smooth test functions need not determine a unique Markov semigroup, even when the operator is linear; see \cite{bass2007pathwise,lejan2004flows,howitt2009consistent}. Such nonuniqueness arises naturally in probability and a wide range of applications. However, much of the existing work is carried out in settings where uniqueness is established, often through a comparison principle for the associated evolution equation. This paper investigates the relationships among sublinear local operators, sub-Markov semigroups, and their underlying structures of laws without assuming uniqueness. We seek to understand what information about global evolution can be recovered from a given infinitesimal description and what freedom remains. 
In particular, we construct the greatest semigroup compatible with a prescribed sublinear local operator satisfying a Lyapunov condition and characterize the semigroups dominated by it. Our results are novel even when the local operator is linear.

To illustrate nonuniqueness, consider the linear local operator
\[
\mathcal L f(x)
=\frac12\operatorname{tr}\,\bigl(a(x)D^2 f(x)\bigr)
\]
on the smooth test class $C_b^\infty(\mathbb R^d)$, where \(a\) is a bounded continuous positive semidefinite matrix-valued function.
This operator describes the infinitesimal behavior of weak solutions to the stochastic differential equation
\[
dX_t=\sqrt{a(X_t)}\,dW_t,
\]
where \(W\) is a \(d\)-dimensional Brownian motion.
The solution laws are characterized by the martingale problem associated with \(\mathcal L\).
When degeneracy is allowed, continuity of the coefficients does not ensure uniqueness in law.
Distinct Markov solutions may therefore yield different semigroups whose infinitesimal actions agree with \(\mathcal L\) on smooth test functions; see Example~\ref{ex:singular_diffusion_multiple_branches}.

We clarify how a prescribed local operator on smooth test functions constrains the class of admissible semigroups.
A central role in this analysis is played by the greatest semigroup consistent with that operator.
The key idea is to take the upper expectation over all compatible solution laws.
For the linear operator \(\mathcal L\) above, let \(\mathcal C_x\) denote the family of all solution laws of the associated martingale problem starting from \(x\), and define
\begin{equation}
	\label{eqn:intro_sto}
	\mathcal T_t f(x)
	:=\sup_{\mathbb P\in\mathcal C_x}
	\mathbb E^{\mathbb P}[f(X_t)],
\end{equation}
where \(X\) is the coordinate process.
We call the indexed family \(\mathcal C=\{\mathcal C_x\}_{x}\) a \emph{structure of laws}.
Under suitable compactness and measurability conditions, stability of \(\mathcal C\) under conditioning and concatenation yields the semigroup property.
The resulting operators form a sublinear Markov semigroup  with infinitesimal action
\[
\lim_{t\searrow0}
\frac{\mathcal T_t f(x)-f(x)}{t}
=\mathcal Lf(x),
\qquad f\in C_b^\infty(\mathbb R^d).
\]
Since the supremum is taken over all compatible solution laws, \(\mathcal T\) is a natural candidate for the greatest semigroup whose infinitesimal action on smooth test functions agrees with \(\mathcal L\).

The standard \(C_b\)-Feller setting is too restrictive for this purpose.
The semigroup \(\mathcal T\) defined in \eqref{eqn:intro_sto} need not belong to this class.
Even for a linear local operator with continuous coefficients, it can happen that the class of compatible \(C_b\)-Feller semigroups is nonempty yet has no greatest element; see Example~\ref{ex:no_greatest_feller_valuation}.
We discover that \(\operatorname{USC}_b\) provides a natural setting for this analysis.
This setting retains the canonical greatest semigroup associated with prescribed local constraints and accommodates the upper-expectation construction in \eqref{eqn:intro_sto}.

The sublinearity of the upper expectation $\mathcal T$ makes sublinear semigroups a natural framework for our study.
Within this framework, the relationships among sublinear semigroups, structures of laws, and sublinear local operators become explicit.
These objects are linked through three maps: the semigroup map, globalization, and localization.
Analyzing these maps clarifies how distinct sublinear semigroups can share the same infinitesimal description and what information about their underlying stochastic dynamics can still be recovered from local data.

The semigroup map sends structures of laws to semigroups and establishes an order isomorphism between the corresponding classes.
In particular, each semigroup uniquely determines its representing structure of laws within a natural class.
Globalization assigns to each local operator the structure of all compatible solution laws, while localization extracts the local operator from a semigroup.
Under the identification provided by the semigroup map, globalization and localization connect local operators with global evolutions, but they need not be inverses.
Localization may lose information that distinguishes semigroups sharing the same local operator.

Within this framework, we construct the greatest sublinear sub-Markov semigroup whose infinitesimal action on smooth test functions agrees, in a relaxed sense, with the prescribed operator.
The construction proceeds through globalization and yields a stochastic representation of the form \eqref{eqn:intro_sto}, while localization establishes the required infinitesimal consistency.
Moreover, we obtain two characterizations of all sublinear sub-Markov semigroups dominated by the greatest semigroup: one through a one-to-one correspondence with stable subfamilies of the laws appearing in its stochastic representation, and the other through an order relation between their infinitesimal generators and the prescribed local operator.

A sublinear local operator encodes local information about infinitesimal behavior, while a semigroup and its structure of laws describe global evolution over time.
The paper thus investigates how much of this global description can be recovered from local information.
The resulting framework characterizes the constraints imposed by a local operator and the freedom remaining among compatible evolutions, and determines what information is preserved or lost when passing between local and global descriptions.
These results clarify the scope and limits of recovering global dynamics from infinitesimal information without assuming a comparison principle or uniqueness of the global evolution.

\subsection{Related literature}

We draw on and develop two complementary perspectives to study the relationships among the semigroup map, globalization, and localization.
The first approach is to recover structures of laws from semigroups through representations of the form \eqref{eqn:intro_sto}.
Abstract representation theorems have been developed by many authors; see
\cite{bartl2020conditional, bartl_cheridito_kupper_19,
	criens2025representation}.
Such representations are widely used in robust finance, dynamic risk measurement, and stochastic control; see, for example,
\cite{neufeld2014measurability,neufeld_nutz_17, delbaen2006structure, Delbaen_Peng_Rosazza_10, Fleming_Soner_06, zbMATH06323249, peng2019nonlinear, nutz2013random, nutz2013constructing}.
Beyond these abstract results, we provide a constructive representation of the greatest semigroup in terms of the largest structure of laws compatible with a prescribed local operator.
Moreover, we characterize how prescribed local operators constrain the correspondence between semigroups and their representing structures of laws, an aspect not captured by the abstract correspondence alone.

The second approach concerns the relationship between sublinear local operators and semigroups.
An established method starts with a sublinear Markov semigroup on \(C_b(\mathbb R^d)\), extracts its infinitesimal description, and constructs a compatible stochastic evolution.
This evolution is then identified with the original semigroup through uniqueness for the associated generator equation.
Viscosity comparison provides an important route to this identification; see
\cite{criensniemann2024markov,
	criens_niemann_JEEQ, criens2025representation,
	kuhn2021infinitesimal, denishu2011function,
	peng2019nonlinear}.
This identification relies on regularity and structural assumptions that support the relevant uniqueness argument.
In contrast, we allow a prescribed local operator to admit multiple compatible semigroups and establish both the greatest-element property of the semigroup obtained by globalization and a characterization of all semigroups dominated by this greatest element. 
To the best of our knowledge, no precise theorem establishing this extremal property and the resulting characterization has previously been stated, even when the prescribed local operator is linear.

Related extremal and comparison results have also been developed in other frameworks; see the early work \cite{nisio_82,lions_nisio_82} and the more recent treatments \cite{criens2025stochastic,blessing2025gamma}, as well as the extremal characterization of PDE solutions in \cite{fernholz_karatzas_11}.
These studies differ from ours in their starting data, the objects being compared, or the domain and regularity assumptions underlying generator comparison.
Our focus is on characterizing the entire class of semigroups compatible with a prescribed local upper bound, without requiring uniqueness.

\subsection{Overview}

We develop our theory in a time-homogeneous Markovian framework
that allows for killing. We consider sublinear sub-Markov
semigroups $\mathcal T=\{\mathcal T_t\}_{t\geq0}$ on the cone
$\USC_b(D)$ of bounded upper semicontinuous functions on a
convex open state space $D\subset\mathbb R^d$, satisfying the regularity
axioms in Definition~\ref{def:SG}. 
The infinitesimal behaviour of a semigroup \(\mathcal{T}\) is described by its \emph{upper generator} $\overline{\mathcal G}^{\mathcal T}$, introduced 
in Definition~\ref{def:upper_generator}.
In our setting, 
finiteness and locality of this generator on
the bounded smooth test
class $C_b^\infty(D)$ yield a unique second-order jet representation:
\[
\overline{\mathcal G}^{\mathcal T} f(x)
=
G_{\mathcal T}
\bigl(x,f(x),\nabla f(x),\nabla^2f(x)\bigr)
\]
for all $f$ in this test class and $x\in D$.
The structural conditions \ref{item:G1}--\ref{item:G3} on
$G_{\mathcal T}$ follow from the semigroup axioms
(Proposition~\ref{prop:local_upper_generator_representation}). 

We first introduce three classes of objects associated with the
Lyapunov condition in Assumption~\ref{assume:lyapunov}.
We call functions satisfying \ref{item:G1}--\ref{item:G3}
\emph{local specifications} and denote by
$\mathfrak G_{\mathrm{Lyap}}$ the class of those satisfying
the Lyapunov condition.
For each $G\in\mathfrak G_{\mathrm{Lyap}}$, let $\mathcal P(G)$ be the class of laws solving the corresponding $G$-supermartingale problem, which encodes the local constraints
prescribed by $G$. 
We denote by $\mathfrak V_{\mathrm{loc,Lyap}}$ the class of
semigroups whose upper generators satisfy the finiteness
and locality requirements and whose local
specifications $G_{\mathcal T}$ belong to
$\mathfrak G_{\mathrm{Lyap}}$.

A structure of laws is an indexed family
$\mathcal C=\{ \mathcal C_{s,x} \}_{s\geq0,\;x\in D}$
whose fibers $\mathcal C_{s,x}$ consist of probability laws
of processes taking values in $D$ up to killing.
We work with stable structures of laws, whose conditioning
and concatenation properties are specified in
Definition~\ref{def:stable_virtual_uncertainty_structure}.
Let $\mathfrak S_{\mathrm{loc,Lyap}}$ denote the class of
fiberwise nonempty, time-homogeneous stable structures of laws $\mathcal C$ such that
$\mathcal C\subseteq\mathcal P(G)$ for some
$G\in\mathfrak G_{\mathrm{Lyap}}$, where inclusion is understood
fiberwise.

Three maps connect these objects.
The \emph{localization map} assigns to a semigroup the
local specification associated with its upper generator.
The \emph{globalization map} assigns to a local specification
$G$ the indexed family $\mathcal P(G)$.
Thus,
\[
\Loc(\mathcal T):=G_{\mathcal T},
\qquad
\Glob(G):=\mathcal P(G).
\]
The \emph{semigroup map} assigns to a structure of laws
$\mathcal C\in\mathfrak S_{\mathrm{loc,Lyap}}$ the semigroup defined by
\begin{equation}
	\label{eq:intro_robust_valuation_earlier}
	\bigl(\operatorname{Sem}(\mathcal C)\bigr)_t f(x)
	:=
	\sup_{\mathbb P\in\mathcal C_{0,x}}
	\mathbb E^{\mathbb P}
	\bigl[f(X_t)\mathbb I_{\{t<\tau_{\mathrm{kill}}\}}\bigr],
\end{equation}
where $\tau_{\mathrm{kill}}$ is the killing time of the coordinate
process $X$, and the integrand is taken to be zero at and after killing.
Within this framework, we study the correspondence between
semigroups and their representing structures of laws,
as well as the order and extremal relations connecting these
global objects to local specifications.

Our first main result establishes that the semigroup map
is an order isomorphism
\begin{equation}
	\label{eq:intro_valuation_structure_correspondence_earlier}
	\operatorname{Sem}:
	\mathfrak S_{\mathrm{loc,Lyap}}
	\longrightarrow
	\mathfrak V_{\mathrm{loc,Lyap}},
	\qquad
	\operatorname{Sem}^{-1}(\mathcal T)=\mathcal R^{\mathcal T}
\end{equation}
(Corollary~\ref{thm:local_lyapunov_valuation_structure_correspondence}).
Thus, although arbitrary representing families need not be unique,
a semigroup determines its representing stable structure of laws
uniquely within $\mathfrak S_{\mathrm{loc,Lyap}}$.
This uniqueness concerns the passage from a semigroup to its
structure of laws; it does not assert that a local
specification determines a unique semigroup.
The recovered family $\mathcal R^{\mathcal T}$ is intrinsic to the
semigroup: a law $\mathbb P$ with initial condition $(t,x)$ belongs to
$\mathcal R_{t,x}^{\mathcal T}$ precisely when every semigroup-orbit process
\[
s\longmapsto
\bigl(\mathcal T_{R-s}g\bigr)(X_s)
\mathbb I_{\{s<\tau_\infty\}},
\]
is a $\mathbb P$-supermartingale on $[t,R]$, for every $R\ge t$ and
$g\in\USC_b(D)$, where $\tau_\infty$ denotes the lifetime.
The representation is attained:
for every initial state, horizon, and bounded upper semicontinuous
terminal function, 
the supremum is achieved by a law in the corresponding fiber
of $\mathcal R^{\mathcal T}$.

Our second main result describes how local specifications constrain global semigroups and structures of laws, and what information localization and globalization preserve.
We introduce the localization and globalization maps
\[
    \Loc:\mathcal T\longmapsto G_{\mathcal T},
    \qquad
    \Glob:G\longmapsto\mathcal P(G),
\]
which respectively extract the local specification of a semigroup and assign to each specification 
the family of all corresponding laws.
For $\mathcal T\in\mathfrak V_{\mathrm{loc,Lyap}}$ and $G\in\mathfrak G_{\mathrm{Lyap}}$, we establish
\begin{equation}
    \label{eq:intro_local_global_correspondence_earlier}
    G_{\mathcal T}\le G
    \quad\Longleftrightarrow\quad
    \mathcal T\le\mathcal V^G
    \quad\Longleftrightarrow\quad
    \mathcal R^{\mathcal T}\subseteq\mathcal P(G),
\end{equation}
where $\mathcal V^G:=\operatorname{Sem}(\Glob(G))$ (Corollary~\ref{cor:local_global_structure_order}).
The inequalities are understood pointwise and the inclusion fiberwise.
Thus a local upper bound $G$ determines the greatest semigroup
whose upper generating function is bounded by $G$, namely
$\mathcal V^G$, and the largest corresponding stable structure of laws, namely $\mathcal P(G)$.
This characterization requires neither uniqueness of the global
evolution nor global realizability of $G$.
For each fixed $G$, the semigroup map therefore restricts
to an order isomorphism between the stable substructures of
$\mathcal P(G)$ in $\mathfrak S_{\mathrm{loc,Lyap}}$
and the semigroups in $\mathfrak V_{\mathrm{loc,Lyap}}$
whose local specifications are bounded above by $G$.
Figure~\ref{fig:intro_three_classes} summarizes these relations.

\begin{figure}[htbp]
\centering
\resizebox{0.98\linewidth}{!}{%
\begin{tikzpicture}[
  >=Stealth,
  map/.style={->,line width=0.65pt},
  object/.style={align=center,inner sep=6pt,text width=4.0cm},
  maplabel/.style={align=center,inner sep=1pt,font=\footnotesize}
]
\node[object] (V) at (0,0) {
  {\large $\mathfrak V_{\mathrm{loc,Lyap}}$}\\[5pt]
  {\small Sublinear sub-Markov}\\[-1pt]{\small semigroups}
};
\node[object] (S) at (9.6,0) {
  {\large $\mathfrak S_{\mathrm{loc,Lyap}}$}\\[5pt]
  {\small Stable structures}\\[-1pt]{\small of laws}
};
\node[object] (G) at (4.8,-3.1) {
  {\large $\mathfrak G_{\mathrm{Lyap}}$}\\[5pt]
  {\small Local specifications}
};
\node[font=\small] at (4.8,1.12) {Order isomorphism};
\draw[map] ($(S.west)+(0,0.24)$) --
  node[above=3pt,font=\small] {$\operatorname{Sem}$}
  ($(V.east)+(0,0.24)$);
\draw[map] ($(V.east)+(0,-0.24)$) --
  node[below=3pt,font=\small]
    {$\operatorname{Sem}^{-1}:\ \mathcal T\mapsto\mathcal R^{\mathcal T}$}
  ($(S.west)+(0,-0.24)$);
\draw[map] (V.south) --
  node[maplabel,sloped,below=6pt,pos=0.49] {
    $\Loc$\\[1pt]$\mathcal T\mapsto G_{\mathcal T}$
  } (G.north west);
\draw[map] (G.north east) --
  node[maplabel,sloped,below=6pt,pos=0.51] {
    $\Glob$\\[1pt]$G\mapsto\mathcal P(G)$
  } (S.south);
\node[align=center] at (4.8,-4.35) {
  $\displaystyle
    G_{\mathcal T}\le G
    \quad\Longleftrightarrow\quad
    \mathcal T\le\mathcal V^G
    \quad\Longleftrightarrow\quad
    \mathcal R^{\mathcal T}\subseteq\mathcal P(G)$
};
\end{tikzpicture}%
}
\caption{Relations among semigroups, stable structures of laws, and local specifications.
The top maps are inverse order isomorphisms.
Localization and globalization preserve order and satisfy the displayed equivalence, but their compositions need not recover the original semigroup or specification.}
\label{fig:intro_three_classes}
\end{figure}
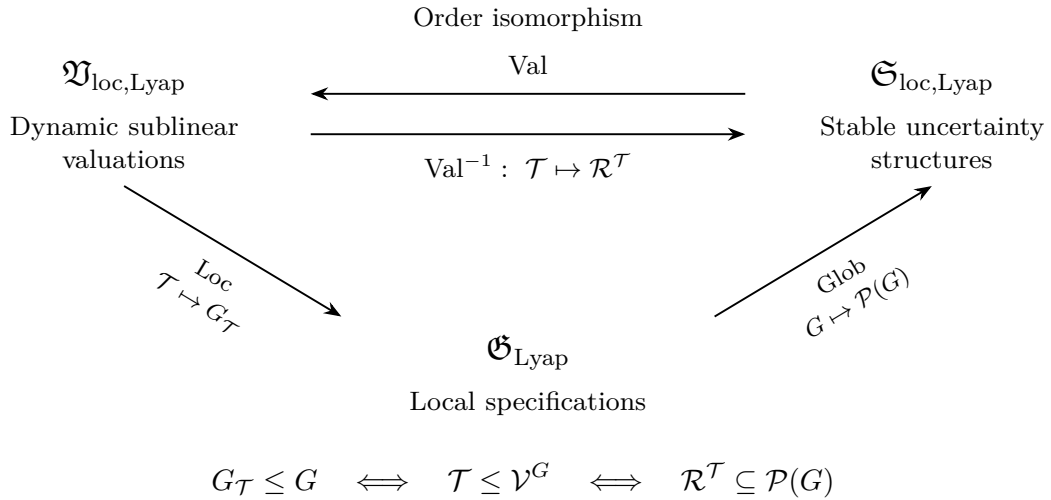

The compositions of these maps describe precisely what is retained
when passing from global to local information and back, or in the
opposite direction.
Starting from a semigroup $\mathcal T$, localization retains its
local specification but may lose information about its global
evolution.
Applying globalization and then the semigroup map yields
\[
    (\operatorname{Sem}\circ\Glob\circ\Loc)(\mathcal T)
    =
    \max\bigl\{
        \mathcal S\in\mathfrak V_{\mathrm{loc,Lyap}}:
        \Loc(\mathcal S)=\Loc(\mathcal T)
    \bigr\},
\]
the greatest semigroup with the same local specification
(Theorem~\ref{thm:maximal_valuation_same_specification}).
Thus this composition preserves the local specification,
while it may strictly enlarge the semigroup and its recovered
structure of laws.

Starting from a specification $G$, the reverse composition produces
\[
    G^\downarrow
    :=(\Loc\circ\operatorname{Sem}\circ\Glob)(G)
    =G_{\mathcal V^G}.
\]
We show that
\[
    G^\downarrow
    =
    \max\bigl\{
        H\in\Loc(\mathfrak V_{\mathrm{loc,Lyap}}):H\le G
    \bigr\}
    =\min\bigl\{
        H\in\mathfrak G_{\mathrm{Lyap}}:
        \mathcal P(H)=\mathcal P(G)
    \bigr\}
\]
(Theorem~\ref{thm:realizable_specification_reduction}).
Hence $G^\downarrow$ is both the greatest globally realizable
specification below $G$ and the least specification defining
the same global structure of laws.
Here, a local specification is called \emph{globally realizable} if it is the local specification of a semigroup in $\mathfrak V_{\mathrm{loc,Lyap}}$.
In particular, this reverse composition preserves the entire family of corresponding laws:
\[
    \mathcal P(G^\downarrow)=\mathcal P(G).
\]

These two compositions need not be identities.
Distinct semigroups may share the same local specification,
whereas distinct specifications may define the same global
structure of laws.
Examples~\ref{ex:singular_diffusion_multiple_branches}
and~\ref{ex:unrealizable_drift_spike} illustrate these two
phenomena, respectively.
The first phenomenon can occur even for a continuous
diffusion specification.

Our third main result identifies the evolution selected by globalization as the greatest viscosity subsolution of the associated Hamilton--Jacobi--Bellman equation. 
It provides a stochastic representation of this subsolution without requiring a comparison principle.
Fix $G\in\mathfrak G_{\mathrm{Lyap}}$ and $f\in\USC_b(D)$.
Let $\operatorname{Sub}_G^+(f)$ denote the class of bounded upper semicontinuous functions on $[0,\infty)\times D$ that are viscosity subsolutions of
\begin{equation}
	\label{eq:intro_pde_earlier}
	\partial_tu-G\bigl(x,u,\nabla u,\nabla^2u\bigr)=0
\end{equation}
on $(0,\infty)\times D$, with initial upper bound $u(0,\cdot \,)\le f$.
Then
\begin{equation}
	\label{eq:intro_perron_earlier}
	\max_{u\in\operatorname{Sub}_G^+(f)}u(t,x)
	=
	\sup_{\mathbb P\in\mathcal P_{0,x}(G)}
	\mathbb E^{\mathbb P}
	\bigl[f(X_t)\mathbb I_{\{t<\tau_{\mathrm{kill}}\}}\bigr]
	=
	\mathcal V_t^Gf(x)
\end{equation}
(Theorem~\ref{thm:upper_perron_feynman_kac}).
The semigroup orbit $(t,x)\mapsto\mathcal V_t^Gf(x)$ itself belongs to $\operatorname{Sub}_G^+(f)$ and is its greatest element.
Uniqueness of viscosity solutions in the discontinuous
convention identifies every dominated semigroup satisfying lower
infinitesimal consistency with $G$ as $\mathcal V^G$
(Theorem~\ref{thm:canonical_uniqueness_under_viscosity_uniqueness}).
Parabolic comparison additionally yields the $C_b$-Feller property
(Theorem~\ref{thm:canonical_feller_under_comparison}).

Consequently, without assuming the existence or uniqueness of a semigroup realizing \(G\), we obtain the following results.
\begin{enumerate}[label=(\roman*)]
    \item \emph{A stochastic representation:}
    \(\mathcal V^G\) admits an explicit stochastic representation
    in terms of the sublinear expectations induced by \(\mathcal P(G)\):
    \begin{equation}
        \mathcal V_t^G f(x)
        = \sup_{\mathbb P\in\mathcal P_{0,x}(G)}
        \mathbb E^{\mathbb P}
        \bigl[f(X_t)\mathbb I_{\{t<\tau_{\mathrm{kill}}\}}\bigr],
    \end{equation}
    for all \(f\in\USC_b(D)\), \(t\ge0\), and \(x\in D\).    
    \item \emph{An existence criterion:}
    \(G\) is globally realizable if and only if \(G_{\mathcal V^G}=G\).
    Equivalently, a semigroup with exact generating function \(G\)
    exists if and only if \(\mathcal V^G\) has exact generating function \(G\).  
     \item \emph{A characterization of dominated semigroups through structures of laws:}
    A sub-Markov semigroup \(\mathcal T\) is dominated by
    \(\mathcal V^G\) if and only if there exists a time-homogeneous
    stable subfamily \(\mathcal C\) of \(\mathcal P(G)\) such that
    \[
        \mathcal T_t f(x)
        = \sup_{\mathbb P\in\mathcal C_{0,x}}
        \mathbb E^{\mathbb P}
        \bigl[f(X_t)\mathbb I_{\{t<\tau_{\mathrm{kill}}\}}\bigr]
    \]
    for all \(f\in\USC_b(D)\), \(t\ge0\), and \(x\in D\).
    Moreover, one may take \(\mathcal C=\mathcal R^{\mathcal T}\). 
    \item \emph{A characterization of dominated semigroups through infinitesimal upper bounds:}
    A sub-Markov semigroup \(\mathcal T\) is dominated by
    \(\mathcal V^G\) if and only if \(G_{\mathcal T}\le G\).
    The canonical semigroup \(\mathcal V^G\) is the greatest semigroup
    satisfying this bound and, when \(G\) is globally realizable,
    the greatest semigroup realizing \(G\).   
    \item \emph{An analytic extremal characterization:}
    The orbits of \(\mathcal V^G\) are the greatest bounded upper
    semicontinuous viscosity subsolutions of the associated
    Hamilton--Jacobi--Bellman Cauchy problems, whether or not
    \(G\) is globally realizable.   
\end{enumerate}

A central technical ingredient is the stochastic realization theorem (Theorem~\ref{prop:one_step_usc_subsolution_realization}), which converts an upper semicontinuous viscosity subsolution into a stopped expectation inequality under an admissible law.
This result is closely related to stochastic sub- and superoptimality principles for viscosity solutions \cite{swiech1996another,koike_swiech2013representation,cesaroni2006lyapunov}.
Our theorem establishes such a realization for merely upper semicontinuous local specifications satisfying a Lyapunov condition, allowing degenerate diffusions without imposing spatial Lipschitz regularity on the coefficients or assuming uniqueness of the associated martingale problem.
The difficulty is to construct a realizing law while preserving both the local admissibility constraints and the stopped expectation inequality through approximation, when neither Lipschitz stability estimates for the original dynamics nor viscosity comparison are available.
In particular, weak convergence alone does not justify passage to the limit in payoffs evaluated at exit times.
The proof combines outer regularization of the local specification, smooth approximations and nondegenerate auxiliary dynamics, and a compactness argument with suitable inner approximations of the stopping domain.
This yields stochastic realization in a setting where the local specification need not determine a unique global evolution.
Applied to semigroup orbits, the theorem turns a local upper-generator bound into global domination by $\mathcal V^G$. 
More generally, it yields a Perron characterization of $(t,x)\mapsto\mathcal V_t^G f(x)$ as the largest upper semicontinuous viscosity subsolution with initial data bounded above by $f$.

Section~\ref{sec:valuation_uncertainty} introduces semigroups, local
specifications, and stable structures of laws.
Section~\ref{sec:valuation_orbit_nonemptiness} establishes the
semigroup--structure correspondence.
Section~\ref{sec:local_global_correspondence} proves the local-to-global
characterization and its extremal consequences.
Section~\ref{sec:exact_generation_comparison} treats exact generation,
the upper Perron representation, viscosity uniqueness, and comparison.
Appendices~\ref{app:local_upper_generators},
\ref{app:characteristic_recovery}, and~\ref{sec:stochastic_realization_perron}
prove localization, characteristic recovery, and stochastic realization,
respectively. The Supplementary Material contains auxiliary selection,
compactness, martingale-problem, and approximation results.

\section{Semigroups, Local Specifications, and Stable Structures of Laws}
\label{sec:valuation_uncertainty}

In this section, we define semigroups, local specifications, and stable structures of laws, together with the local Lyapunov subclasses related in the next section.
Let $D\subset\mathbb R^d$ be convex and open, with an exhaustion $D=\bigcup_nD_n$ by convex smooth subdomains satisfying $D_n\Subset D_{n+1}$.

We use $C$, $\USC$, and $\LSC$ for continuous, upper semicontinuous, and lower semicontinuous functions; a subscript $b$ means bounded.
The space $C_b^\infty(D)$ consists of smooth functions whose derivatives of every order, including order zero, are bounded.
We write $K\Subset O$ if $\overline K$ is compact in $O$, and $C^{1,2}$ tests are defined on a neighborhood of the contact point.
The notation $\mathcal P(E)$ denotes probability laws on a Polish space $E$, $\Rightarrow$ denotes weak convergence, and $\operatorname{Gr}(\Gamma)$ denotes the graph of a correspondence.
For locally bounded functions, $v^*$ and $v_*$ are the semicontinuous envelopes.
Our half-relaxed-limit convention is
\[
 (\limsup_n^*v_n)(x):=\limsup_{n\to\infty,\,y\to x}v_n(y),
 \qquad \liminf_{n,*}v_n:=-\limsup_n^*(-v_n),
\]
with the same convention for continuous parameters tending to zero.
For varying domains, the joint limit is restricted to points where the functions are defined.

Let $\mathbb S(d)$ and $\mathbb S^+(d)$ denote symmetric and positive semidefinite matrices, with $H_1\ge H_2$ meaning $H_1-H_2\in\mathbb S^+(d)$. 
Set
\[
 \mathfrak J:=\mathbb R\times\mathbb R^d\times\mathbb S(d),
 \qquad \|(z,p,H)\|^2:=|z|^2+|p|^2+\tfrac12\operatorname{tr}(H^2).
\]
For $f\in C^2(D)$, write $J_x^2f:=(f(x),\nabla f(x),\nabla^2f(x))$.

\subsection{Sublinear sub-Markov semigroups}
\label{sec:dsvr}

We now formulate sublinear sub-Markov semigroups axiomatically, allowing time inhomogeneity.
The operators act on the cone \(\USC_b(D)\) of bounded upper semicontinuous functions on \(D\).
Here and below, sublinearity on \(\USC_b(D)\) means positive homogeneity and subadditivity on this cone.
 
\begin{definition}\label{def:SG}
A sublinear sub-Markov semigroup on \(\USC_b(D)\) is a family of operators
\[
    \mathcal T
    =
    \{\mathcal T_{t,T}\}_{0\le t\le T<\infty},
    \qquad
    \mathcal T_{t,T}:\USC_b(D)\longrightarrow\USC_b(D),
\]
such that $\mathcal T_{t,t}=\operatorname{id}_{\USC_b(D)}$ for all $t\ge0$, and the following properties hold.
\begin{enumerate}[label=(V\arabic*), ref=(V\arabic*)]
    \item\label{V1}
    \(\mathcal T_{t,T}\) is monotone and sublinear on \(\USC_b(D)\) for all \(0\le t\le T<\infty\).

    \item\label{V2}
    \(\lVert\mathcal T_{t,T}f\rVert_\infty\le\lVert f\rVert_\infty\) for all \(0\le t\le T<\infty\) and \(f\in\USC_b(D)\).

    \item\label{V3}
    \(\mathcal T_{t,T}\) is continuous from above for all \(0\le t\le T<\infty\); that is, if \(f_n,f\in\USC_b(D)\) and \(f_n\searrow f\) pointwise on \(D\), then \(\mathcal T_{t,T}f_n\searrow\mathcal T_{t,T}f\) pointwise on \(D\).

    \item\label{V4}
    For every \(f\in\USC_b(D)\), \(\limsup_{n\to\infty}^{*}\mathcal T_{t_n,T_n}f\le\mathcal T_{t,T}f\) whenever \(0\le t_n\le T_n<\infty\) and \((t_n,T_n)\to(t,T)\).

    \item\label{V5}
    The time-consistency property holds, that is, $\mathcal{T}_{t,T}=\mathcal{T}_{t,s}\mathcal{T}_{s,T}$ for all $0\le t\le s\le T$.
\end{enumerate}
The semigroup is called \emph{time-homogeneous} if $\mathcal{T}_{t,T}$ depends only on $T-t$.
In this case, we use the one-parameter notation
$$\{\mathcal T_t\}_{t\ge0}:=\{\mathcal T_{0,t}\}_{t\ge0}.$$
We write $\mathfrak V$ for the class of time-homogeneous sublinear sub-Markov semigroups on $\USC_b(D)$.
\end{definition}

The above definition collects the order and regularity properties of a sublinear sub-Markov semigroup.
Condition~\ref{V1} imposes monotonicity and sublinearity; related formulations appear in \cite[Section~2.3]{artzner1999coherent} and \cite[Theorems~2.1--2.2]{burzoni2021viability}.
Condition~\ref{V2} gives a uniform norm bound and allows loss of mass through killing.
Conditions~\ref{V3} and~\ref{V4} impose continuity from above and temporal upper semicontinuity, respectively.
Finally, condition~\ref{V5} expresses time consistency through the two-parameter semigroup property; under time homogeneity, it reduces to the usual one-parameter semigroup property.
In the remainder of the paper, semigroups are understood to be time-homogeneous unless stated otherwise.

The choice of $\USC_b(D)$ allows evolutions whose upper expectations
preserve upper semicontinuity but not continuity, as the following
elementary example illustrates.

\begin{example}[A semigroup that does not preserve continuity]
\label{ex:peano_nonfeller_semigroups}
On $D=\mathbb R$, consider the classical Peano equation
\[
 \dot X_t=b(X_t),\qquad
 b(x)=2\operatorname{sgn}(x)\sqrt{|x|},\qquad b(0)=0,
\]
whose continuity-preservation obstruction is discussed in
\cite[Discussion~4.9(c)]{criens2025representation}.
For $x\ne0$, its unique solution is
$\Phi_t(x)=\operatorname{sgn}(x)(\sqrt{|x|}+t)^2$.
From zero, a solution may wait for an arbitrary time and then depart
in either direction, so its reachable set at time $t$ is $[-t^2,t^2]$.
Upper expectations over all these solutions therefore give
\begin{equation}
\label{eq:peano_usc_semigroup}
 \mathcal T_t f(x)=
 \begin{cases}
  f(\Phi_t(x)),&x\ne0,\\
  \displaystyle\max_{|y|\le t^2}f(y),&x=0,
 \end{cases}
 \qquad f\in\USC_b(\mathbb R).
\end{equation}
This semigroup belongs to $\mathfrak V$, but does not preserve
$C_b(\mathbb R)$. Indeed, for $f(x)=\tanh x$ and $t>0$,
\[
 \mathcal T_t f(0)=\tanh(t^2),\qquad
 \lim_{x\uparrow0}\mathcal T_t f(x)=-\tanh(t^2).
\]
Thus a continuous drift already gives an evolution for which
$\USC_b$ is invariant whereas $C_b$ is not.
\end{example}

\begin{proof}
Let $K_t(x)=\{\Phi_t(x)\}$ for $x\ne0$ and $K_t(0)=[-t^2,t^2]$.
This correspondence has nonempty compact values, is locally bounded
in $(t,x)$, and has a closed graph, including at $t=0$ and $x=0$.
It follows that $\max_{y\in K_t(x)}f(y)$ is jointly upper
semicontinuous for every $f\in\USC_b(\mathbb R)$, proving
preservation of $\USC_b$ and \ref{V4}.
Properties~\ref{V1}--\ref{V2} are immediate, and compactness gives
\ref{V3} by the decreasing-maximum property for upper semicontinuous
functions. Finally, $K_0(x)=\{x\}$ and
\[
 K_{s+t}(x)=\bigcup_{y\in K_s(x)}K_t(y).
\]
For $x\ne0$ this is the flow identity; for $x=0$ the union equals
$[-(s+t)^2,(s+t)^2]$. Hence \ref{V5} holds.
\end{proof}

\subsection{Finite local upper generators and local specifications}
\label{subsec:local_upper_generators}
\label{subsec:generating_function}

We introduce the upper generator of a time-homogeneous sublinear sub-Markov semigroup $\mathcal T$, which characterizes its local behavior.
For $\mathcal T\in\mathfrak V$ and $f\in C_b(D)$, set
\[
    \Delta_h^{\mathcal T}f:=\frac{\mathcal T_hf-f}{h},\qquad h>0.
\]
\begin{definition}
\label{def:upper_generator}
The \emph{upper generator} of $\mathcal T$ is the extended-real operator
\begin{equation}
\label{eq:upper_generator_definition}
    \overline{\mathcal G}^{\mathcal T}f(x)
    :=\limsup_{h\searrow0}^{*}\Delta_h^{\mathcal T}f(x),
    \qquad f\in C_b(D),\quad x\in D.
\end{equation}
Hence, for each fixed test function $f\in C_b(D)$, $\overline{\mathcal G}^{\mathcal T}f$ is upper semicontinuous, with extended-real values allowed.
\end{definition}

We restrict our attention to semigroups whose infinitesimal dynamics are local.
The following assumption requires the upper generator at a state to depend only on the test function in a neighborhood of that state.
This restricts the analysis to diffusion-type dynamics with killing and excludes spatial jumps.
Treating nonlocal generators would require a corresponding inverse theory for jump-type dynamics and is left for future work.

\begin{assume}
\label{assume:local_upper_generator}
For every $f\in C_b^\infty(D)$ and $x\in D$,
$\overline{\mathcal G}^{\mathcal T}f(x)\in\mathbb R$.
Moreover, if $f,g\in C_b^\infty(D)$ coincide in a neighborhood of $x$, then
\[
    \overline{\mathcal G}^{\mathcal T}f(x)
    =\overline{\mathcal G}^{\mathcal T}g(x).
\]
\end{assume}

Under this assumption, the following proposition records the second-order jet dependence of upper generators and defines the generating function.
The proof is given in Appendix~\ref{app:local_upper_generators}.
The structural properties~\ref{item:G1}--\ref{item:G3} below will serve as the key conditions on $G$ throughout this paper.

\begin{proposition}
\label{prop:local_upper_generator_representation}
Let $\mathcal T\in\mathfrak V$ satisfy
Assumption~\ref{assume:local_upper_generator}.
Then there is a unique finite function
$G:D\times\mathfrak J\to\mathbb R$ such that
\begin{equation}
\label{eq:thm:representationgenerator_eq1}
    \overline{\mathcal G}^{\mathcal T}f(x)
    =G(x,J_x^2f),
    \qquad f\in C_b^\infty(D),\quad x\in D.
\end{equation}
Moreover, the function $G$ satisfies the following properties:
\begin{enumerate}[label=(G\arabic*), ref=(G\arabic*)]
\item
\label{item:G1}
The function \(G\) is jointly upper semicontinuous on \(D\times\mathfrak J\), and \(G(x,\cdot)\) is sublinear on \(\mathfrak J\) for every \(x\in D\).
Moreover, \(G\) is locally uniformly Lipschitz in the jet variable: for every compact set \(K\Subset D\), there exists a constant \(L_K<\infty\) such that
\begin{equation}
\label{eq:local_uniform_jet_lipschitz_G}
    \bigl|G(x,Z_1)-G(x,Z_2)\bigr|
    \le
    L_K\lVert Z_1-Z_2\rVert,
    \qquad
    x\in K,\quad Z_1,Z_2\in\mathfrak J.
\end{equation}

\item
\label{item:G2}
\(G(x,z,p,H_1)\geq G(x,z,p,H_2)\) whenever \(H_1\geq H_2\), for all
\(x\in D\), \(z\in\mathbb R\), \(p\in\mathbb R^d\), and
\(H_1,H_2\in\mathbb S(d)\).

\item
\label{item:G3}
\(G(x,z_1,p,H)\leq G(x,z_2,p,H)\) whenever \(z_1\geq z_2\), for all
\(x\in D\), \(z_1,z_2\in\mathbb R\), \(p\in\mathbb R^d\), and
\(H\in\mathbb S(d)\).
\end{enumerate}
The function $G$ is called the \emph{upper generating function} of $\mathcal T$ and denoted by $G_{\mathcal T}$.
\end{proposition}

We denote the class of local specifications by
\[
    \mathfrak G
    :=\{G:D\times\mathfrak J\to\mathbb R:
                  G\text{ satisfies }\ref{item:G1}\text{--}\ref{item:G3}\}
\]
and set
\[
    \mathfrak V_{\mathrm{loc}}
    :=\{\mathcal T\in\mathfrak V:
            \mathcal T\text{ satisfies
            Assumption~\ref{assume:local_upper_generator}}\}.
\]
The function $G_{\mathcal T}$ is therefore defined intrinsically for every $\mathcal T\in\mathfrak V_{\mathrm{loc}}$.

For any $G\in\mathfrak G$, finite-dimensional convex duality represents each sublinear map \(G(x,\cdot)\) through a compact convex correspondence of linear coefficient functionals.
The support correspondence of \(G\) is defined by
\begin{align}
\label{def:supportset}
\mathfrak A_G(x)
:=
\left\{
V\in\mathfrak J:
\ell_V(W)\le G(x,W)
\ \text{for every }W\in\mathfrak J
\right\},
\qquad x\in D.
\end{align}
For \(V=(c,b,a)\) and \(U=(z,p,H)\) in \(\mathfrak J\), define
\[
    \ell_V(U):=\frac12\operatorname{tr}(aH)+b\cdot p+cz,
\]
and denote the coefficient cone by
\[
    \mathfrak K
    :=(-\infty,0]\times\mathbb R^d\times\mathbb S^+(d).
\]
The equality in \eqref{eq:support_function_representation}
provides the Hamilton--Jacobi--Bellman representation of $G$.
This representation is related to the sublinear Courr{\`e}ge--von Waldenfels theorem, which expresses infinitesimal generators of sublinear Markov semigroups on $\mathbb R^d$ with sufficiently rich domains as suprema of L\'evy-type operators \cite{kuhn2021infinitesimal}.

\begin{lemma}
\label{lem:support_correspondence_geometry}
Let \(G\in\mathfrak G\). Then the following statements hold.
\begin{enumerate}[label=(\roman*)]
\item For every \((x,U)\in D\times\mathfrak J\),
\begin{align}
\label{eq:support_function_representation}
    G(x,U)
    =
    \max_{V\in\mathfrak A_G(x)}\ell_V(U).
\end{align}

\item For every \(x\in D\), the set \(\mathfrak A_G(x)\) is a nonempty compact convex subset of \(\mathfrak K\).
Moreover, the correspondence \(\mathfrak A_G:D\rightrightarrows\mathfrak K\) is locally bounded and has a closed graph. 
In particular, \(\operatorname{Gr}(\mathfrak A_G)\) is Borel.
\end{enumerate}
\end{lemma}

\begin{proof}
Finite-dimensional support-function duality gives
\eqref{eq:support_function_representation} and nonempty compact convex
values; see \cite[Propositions~1 and~2]{kusuoka2017supermartingale}.
If $V=(c,b,a)\in\mathfrak A_G(x)$, then
$c=\ell_V(1,0,0)\le G(x,1,0,0)\le0$ by \ref{item:G3}.
For $H\in\mathbb S^+(d)$, condition~\ref{item:G2} gives
$-\tfrac12\operatorname{tr}(aH)=\ell_V(0,0,-H)\le G(x,0,0,-H)\le0$;
hence $a\in\mathbb S^+(d)$ and $V\in\mathfrak K$.
For a compact $K\Subset D$, the jet Lipschitz bound implies
\[
 \|V\|=\sup_{\|U\|\le1}|\ell_V(U)|\le L_K,
 \qquad x\in K,\quad V\in\mathfrak A_G(x),
\]
which proves local boundedness. If $(x_j,V_j)\to(x,V)$ with
$V_j\in\mathfrak A_G(x_j)$, joint upper semicontinuity yields, for every
$U\in\mathfrak J$,
\[
 \ell_V(U)=\lim_j\ell_{V_j}(U)
 \le\limsup_jG(x_j,U)\le G(x,U).
\]
Thus $V\in\mathfrak A_G(x)$, proving that the graph is closed and Borel.
\end{proof}

Now we impose the following additional Lyapunov-type condition on \(G\).
This condition controls live exits from large domains and yields tightness of the corresponding families of laws.
Conditions of this type are standard in the martingale-problem literature and provide control of the process near the boundary of the state space.

\begin{assume}
\label{assume:lyapunov}
There exist $\phi\in C^2(D)$ with $\phi\ge1$ and a constant $C_\phi\ge0$
such that
\[
    m_n^\phi:=\inf_{x\in D\setminus D_n}\phi(x)\longrightarrow\infty
\]
and
\[
    G\bigl(x,\phi(x),\nabla\phi(x),\nabla^2\phi(x)\bigr)
    \le C_\phi\phi(x),
    \qquad x\in D.
\]
\end{assume}

Throughout the paper, we restrict attention to local specifications satisfying Assumption~\ref{assume:lyapunov} and to semigroups whose associated local specifications satisfy this assumption.
We emphasize that this is the only technical condition required for our main results below.
Accordingly, we define
\[
    \mathfrak G_{\mathrm{Lyap}}
    :=\{G\in\mathfrak G:
        G\text{ satisfies Assumption~\ref{assume:lyapunov}}\}
\]
and
\begin{equation}
\label{eq:local_lyapunov_valuation_class}
    \mathfrak V_{\mathrm{loc,Lyap}}
    :=\{\mathcal T\in\mathfrak V_{\mathrm{loc}}:
        G_{\mathcal T}\in\mathfrak G_{\mathrm{Lyap}}\}.
\end{equation}
The class $\mathfrak G_{\mathrm{Lyap}}$ is downward closed in $\mathfrak G$: if $G\in\mathfrak G_{\mathrm{Lyap}}$ and $H\in\mathfrak G$ satisfy $H\le G$, then $H$ satisfies the same Lyapunov bound with the same function $\phi$ and constant $C_\phi$.

\subsection{Canonical paths and stable structures of laws}
\label{subsec:canonical_path_spaces}

We first construct the canonical path space on which the laws in these structures are defined.
Let $\widehat D:=D\cup\{\triangle\}$ be the one-point compactification with a compatible metric $d_{\widehat D}$, and let $\Omega^{\mathrm{cad}}:=D([0,\infty),\widehat D)$ have the Skorokhod $J_1$ topology. 
For $\omega\in\Omega^{\mathrm{cad}}$, put $\tau_\infty(\omega):=\inf\{t\ge0:\omega(t)=\triangle\}$ and define
\[
 \widetilde\Omega:=\left\{\omega\in\Omega^{\mathrm{cad}}:
 \begin{array}{l}
   \omega\text{ is continuous on }[0,\tau_\infty),\\
   \omega(s)=\triangle\text{ for }s\ge\tau_\infty
 \end{array}\right\}.
\]
This is a standard Borel subspace.
Let $X_s(\omega):=\omega(s)$ denote the coordinate process, and let $\widetilde{\mathcal F}$ be the Borel $\sigma$-algebra of $\widetilde\Omega$.
Define the raw filtration $\widetilde{\mathbb F}:=(\widetilde{\mathcal F}_s)_{s\ge0}$ by $\widetilde{\mathcal F}_s:=\sigma(X_r:0\le r\le s)$.
A finite lifetime is a killing time when the path jumps to $\triangle$, and an explosion time when it reaches $\triangle$ continuously. 
Precisely,
\[
 \tau_{\mathrm{kill}}(\omega):=
 \begin{cases}\tau_\infty(\omega),&\omega\text{ has a jump to }\triangle,\\
 \infty,&\text{otherwise},\end{cases}
 \qquad
 \tau_{\mathrm{exp}}(\omega):=
 \begin{cases}\tau_\infty(\omega),&\omega\text{ is continuous},\\
 \infty,&\text{otherwise}.\end{cases}
\]
Thus $\tau_\infty=\tau_{\mathrm{kill}}\wedge\tau_{\mathrm{exp}}$.
Set
\[
 \mathfrak M:=\{\mathbb P\in\mathcal P(\Omega^{\mathrm{cad}}):
                       \mathbb P(\widetilde\Omega)=1\},
\]
with the relative weak topology. 
Denote the usual augmentation under $\mathbb P\in\mathfrak M$ by $\widetilde{\mathbb F}^{\,\mathbb P}$.
For $O\Subset D$ open, use
\[
 \rho_O^{\,t}:=\inf\{s\ge t:X_s\notin O\},
 \qquad \tau_n:=\rho_{D_n}^{\,0},
 \qquad [s]_{t,\rho}:=(s\wedge\rho)\vee t,
\]
with $\inf\varnothing:=\infty$.
For $x\in\widehat D$, let $\mathbf x$ be the constant path at $x$, and let $(\theta_t\omega)(s):=\omega((s-t)^+)$ be the delay map.
For a function $f$ on $D$, the product
$f(X_s)\mathbb I_{\{s<\tau_\infty\}}$ is understood to be zero
when $s\ge\tau_\infty$.

\subsubsection{Stable structures}
For a horizon \(T>0\), an indexed collection of possibly empty subsets
\(\mathcal P=\{\mathcal P_{t,x}\}_{(t,x)\in[0,T]\times\widehat D}\)
of \(\mathfrak M\) is called a \emph{structure of laws} if
\(\mathcal P_{t,\triangle}=\{\delta_{\mathbf\triangle}\}\) and every
\(\mathbb P\in\mathcal P_{t,x}\), \(x\in D\), satisfies
\begin{equation}
\label{def:rdvus_fibers}
    \mathbb P(X_s=x\text{ for every }s\in[0,t])=1.
\end{equation}
Each fiber \(\mathcal P_{t,x}\) is a family of probability laws;
the indexed collection \(\mathcal P\) is the structure.
It is time-homogeneous if
\(\mathcal P_{t,x}=\mathcal P_{0,x}\circ\theta_t^{-1}\).
It is called \emph{fiberwise nonempty} if
\(\mathcal P_{t,x}\neq\varnothing\) for every
\((t,x)\in[0,T]\times\widehat D\).
For a structure indexed by all \((t,x)\in[0,\infty)\times\widehat D\), the same
term means that every fiber on the full time interval is nonempty.

We now impose a natural condition on structures of laws: closure under conditioning and concatenation.
This condition, called \emph{dynamic stability}, ensures that the associated upper expectations satisfy the dynamic programming principle, thereby linking structures of laws to the time-consistency axiom~\ref{V5}.
To formulate this notion precisely, we first introduce a restart map that enforces the constant-past convention in \eqref{def:rdvus_fibers}.
For paths with $\omega(t)=\eta(t)$, define
$(\omega\otimes_t\eta)(s):=\omega(s)$ for $s\le t$ and $\eta(s)$
for $s>t$. Let
\[
 \mathsf R_{t,x}(\eta):=
 \begin{cases}\mathbf x\otimes_t\eta,&\eta(t)=x,\\
 \mathbf x,&\eta(t)\ne x,\end{cases}
 \qquad
 \mathsf R_{\tau,\omega}:=\mathsf R_{\tau(\omega),X_\tau(\omega)}.
\]
Using this restart map, we define two fundamental operations---restarted conditioning and restarted concatenation---and introduce the corresponding notion of stability for structures of laws.

\begin{definition}
\label{def:conditioning_concatenation_virtual_space}
Let $\mathbb P\in\mathfrak M$ and let $\tau$ be a finite
$\widetilde{\mathbb F}$-stopping time. If
$\overline{\mathbb P}^{\tau,\omega}$ is a Borel regular conditional
law given $\widetilde{\mathcal F}_\tau$, its \emph{restarted conditional
law} is
\begin{equation}
\label{eq:restarted_conditional_virtual_law}
 \mathbb P^{\tau,\omega}:=
 \overline{\mathbb P}^{\tau,\omega}\circ\mathsf R_{\tau,\omega}^{-1}.
\end{equation}
A $\widetilde{\mathcal F}_\tau$-measurable kernel
$\nu:\widetilde\Omega\to\mathfrak M$ is \emph{restart-compatible} if
$\nu(\omega)(X_s=X_\tau(\omega)\text{ for all }s\le\tau(\omega))=1$
for every $\omega$. Its \emph{restarted concatenation} with $\mathbb P$ is
\begin{equation}
\label{eq:concatenated_virtual_law}
 (\mathbb P\otimes_\tau\nu)(B)
 :=\iint\mathbb I_B(\omega\otimes_{\tau(\omega)}\eta)
                   \,\nu(\omega,d\eta)\,\mathbb P(d\omega),
 \qquad B\in\widetilde{\mathcal F}.
\end{equation}
\end{definition}

\begin{definition}
\label{def:stable_virtual_uncertainty_structure}
\label{def:regular_dynamic_virtual_uncertainty_structure}
Fix a horizon \(T>0\), and let
\(\mathcal P=\{\mathcal P_{t,x}\}_{(t,x)\in[0,T]\times\widehat D}\)
be a structure of laws on \([0,T]\).

\begin{enumerate}[label=(\roman*), ref=(\roman*)]
\item\label{def:rdvus_regularity}
It is \emph{topologically stable} if every fiber
\(\mathcal P_{t,x}\) is convex, the graph of
\((t,x)\mapsto\mathcal P_{t,x}\) is weakly closed on \([0,T]\times D\), and,
for every compact \(K\subset[0,T]\times D\),
\[
    \mathcal P_K
    :=
    \bigcup_{(t,x)\in K}\mathcal P_{t,x}
\]
is weakly compact.

\item\label{def:rdvus_conditioning}
It is \emph{stable under restarted conditioning} if, whenever
\(\mathbb P\in\mathcal P_{t,x}\) and \(s\in[t,T]\), one has
\(\mathbb P^{s,\widetilde\omega}\in
\mathcal P_{s,X_s(\widetilde\omega)}\) for
\(\mathbb P\)-almost every \(\widetilde\omega\).
It is \emph{strongly stable under restarted conditioning} if the same
property holds for every finite \(\widetilde{\mathbb F}\)-stopping time
\(\widetilde\tau\) satisfying \(t\le\widetilde\tau\le T\), with \(s\)
replaced by \(\widetilde\tau\).

\item\label{def:rdvus_pasting}
It is \emph{stable under restarted concatenation} if, whenever
\(\mathbb P\in\mathcal P_{t,x}\), \(s\in[t,T]\), and
\(\widetilde\nu:\widetilde\Omega\to\mathfrak M\) is an
\(\widetilde{\mathcal F}_s\)-measurable, \(s\)-restart-compatible kernel
such that
\(\widetilde\nu(\widetilde\omega)\in
\mathcal P_{s,X_s(\widetilde\omega)}\) for
\(\mathbb P\)-almost every \(\widetilde\omega\), one has
\(\mathbb P\otimes_s\widetilde\nu\in\mathcal P_{t,x}\).
It is \emph{strongly stable under restarted concatenation} if the same
property holds for every finite \(\widetilde{\mathbb F}\)-stopping time
\(\widetilde\tau\) satisfying \(t\le\widetilde\tau\le T\), with all
corresponding occurrences of \(s\) replaced by \(\widetilde\tau\).

\item\label{def:rdvus_stable} The structure \(\mathcal P\) is called \emph{dynamically stable} if it is
stable under restarted conditioning and restarted concatenation, and
\emph{strongly dynamically stable} if it is strongly stable under both
operations. It is called \emph{stable} if it is both topologically and
dynamically stable, and \emph{strongly stable} if it is both topologically
and strongly dynamically stable.
\end{enumerate}
A structure indexed by all \((t,x)\in[0,\infty)\times\widehat D\) is said to have any of the preceding stability properties if its restriction to every finite horizon
has that property.
\end{definition}

The two notions of stability play distinct roles.
Topological stability provides the closed-graph and compactness properties needed to pass to weak limits and apply measurable selection.
Dynamic stability ensures time consistency, as discussed above.

\subsubsection{The \texorpdfstring{$G$}{G}-supermartingale problem}
The following $G$-supermartingale problem defines the structure of laws associated with a given local specification $G$ without specifying local characteristics of the coordinate process $X$.

\begin{definition}
\label{def:G_supermartingale_problem}
Let $G\in\mathfrak G$, and fix
\((t,\widetilde\omega)\in[0,\infty)\times\widetilde\Omega\).
We say that \(\mathbb P\in\mathfrak M\) solves the \emph{generalized \(G\)-supermartingale problem on \(\widetilde\Omega\) starting from \((t,\widetilde\omega)\)} if it satisfies the prescribed-history condition $\mathbb P(X_s=\widetilde\omega(s)\text{ for }0\le s\le t)=1$ and, for every \(f\in C_b^\infty(D)\) and \(n\ge1\), the process
\begin{equation}
\label{eq:defM_f,nonline}
\widetilde M_s^{f,n}
:=
f\bigl(X_{[s]_{t,\tau_n}}\bigr)
\mathbb I_{\{[s]_{t,\tau_n}<\tau_\infty\}}
-
\int_t^{[s]_{t,\tau_n}}
G\bigl(
    X_r,
    J_{X_r}^2f
\bigr)
\mathbb I_{\{r<\tau_\infty\}}\,dr,
\qquad s\ge t,
\end{equation}
is a \(\mathbb P\)-supermartingale with respect to \(\widetilde{\mathbb F}\).
We denote the corresponding solution set by $\mathcal P_{t,\widetilde\omega}(G)$.
For constant initial histories write $\mathcal P_{t,x}(G):=\mathcal P_{t,\mathbf x}(G)$ and denote the resulting structure by $\mathcal P(G)$.
We abbreviate $\mathcal P_x(G):=\mathcal P_{0,x}(G)$.
\end{definition}

For $G\in\mathfrak G_{\mathrm{Lyap}}$, define
\begin{equation}
\label{eq:generator_compatible_structure_class}
\mathfrak S(G)
:=\left\{\mathcal C\subseteq\mathcal P(G):
\begin{array}{l}
\mathcal C\text{ is a fiberwise nonempty, time-homogeneous}\\
\text{stable structure of laws}
\end{array}\right\}.
\end{equation}
Here and throughout, inclusion between structures is understood fiberwise: for two structures $\mathcal P$ and $\mathcal Q$ with the same index set $I\times\widehat D$, we write
\[
    \mathcal P\subseteq\mathcal Q
    \quad\Longleftrightarrow\quad
    \mathcal P_{t,x}\subseteq\mathcal Q_{t,x}
    \quad\text{for every }(t,x)\in I\times\widehat D,
\]
where $I=[0,T]$ or $I=[0,\infty)$, as appropriate.
Finally, the class of stable structures of laws under a local Lyapunov constraint is denoted by
\begin{equation}
\label{eq:local_lyapunov_structure_class}
\mathfrak S_{\mathrm{loc,Lyap}}
:=\bigcup_{G\in\mathfrak G_{\mathrm{Lyap}}}\mathfrak S(G).
\end{equation}

The three classes $\mathfrak V_{\mathrm{loc,Lyap}},\mathfrak S_{\mathrm{loc,Lyap}},\mathfrak G_{\mathrm{Lyap}}$ are ordered by pointwise semigroup inequalities,
fiberwise inclusion of structures, and pointwise inequalities of local
specifications, respectively.

\section{Semigroup-Structure Correspondence}
\label{sec:valuation_orbit_nonemptiness}
\label{sec:stable_recovery}
\label{sec:stable_representation}

We establish a representation of globally dominated semigroups
by stable substructures of $\mathcal P(G)$.
We first construct the admissible laws, then recover stable structures
from their upper expectations, and finally obtain the representation
theorem by proving fiberwise nonemptiness and stability of the recovered structures.

\subsection{Admissible laws}
\label{subsec:canonical_local_law_properties}
We first establish the local characteristics, stability, and nonemptiness of $\mathcal P(G)$.

\subsubsection{Local characteristics}
The support correspondence gives a coefficient description of the laws
in $\mathcal P(G)$. We formulate this description directly on the
canonical path space.

\begin{definition}
\label{def:virtual_coefficient_fields}
A map
$\beta=(c,b,a):[0,\infty)\times\widetilde\Omega\to\mathfrak K$
is a \emph{coefficient field} if it is progressively measurable with
respect to the raw filtration $\widetilde{\mathbb F}$, satisfies
\[
 \beta(s,\widetilde\omega)=0
 \qquad\text{for }s\ge\tau_\infty(\widetilde\omega),
\]
and is locally bounded: for every $T>0$ and $n\ge1$,
\[
 \sup\bigl\{
 |c(s,\widetilde\omega)|+|b(s,\widetilde\omega)|+\|a(s,\widetilde\omega)\|:
 0\le s\le T,\ \widetilde\omega\in\widetilde\Omega,
 X_s(\widetilde\omega)\in\overline D_n
 \bigr\}<\infty.
\]
For such a field, set $k^\beta:=-c\ge0$ and
\[
 L^\beta(s,\widetilde\omega,U):=\ell_{\beta(s,\widetilde\omega)}(U),
 \qquad U\in\mathfrak J.
\]
\end{definition}
For a path in $\widetilde\Omega$, define its continuous coordinate by freezing its
last live state after killing:
\begin{equation}
\label{eq:continuous_coordinate_virtual_path}
 \overline X_s(\omega):=
 \begin{cases}
 X_s(\omega),&s<\tau_{\mathrm{kill}}(\omega),\\
 X_{\tau_{\mathrm{kill}}(\omega)-}(\omega),
     &s\ge\tau_{\mathrm{kill}}(\omega).
 \end{cases}
\end{equation}
When $\tau_{\mathrm{kill}}=\infty$, we set $\overline X=X$.
The process $\overline X$ is continuous with values in $\widehat D$ and
is $D$-valued before continuous explosion. Write
$N_s:=\mathbb I_{\{\tau_{\mathrm{kill}}\le s\}}$ for the killing
indicator.
We also write $N_s^t:=N_s-N_t$.
For $O\Subset D$ open, define
$\widehat\tau_O^{\,t}:=\inf\{s\ge t:\overline X_s\notin O\}$ and
$\widehat\tau_n:=\widehat\tau_{D_n}^{\,0}$; thus
$\rho_O^{\,t}=\widehat\tau_O^{\,t}\wedge\tau_{\mathrm{kill}}$
on paths alive at time $t$.

\begin{definition}
\label{def:generalized_linear_martingale_problems}
Let $\beta$ be a coefficient field and fix $t\ge0$.
Let $\widetilde\omega\in\widetilde\Omega$. We say a probability measure
$\mathbb P\in\mathfrak M$ is a solution to  the
\emph{generalized $L^\beta$-martingale problem on $\widetilde\Omega$
starting from $(t,\widetilde\omega)$} if
\begin{align}\label{eq:initialconditionP}
\mathbb P\bigl(
X_s=\widetilde\omega(s)\ \text{for all }0\leq s\leq t
\bigr)=1
\end{align}
and, for every $f\in C_b^\infty(D)$ and $n\geq1$, the process
\begingroup
\mathtoolsset{showonlyrefs=false}
\begin{equation}
\label{eq:defM_f,nlinearmtg}
\widetilde M_s^{f,n,\beta}
:=
f\bigl(X_{[s]_{t,\tau_n}}\bigr)
\mathbb I_{\{[s]_{t,\tau_n}<\tau_\infty\}}
-
\int_t^{[s]_{t,\tau_n}}
L^\beta\bigl(r,X,J_{X_r}^2f
\bigr)
\mathbb I_{\{r<\tau_\infty\}}\,dr,
\quad s\geq t,
\end{equation}
\endgroup
is a $\mathbb P$-martingale with respect to $\widetilde{\mathbb F}$.
We denote the corresponding solution set by $\widetilde{\mathcal P}_{t,\widetilde\omega}(L^\beta)$.
\end{definition}
For either a state or a history $\xi$, we abbreviate
$\widetilde{\mathcal P}_\xi(L^\beta)
:=\widetilde{\mathcal P}_{0,\xi}(L^\beta)$ and use a state subscript for
the corresponding constant history.

\begin{definition}
\label{def:generator_compatible_coefficient_fields}
Let $\beta=(c,b,a)$ be a coefficient field and let $G\in\mathfrak G$.
The coefficient field $\beta$ is \emph{generator-compatible with $G$} if
\begin{equation}
\beta(t,\omega)\in \mathfrak A_G\bigl(X_t(\omega)\bigr)
\qquad
\text{for }t<\tau_\infty(\omega).
\end{equation}
We denote the class of all such coefficient fields by $\mathfrak B_{\mathrm{gc}}(G)$.
\end{definition}\noindent Equivalently, a coefficient field \(\beta\) is generator-compatible if and only if
\begin{align}\label{eq:coefficient_generator_compatibility}
    L^\beta(t,\omega,U)
    \le
    G(\omega(t),U)
      \text{ for all }(t,\omega,U)\text{ with }t<\tau_\infty(\omega).
\end{align}

The following proposition recovers the killing intensity, drift, and
covariance from each admissible law.
Its proof, including the recovery and aggregation argument, is provided
in Appendix~\ref{app:characteristic_recovery}.

\begin{proposition}
\label{prop:recovery_aggregation_generator_compatible_characteristics}
Let $G\in\mathfrak G$.
Then every $\mathbb P\in\mathcal P_{t,\widetilde\omega}(G)$ with
$t<\tau_\infty(\widetilde\omega)$ admits a coefficient field
$\beta^{\mathbb P}=(c^{\mathbb P},b^{\mathbb P},a^{\mathbb P})
\in\mathfrak B_{\mathrm{gc}}(G)$ such that
$\mathbb P\in\widetilde{\mathcal P}_{t,\widetilde\omega}
(L^{\beta^{\mathbb P}})$ and the following properties hold.
\begin{enumerate}[label=(\roman*), ref=(\roman*)]
\item\label{prop:recovery_killing_compensator}
The predictable compensator of
$N$ with respect to $\widetilde{\mathbb F}^{\,\mathbb P}$ is
\[
 A_s^{\mathbb P}
 :=\int_t^{s\vee t}k^{\beta^{\mathbb P}}(r,X)\mathbb I_{\{r<\tau_\infty\}}\,dr,
 \qquad s\ge0.
\]
Thus $N-A^{\mathbb P}$ is a $\mathbb P$-local martingale, and
$k^{\beta^{\mathbb P}}=-c^{\mathbb P}$ is the killing intensity before
the lifetime. In particular,
$\mathbb P(\tau_{\mathrm{kill}}=a)=0$ for every deterministic $a\ge0$.

\item\label{prop:recovery_continuous_characteristics}
For every $n$ such that $\widetilde\omega([0,t])\subset D_n$, there is a continuous
$\widetilde{\mathbb F}^{\,\mathbb P}$-local martingale
$M^{\mathbb P,n}$, null on $[0,t]$, such that, for $s\ge t$,
\[
\begin{aligned}
 \overline X_{s\wedge\tau_n}
 &=\widetilde\omega(t)
   +\int_t^{s\wedge\tau_n}b^{\mathbb P}(r,X)\,dr
   +M_s^{\mathbb P,n},\\
 \bigl\langle M^{\mathbb P,n}\bigr\rangle_s
 &=\int_t^{s\wedge\tau_n}a^{\mathbb P}(r,X)\,dr.
\end{aligned}
\]
Hence $b^{\mathbb P}$ and $a^{\mathbb P}$ are the local drift and
covariance densities of $\overline X$ after time $t$.
Together with \ref{prop:recovery_killing_compensator}, these identities
determine the recovered coefficients $dr\otimes d\mathbb P$-a.e. on
$\{(r,\widetilde\eta):t<r<\tau_\infty(\widetilde\eta)\}$.

\end{enumerate}
\end{proposition}

The recovered coefficient field yields the following equivalence between the coefficient-free $G$-supermartingale problem and the generator-compatible linear martingale problems.

\begin{corollary}
\label{cor:recovery_martingale_problem_equivalence}
Let $G\in\mathfrak G$. Then, for every initial history $(t,\widetilde\omega)\in[0,\infty)\times\widetilde\Omega$,
\begin{equation}
\label{eq:virtual_model_classes_from_fields}
\mathcal P_{t,\widetilde\omega}(G)
=\bigcup_{\beta\in\mathfrak B_{\mathrm{gc}}(G)}
  \widetilde{\mathcal P}_{t,\widetilde\omega}(L^\beta).
\end{equation}
\end{corollary}
\begin{proof}
Fix $(t,\widetilde\omega)$.  If
$\mathbb P\in\widetilde{\mathcal P}_{t,\widetilde\omega}(L^\beta)$
for $\beta\in\mathfrak B_{\mathrm{gc}}(G)$, then
\begin{equation}
\label{eq:linear_to_nonlinear_supermartingale_identity}
\widetilde M_s^{f,n}
=\widetilde M_s^{f,n,\beta}
-\int_t^{[s]_{t,\tau_n}}
 [G(X_r,J_{X_r}^2f)-L^\beta(r,X,J_{X_r}^2f)]\mathbb I_{\{r<\tau_\infty\}}\,dr.
\end{equation}
The integral is nondecreasing, so $\widetilde M^{f,n}$ is a $\mathbb P$-supermartingale.  
This proves the inclusion $\supseteq$ in \eqref{eq:virtual_model_classes_from_fields}.
Conversely, when $t<\tau_\infty(\widetilde\omega)$, Proposition~\ref{prop:recovery_aggregation_generator_compatible_characteristics} provides a field $\beta^{\mathbb P}$ such that $\mathbb P\in\widetilde{\mathcal P}_{t,\widetilde\omega}(L^{\beta^{\mathbb P}})$ for every $\mathbb P\in\mathcal P_{t,\widetilde\omega}(G)$.
When $t\ge\tau_\infty(\widetilde\omega)$, the prescribed history and
absorption give $\mathbb P=\delta_{\widetilde\omega}$.  Choose any
$\beta^0\in\mathfrak B_{\mathrm{gc}}(G)$, whose existence follows from
a Borel selector of $\mathfrak A_G$.  Each process
$\widetilde M^{f,n,\beta^0}$ in \eqref{eq:defM_f,nlinearmtg} is
identically zero on $[t,\infty)$, so
$\delta_{\widetilde\omega}\in\widetilde{\mathcal P}_{t,\widetilde\omega}(L^{\beta^0})$.
This proves the reverse inclusion.
\end{proof}

The following lemma shows that the Lyapunov condition excludes continuous explosion. 

\begin{lemma}
\label{lem:admissible_nonexplosion}
Let $G\in\mathfrak G_{\mathrm{Lyap}}$ and
$\mathbb P\in\mathcal P_{t,\widetilde\omega}(G)$ with
$t<\tau_\infty(\widetilde\omega)$. Then
$\mathbb P(\tau_{\mathrm{exp}}=\infty)=1$.
\end{lemma}

\begin{proof}
Choose $\beta^{\mathbb P}$ as in
Proposition~\ref{prop:recovery_aggregation_generator_compatible_characteristics}.
Then since $\beta^{\mathbb P}\in\mathfrak B_{\mathrm{gc}}(G)$, we have
\begin{align}
    L^{\beta^{\mathbb P}}(r,X,J_{X_r}^2\phi)\mathbb I_{\{\tau_\infty>r\}}\le C_\phi\phi(X_r)\mathbb I_{\{\tau_\infty>r\}}.
\end{align}
The Lyapunov estimate in Proposition~\ref{prop:lyap_tail_bounds} of the Supplementary Material therefore excludes finite continuous explosion.
\end{proof}

\subsubsection{Time homogeneity and strong stability}

We now establish the time homogeneity and strong stability of $\mathcal P(G)$.

\begin{proposition}
\label{prop:propertiesP_t,x_revised}
For every $G\in\mathfrak G_{\mathrm{Lyap}}$, the structure $\mathcal P(G)$ is time-homogeneous and strongly stable.
\end{proposition}

\begin{proof}
We first prove that $\mathcal P(G)$ is topologically stable.
Assume that $G\in\mathfrak G_{\mathrm{Lyap}}$, and fix $\phi$ and $C_\phi$ as in Assumption~\ref{assume:lyapunov}.
Each fiber $\mathcal P_{t,x}(G)$ is convex because the prescribed-history
condition and the inequalities
$\mathbb E^{\mathbb P}[F(\widetilde M_s^{f,n}-\widetilde M_r^{f,n})]\le0$
are linear in $\mathbb P$, for every $f\in C_b^\infty(D)$, $n\ge1$,
$t\le r\le s$, and bounded nonnegative
$\widetilde{\mathcal F}_r$-measurable $F$, with
$\widetilde M^{f,n}$ defined in \eqref{eq:defM_f,nonline}.
Now we prove that the correspondence $(t,x)\mapsto\mathcal P_{t,x}(G)$ has a closed graph.
Let
\[
    (t_j,x_j)
    \longrightarrow
    (t,x)
    \quad\text{in }[0,\infty)\times D,
    \qquad
    \mathbb P_j\Rightarrow\mathbb P,
    \qquad
    \mathbb P_j\in
    \mathcal P_{t_j,x_j}(G)
    \quad\text{for every }j\geq1.
\]
By Proposition~\ref{prop:recovery_aggregation_generator_compatible_characteristics},
for every \(j\ge1\), choose
\(\beta^j=(c^j,b^j,a^j)\in\mathfrak B_{\mathrm{gc}}(G)\) such that
\(\mathbb P_j\in
\widetilde{\mathcal P}_{t_j,x_j}(L^{\beta^j})\), and set
\[
    \kappa_r^j
    :=
    k^{\beta^j}(r,X)\mathbb I_{\{r>t_j\}},
    \qquad r\ge0.
\]
To prove that $\widetilde M^{f,n}$ is a $\mathbb P$-supermartingale
for every $f\in C_b^\infty(D)$ and $n\ge1$, fix such $f$ and $n$
and note the following.
\begin{itemize}
    \item
    Generator compatibility and Assumption~\ref{assume:lyapunov} give
    \[
        L^{\beta^j}\bigl(r,X,J_{X_r}^2\phi\bigr)
        \le G\bigl(X_r,J_{X_r}^2\phi\bigr)
        \le C_\phi\phi(X_r)
    \]
    before the lifetime $\tau_\infty$.

    \item
    By Proposition~\ref{prop:recovery_aggregation_generator_compatible_characteristics}%
    \ref{prop:recovery_killing_compensator}, the process
    \[
        N_s-\int_0^s\kappa_r^j\mathbb I_{\{r<\tau_\infty\}}\,dr,
        \qquad s\ge0,
    \]
    is a $\mathbb P_j$-local martingale.

    \item
    The local boundedness of $\mathfrak A_G$ and generator compatibility
    imply that, for every $H>0$ and compact $K\Subset D$,
    \[
        \sup_{j\ge1}
        \left\|
            \kappa^j\mathbb I_{\{\cdot<\tau_\infty\}}\mathbb I_{\{X\in K\}}
        \right\|_{L^\infty(
            [0,H]\times\widetilde\Omega,
            dr\otimes d\mathbb P_j
        )}
        <\infty.
    \]
\end{itemize}
Proposition~\ref{prop:weak_passage_martingale_relations} in the Supplementary Material, applied with the preceding $\kappa^j$ and
\[
    q_j(s,x)=q(s,x):=G(x,J_x^2f),
    \qquad (s,x)\in[0,\infty)\times D,
\]
therefore gives both
$\mathbb P(X_s=x\text{ for all }0\le s\le t)=1$
and the $\mathbb P$-supermartingale property of $\widetilde M^{f,n}$.
Thus $\mathbb P\in\mathcal P_{t,x}(G)$, proving that the graph of
$(t,x)\mapsto\mathcal P_{t,x}(G)$ is closed on $[0,\infty)\times D$.

Now let $K\subset[0,\infty)\times D$ be compact, and write
\[
 \mathcal P_K(G):=\bigcup_{(t,x)\in K}\mathcal P_{t,x}(G).
\]
Given a sequence $(\mathbb P_j)_{j\ge1}$ in $\mathcal P_K(G)$, choose $(t_j,x_j)\in K$ with $\mathbb P_j\in\mathcal P_{t_j,x_j}(G)$.
Then the precompactness of initial condition $\{(t_j,x_j)\}_{j\ge1}$ and Lyapunov bound yield the tightness of $(\mathbb P_j)_{j\ge1}$ (see Proposition~\ref{prop:compactness_criterion} in the Supplementary Material). 
Prokhorov's theorem and the closed-graph property proved above yield a weakly convergent subsequence whose limit belongs to $\mathcal P_{t,x}(G)\subseteq\mathcal P_K(G)$.
Hence $\mathcal P_K(G)$ is weakly compact. 

Time homogeneity follows directly from the time independence of $G$.
For strong dynamic stability, we give the localization details in the
classical conditioning and pasting arguments; compare
\cite[Lemma~12.2.1 and Theorems~6.1.2--6.1.3 and~6.2.1]{stroock1997multidimensional}.
Let $\mathbb P\in\mathcal P_{t,x}(G)$ and let $\tau$ be a finite raw
stopping time with $t\le\tau\le T$. Set
\[
 \sigma_n:=\inf\{s\ge\tau:X_s\notin D_n\},
 \qquad B_m:=\{\tau<\tau_m\}.
\]
On $B_m\cap\{X_\tau\in D_n\}$, with $m\ge n$, one has
$\sigma_n\le\tau_m$.
Thus the increments after $\tau$ of $\widetilde M^{f,m}$ stopped at
$\sigma_n$ coincide with the $n$-localized increments after restarting.
Choose regular conditional laws given $\widetilde{\mathcal F}_\tau$
that preserve the stopped history outside one null set, as permitted by
Galmarino's test on the canonical path space.
Apply optional sampling at rational elapsed times $\tau+u$ and
$\tau+v$, with $0\le u\le v$, to these stopped increments and
condition on $\widetilde{\mathcal F}_\tau$.
Use $f\in\mathscr F_{\mathrm{sm}}$ from Proposition~\ref{lem:countable_smooth_determining_class} in the Supplementary Material and multipliers from a fixed countable
algebra of coordinate cylinders at rational elapsed times.
Together with the countable indices $m,n$, these tests give a single
exceptional null set. Since $B_m\uparrow\{\tau<\tau_\infty\}$,
the resulting inequalities hold under the restarted conditional law
whenever $X_\tau$ is alive.
If $X_\tau\notin D_n$, the corresponding restarted localized process
is constant; if $X_\tau=\triangle$, the restarted conditional law is
$\delta_{\mathbf\triangle}$.
The monotone class theorem, right continuity of the stopped processes,
and the local $C^2$ approximation in the same proposition extend these
inequalities to all bounded nonnegative multipliers, all times, and all
$f\in C_b^\infty(D)$.
This proves strong stability under restarted conditioning.

For restarted concatenation, let $\nu$ be a restart-compatible kernel
with
$\nu(\omega)\in{\mathcal P}_{\tau(\omega),X_\tau(\omega)}(G)$
for $\mathbb P$-almost every $\omega$.
The concatenated law agrees with $\mathbb P$ up to $\tau$.
For each $n$, the original $n$-localized process is already constant
after $\tau$ on $\{\tau_n\le\tau\}$.
On $\{\tau<\tau_n\}$, its increments after $\tau$ coincide with the
$n$-localized increments under the continuation law $\nu(\omega)$.
Splitting each increment at $\tau$, optional sampling before $\tau$
and the supermartingale inequalities under $\nu(\omega)$ after $\tau$,
followed by the tower property, give the required inequalities under
$\mathbb P\otimes_\tau\nu$.
All these stopped processes are bounded on finite horizons, so the
conditional expectations and integrations are justified.
This proves strong stability under restarted concatenation.
\end{proof}

\subsubsection{Fiberwise nonemptiness}

To establish the fiberwise nonemptiness of $\mathcal P(G)$, we use the viscosity solution notions associated with the nonlinear parabolic equation
\begin{equation}
\label{eq:mainHJB}
    \partial_t u(t,x)
    =
    G\bigl(x,J_x^2u(t,\cdot)\bigr),
    \qquad (t,x)\in(0,\infty)\times D,
\end{equation}
where $G\in\mathfrak G$.
For later use, we define subsolutions for the more general equation
\begin{equation}
\label{eq:usc_subsolution_with_source}
    \partial_tu-G\bigl(x,J_x^2u(t,\cdot)\bigr)-q(t,x)=0,
    \qquad (t,x)\in(0,\infty)\times D,
\end{equation}
where $q\in\USC_b([0,\infty)\times D)$.
Since $G$, $q$, and the semigroup orbit may be discontinuous, we use
semicontinuous envelopes. Regarding $G$ as a function on
$D\times\mathfrak J$, we write $G^*$ and $G_*$ for its upper and lower
semicontinuous envelopes, respectively. In the present setting, $G^*=G$.

We write $C_b^\infty([0,\infty)\times D)$ for the space of bounded
functions that are smooth on $(0,\infty)\times D$ and whose partial
derivatives of all orders admit bounded continuous extensions to
$[0,\infty)\times D$.
Let $\mathcal O\subset(0,\infty)\times D$ be open.
A function $u\in\USC_b(\mathcal O)$ is a viscosity subsolution of
\eqref{eq:usc_subsolution_with_source} on $\mathcal O$ if, whenever
$(t_0,x_0)\in\mathcal O$ and
$\psi\in C_b^\infty([0,\infty)\times D)$ satisfy
$u(t_0,x_0)=\psi(t_0,x_0)$ and $u-\psi$ attains a global maximum at
$(t_0,x_0)$, one has
\begin{equation}
\label{eq:viscosity_subsolution_inequality}
    \partial_t\psi(t_0,x_0)
    -G\bigl(x_0,J_{x_0}^2\psi(t_0,\cdot)\bigr)
    -q(t_0,x_0)
    \le0.
\end{equation}
Taking $q=0$ gives the notion of a viscosity subsolution of
\eqref{eq:mainHJB}.

A function $u\in\LSC_b(\mathcal O)$ is a viscosity supersolution of
\eqref{eq:mainHJB} on $\mathcal O$ if, whenever
$(t_0,x_0)\in\mathcal O$ and
$\psi\in C_b^\infty([0,\infty)\times D)$ satisfy
$u(t_0,x_0)=\psi(t_0,x_0)$ and $u-\psi$ attains a global minimum at
$(t_0,x_0)$, one has
\begin{equation}
\label{eq:viscosity_supersolution_inequality}
    \partial_t\psi(t_0,x_0)
    -G_*\bigl(x_0,J_{x_0}^2\psi(t_0,\cdot)\bigr)
    \ge0.
\end{equation}
A bounded function $u:\mathcal O\to\mathbb R$ is a viscosity solution
of \eqref{eq:mainHJB} in the discontinuous convention if $u^*$ is a
viscosity subsolution and $u_*$ is a viscosity supersolution of that equation.
We refer to \cite{crandall1992user} for the general viscosity-solution
framework.

For bounded functions, these definitions are equivalent to the usual
definitions using local $C^{1,2}$ test functions.
Indeed, one may make a local contact strict, approximate the test function
smoothly together with its time derivative and spatial derivatives up to
order two, and globalize the resulting tests by a cutoff to a constant.
Passing to the limit at the resulting contact points uses the upper
semicontinuity of $G$ and $q$ for subsolutions and the lower
semicontinuity of $G_*$ for supersolutions.

We also fix the initial-trace conventions, which will be used independently
of the interior inequalities.  For a bounded initial datum
$f:D\to\mathbb R$, a bounded function
$u:(0,\infty)\times D\to\mathbb R$ satisfies the relaxed subsolution
initial condition associated with $f$ if
\begin{equation}
\label{eq:relaxed_subsolution_initial_condition}
    \limsup_{\substack{(s,y)\to(0,x)\\s>0}}u^*(s,y)
    \le f^*(x),\qquad x\in D,
\end{equation}
and the relaxed supersolution initial condition associated with $f$ if
\begin{equation}
\label{eq:relaxed_supersolution_initial_condition}
    \liminf_{\substack{(s,y)\to(0,x)\\s>0}}u_*(s,y)
    \ge f_*(x),\qquad x\in D.
\end{equation}
It has initial datum $f$ in the relaxed sense if it satisfies both
conditions.

For $f\in\USC_b(D)$, let
\begin{equation}
\label{eq:upper_perron_class}
\operatorname{Sub}_G^+(f)
:=
\left\{
u\in\USC_b([0,\infty)\times D):
\begin{aligned}
 &u\text{ is a viscosity subsolution}\\[-0.1em]
 &\text{of \eqref{eq:mainHJB} on }(0,\infty)\times D,\\[-0.1em]
 &u(0,x)\le f(x),\qquad x\in D
\end{aligned}
\right\}.
\end{equation}
Upper semicontinuity on the closed time--space domain ensures that the
boundary inequality in \eqref{eq:upper_perron_class} implies
\eqref{eq:relaxed_subsolution_initial_condition}.

The following stopped subsolution realization theorem yields the terminal realization and fiberwise nonemptiness results below.
Its proof is deferred to Appendix~\ref{sec:stochastic_realization_perron}.

\begin{theorem}
\label{prop:one_step_usc_subsolution_realization}
Let $G\in\mathfrak G_{\mathrm{Lyap}}$ and
$u,q\in\USC_b([0,\infty)\times D)$, where $u$ is a viscosity subsolution of
\eqref{eq:usc_subsolution_with_source} on $(0,\infty)\times D$.
Fix \(T>0\), \(0<h\le T\), a bounded open set \(O\Subset D\), and
\(x\in O\). Then there exists
\(\mathbb P^{T,h,O,x}\in\mathcal P_x(G)\) such that
\begin{equation}
\label{eq:one_step_usc_realization}
    u(T,x)
    \le
    \mathbb E^{\mathbb P^{T,h,O,x}}\!\Bigg[
        u\bigl(T-(h\wedge\rho_O^{\,0}),X_{h\wedge\rho_O^{\,0}}\bigr)
        \mathbb I_{\{h\wedge\rho_O^{\,0}<\tau_\infty\}}
        +\int_0^{h\wedge\rho_O^{\,0}}q(T-r,X_r)\mathbb I_{\{r<\tau_\infty\}}\,dr
    \Bigg].
\end{equation}
\end{theorem}

\begin{remark}
Theorem~\ref{prop:one_step_usc_subsolution_realization} is the technical core of this paper.
It converts a viscosity subsolution inequality into a stopped expectation inequality under an admissible law, allowing upper semicontinuous evolution and degenerate diffusion coefficients, without assuming a comparison principle or uniqueness for the martingale problem.
This generality requires a careful approximation construction.
The proof regularizes the local specification and the subsolution, introduces nondegenerate approximations, and constructs laws satisfying the corresponding approximate realization inequalities.
It then establishes that a weak limit belongs to $\mathcal P_x(G)$ and satisfies the desired stopped expectation inequality.
Appendix~\ref{sec:stochastic_realization_perron} presents the main construction and limiting argument, while the supporting approximation results and estimates are developed in Appendix~F of the Supplementary Material.
\end{remark}

Taking $q\equiv0$ and exhausting the state space give the following terminal realization result.

\begin{corollary}
\label{cor:usc_terminal_realization}
Let $G\in\mathfrak G_{\mathrm{Lyap}}$, let $f\in\USC_b(D)$, and let
$u\in\operatorname{Sub}_G^+(f)$. For every $T>0$ and $x\in D$, there exists
$\mathbb P^{T,x}\in\mathcal P_x(G)$ such that
\begin{equation}
\label{eq:usc_terminal_realization_conclusion}
    u(T,x)\le
    \mathbb E^{\mathbb P^{T,x}}[f(X_T)\mathbb I_{\{T<\tau_\infty\}}].
\end{equation}
\end{corollary}

\begin{proof}
Choose \(n_0\) with \(x\in D_{n_0}\).
For \(n\ge n_0\), apply
Theorem~\ref{prop:one_step_usc_subsolution_realization} with \(q=0\), \(h=T\),
and \(O=D_n\) to obtain \(\mathbb P_n\in\mathcal P_x(G)\).  With
\(C_{u,f}:=\|u\|_\infty+\|f\|_\infty\), the boundary domination
$u(0,\cdot)\le f$ gives
\begin{equation}
\label{eq:cor:usc_terminal_realization}
\begin{aligned}
    u(T,x)
    &\le
    \mathbb E^{\mathbb P_n}\!\left[
        u\bigl(T-T\wedge\tau_n,X_{T\wedge\tau_n}\bigr)
        \mathbb I_{\{T\wedge\tau_n<\tau_\infty\}}
    \right]\\
    &\le
    \mathbb E^{\mathbb P_n}[f(X_T)\mathbb I_{\{T<\tau_\infty\}}]
    +C_{u,f}\,
      \mathbb P_n(\tau_n\le T,\ \tau_n<\tau_\infty).
\end{aligned}
\end{equation}

By weak compactness of \(\mathcal P_x(G)\), after passing to a subsequence, we may assume that \(\mathbb P_n\Rightarrow\mathbb P^{T,x}\in\mathcal P_x(G)\).  
Moreover, Proposition~\ref{prop:lyap_tail_bounds} in the Supplementary Material gives
\begin{align}\label{eq:cor:usc_termial_realization:terminalterm}
    \mathbb P_n(\tau_n\le T,\ \tau_n<\tau_\infty)
    \le
    \frac{e^{C_\phi T}\phi(x)}{m_n^\phi}
    \longrightarrow0.
\end{align}
Passing to the limit as $n\to\infty$ in \eqref{eq:cor:usc_terminal_realization} and using \eqref{eq:cor:usc_termial_realization:terminalterm} together with Proposition~\ref{lem:killed_cylinder_payoff_upper_stability} in the Supplementary Material, applied with $m=1$ and $r_{n,1}=r_1=T$, we obtain
\[
\begin{aligned}
    u(T,x)
    &\le
    \liminf_{n\to\infty}
        \mathbb E^{\mathbb P_n}[f(X_T)\mathbb I_{\{T<\tau_\infty\}}]
    \le
    \limsup_{n\to\infty}
        \mathbb E^{\mathbb P_n}[f(X_T)\mathbb I_{\{T<\tau_\infty\}}]
    \\
    &\le
    \mathbb E^{\mathbb P^{T,x}}[f(X_T)\mathbb I_{\{T<\tau_\infty\}}].
\end{aligned}
\]
Indeed, the applicability of Proposition~\ref{lem:killed_cylinder_payoff_upper_stability} in the Supplementary Material follows from Proposition~\ref{prop:recovery_aggregation_generator_compatible_characteristics}\ref{prop:recovery_killing_compensator} and Lemma~\ref{lem:admissible_nonexplosion}.
This completes the proof.
\end{proof}

\begin{corollary}
\label{cor:canonical_regular_dynamic_virtual_structure}
For $G\in\mathfrak G_{\mathrm{Lyap}}$, the structure $\mathcal P(G)$ is
fiberwise nonempty.
\end{corollary}
\begin{proof}
Since $G(x,0)=0$, the zero function lies in $\operatorname{Sub}_G^+(0)$.
Corollary~\ref{cor:usc_terminal_realization} gives a law in $\mathcal P_x(G)$ for every $x\in D$.
Time homogeneity from Proposition~\ref{prop:propertiesP_t,x_revised}
and the cemetery convention give nonemptiness for every initial parameter.
\end{proof}

\subsection{Upper-expectation representation and stable recovery}
\label{subsec:abstract_stable_recovery}
For $\mathcal C\in\mathfrak S_{\mathrm{loc,Lyap}}$, define the semigroup map $\operatorname{Sem}$ and the canonical semigroup associated with $G$ through the upper-expectation formula
\begin{equation}
\label{eq:robust_valuation_mapping_definition}
 (\operatorname{Sem}(\mathcal C))_t f(x)
 :=\sup_{\mathbb P\in\mathcal C_{0,x}}
       \mathbb E^{\mathbb P}[f(X_t)\mathbb I_{\{t<\tau_\infty\}}],
 \qquad \mathcal V^G:=\operatorname{Sem}(\mathcal P(G)).
\end{equation}
We say that $\mathcal C$ represents $\mathcal T$ when $\operatorname{Sem}(\mathcal C)=\mathcal T$.

We first verify the semigroup axioms and then recover each stable
structure of laws from its semigroup. This shows that $\operatorname{Sem}$
is an order embedding from $\mathfrak S_{\mathrm{loc,Lyap}}$ into $\mathfrak V$.

\begin{proposition}
\label{prop:stable_virtual_structure_generates_valuation}
For $\mathcal C\in\mathfrak S_{\mathrm{loc,Lyap}}$,
$\operatorname{Sem}(\mathcal C)$ is a time-homogeneous sublinear sub-Markov semigroup, and the defining supremum is attained for every bounded upper semicontinuous terminal function.
\end{proposition}
\begin{proof}
Choose $G\in\mathfrak G_{\mathrm{Lyap}}$ with $\mathcal C\subseteq\mathcal P(G)$.
For $0\le t\le T$, write
\[
    (\operatorname{Sem}(\mathcal C))_{t,T}f(x)
    :=\sup_{\mathbb P\in\mathcal C_{t,x}}
      \mathbb E^{\mathbb P}[f(X_T)\mathbb I_{\{T<\tau_\infty\}}].
\]
The prescribed-history condition on $\mathcal C_{t,x}$ gives
$\operatorname{Sem}(\mathcal C)_{t,t}=\operatorname{id}_{\USC_b(D)}$ for all $t\ge0$.
Moreover, the time homogeneity of $\mathcal C$ implies that
$\operatorname{Sem}(\mathcal C)$ is time homogeneous:
\[
    \operatorname{Sem}(\mathcal C)_{t,T}
    =\operatorname{Sem}(\mathcal C)_{0,T-t}:=\operatorname{Sem}(\mathcal C)_{T-t},
    \qquad 0\le t\le T.
\]

Next, Proposition~\ref{lem:killed_cylinder_payoff_upper_stability} in the Supplementary Material, applied with $m=1$, shows that the map
\[
    \mathcal C_x\ni\mathbb P
    \longmapsto\mathbb E^{\mathbb P}[f(X_t)\mathbb I_{\{t<\tau_\infty\}}]
\]
is upper semicontinuous for every fixed $t\ge0$, $x\in D$, and
$f\in\USC_b(D)$.
Compactness of each fiber therefore implies that the supremum in the
representation of $\operatorname{Sem}(\mathcal C)_tf(x)$ is attained.

Moreover, $\operatorname{Sem}(\mathcal C)_tf\in\USC_b(D)$.
Boundedness follows from
$|\operatorname{Sem}(\mathcal C)_tf(x)|\le\|f\|_\infty$.
To prove upper semicontinuity, let $x_n\to x$ and choose maximizing laws
$\mathbb P_n\in\mathcal C_{x_n}$.
Compactness over compact sets of initial states and closedness of the
graph of $\mathcal C$ yield, along a subsequence realizing the upper
limit, a further subsequence converging weakly to some
$\mathbb P\in\mathcal C_x$.
Applying Proposition~\ref{lem:killed_cylinder_payoff_upper_stability} in the Supplementary Material gives
\[
    \limsup_{n\to\infty}\operatorname{Sem}(\mathcal C)_tf(x_n)
    \le\mathbb E^{\mathbb P}[f(X_t)\mathbb I_{\{t<\tau_\infty\}}]
    \le\operatorname{Sem}(\mathcal C)_tf(x).
\]

Properties \ref{V1}--\ref{V2} are immediate from the fiberwise nonemptiness of $\mathcal C$ and $0\le \mathbb I_{\{t<\tau_\infty\}}\le 1$.
To prove \ref{V3}, fix $t\ge0$ and $x\in D$, and let $f_j,f\in\USC_b(D)$ satisfy $f_j\searrow f$.
Then $\limsup_j^*f_j=f$.
For each $j$, choose a law $\mathbb P_j\in\mathcal C_x$ attaining the supremum in the representation of $(\operatorname{Sem}(\mathcal C))_tf_j(x)$.
Weak compactness of $\mathcal C_x$ gives a subsequence converging weakly to some $\mathbb P\in\mathcal C_x$.
Passing to the limit along this subsequence and applying Proposition~\ref{lem:killed_cylinder_payoff_upper_stability} in the Supplementary Material yields
\[
    \lim_{j\to\infty}(\operatorname{Sem}(\mathcal C))_tf_j(x)
    \le\mathbb E^{\mathbb P}[f(X_t)\mathbb I_{\{t<\tau_\infty\}}]
    \le(\operatorname{Sem}(\mathcal C))_tf(x).
\]
Monotonicity gives the reverse inequality, proving \ref{V3}.

A similar argument proves \ref{V4}.
Let $(t_j,x_j)\to(t,x)$ and choose laws $\mathbb P_j\in\mathcal C_{x_j}$ attaining the suprema in the representations of $(\operatorname{Sem}(\mathcal C))_{t_j}f(x_j)$.
Weak compactness of $(\bigcup_{j\ge1}\mathcal C_{x_j})\cup\mathcal C_x$ and closedness of the graph of $\mathcal C$, together with Proposition~\ref{lem:killed_cylinder_payoff_upper_stability} in the Supplementary Material, yield \ref{V4}.

To prove \ref{V5}, write $\mathcal T=\operatorname{Sem}(\mathcal C)$.
Stability under restarted conditioning at time $s$ and time homogeneity of $\mathcal C$ give
\[
    \mathbb E^{\mathbb P}
    [f(X_{s+t})\mathbb I_{\{s+t<\tau_\infty\}}\mid\widetilde{\mathcal F}_s]
    \le\mathcal T_tf(X_s)\mathbb I_{\{s<\tau_\infty\}},
    \qquad \mathbb P\in\mathcal C_{0,x},
\]
and hence $\mathcal T_{s+t}f\le\mathcal T_s(\mathcal T_tf)$.
For the reverse inequality, observe that the maximizing correspondence
\[
    y\longmapsto
    \argmax_{\mathbb Q\in\mathcal C_{s,y}}
    \mathbb E^{\mathbb Q}[f(X_{s+t})\mathbb I_{\{s+t<\tau_\infty\}}]
\]
has nonempty compact values and a Borel graph, since the expectation
map and $\mathcal T_tf$ are Borel.
Proposition~\ref{lem:measurable_model_selectors} in the Supplementary Material supplies a Borel selector, which we extend to the cemetery state by $\delta_{\mathbf\triangle}$.
Pasting this selector at time $s$ onto any law in $\mathcal C_{0,x}$ and taking suprema yields the reverse inequality, proving \ref{V5}.
\end{proof}

The (strong) stability and fiberwise nonemptiness of $\mathcal P(G)$ were established in Proposition~\ref{prop:propertiesP_t,x_revised} and Corollary~\ref{cor:canonical_regular_dynamic_virtual_structure}, respectively.
Applying Proposition~\ref{prop:stable_virtual_structure_generates_valuation} to $\mathcal P(G)$ therefore yields the following corollary.

\begin{corollary}
\label{cor:canonical_robust_valuation}
For every $G\in\mathfrak G_{\mathrm{Lyap}}$,
$\mathcal V^G=\operatorname{Sem}(\mathcal P(G))$ is a sublinear sub-Markov semigroup.
\end{corollary}

We next recover each stable structure of laws from its semigroup.
For $\mathcal T\in\mathfrak V$, introduce the semigroup orbits
\begin{equation}
\label{eq:def_semigroup_orbit_process}
    Y_s^{\mathcal T;R,g}
    :=\mathcal T_{R-s}g(X_s)\mathbb I_{\{s<\tau_\infty\}},
    \qquad 0\le s\le R,
\end{equation}
and define the semigroup-compatible family of laws at $(t,x)$ by
\begin{equation} \label{eq:def_valuation_compatible_family}
\begin{aligned}
\mathcal R_{t,x}^{\mathcal T}:=\{\mathbb P\in\mathfrak M:\; &\mathbb P(X_{\cdot\wedge t}=\mathbf x)=1,\\
&Y^{\mathcal T;R,g}\text{ is a }\mathbb P\text{-supermartingale on }[t,R] \\
&\qquad\text{ for all }R\ge t,\ g\in\USC_b(D)\}. 
\end{aligned} 
\end{equation}
We write $\mathcal R^{\mathcal T}:=\{\mathcal R_{t,x}^{\mathcal T}\}_{t,x}$ for the corresponding structure and $\mathcal R_x^{\mathcal T}:=\mathcal R_{0,x}^{\mathcal T}$, and denote the restriction of $\mathcal R^{\mathcal T}$ to initial times in $[0,T]$ by $\mathcal R_{[0,T]}^{\mathcal T}$.

First observe that, for $\mathcal T\in\mathfrak V$ and a law
$\mathbb P$ with the prescribed history at $(t,x)$,
$\mathbb P\in\mathcal R_{t,x}^{\mathcal T}$ if and only if
\begin{equation}
\label{eq:valuation_terminal_conditional_domination}
 \mathbb E^{\mathbb P}
 [g(X_r)\mathbb I_{\{r<\tau_\infty\}}\mid\widetilde{\mathcal F}_s]
 \le\mathcal T_{r-s}g(X_s)\mathbb I_{\{s<\tau_\infty\}},
 \qquad t\le s<r,\quad g\in\USC_b(D).
\end{equation}
Necessity follows from
$\mathbb E^{\mathbb P}[Y_r^{\mathcal T;r,g}\mid\widetilde{\mathcal F}_s]
\le Y_s^{\mathcal T;r,g}$.
Conversely, applying \eqref{eq:valuation_terminal_conditional_domination} to
$\mathcal T_{R-r}g\in\USC_b(D)$ and using the semigroup property gives
$\mathbb E^{\mathbb P}[Y_r^{\mathcal T;R,g}\mid\widetilde{\mathcal F}_s]
\le Y_s^{\mathcal T;R,g}$ for $t\le s<r\le R$.
In particular,
\[
 \mathcal T_1\le\mathcal T_2
 \quad\Longrightarrow\quad
 \mathcal R^{\mathcal T_1}\subseteq\mathcal R^{\mathcal T_2}.
\]

We first characterize admissible one-time marginals through semigroup inequalities and construct a measurable lifting of these marginals to path laws.
For $\mathcal C\in\mathfrak S_{\mathrm{loc,Lyap}}$, set $\mathcal T:=\operatorname{Sem}(\mathcal C)$.
It is useful to retain the cemetery mass explicitly.
For $h\in C(\widehat D)$ and $q\ge0$, define
\begin{equation}
\label{eq:cemetery_completed_valuation}
\widehat{\mathcal T}_q h(y)
:=
\begin{cases}
h(\triangle)+\mathcal T_q\bigl(h|_D-h(\triangle)\bigr)(y),&y\in D,\\
h(\triangle),&y=\triangle.
\end{cases}
\end{equation}
Since $h(X_r)=h(\triangle)+\bigl(h(X_r)-h(\triangle)\bigr)\mathbb I_{\{r<\tau_\infty\}}$, time homogeneity gives
\begin{equation}
\label{eq:completed_valuation_support}
\widehat{\mathcal T}_{r-s}h(y)
=
\sup_{\mathbb P\in\mathcal C_{s,y}}
\mathbb E^{\mathbb P}[h(X_r)],
\qquad 0\le s<r.
\end{equation}
Clearly, $\widehat{\mathcal T}_q h$ is Borel on $\widehat D$.

\begin{lemma}
\label{lem:stable_marginal_lifting}
Let $\mathcal C\in\mathfrak S_{\mathrm{loc,Lyap}}$ and $0\le s<r$.  
For $y\in\widehat D$, set
\[
K_{s,r}^{\mathcal C}(y)
:=\bigl\{\mathbb P\circ X_r^{-1}:\mathbb P\in\mathcal C_{s,y}\bigr\}
\subseteq\mathcal P(\widehat D).
\]
Then $K_{s,r}^{\mathcal C}(y)$ is nonempty, compact and convex, and
\begin{equation}
\label{eq:stable_marginal_separation}
\mu\in K_{s,r}^{\mathcal C}(y)
\quad\Longleftrightarrow\quad
\int_{\widehat D}h\,d\mu
\le\widehat{\mathcal T}_{r-s}h(y)
\quad\text{for every }h\in C(\widehat D).
\end{equation}
The graph of $K_{s,r}^{\mathcal C}$ is Borel, and there exists a Borel map
\[
L_{s,r}:\operatorname{Gr}(K_{s,r}^{\mathcal C})\longrightarrow\mathfrak M
\]
satisfying $L_{s,r}(y,\mu)\in\mathcal C_{s,y}$ and $L_{s,r}(y,\mu)\circ X_r^{-1}=\mu$.
Moreover, for every $(t,x)$ and every finite collection $t\le r_1<\cdots<r_m$, the map
\begin{equation}
\label{eq:stable_finite_marginal_continuity}
\mathbb P\longmapsto
\mathbb P\circ(X_{r_1},\ldots,X_{r_m})^{-1}
\end{equation}
is weakly continuous on $\mathcal C_{t,x}$.
\end{lemma}

\begin{proof}
We first verify the last assertion, including observations after killing.
The case $x=\triangle$ is immediate.
Suppose $x\in D$, and let $\mathbb P_n\Rightarrow\mathbb P$ in $\mathcal C_{t,x}$.
Choose $G\in\mathfrak G_{\mathrm{Lyap}}$ with $\mathcal C\subseteq\mathcal P(G)$.
Fix $\psi\in C(\widehat D^m)$.
It suffices to show that
\begin{equation}
\label{eq:stable_finite_marginal_continuity_target}
    \lim_{n\to\infty}
    \mathbb E^{\mathbb P_n}[\psi(X_{r_1},\ldots,X_{r_m})]
    =
    \mathbb E^{\mathbb P}[\psi(X_{r_1},\ldots,X_{r_m})].
\end{equation}
For $1\le j\le m$, define $F_j:D^j\to\mathbb R$ by
\[
F_j(y_1,\ldots,y_j)
:=\psi(y_1,\ldots,y_j,\triangle,\ldots,\triangle)
-\psi(y_1,\ldots,y_{j-1},\triangle,\ldots,\triangle)
\]
for $(y_1,\ldots,y_j)\in D^j$.
Each $F_j$ is bounded and continuous.
Moreover, absorption at $\triangle$ yields the pathwise identity
\begin{equation}
\label{eq:full_cylinder_killed_decomposition}
\psi(X_{r_1},\ldots,X_{r_m})
=
\psi(\triangle,\ldots,\triangle)
+\sum_{j=1}^m F_j(X_{r_1},\ldots,X_{r_j})\mathbb I_{\{r_j<\tau_\infty\}}.
\end{equation}
Each summand is understood to be zero when $r_j\ge\tau_\infty$.
Applying Proposition~\ref{lem:killed_cylinder_payoff_upper_stability} in the Supplementary Material to $F_j$ and $-F_j$ shows that the expectation of each summand converges.
The applicability of this proposition follows from Proposition~\ref{prop:recovery_aggregation_generator_compatible_characteristics}\ref{prop:recovery_killing_compensator} and Lemma~\ref{lem:admissible_nonexplosion}, which ensure that these laws satisfy the required lifetime conditions.
Summing the resulting limits proves \eqref{eq:stable_finite_marginal_continuity_target}.
In particular, for every $0\le s<r$, the one-time marginal map $\mathbb P\mapsto\mathbb P\circ X_r^{-1}$ is continuous on each compact fiber $\mathcal C_{s,y}$.
Its image $K_{s,r}^{\mathcal C}(y)$ is therefore compact and is convex because the marginal map is affine.
Its nonemptiness follows directly from the fiberwise nonemptiness of $\mathcal C$.

Next, we prove \eqref{eq:stable_marginal_separation}.
By \eqref{eq:completed_valuation_support},
\begin{align}
\widehat{\mathcal T}_{r-s}h(y)
=
\sup_{\mathbb P\in\mathcal C_{s,y}}
\mathbb E^{\mathbb P}[h(X_r)]
=
\sup_{\mu\in K_{s,r}^{\mathcal C}(y)}
\int_{\widehat D}h\,d\mu,
\qquad 0\le s<r.
\end{align}
The equivalence in \eqref{eq:stable_marginal_separation} now follows
from the separation theorem for compact convex sets of probability measures.

To prove that the graph of $K_{s,r}^{\mathcal C}$ is Borel, fix a countable uniformly dense subset $\{h_j:j\ge1\}$ of $C(\widehat D)$.
Then
\begin{equation}
    \operatorname{Gr}(K_{s,r}^{\mathcal C})
    =
    \bigcap_{j\ge1}
    \left\{
    (y,\mu):
    \int_{\widehat D}h_j\,d\mu
    \le\widehat{\mathcal T}_{r-s}h_j(y)
    \right\}.
\end{equation}
Thus $\operatorname{Gr}(K_{s,r}^{\mathcal C})$ is Borel, being a countable intersection of Borel sets.

Finally, define a correspondence $\Gamma_{s,r}$ on $\operatorname{Gr}(K_{s,r}^{\mathcal C})$ by
\[
\Gamma_{s,r}(y,\mu)
:=\bigl\{\mathbb P\in\mathcal C_{s,y}:
                \mathbb P\circ X_r^{-1}=\mu\bigr\}.
\]
Its values are nonempty and compact.
Its graph is Borel because the graph of
$y\mapsto\mathcal C_{s,y}$ is Borel by topological stability and the cemetery convention, and the pushforward map $\mathbb P\mapsto\mathbb P\circ X_r^{-1}$ is Borel.
We regard its values as subsets of the Polish space $\mathcal P(\Omega^{\mathrm{cad}})$.
Since $\widetilde\Omega$ is Borel, $\mathfrak M$ is a Borel subset of this space.
Proposition~\ref{lem:measurable_model_selectors} in the Supplementary Material therefore supplies the asserted Borel selector $L_{s,r}$.
\end{proof}

We can now recover each stable structure and compare structures through their semigroups.

\begin{proposition}
\label{thm:abstract_stable_recovery}
For every $\mathcal C\in\mathfrak S_{\mathrm{loc,Lyap}}$,
\begin{equation}
\label{eq:abstract_stable_recovery}
 \mathcal C=\mathcal R^{\operatorname{Sem}(\mathcal C)}.
\end{equation}
Consequently, for $\mathcal C_1,\mathcal C_2\in\mathfrak S_{\mathrm{loc,Lyap}}$,
\begin{equation}
\label{eq:valuation_order_embedding}
 \operatorname{Sem}(\mathcal C_1)\le\operatorname{Sem}(\mathcal C_2)
 \quad\Longleftrightarrow\quad\mathcal C_1\subseteq\mathcal C_2.
\end{equation}
In particular, $\operatorname{Sem}$ is injective on
$\mathfrak S_{\mathrm{loc,Lyap}}$.
\end{proposition}

\begin{proof}
Fix $\mathcal C\in\mathfrak S_{\mathrm{loc,Lyap}}$ and write
$\mathcal T:=\operatorname{Sem}(\mathcal C)$.
Stability under restarted conditioning and time homogeneity imply
\eqref{eq:valuation_terminal_conditional_domination} for every
$\mathbb P\in\mathcal C_{t,x}$, and hence
$\mathcal C\subseteq\mathcal R^{\mathcal T}$.

Conversely, fix $(t,x)$ and $\mathbb P\in\mathcal R_{t,x}^{\mathcal T}$.
The cemetery case is immediate, so suppose $x\in D$.
By \eqref{eq:valuation_terminal_conditional_domination},
\begin{equation}
\label{eq:abstract_one_step_recovery}
\mathbb E^{\mathbb P}[g(X_r)\mathbb I_{\{r<\tau_\infty\}}\mid\widetilde{\mathcal F}_s]
\le\mathcal T_{r-s}g(X_s)\mathbb I_{\{s<\tau_\infty\}}
\end{equation}
for every $t\le s<r$ and $g\in C_b(D)$.
For $h\in C(\widehat D)$, apply \eqref{eq:abstract_one_step_recovery} to $g=h|_D-h(\triangle)$ and add $h(\triangle)$.
This gives
\begin{equation}
\label{eq:completed_conditional_domination}
\mathbb E^{\mathbb P}
\bigl[h(X_r)\mid\widetilde{\mathcal F}_s\bigr]
\le\widehat{\mathcal T}_{r-s}h(X_s), \qquad t\le s<r.
\end{equation}

Fix a finite grid $\pi:=\{t=r_0<r_1<\cdots<r_m\}$.
We claim that
\begin{align}\label{eq:thm:abstract_stable_recovery_main_claim}
F_\pi:=\bigl\{\mathbb P'\in\mathcal C_{t,x}:
\mathbb P'\circ(X_r)_{r\in\pi}^{-1}
=\mathbb P\circ(X_r)_{r\in\pi}^{-1}\bigr\}\neq\varnothing.
\end{align}
Let $H_i:=(X_{r_0},\ldots,X_{r_i})$ for $0\le i\le m$.
For each $0\le i<m$, let $\mu_i(z,\cdot)$ be a Borel regular conditional distribution of $X_{r_{i+1}}$ given $H_i=z$ under $\mathbb P$.
For every $h\in C(\widehat D)$, conditioning \eqref{eq:completed_conditional_domination} further on $H_i$ gives a $\mathbb P$-null set $N_{i,h}\subset\widetilde\Omega$ such that
\begin{align}\label{eq:integral_bounded_each_h}
    \int_{\widehat D}h\,d\mu_i(H_i)
    \le \widehat{\mathcal T}_{r_{i+1}-r_i}h(X_{r_i})
    \qquad \text{outside }N_{i,h}.
\end{align}
Choose a countable uniformly dense subset $\{h_j:j\ge1\}$ of $C(\widehat D)$ and set $N_i:=\bigcup_{j\ge1}N_{i,h_j}$.
Then $N_i$ is a $\mathbb P$-null set.
By uniform approximation, \eqref{eq:integral_bounded_each_h} extends to
\begin{align}\label{eq:integral_bounded_all_h}
    \int_{\widehat D}h\,d\mu_i(H_i)
    \le \widehat{\mathcal T}_{r_{i+1}-r_i}h(X_{r_i})
    \qquad \text{for all }h\in C(\widehat D),
    \quad \text{outside }N_i.
\end{align}
The characterization \eqref{eq:stable_marginal_separation} in Lemma~\ref{lem:stable_marginal_lifting} therefore yields
\[
    \mu_i(H_i)\in
    K_{r_i,r_{i+1}}^{\mathcal C}(X_{r_i}),
    \qquad \mathbb P\text{-almost surely}.
\]
Define
\[
    B_i
    :=
    \bigl\{z=(z_0,\ldots,z_i)\in\widehat D^{i+1}:
    \mu_i(z)\in K_{r_i,r_{i+1}}^{\mathcal C}(z_i)\bigr\}.
\]
Since $\mu_i$ is a Borel kernel, the map $z\mapsto(z_i,\mu_i(z))\in\widehat D\times\mathcal P(\widehat D)$ is Borel.
By Lemma~\ref{lem:stable_marginal_lifting}, the graph of $K_{r_i,r_{i+1}}^{\mathcal C}$ is Borel, so its inverse image under this map, namely $B_i$, is Borel.
The preceding inclusion gives $\mathbb P(H_i\in B_i)=1$.

Choose a Borel selector
$y\mapsto\eta_i(y)\in\mathcal C_{r_i,y}$, whose existence follows
from Proposition~\ref{lem:measurable_model_selectors} in the Supplementary Material, and define
\[
    \kappa_i(z)
    :=
    \begin{cases}
    L_{r_i,r_{i+1}}(z_i,\mu_i(z)),&z\in B_i,\\
    \eta_i(z_i),&z\notin B_i.
    \end{cases}
\]
Then $\kappa_i$ is Borel and
$\kappa_i(z)\in\mathcal C_{r_i,z_i}$ for every $z$.
Moreover,
\begin{align}\label{eq:kappa_i_mu_i}
    \kappa_i(z)\circ X_{r_{i+1}}^{-1}
    =\mu_i(z),
    \qquad z\in B_i.
\end{align}

Take any $\mathbb P_0\in\mathcal C_{t,x}$ and define recursively
\begin{align}\label{eq:concatenation_P_i_nu_i}
    \mathbb P_{i+1}
    :=\mathbb P_i\otimes_{r_i}\nu_i,
    \qquad
    \nu_i(\omega):=\kappa_i(H_i(\omega)),
    \qquad 0\le i<m.
\end{align}
Each $\nu_i$ is $\widetilde{\mathcal F}_{r_i}$-measurable and $r_i$-restart-compatible, with $\nu_i(\omega)\in\mathcal C_{r_i,X_{r_i}(\omega)}$.
Stability under concatenation therefore gives $\mathbb P_i\in\mathcal C_{t,x}$ for every $0\le i\le m$.
We prove inductively that
\begin{align}
    \mathbb P_i\circ H_i^{-1}
    =\mathbb P\circ H_i^{-1},
    \qquad 0\le i\le m.
\end{align}
The case $i=0$ follows from the prescribed initial history.
Suppose $\mathbb P_i\circ H_i^{-1}=\mathbb P\circ H_i^{-1}$ holds for some $i<m$.
Then $\mathbb P_i(H_i\in B_i)=\mathbb P(H_i\in B_i)=1$.
On this set, the conditional distribution of $X_{r_{i+1}}$
under the concatenated law is $\mu_i(H_i)$.
Since concatenation preserves the observations up to time $r_i$, \eqref{eq:kappa_i_mu_i} and \eqref{eq:concatenation_P_i_nu_i} yield
\[
    \mathbb P_{i+1}\circ H_{i+1}^{-1}
    =\mathbb P\circ H_{i+1}^{-1}.
\]
This completes the induction.
In particular, $\mathbb P_m\in F_\pi$, proving \eqref{eq:thm:abstract_stable_recovery_main_claim}.

Now let $\pi$ range over finite grids containing $t$ and otherwise
consisting of rational times greater than $t$.
By Lemma~\ref{lem:stable_marginal_lifting}, each $F_\pi$ is closed
in the compact set $\mathcal C_{t,x}$.
These sets have the finite-intersection property, since
\[
    F_{\pi_1}\cap\cdots\cap F_{\pi_k}
    \supseteq F_{\pi_1\cup\cdots\cup\pi_k}\ne\varnothing.
\]
Compactness therefore gives $\mathbb P_*\in\bigcap_\pi F_\pi$.
The laws $\mathbb P_*$ and $\mathbb P$ have the same finite-dimensional distributions at all rational times after $t$ and the same constant history up to $t$.
Right-continuity of the coordinate paths implies that these observations generate the canonical Borel $\sigma$-field.
Hence $\mathbb P_*=\mathbb P$, so $\mathbb P\in\mathcal C_{t,x}$.
This proves \eqref{eq:abstract_stable_recovery}.

If $\operatorname{Sem}(\mathcal C_1)\le\operatorname{Sem}(\mathcal C_2)$,
the monotonicity of $\mathcal R^{\mathcal T}$ established above gives
\[
 \mathcal C_1
 =\mathcal R^{\operatorname{Sem}(\mathcal C_1)}
 \subseteq\mathcal R^{\operatorname{Sem}(\mathcal C_2)}
 =\mathcal C_2.
\]
The reverse implication follows by taking suprema.
Thus \eqref{eq:valuation_order_embedding} holds, and injectivity follows.
\end{proof}

For each $G\in\mathfrak G_{\mathrm{Lyap}}$, define
\begin{equation}
\label{eq:globally_dominated_valuation_class}
    \mathfrak V_{\le}(G)
    :=\{\mathcal T\in\mathfrak V:\mathcal T\le\mathcal V^G\}.
\end{equation}
Applying the recovery identity to $\mathcal P(G)$ yields the following inclusion for every semigroup dominated by $\mathcal V^G$.

\begin{corollary}
\label{prop:valuation_compatible_generator_substructure}
Let $G\in\mathfrak G_{\mathrm{Lyap}}$ and $\mathcal T\in\mathfrak V_{\le}(G)$.
Then $\mathcal R^{\mathcal T}\subseteq\mathcal P(G)$.
\end{corollary}

\subsection{Representing structure of laws}
\label{subsec:recovered_fiber_properties_proofs}

Throughout this subsection, fix $G\in\mathfrak G_{\mathrm{Lyap}}$ and $\mathcal T\in\mathfrak V_{\le}(G)$.
We first show that the recovered structure $\mathcal R^{\mathcal T}$ is time-homogeneous, provides an attained representation of $\mathcal T$, and is stable.
These properties will yield the representation theorem at the end of the subsection.
We begin with time homogeneity.

\begin{proposition}
\label{prop:recovered_time_homogeneity}
Let $\mathcal T\in\mathfrak V$.
For every $g\in\USC_b(D)$ and $R\ge t\ge0$, the semigroup orbits satisfy
\begin{equation}
\label{eq:valuation_orbit_delay_identity}
 Y_{t+s}^{\mathcal T;R,g}\circ\theta_t
 =
 Y_s^{\mathcal T;R-t,g},
 \qquad 0\le s\le R-t.
\end{equation}
Consequently, $\mathcal R^{\mathcal T}$ is a time-homogeneous structure of laws.
\end{proposition}

\begin{proof}
Fix $t\ge0$ and $x\in D$.
For $s\ge0$, the definition of $\theta_t$ gives $X_{t+s}\circ\theta_t=X_s$ and $\mathbb I_{\{t+s<\tau_\infty\}}\circ\theta_t
 =\mathbb I_{\{s<\tau_\infty\}}$.
Consequently, for $R\ge t$ and $0\le s\le R-t$,
\[
Y_{t+s}^{\mathcal T;R,g}\circ\theta_t
=
\mathcal T_{R-t-s}g(X_s)
\mathbb I_{\{s<\tau_\infty\}}
=Y_s^{\mathcal T;R-t,g},
\]
which proves \eqref{eq:valuation_orbit_delay_identity}.

We now show that $\mathcal R^{\mathcal T}$ is time-homogeneous, that is,
\[
    \mathcal R_{t,x}^{\mathcal T}
    =
    \mathcal R_x^{\mathcal T}\circ\theta_t^{-1},
    \qquad
    \mathcal R_x^{\mathcal T}:=\mathcal R_{0,x}^{\mathcal T}.
\]
Fix $t\ge0$.
The cemetery case follows from
$\mathcal R_{t,\triangle}^{\mathcal T}
=\mathcal R_{\triangle}^{\mathcal T}
=\{\delta_{\mathbf\triangle}\}$,
since $\theta_t$ fixes the constant cemetery path.

Let $x\in D$.
We first prove the inclusion $\mathcal R_x^{\mathcal T}\circ\theta_t^{-1}\subseteq\mathcal R_{t,x}^{\mathcal T}$.
Let $\mathbb P\in\mathcal R_x^{\mathcal T}$ and set $\mathbb P':=\mathbb P\circ\theta_t^{-1}$.
By construction, $\mathbb P'$ has the prescribed history at $(t,x)$:
\[
    \mathbb P'(X_s=x\text{ for all }s\in[0,t])=1.
\]
Moreover, $\theta_t^{-1}(\widetilde{\mathcal F}_{t+r})=\widetilde{\mathcal F}_r$ for every $r\ge0$.
Thus, for every $R\ge t$, $g\in\USC_b(D)$, $0\le r\le s\le R-t$, and bounded nonnegative $\widetilde{\mathcal F}_{t+r}$-measurable $F$, identity~\eqref{eq:valuation_orbit_delay_identity} gives
\[
\mathbb E^{\mathbb P'}\!\left[
F\bigl(Y_{t+s}^{\mathcal T;R,g}
       -Y_{t+r}^{\mathcal T;R,g}\bigr)\right]=
\mathbb E^{\mathbb P}\!\left[
(F\circ\theta_t)
\bigl(Y_s^{\mathcal T;R-t,g}
       -Y_r^{\mathcal T;R-t,g}\bigr)\right]
\le0.
\]
The inequality follows from $\mathbb P\in\mathcal R_x^{\mathcal T}$, since $F\circ\theta_t$ is $\widetilde{\mathcal F}_r$-measurable.
Hence $\mathbb P'\in\mathcal R_{t,x}^{\mathcal T}$, proving $\mathcal R_x^{\mathcal T}\circ\theta_t^{-1}\subseteq\mathcal R_{t,x}^{\mathcal T}$.

To prove the reverse inclusion $\mathcal R_{t,x}^{\mathcal T}\subseteq\mathcal R_x^{\mathcal T}\circ\theta_t^{-1}$, let $\mathbb P\in\mathcal R_{t,x}^{\mathcal T}$.
Define the forward shift by $(\eta_t\omega)(q):=\omega(t+q)$ and set $\mathbb P':=\mathbb P\circ\eta_t^{-1}$.
Since $\mathbb P$ has the prescribed history at $(t,x)$, we have $\mathbb P'(X_0=x)=1$ and $\mathbb P'\circ\theta_t^{-1}=\mathbb P$.
It therefore remains to show that $\mathbb P'\in\mathcal R_x^{\mathcal T}$.

Fix $R\ge0$, $g\in\USC_b(D)$, $0\le r\le s\le R$, and a bounded nonnegative $\widetilde{\mathcal F}_r$-measurable $H$.
Since $\eta_t^{-1}(\widetilde{\mathcal F}_r)\subseteq\widetilde{\mathcal F}_{t+r}$, the function $H\circ\eta_t$ is $\widetilde{\mathcal F}_{t+r}$-measurable.
Using $\mathbb P=\mathbb P'\circ\theta_t^{-1}$, $\eta_t\circ\theta_t=\mathrm{id}$, and \eqref{eq:valuation_orbit_delay_identity}, we obtain
\[
\mathbb E^{\mathbb P'}\!\left[
H\bigl(Y_s^{\mathcal T;R,g}
       -Y_r^{\mathcal T;R,g}\bigr)\right]=
\mathbb E^{\mathbb P}\!\left[
(H\circ\eta_t)
\bigl(Y_{t+s}^{\mathcal T;t+R,g}
       -Y_{t+r}^{\mathcal T;t+R,g}\bigr)\right]
\le0.
\]
Thus $\mathbb P'\in\mathcal R_x^{\mathcal T}$, and consequently $\mathbb P=\mathbb P'\circ\theta_t^{-1}\in\mathcal R_x^{\mathcal T}\circ\theta_t^{-1}$.
This proves the reverse inclusion and completes the proof of time homogeneity.
\end{proof}

To show that $\mathcal R^{\mathcal T}$ provides an attained representation of $\mathcal T$, we first establish two preparatory lemmas.
For $h>0$, set
\begin{equation}
\label{eq:def_one_step_valuation_dominated_fiber}
 \mathcal D_h^{\mathcal T}(x):=
 \{\mathbb P\in\mathcal P_x(G):
   \mathbb E^{\mathbb P}[g(X_h)\mathbb I_{\{h<\tau_\infty\}}]\le\mathcal T_hg(x)
   \text{ for all }g\in\USC_b(D)\},
\end{equation}
and, for $\varphi\in\USC_b(D)$, set
\begin{equation}
\label{eq:def_one_step_valuation_calibration_fiber}
 \mathcal C_h^{\mathcal T}(x;\varphi):=
 \{\mathbb P\in\mathcal D_h^{\mathcal T}(x):
   \mathbb E^{\mathbb P}[\varphi(X_h)\mathbb I_{\{h<\tau_\infty\}}]=\mathcal T_h\varphi(x)\}.
\end{equation}

\begin{lemma}
\label{prop:one_step_polar_calibration_valuation}
For $h>0$ and $\varphi\in\USC_b(D)$,
$\mathcal C_h^{\mathcal T}(x;\varphi)$ is nonempty and compact for every
$x\in D$. Its graph is Borel, and it admits a Borel selector.
\end{lemma}
\begin{proof}
Fix $h>0$ and $x\in D$, and write
\[
    \Pi(f):=\mathcal T_hf(x),
    \qquad
    \Lambda_f(\mathbb P)
    :=\mathbb E^{\mathbb P}
      [f(X_h)\mathbb I_{\{h<\tau_\infty\}}],
    \qquad f\in\USC_b(D).
\]
By Proposition~\ref{lem:killed_cylinder_payoff_upper_stability} in the Supplementary Material, applied with $m=1$,
$\Lambda_f$ is weakly upper semicontinuous on $\mathcal P_x(G)$
for every $f\in\USC_b(D)$.
For $f\in C_b(D)$, applying the same result to $f$ and $-f$
shows that $\Lambda_f$ is weakly continuous.

We first prove that
\begin{equation}\label{eq:D_h^T_continuoustest}
    \mathcal D_h^{\mathcal T}(x)
    =
    \left\{
    \mathbb P\in\mathcal P_x(G):
    \Lambda_g(\mathbb P)\le\Pi(g)
    \text{ for every }g\in C_b(D)
    \right\}.
\end{equation}
Denote the right-hand side by $\mathcal D$.
The inclusion $\mathcal D_h^{\mathcal T}(x)\subseteq\mathcal D$ is immediate.
For the reverse inclusion, let $\mathbb P\in\mathcal D$ and fix $f\in\USC_b(D)$.
Choose uniformly bounded functions $f_n\in C_b(D)$ decreasing to $f$.
Since $\Lambda_{f_n}(\mathbb P)\le\Pi(f_n)$ for every $n\ge1$, bounded convergence and \ref{V3} give
\[
    \Lambda_f(\mathbb P)
    =\lim_{n\to\infty}\Lambda_{f_n}(\mathbb P)
    \le\lim_{n\to\infty}\Pi(f_n)
    =\Pi(f).
\]
Thus $\mathbb P\in\mathcal D_h^{\mathcal T}(x)$, proving \eqref{eq:D_h^T_continuoustest}.
Since $\Lambda_g$ is weakly continuous for every $g\in C_b(D)$, this characterization shows that $\mathcal D_h^{\mathcal T}(x)$ is a closed subset of the compact set $\mathcal P_x(G)$, and hence is compact.

Fix $\varphi\in\USC_b(D)$.
For $m\ge1$ and $g_1,\ldots,g_m\in C_b(D)$, define
\[
    \widehat{\mathcal K}(\varphi;g_1,\ldots,g_m)
    :=
    \left\{
    \mathbb P\in\mathcal P_x(G):
    \begin{aligned}
        &\Lambda_\varphi(\mathbb P)\ge\Pi(\varphi),\\
        &\Lambda_{g_i}(\mathbb P)\le\Pi(g_i),
          \quad 1\le i\le m
    \end{aligned}
    \right\}.
\]
These sets are compact by the semicontinuity properties of $\Lambda_\cdot$ established above.
We claim that each is nonempty.
Suppose, for a contradiction, that $\widehat{\mathcal K}(\varphi;g_1,\ldots,g_m)=\varnothing$.
Set $M:=\|\varphi\|_\infty$ and consider
\[
\begin{aligned}
    A
    &:=
    \left\{
    \bigl(u,\Lambda_{g_1}(\mathbb P),\ldots,
             \Lambda_{g_m}(\mathbb P)\bigr):
    \mathbb P\in\mathcal P_x(G),\
    -M\le u\le\Lambda_\varphi(\mathbb P)
    \right\},\\
    B
    &:=
    [\Pi(\varphi),\infty)
    \times\prod_{i=1}^m(-\infty,\Pi(g_i)].
\end{aligned}
\]
The set $A$ is compact and convex.
Indeed, the set of pairs $(\mathbb P,u)$ satisfying
$\mathbb P\in\mathcal P_x(G)$ and
$-M\le u\le\Lambda_\varphi(\mathbb P)$ is compact and convex,
and its defining map into $A$ is continuous and affine.
The set $B$ is closed and convex, and our supposition gives
$A\cap B=\varnothing$.

By the strict separation, there exist $a\ge0$, $\lambda_1,\ldots,\lambda_m\ge0$, and $\varepsilon>0$ such that
\begin{align}\label{eq:separating_ineq}
    a\Lambda_\varphi(\mathbb P)
    -\sum_{i=1}^m\lambda_i\Lambda_{g_i}(\mathbb P)
    \le
    a\Pi(\varphi)-\sum_{i=1}^m\lambda_i\Pi(g_i)-\varepsilon
\end{align}
for every $\mathbb P\in\mathcal P_x(G)$.
Since $a\ge0$, the function $\psi:=a\varphi-\sum_{i=1}^m\lambda_i g_i$ belongs to $\USC_b(D)$.
By sublinearity of $\Pi$, we have
\[
    a\Pi(\varphi)
    =\Pi(a\varphi)
    \le\Pi(\psi)+\sum_{i=1}^m\lambda_i\Pi(g_i).
\]
Taking the supremum over $\mathbb P$ in the inequality \eqref{eq:separating_ineq} and using $\mathcal T\le\mathcal V^G$, we obtain
\[
    \Pi(\psi)
    \le\mathcal V_h^G\psi(x)
    \le
    a\Pi(\varphi)-\sum_{i=1}^m\lambda_i\Pi(g_i)-\varepsilon
    \le\Pi(\psi)-\varepsilon,
\]
a contradiction. 
Hence each $\widehat{\mathcal K}(\varphi;g_1,\ldots,g_m)$ is nonempty and compact.

The compact sets $\widehat{\mathcal K}(\varphi;g)$, $g\in C_b(D)$, therefore have the finite-intersection property.
Their intersection is nonempty and compact.
By \eqref{eq:D_h^T_continuoustest}, this intersection equals
\[
    \mathcal D_h^{\mathcal T}(x)
    \cap
    \{\mathbb P:\Lambda_\varphi(\mathbb P)\ge\Pi(\varphi)\}
    =
    \mathcal C_h^{\mathcal T}(x;\varphi).
\]
This proves nonemptiness and compactness of $\mathcal C_h^{\mathcal T}(x;\varphi)$.

We next show that $x\mapsto\mathcal D_h^{\mathcal T}(x)$ has a closed graph.
Let $x_n\to x$, $\mathbb P_n\Rightarrow\mathbb P$, and $\mathbb P_n\in\mathcal D_h^{\mathcal T}(x_n)$.
The closed-graph property of $x\mapsto\mathcal P_x(G)$ gives $\mathbb P\in\mathcal P_x(G)$.
For every $g\in C_b(D)$, Proposition~\ref{lem:killed_cylinder_payoff_upper_stability} in the Supplementary Material, applied to $g$ and $-g$, together with upper semicontinuity of $\mathcal T_hg$, gives
\[
    \Lambda_g(\mathbb P)
    =\lim_{n\to\infty}\Lambda_g(\mathbb P_n)
    \le\limsup_{n\to\infty}\mathcal T_hg(x_n)
    \le\mathcal T_hg(x).
\]
By \eqref{eq:D_h^T_continuoustest}, we therefore have $\mathbb P\in\mathcal D_h^{\mathcal T}(x)$.
Thus $x\mapsto\mathcal D_h^{\mathcal T}(x)$ has a closed graph.

Finally, the maps $\mathbb P\mapsto\Lambda_\varphi(\mathbb P)$ and $x\mapsto\mathcal T_h\varphi(x)$ are Borel.
Since $x\mapsto\mathcal D_h^{\mathcal T}(x)$ has a closed graph, the set
\[
    \left\{
    (x,\mathbb P):
    \mathbb P\in\mathcal D_h^{\mathcal T}(x),\
    \Lambda_\varphi(\mathbb P)=\mathcal T_h\varphi(x)
    \right\}
\]
is Borel. 
This is precisely the graph of $x\mapsto\mathcal C_h^{\mathcal T}(x;\varphi)$.
Viewing its nonempty compact values as subsets of the Polish space $\mathcal P(\Omega^{\mathrm{cad}})$, Proposition~\ref{lem:measurable_model_selectors} in the Supplementary Material provides a Borel selector.
\end{proof}

\begin{lemma}
\label{lem:continuous_testing_time_shifts}
Let $\mathcal T\in\mathfrak V$, $t\ge0$, and $x\in D$.
For a law $\mathbb P$ with the prescribed history at $(t,x)$,
membership in $\mathcal R_{t,x}^{\mathcal T}$ is equivalent to
\begin{equation}
\label{eq:primitive_valuation_compatibility}
 \mathbb E^{\mathbb P}[FY_s^{\mathcal T;R,g}]
 \le
 \mathbb E^{\mathbb P}[FY_r^{\mathcal T;R,g}]
\end{equation}
for every $R\ge t$, $g\in C_b(D)$, $t\le r\le s\le R$,
and bounded nonnegative $\widetilde{\mathcal F}_r$-measurable $F$.
\end{lemma}

\begin{proof}
Necessity of \eqref{eq:primitive_valuation_compatibility} follows because
each $Y^{\mathcal T;R,g}$ is a $\mathbb P$-supermartingale when
$\mathbb P\in\mathcal R_{t,x}^{\mathcal T}$.
Conversely, suppose that \eqref{eq:primitive_valuation_compatibility} holds for continuous test functions.
For $g\in\USC_b(D)$, choose uniformly bounded functions $g_j\in C_b(D)$ decreasing to $g$.
By \ref{V3}, we have $\mathcal T_qg_j\searrow\mathcal T_qg$ for every $q\ge0$.
Bounded convergence therefore extends \eqref{eq:primitive_valuation_compatibility} to every $g\in\USC_b(D)$.
Since $Y^{\mathcal T;R,g}$ is adapted and bounded by $\|g\|_\infty$,
\eqref{eq:primitive_valuation_compatibility} makes it a
$\mathbb P$-supermartingale on $[t,R]$ for every $R\ge t$ and
$g\in\USC_b(D)$, proving $\mathbb P\in\mathcal R_{t,x}^{\mathcal T}$.
\end{proof}

\begin{proposition}
\label{prop:recovered_calibrated_nonemptiness}
For every $0\le t\le T$, $x\in D$, and $f\in\USC_b(D)$, 
\[
\mathcal T_{T-t}f(x)=
\max_{\mathbb P\in\mathcal R^{\mathcal T}_{t,x}}
\mathbb E^{\mathbb P}[f(X_T)\mathbb I_{\{T<\tau_\infty\}}].
\]
In particular, $\mathcal R^{\mathcal T}$ is fiberwise nonempty.
\end{proposition}
\begin{proof}
By time homogeneity of $\mathcal T$ and $\mathcal R^{\mathcal T}$,
it suffices to consider $t=0$.
Fix $x\in D$ and $f\in\USC_b(D)$.
We first treat the case $T>0$.
By the definition of $\mathcal R^{\mathcal T}$,
\begin{align}\label{eq:sup_le_T_T}
    \sup_{\mathbb P\in\mathcal R_x^{\mathcal T}}
    \mathbb E^{\mathbb P}
    [f(X_T)\mathbb I_{\{T<\tau_\infty\}}]
    \le\mathcal T_Tf(x).
\end{align}
It therefore remains to construct a law $\mathbb P\in\mathcal R_x^{\mathcal T}$ such that
\begin{align}\label{eq:attaining_maximum_R_T}
    \mathbb E^{\mathbb P}
    [f(X_T)\mathbb I_{\{T<\tau_\infty\}}]
    =\mathcal T_Tf(x).
\end{align}

By Proposition~\ref{prop:parameterized_closed_valuation_relations}\textup{(iii)} in the Supplementary Material and Lemma~\ref{lem:continuous_testing_time_shifts}, we can choose a countable family
\[
    \mathscr I_0=\{(g_i,r_i,s_i):i\ge1\},
    \qquad
    g_i\in C_b(D),\quad 0\le r_i\le s_i<\infty,
\]
such that, for every $\mathbb P\in\mathcal P_x(G)$, membership in $\mathcal R_x^{\mathcal T}$ is equivalent to
\begin{equation}\label{eq:recovered_countable_tests}
    \mathbb E^{\mathbb P}
    [F g_i(X_{s_i})\mathbb I_{\{s_i<\tau_\infty\}}]
    \le
    \mathbb E^{\mathbb P}
    [F\mathcal T_{s_i-r_i}g_i(X_{r_i})
       \mathbb I_{\{r_i<\tau_\infty\}}],
\end{equation}
for every $i\ge1$ and every bounded nonnegative $\widetilde{\mathcal F}_{r_i}$-measurable random variable $F$.
For each $m\ge1$, choose a finite grid $\pi_m=\{0=u_0<\cdots<u_N\}$ containing $T$ and all times $r_i,s_i$ with $1\le i\le m$, where $N=N(m)$.
For $0\le k<N$, set
\[
    h_k:=u_{k+1}-u_k,
    \qquad
    \varphi_k:=
    \begin{cases}
        \mathcal T_{T-u_{k+1}}f,&u_{k+1}\le T,\\
        0,&u_{k+1}>T.
    \end{cases}
\]
By Lemma~\ref{prop:one_step_polar_calibration_valuation}, there exist Borel selectors
\begin{align}\label{eq:kappa_in_C}
    \kappa_{k,m}(y)
    \in\mathcal C_{h_k}^{\mathcal T}(y;\varphi_k),
    \qquad y\in D.
\end{align}
Extend each selector to $\widehat D$ by setting $\kappa_{k,m}(\triangle):=\delta_{\mathbf\triangle}$.
For $1\le k<N$, define the $\widetilde{\mathcal F}_{u_k}$-measurable kernel
\[
    \nu_{k,m}(\omega)
    :=
    \kappa_{k,m}(X_{u_k}(\omega))\circ\theta_{u_k}^{-1}.
\]
Time homogeneity of $\mathcal P(G)$ ensures that
$\nu_{k,m}(\omega)\in\mathcal P_{u_k,X_{u_k}(\omega)}(G)$,
and these kernels are restart-compatible.
Starting with $\kappa_{0,m}(x)$, successively concatenate
$\nu_{k,m}$ at $u_k$ for $k=1,\ldots,N-1$,
and denote the resulting law by $\mathbb P_m$.
Then, by Proposition~\ref{prop:propertiesP_t,x_revised}, $\mathbb P_m\in\mathcal P_x(G)$.

We now verify the properties of $\mathbb P_m$.
For every $0\le k<N$ and $\psi\in\USC_b(D)$, by \eqref{eq:kappa_in_C}, we have
\begin{equation}\label{eq:grid_one_step_domination}
    \mathbb E^{\mathbb P_m}\!\left[
        \psi(X_{u_{k+1}})
        \mathbb I_{\{u_{k+1}<\tau_\infty\}}
        \,\middle|\,\widetilde{\mathcal F}_{u_k}
    \right]\le
    \mathcal T_{h_k}\psi(X_{u_k})
    \mathbb I_{\{u_k<\tau_\infty\}},
\end{equation}
with equality when $\psi=\varphi_k$.
Inequality~\eqref{eq:grid_one_step_domination}, including equality for
$\psi=\varphi_k$, remains valid under the final law $\mathbb P_m$
because subsequent concatenations preserve the law of the path
up to $u_{k+1}$.

Fix $i\le m$.
Whenever $u_{k+1}\le s_i$, apply
\eqref{eq:grid_one_step_domination} with
$\psi=\mathcal T_{s_i-u_{k+1}}g_i$.
By \ref{V5},
\[
    \mathbb E^{\mathbb P_m}\!\left[
        Y_{u_{k+1}}^{\mathcal T;s_i,g_i}
        \,\middle|\,\widetilde{\mathcal F}_{u_k}
    \right]
    \le Y_{u_k}^{\mathcal T;s_i,g_i}.
\]
Since $r_i,s_i\in\pi_m$, iteration between these two grid points gives
\[
    \mathbb E^{\mathbb P_m}\!\left[
        g_i(X_{s_i})\mathbb I_{\{s_i<\tau_\infty\}}
        \,\middle|\,\widetilde{\mathcal F}_{r_i}
    \right]
    \le
    \mathcal T_{s_i-r_i}g_i(X_{r_i})
    \mathbb I_{\{r_i<\tau_\infty\}}.
\]
Multiplying by any bounded nonnegative $\widetilde{\mathcal F}_{r_i}$-measurable $F$ and taking expectations proves \eqref{eq:recovered_countable_tests} with $\mathbb P=\mathbb P_m$ for every $i\le m$.

Likewise, whenever $u_{k+1}\le T$, the equality in \eqref{eq:grid_one_step_domination} for $\psi=\varphi_k$ and \ref{V5} give
\[
    \mathbb E^{\mathbb P_m}\!\left[
        Y_{u_{k+1}}^{\mathcal T;T,f}
        \,\middle|\,\widetilde{\mathcal F}_{u_k}
    \right]
    =Y_{u_k}^{\mathcal T;T,f}.
\]
Since $T\in\pi_m$, iteration from $0$ to $T$ yields
\begin{equation}\label{eq:grid_target_terminal_calibration}
    \mathbb E^{\mathbb P_m}
    [f(X_T)\mathbb I_{\{T<\tau_\infty\}}]
    =\mathcal T_Tf(x).
\end{equation}

By compactness of $\mathcal P_x(G)$, there exists a subsequence
$\mathbb P_{m_j}\Rightarrow\mathbb P\in\mathcal P_x(G)$.
For each fixed $i$, inequality~\eqref{eq:recovered_countable_tests}
holds under $\mathbb P_{m_j}$ for all sufficiently large $j$,
since $m_j\ge i$.
Proposition~\ref{prop:parameterized_closed_valuation_relations}\textup{(i)} in the Supplementary Material shows that the set
of laws satisfying \eqref{eq:recovered_countable_tests} for every bounded
nonnegative $\widetilde{\mathcal F}_{r_i}$-measurable $F$
is weakly closed in $\mathcal P_x(G)$.
Hence $\mathbb P$ satisfies
\eqref{eq:recovered_countable_tests} for every $i\ge1$.
The defining property of $\mathscr I_0$ therefore gives
$\mathbb P\in\mathcal R_x^{\mathcal T}$.

Finally, Proposition~\ref{lem:killed_cylinder_payoff_upper_stability} in the Supplementary Material, applied with $m=1$,
together with \eqref{eq:sup_le_T_T} and \eqref{eq:grid_target_terminal_calibration}, gives
\[
    \mathcal T_Tf(x)
    =\lim_{j\to\infty}
      \mathbb E^{\mathbb P_{m_j}}
      [f(X_T)\mathbb I_{\{T<\tau_\infty\}}]
    \le
      \mathbb E^{\mathbb P}
      [f(X_T)\mathbb I_{\{T<\tau_\infty\}}]
    \le\mathcal T_Tf(x).
\]
Thus $\mathbb P$ satisfies \eqref{eq:attaining_maximum_R_T}, proving the maximum representation for $t=0$ and $T>0$.

The construction also proves that $\mathcal R_x^{\mathcal T}$
is nonempty.
For $T=0$, the representation follows from the prescribed initial
state and $\mathcal T_0f=f$.
Time homogeneity gives the assertion for every $0\le t\le T$.
Together with the cemetery convention, this also proves that
$\mathcal R^{\mathcal T}$ is fiberwise nonempty.
\end{proof}

It remains to establish the stability of the recovered structure $\mathcal R^{\mathcal T}$.

\begin{proposition}
\label{prop:valuation_relation_substructures}
The recovered structure $\mathcal R^{\mathcal T}$ is stable.
In particular, $\mathcal R^{\mathcal T}\in\mathfrak S(G)$.
\end{proposition}
\begin{proof}
Corollary~\ref{prop:valuation_compatible_generator_substructure} gives
$\mathcal R^{\mathcal T}\subseteq\mathcal P(G)$.
The assertions for the constant cemetery law are immediate,
so we consider initial states in $D$.
We first prove stability under restarted conditioning.
Fix $\mathbb P\in\mathcal R_{t,x}^{\mathcal T}$ and a deterministic $a\ge t$.
By Proposition~\ref{prop:propertiesP_t,x_revised},
\[
    \mathbb P^{a,\omega}
    \in\mathcal P_{a,X_a(\omega)}(G)
    \qquad\text{for }\mathbb P\text{-almost every }\omega.
\]
By Proposition~\ref{prop:parameterized_closed_valuation_relations}\textup{(iii)} and Proposition~\ref{lem:countable_determining_cylinder_classes} in the Supplementary Material, together with Lemma~\ref{lem:continuous_testing_time_shifts}, we can choose a countable family
\[
    \mathscr I_a=\{(g_i,r_i,s_i):i\ge1\},
    \qquad
    g_i\in C_b(D),\quad a\le r_i\le s_i<\infty,
\]
and, for each $i$, a countable class $\mathscr H_{r_i,0}$ of bounded nonnegative $\widetilde{\mathcal F}_{r_i}$-measurable cylinder functions, such that the following characterization holds.
For every $y\in D$ and $\mathbb P'\in\mathcal P_{a,y}(G)$, membership in $\mathcal R_{a,y}^{\mathcal T}$ is equivalent to
\begin{equation}\label{eq:conditioning_selected_tests}
    \mathbb E^{\mathbb P'}
    [F g_i(X_{s_i})\mathbb I_{\{s_i<\tau_\infty\}}]
    \le
    \mathbb E^{\mathbb P'}
    [F\mathcal T_{s_i-r_i}g_i(X_{r_i})
       \mathbb I_{\{r_i<\tau_\infty\}}],
\end{equation}
for every $i\ge1$ and $F\in\mathscr H_{r_i,0}$.

Fix $i\ge1$ and $F\in\mathscr H_{r_i,0}$, and define $F^{[a]}(\omega):=F\bigl(\mathsf R_{a,X_a(\omega)}(\omega)\bigr)$.
Since $a\le r_i$, this function is bounded, nonnegative, and $\widetilde{\mathcal F}_{r_i}$-measurable.
Restarting at $a$ leaves the semigroup orbits unchanged at times
$r_i$ and $s_i$.
Thus, for every $B\in\widetilde{\mathcal F}_a$, we have
\[
    \mathbb E^{\mathbb P}\!\left[
        \mathbb I_B
        \mathbb E^{\mathbb P^{a,\omega}}\!\left[
            F\bigl(
                Y_{s_i}^{\mathcal T;s_i,g_i}
                -Y_{r_i}^{\mathcal T;s_i,g_i}
            \bigr)
        \right]
    \right]
    =
    \mathbb E^{\mathbb P}\!\left[
        \mathbb I_B F^{[a]}
        \bigl(
            Y_{s_i}^{\mathcal T;s_i,g_i}
            -Y_{r_i}^{\mathcal T;s_i,g_i}
        \bigr)
    \right]
    \le0.
\]
The last inequality follows from $\mathbb P\in\mathcal R_{t,x}^{\mathcal T}$, since $\mathbb I_BF^{[a]}$ is $\widetilde{\mathcal F}_{r_i}$-measurable.
As this holds for every $B\in\widetilde{\mathcal F}_a$, we obtain
\[
    \mathbb E^{\mathbb P^{a,\omega}}\!\left[
        F\bigl(
            Y_{s_i}^{\mathcal T;s_i,g_i}
            -Y_{r_i}^{\mathcal T;s_i,g_i}
        \bigr)
    \right]
    \le0
\]
for $\mathbb P$-almost every $\omega$.
There are only countably many pairs $(i,F)$ with $F\in\mathscr H_{r_i,0}$.
We may therefore choose a single $\mathbb P$-null set outside which $\mathbb P^{a,\omega}\in\mathcal P_{a,X_a(\omega)}(G)$ and \eqref{eq:conditioning_selected_tests} holds under $\mathbb P^{a,\omega}$ for every $i\ge1$ and $F\in\mathscr H_{r_i,0}$.
By the preceding characterization, it follows that
\[
    \mathbb P^{a,\omega}
    \in\mathcal R_{a,X_a(\omega)}^{\mathcal T}
    \qquad\text{whenever }X_a(\omega)\in D.
\]
When $X_a(\omega)=\triangle$, the restarted law is $\delta_{\mathbf\triangle}$, so the same conclusion holds.
This proves stability under restarted conditioning.

We next prove stability under restarted concatenation.
Let $\nu$ be a $\widetilde{\mathcal F}_a$-measurable, restart-compatible kernel such that $\nu(\omega)\in\mathcal R_{a,X_a(\omega)}^{\mathcal T}$ for $\mathbb P$-almost every $\omega$, and set
$\mathbb P':=\mathbb P\otimes_a\nu$.
By Proposition~\ref{prop:propertiesP_t,x_revised},
$\mathbb P'\in\mathcal P_{t,x}(G)$.
Fix $R\ge t$ and $g\in\USC_b(D)$, and write $Y:=Y^{\mathcal T;R,g}$.
For $t\le r\le s\le R$ and bounded nonnegative
$\widetilde{\mathcal F}_r$-measurable $F$, we verify
$\mathbb E^{\mathbb P'}[F(Y_s-Y_r)]\le0$.
When $s\le a$, this follows from the same inequality under $\mathbb P$,
since the two laws agree on $\widetilde{\mathcal F}_a$.
When $a\le r$, set $F^\omega(\eta):=F(\omega\otimes_a\eta)$ and
integrate $\mathbb E^{\nu(\omega)}[F^\omega(Y_s-Y_r)]\le0$
using \eqref{eq:concatenated_virtual_law}.
When $r<a<s$, the tower property gives
$\mathbb E^{\mathbb P'}[F(Y_s-Y_a)]\le0$ and
$\mathbb E^{\mathbb P'}[F(Y_a-Y_r)]\le0$; adding them proves the claim.
Thus $Y^{\mathcal T;R,g}$ is a $\mathbb P'$-supermartingale on $[t,R]$, proving $\mathbb P'\in\mathcal R_{t,x}^{\mathcal T}$.

Finally, we prove topological stability.
By Corollary~\ref{prop:valuation_compatible_generator_substructure},
$\mathcal R_{t,x}^{\mathcal T}
=\mathcal R_{t,x}^{\mathcal T}\cap\mathcal P_{t,x}(G)$
for every $(t,x)$.
Proposition~\ref{prop:parameterized_closed_valuation_relations}\textup{(ii)} in the Supplementary Material therefore shows that
$(t,x)\mapsto\mathcal R_{t,x}^{\mathcal T}$ has a closed graph
on $[0,T_0]\times D$ for every $T_0>0$.
Let $K\subset[0,T_0]\times D$ be compact.
By Proposition~\ref{prop:propertiesP_t,x_revised}, the set
\[
    \mathcal P_K(G)
    :=\bigcup_{(t,x)\in K}\mathcal P_{t,x}(G)
\]
is compact.
The set $\left\{(t,x,\mathbb P):(t,x)\in K,\ \mathbb P\in\mathcal R_{t,x}^{\mathcal T}\right\}$ is closed in the compact product $K\times\mathcal P_K(G)$
and hence is compact.
Its projection onto the law coordinate is
\[
    \bigcup_{(t,x)\in K}\mathcal R_{t,x}^{\mathcal T},
\]
which is therefore compact.
Each fiber is convex because the prescribed-history condition
and the inequalities \eqref{eq:primitive_valuation_compatibility} are preserved by
convex combinations of laws.
This proves topological stability and hence stability of
$\mathcal R^{\mathcal T}$.

Combining this with time homogeneity from
Proposition~\ref{prop:recovered_time_homogeneity},
fiberwise nonemptiness from
Proposition~\ref{prop:recovered_calibrated_nonemptiness},
and the inclusion $\mathcal R^{\mathcal T}\subseteq\mathcal P(G)$
gives $\mathcal R^{\mathcal T}\in\mathfrak S(G)$.
\end{proof}

Combining the preceding results identifies all semigroups dominated by
$\mathcal V^G$ with their representing stable structures.

\begin{theorem}
\label{thm:globally_dominated_structure_correspondence}
Let $G\in\mathfrak G_{\mathrm{Lyap}}$.
The semigroup map $\operatorname{Sem}$ is an order isomorphism between
$\mathfrak S(G)$ and $\mathfrak V_{\le}(G)$, with inverse
$\operatorname{Sem}^{-1}(\mathcal T)=\mathcal R^{\mathcal T}$.
Moreover, for every $\mathcal T\in\mathfrak V_{\le}(G)$,
$0\le t\le T$, $x\in D$, and $f\in\USC_b(D)$,
\begin{equation}
\label{eq:stable_recovery_representation}
    \mathcal T_{T-t}f(x)
    =
    \max_{\mathbb P\in\mathcal R_{t,x}^{\mathcal T}}
    \mathbb E^{\mathbb P}[f(X_T)\mathbb I_{\{T<\tau_\infty\}}].
\end{equation}
\end{theorem}

\begin{proof}
Fix $G\in\mathfrak G_{\mathrm{Lyap}}$.
Proposition~\ref{prop:stable_virtual_structure_generates_valuation}
and inclusion in $\mathcal P(G)$ give
$\operatorname{Sem}(\mathfrak S(G))\subseteq\mathfrak V_{\le}(G)$.
Conversely, for $\mathcal T\in\mathfrak V_{\le}(G)$,
Propositions~\ref{prop:recovered_time_homogeneity},
\ref{prop:recovered_calibrated_nonemptiness}, and~\ref{prop:valuation_relation_substructures}, together with
$\mathbb E^{\mathbb P}[Y_T^{\mathcal T;T,f}]
\le\mathbb E^{\mathbb P}[Y_t^{\mathcal T;T,f}]
=\mathcal T_{T-t}f(x)$ for
$\mathbb P\in\mathcal R_{t,x}^{\mathcal T}$, give
$\mathcal R^{\mathcal T}\in\mathfrak S(G)$ and
\eqref{eq:stable_recovery_representation}.
Thus $\operatorname{Sem}$ maps $\mathfrak S(G)$ onto
$\mathfrak V_{\le}(G)$, and
Proposition~\ref{thm:abstract_stable_recovery} supplies the order
isomorphism and inverse formula.
\end{proof}

\section{Local Specifications and Global Semigroups}
\label{sec:local_global_correspondence}
\label{sec:local_recovery}

Section~\ref{sec:valuation_orbit_nonemptiness} establishes the representation
of semigroups under global domination. We now characterize this domination
using local upper-generator bounds and study the extremal objects
determined by a local specification.

\subsection{Local constraints and global domination}
\label{subsec:local_constraints_global_domination}
\label{subsec:local_domination_bridge}

The following theorem identifies the globally dominated class using only
upper-generator bounds on smooth test functions. These bounds also imply
finiteness and locality of the upper generator.

\begin{theorem}
\label{cor:local_global_domination}
Let $G\in\mathfrak G_{\mathrm{Lyap}}$.  For every
$\mathcal T\in\mathfrak V$,
\begin{equation}
\label{eq:local_global_valuation_domination}
\begin{aligned}
    \mathcal T\in\mathfrak V_{\le}(G)
    &\quad\Longleftrightarrow\quad
    \overline{\mathcal G}^{\mathcal T}\varphi(x)
      \le G(x,J_x^2\varphi)
      \quad\text{for all }\varphi\in C_b^\infty(D),\ x\in D
    \\
    &\quad\Longleftrightarrow\quad
    \mathcal T\in\mathfrak V_{\mathrm{loc}}
    \ \text{and}\ G_{\mathcal T}\le G.
\end{aligned}
\end{equation}
\end{theorem}

We first establish four auxiliary lemmas for its proof.

\begin{lemma}
\label{prop:ambient_valuation_generator_bounds}
For $G\in\mathfrak G_{\mathrm{Lyap}}$ and $f\in C_b^\infty(D)$, $\overline{\mathcal G}^{\mathcal V^G}f\le G(\cdot,J^2f)$.
\end{lemma}
\begin{proof}
Fix $f\in C_b^\infty(D)$, $h_j\searrow0$, and $x_j\to x$.
Set $q_f:=G(\cdot,J^2f)$.
Choose $B_{2r}(x)\Subset D$ and set $\rho=\rho_{B_{2r}(x)}^{\,0}$.
For all sufficiently large $j$, choose $\mathbb P_j\in\mathcal P_{x_j}(G)$ attaining $\mathcal V_{h_j}^Gf(x_j)$. 
Choose $n$ such that $\overline{B_{2r}(x)}\subset D_n$.
For all sufficiently large $j$, optional sampling of the bounded
$D_n$-stopped test supermartingale at $h_j\wedge\rho$ gives
\begin{equation}
\label{eq:canonical_generator_local_supermartingale_bound}
 \mathbb E^{\mathbb P_j}[f(X_{h_j\wedge\rho})\mathbb I_{\{h_j\wedge\rho<\tau_\infty\}}]-f(x_j)
 \le\mathbb E^{\mathbb P_j}
       \left[\int_0^{h_j\wedge\rho}q_f(X_s)\mathbb I_{\{s<\tau_\infty\}}\,ds\right].
\end{equation}
The live-exit estimate in Proposition~\ref{prop:short_time_control_bounded_coefficients}\textup{(ii)} of the
Supplementary Material implies
\begin{equation}
\label{eq:canonical_generator_stopping_error}
 \mathcal V_{h_j}^Gf(x_j)
 -\mathbb E^{\mathbb P_j}[f(X_{h_j\wedge\rho})\mathbb I_{\{h_j\wedge\rho<\tau_\infty\}}]
 =o(h_j).
\end{equation}
For $\varepsilon>0$, choose $r_\varepsilon\in(0,r]$ so that
$q_f(y)\le q_f(x)+\varepsilon$ on $B_{r_\varepsilon}(x)$,
keeping $\rho$ fixed.
Local boundedness of $q_f$ and the killing and oscillation estimates in
Proposition~\ref{prop:short_time_control_bounded_coefficients}\textup{(i)} of the Supplementary Material yield
\[
 \frac1{h_j}\mathbb E^{\mathbb P_j}
       \left[\int_0^{h_j\wedge\rho}q_f(X_s)\mathbb I_{\{s<\tau_\infty\}}\,ds\right]
 \le q_f(x)+\varepsilon+o(1).
\]
Letting $\varepsilon\searrow0$ proves $\overline{\mathcal G}^{\mathcal V^G}f\le q_f$.
\end{proof}

\begin{lemma}
\label{prop:upper_domination_jet_structure}
Let $G\in\mathfrak G$, and let
$\mathcal T\in\mathfrak V$ satisfy
$\overline{\mathcal G}^{\mathcal T}f(x)\le G(x,J_x^2f)$ for every
$f\in C_b^\infty(D)$ and $x\in D$.
Then $\mathcal T$ satisfies Assumption~\ref{assume:local_upper_generator}.
Its upper generating function satisfies
$G_{\mathcal T}\in\mathfrak G$ and $G_{\mathcal T}\le G$.
Moreover, the jet-Lipschitz bound
\eqref{eq:local_uniform_jet_lipschitz_G} holds with the same local
constants $L_K$ as in \eqref{eq:local_uniform_jet_lipschitz_G} for $G$.

\end{lemma}

\begin{proof}
Applying the assumed upper bound to both $f$ and $-f$, and using
$\Delta_h^{\mathcal T}f+\Delta_h^{\mathcal T}(-f)\ge0$, yields
\[
    -G(x,-J_x^2f)
    \le\liminf_{h\searrow0,*}\Delta_h^{\mathcal T}f(x)
    \le \overline{\mathcal G}^{\mathcal T}f(x)\le G(x,J_x^2f).
\]
Thus $\overline{\mathcal G}^{\mathcal T}$ is finite and sublinear on the smooth core $C_b^\infty(D)$.
If $J_x^2r=0$, the upper bound gives $\overline{\mathcal G}^{\mathcal T}(\pm r)(x)\le G(x,0)=0$.
Their sum is nonnegative, so both values are zero.
Sublinearity now implies $\overline{\mathcal G}^{\mathcal T}f(x)=\overline{\mathcal G}^{\mathcal T}g(x)$ whenever $J_x^2f=J_x^2g$; in particular, $\overline{\mathcal G}^{\mathcal T}$ is local.
Proposition~\ref{prop:local_upper_generator_representation} supplies $G_{\mathcal T}\in\mathfrak G$.
Realizing each jet by a smooth function gives $G_{\mathcal T}\le G$.
Finally, for $x\in K$,
\[
    G_{\mathcal T}(x,U)-G_{\mathcal T}(x,V)
    \le G_{\mathcal T}(x,U-V)
    \le G(x,U-V)\le L_K\|U-V\|.
\]
Interchanging $U,V$ proves the stated uniform estimate.
\end{proof}

The preceding two lemmas yield the following corollary.

\begin{corollary}\label{cor:V_G_is_local_Lyapunov}
If $G\in\mathfrak G_{\mathrm{Lyap}}$, then $\mathcal V^G\in\mathfrak V_{\mathrm{loc,Lyap}}$ and $G_{\mathcal V^G}\in\mathfrak G_{\mathrm{Lyap}}$.
\end{corollary}

\begin{proof}
Lemmas~\ref{prop:ambient_valuation_generator_bounds} and~\ref{prop:upper_domination_jet_structure} imply that $\mathcal V^G\in\mathfrak V_{\mathrm{loc}}$, $G_{\mathcal V^G}\in\mathfrak G$, and $G_{\mathcal V^G}\le G$.
The last inequality ensures that $G_{\mathcal V^G}$ satisfies the same Lyapunov bound as $G$, proving the claim.
\end{proof}

For $\mathcal T\in\mathfrak V$ and $f\in\USC_b(D)$, define the semigroup orbit by
\begin{equation}
\label{eqn:PDE}
    v_f^{\mathcal T}(t,x):=\mathcal T_tf(x),
    \qquad (t,x)\in[0,\infty)\times D.
\end{equation}
The following lemma shows that the regularity and relaxed initial conditions of this orbit follow directly from the semigroup axioms.

\begin{lemma}
\label{lemma:valuation_initial_continuity}
Let $\mathcal T\in\mathfrak V$.
For every $f\in\USC_b(D)$, $v_f^{\mathcal T}\in\USC_b([0,\infty)\times D)$ and $v_f^{\mathcal T}$ has initial datum $f$ in the relaxed sense.  
Moreover, for every $r\in C_b(D)$,
\begin{equation}
\label{eq:continuous_payoff_zero_time_convergence}
    \mathcal T_hr\longrightarrow r
    \quad\text{locally uniformly as }h\searrow0.
\end{equation}
\end{lemma}
\begin{proof}
Conditions~\ref{V2} and~\ref{V4} imply boundedness and upper semicontinuity of $v_f^{\mathcal T}$.
For $r\in C_b(D)$, sublinearity and \ref{V4} at zero give
\[
 r\le-\limsup_{h\searrow0}^{*}\mathcal T_h(-r)
 \le\liminf_{h\searrow0,*}\mathcal T_hr
 \le\limsup_{h\searrow0}^{*}\mathcal T_hr\le r.
\]
Equality of the two half-relaxed limits with the continuous function $r$
implies \eqref{eq:continuous_payoff_zero_time_convergence}.
The relaxed upper initial condition follows from upper semicontinuity
and $v_f^{\mathcal T}(0,\cdot)=f$.

It remains to prove \eqref{eq:relaxed_supersolution_initial_condition}.
For any $g\in C_b(D)$ with $g\le f$, \ref{V1} give
$-v_{-g}^{\mathcal T}\le v_g^{\mathcal T}\le v_f^{\mathcal T}$.
Since $-v_{-g}^{\mathcal T}$ is lower semicontinuous, it is also bounded above by $(v_f^{\mathcal T})_*$.
Taking the lower limit at $(0,x)$ and using
\eqref{eq:continuous_payoff_zero_time_convergence} gives a lower bound
of $g(x)$.  The supremum of $g(x)$ over all bounded continuous minorants
$g$ of $f$ equals $f_*(x)$, proving
\eqref{eq:relaxed_supersolution_initial_condition}.
\end{proof}

The final lemma shows that the semigroup orbit is a viscosity subsolution of the nonlinear parabolic equation \eqref{eq:mainHJB} associated with a local specification $G$ that bounds the semigroup's upper generating function from above.

\begin{lemma}
\label{prop:PDErepn}
Let $G\in\mathfrak G$, and let
$\mathcal T\in\mathfrak V$ satisfy
\begin{equation}
\label{eq:upper_generator_domination}
    \overline{\mathcal G}^{\mathcal T}\varphi(x)
    \le G(x,J_x^2\varphi),
    \qquad \varphi\in C_b^\infty(D),\quad x\in D.
\end{equation}
For every $f\in\USC_b(D)$, the orbit $v_f^{\mathcal T}$ is a viscosity subsolution of \eqref{eq:mainHJB} and has initial datum $f$ in the relaxed sense.
In particular, $v_f^{\mathcal T}\in\operatorname{Sub}_G^+(f)$.
\end{lemma}
\begin{proof}
Let $\psi\in C_b^\infty([0,\infty)\times D)$ touch $v_f$ globally from above
at $(t_0,x_0)$, where $t_0>0$.
Choose $h_n\searrow0$ with $h_n<t_0$, and set
\[
    \varphi=\psi(t_0,\cdot),\qquad
    q_n=\psi(t_0-h_n,\cdot),\qquad
    r=-\partial_t\psi(t_0,\cdot).
\]
Then $\varphi,r,q_n\in C_b(D)$, and the bounded derivatives of $\psi$
imply
\begin{equation}
\label{eq:moving_payoff_first_order_expansion}
    e_n:=q_n-\varphi-h_nr,
    \qquad \|e_n\|_\infty=o(h_n).
\end{equation}
For any $x\in D$, properties~\ref{V1} and~\ref{V2} give
\begin{equation}
\label{eq:moving_payoff_relaxed_generator_sandwich}
\begin{aligned}
    \Delta_{h_n}^{\mathcal T}q_n(x)
    &\le \Delta_{h_n}^{\mathcal T}\varphi(x)
       +\mathcal T_{h_n}r(x)-r(x)
       +2\|e_n\|_\infty/h_n,\\
    \Delta_{h_n}^{\mathcal T}q_n(x)
    &\ge \Delta_{h_n}^{\mathcal T}\varphi(x)
       -\mathcal T_{h_n}(-r)(x)-r(x)
       -2\|e_n\|_\infty/h_n.
\end{aligned}
\end{equation}
By Lemma~\ref{lemma:valuation_initial_continuity} and the upper bound in \eqref{eq:moving_payoff_relaxed_generator_sandwich}, applied with $x=x_0$, and \eqref{eq:upper_generator_domination}, we obtain
\begin{align}\label{eq:PDErepn_keyineq_1}
    \limsup_{n\to\infty}
        \Delta_{h_n}^{\mathcal T}q_n(x_0)
    \le G(x_0,J_{x_0}^2\varphi).
\end{align}
At the contact point $(t_0,x_0)$, we have
$\psi(t_0,x_0)=v_f^{\mathcal T}(t_0,x_0)$ and, for all sufficiently large $n$,
\[
    q_n=\psi(t_0-h_n,\cdot)
    \ge v_f^{\mathcal T}(t_0-h_n,\cdot)
    =\mathcal T_{t_0-h_n}f.
\]
Monotonicity and \ref{V5} therefore give
\[
    \mathcal T_{h_n}q_n(x_0)
    \ge
    \mathcal T_{h_n}\bigl(\mathcal T_{t_0-h_n}f\bigr)(x_0)
    =\mathcal T_{t_0}f(x_0)
    =\psi(t_0,x_0).
\]
Subtracting $q_n(x_0)=\psi(t_0-h_n,x_0)$ and dividing by $h_n$ yields
\begin{align}\label{eq:PDErepn_keyineq_2}
    \Delta_{h_n}^{\mathcal T}q_n(x_0)
    \ge
    \frac{\psi(t_0,x_0)-\psi(t_0-h_n,x_0)}{h_n}
    \longrightarrow \partial_t\psi(t_0,x_0).
\end{align}
Combining \eqref{eq:PDErepn_keyineq_1} and \eqref{eq:PDErepn_keyineq_2} yields \eqref{eq:viscosity_subsolution_inequality} with $q=0$.
The regularity and initial conditions follow from Lemma~\ref{lemma:valuation_initial_continuity}.
\end{proof}

We now prove Theorem~\ref{cor:local_global_domination}.

\begin{proof}[Proof of Theorem~\ref{cor:local_global_domination}]
If $\mathcal T\in\mathfrak V_{\le}(G)$, by Lemma~\ref{prop:ambient_valuation_generator_bounds}, we have
\[
    \overline{\mathcal G}^{\mathcal T}\varphi
    \le\overline{\mathcal G}^{\mathcal V^G}\varphi
    \le G(\cdot,J^2\varphi),\qquad \varphi\in C_b^\infty(D).
\]
Conversely, suppose the upper-generator inequality holds.  For
$f\in\USC_b(D)$, Lemma~\ref{prop:PDErepn} shows that
$v_f^{\mathcal T}(s,y):=\mathcal T_s f(y)$ is a bounded upper semicontinuous viscosity subsolution with $v_f^{\mathcal T}(0,\cdot)=f$.
For $t>0$ and $x\in D$, Corollary~\ref{cor:usc_terminal_realization} therefore provides $\mathbb P\in\mathcal P_x(G)$ such that
\[
    \mathcal T_t f(x)
    \le\mathbb E^{\mathbb P}[f(X_t)\mathbb I_{\{t<\tau_\infty\}}]
    \le\mathcal V_t^Gf(x).
\]
When $t=0$, it is immediate to see $\mathcal T_0f=\mathcal V_0^Gf=f$, proving first equivalence.
Finally, Lemma~\ref{prop:upper_domination_jet_structure} shows that the upper-generator bound implies $\mathcal T\in\mathfrak V_{\mathrm{loc}}$ and $G_{\mathcal T}\le G$; the converse follows from the jet representation of the upper generator.
\end{proof}

In particular,
\begin{equation}
\label{eq:local_lyapunov_union_dominated_classes}
 \mathfrak V_{\mathrm{loc,Lyap}}
 =\bigcup_{G\in\mathfrak G_{\mathrm{Lyap}}}\mathfrak V_{\le}(G).
\end{equation}
Indeed, $\mathcal T\in\mathfrak V_{\le}(G)$ implies
$G_{\mathcal T}\le G$, so $G_{\mathcal T}$ satisfies the same Lyapunov
bound as $G$. Conversely, for
$\mathcal T\in\mathfrak V_{\mathrm{loc,Lyap}}$, apply
Theorem~\ref{cor:local_global_domination} with $G=G_{\mathcal T}$.

Combining this characterization with
Theorem~\ref{thm:globally_dominated_structure_correspondence} gives the
correspondence between the local Lyapunov classes.

\begin{corollary}
\label{thm:local_lyapunov_valuation_structure_correspondence}
The semigroup map $\operatorname{Sem}$ is an order isomorphism between
$\mathfrak S_{\mathrm{loc,Lyap}}$ and
$\mathfrak V_{\mathrm{loc,Lyap}}$, with inverse
$\operatorname{Sem}^{-1}(\mathcal T)=\mathcal R^{\mathcal T}$.
The attained representation \eqref{eq:stable_recovery_representation}
holds for every $\mathcal T\in\mathfrak V_{\mathrm{loc,Lyap}}$.
\end{corollary}

\begin{proof}
Theorem~\ref{cor:local_global_domination} gives the union identity
\eqref{eq:local_lyapunov_union_dominated_classes}.
Together with \eqref{eq:local_lyapunov_structure_class},
Theorem~\ref{thm:globally_dominated_structure_correspondence} therefore
shows that $\operatorname{Sem}$ maps
$\mathfrak S_{\mathrm{loc,Lyap}}$ onto
$\mathfrak V_{\mathrm{loc,Lyap}}$, with the common inverse
$\mathcal T\mapsto\mathcal R^{\mathcal T}$.
Proposition~\ref{thm:abstract_stable_recovery} supplies injectivity
and order reflection, while preservation of order follows directly
from the upper-expectation formula.
For each $\mathcal T\in\mathfrak V_{\mathrm{loc,Lyap}}$, choosing
$G$ such that $\mathcal T\in\mathfrak V_{\le}(G)$ also yields
\eqref{eq:stable_recovery_representation}.
\end{proof}

We now express the local constraints through localization and globalization. Proposition~\ref{prop:local_upper_generator_representation}
defines the \emph{localization map}
\begin{equation}
\label{eq:localization_mapping}
    \Loc:\mathfrak V_{\mathrm{loc}}\longrightarrow\mathfrak G,
    \qquad \Loc(\mathcal T):=G_{\mathcal T}.
\end{equation}
We call $G\in\mathfrak G$ \emph{globally realizable} if
$G\in\Loc(\mathfrak V_{\mathrm{loc}})$.
Conversely, Proposition~\ref{prop:propertiesP_t,x_revised}
and Corollary~\ref{cor:canonical_regular_dynamic_virtual_structure}
ensure that $\mathcal P(G)\in\mathfrak S(G)$ for every
$G\in\mathfrak G_{\mathrm{Lyap}}$.
This defines the \emph{globalization map}
\begin{equation}
\label{eq:globalization_mapping_definition}
    \Glob:\mathfrak G_{\mathrm{Lyap}}
       \longrightarrow\mathfrak S_{\mathrm{loc,Lyap}},
    \qquad \Glob(G):=\mathcal P(G),
\end{equation}
with associated semigroup
$\mathcal V^G=\operatorname{Sem}(\Glob(G))$.

Both maps preserve order on their domains. For $\Loc$, this follows
from monotonicity of the upper half-relaxed limit and the realization of
every jet by a bounded smooth function. If $H\le G$, the process in
\eqref{eq:defM_f,nonline} is obtained from its counterpart for $H$ by
subtracting the nondecreasing integral
$\int_t^{[s]_{t,\tau_n}}(G-H)(X_r,J_{X_r}^2f)
\mathbb I_{\{r<\tau_\infty\}}\,dr$.
It is therefore a $\mathbb P$-supermartingale for every
$\mathbb P\in\mathcal P_{t,x}(H)$, giving
$\mathcal P(H)\subseteq\mathcal P(G)$; hence $\Glob$ also preserves order.
Moreover, any $H\in\mathfrak G$ below
$G\in\mathfrak G_{\mathrm{Lyap}}$ satisfies the same Lyapunov bound.
The following result identifies exactly how a local upper bound
constrains a semigroup and its recovered structure.

\begin{corollary}
\label{cor:local_global_structure_order}
For $\mathcal T\in\mathfrak V_{\mathrm{loc,Lyap}}$ and
$G\in\mathfrak G_{\mathrm{Lyap}}$,
\begin{equation}
\label{eq:local_global_recovered_structure_order}
    \Loc(\mathcal T)\le G
    \quad\Longleftrightarrow\quad
       \mathcal T\le\mathcal V^G
    \quad\Longleftrightarrow\quad
       \mathcal R^{\mathcal T}\subseteq\mathcal P(G).
\end{equation}
Equivalently, for $\mathcal C\in\mathfrak S_{\mathrm{loc,Lyap}}$,
\begin{equation}
\label{eq:local_global_structure_order}
    \Loc(\operatorname{Sem}(\mathcal C))\le G
    \quad\Longleftrightarrow\quad
    \mathcal C\subseteq\Glob(G).
\end{equation}
\end{corollary}

\begin{proof}
The first equivalence in
\eqref{eq:local_global_recovered_structure_order} is
Theorem~\ref{cor:local_global_domination}.
The second follows from the order isomorphism in
Corollary~\ref{thm:local_lyapunov_valuation_structure_correspondence}, since
$\operatorname{Sem}(\mathcal R^{\mathcal T})=\mathcal T$.
Applying this equivalence to $\mathcal T=\operatorname{Sem}(\mathcal C)$
and using the inverse formula \eqref{eq:abstract_stable_recovery}
gives \eqref{eq:local_global_structure_order}.
\end{proof}

Thus $G$ determines the greatest semigroup whose upper generating
function is bounded by $G$, namely $\mathcal V^G$, and the full structure
$\mathcal P(G)$ containing every corresponding recovered structure.
This conclusion requires neither uniqueness of the global evolution
nor global realizability of $G$. The next two subsections examine what
happens when we compose localization and globalization in either order.

\subsection{Global information lost under localization}
\label{subsec:global_information_localization}

Start with a semigroup $\mathcal T$ and retain only its local specification $G_{\mathcal T}$. 
Globalization then includes all laws compatible with this specification, and taking upper expectations yields $\mathcal V^{G_{\mathcal T}}$. 
The next result shows that this operation preserves the upper generating function while producing the greatest semigroup with that function.
It also describes the corresponding enlargement of the recovered structure.

\begin{theorem}
\label{thm:maximal_valuation_same_specification}
For every $\mathcal T\in\mathfrak V_{\mathrm{loc,Lyap}}$,
\begin{equation}
\label{eq:greatest_valuation_same_upper_generator}
    \mathcal V^{G_{\mathcal T}}
    =\max\{\mathcal S\in\mathfrak V_{\mathrm{loc}}:
                    G_{\mathcal S}=G_{\mathcal T}\}.
\end{equation}
Its recovered structure is
$\mathcal R^{\mathcal V^{G_{\mathcal T}}}=\mathcal P(G_{\mathcal T})$,
and
\begin{equation}
\label{eq:smallest_full_local_family}
    \mathcal P(G_{\mathcal T})
    =\min\{\mathcal P(H):H\in\mathfrak G_{\mathrm{Lyap}},
              \ \mathcal R^{\mathcal T}\subseteq\mathcal P(H)\},
\end{equation}
where the minimum is taken in fiberwise inclusion.
\end{theorem}

\begin{proof}
Write $G=G_{\mathcal T}$. The order correspondence gives
$\mathcal T\le\mathcal V^G$. Since $\Loc$ preserves order and
Lemma~\ref{prop:ambient_valuation_generator_bounds} gives
$G_{\mathcal V^G}\le G$, we obtain
\[
    G=\Loc(\mathcal T)
    \le\Loc(\mathcal V^G)\le G.
\]
Consequently, $\mathcal V^G$ has the same upper generating function as $\mathcal T$ and belongs to the set in \eqref{eq:greatest_valuation_same_upper_generator}.
Every $\mathcal S$ in that set has local specification $G$ and hence
belongs to $\mathfrak V_{\mathrm{loc,Lyap}}$.
By Corollary~\ref{cor:local_global_structure_order}, we have $\mathcal S\le\mathcal V^G$, proving \eqref{eq:greatest_valuation_same_upper_generator}.

The identity $\mathcal R^{\mathcal V^G}=\mathcal P(G)$ follows from the inverse formula in Theorem~\ref{thm:globally_dominated_structure_correspondence}.
Moreover, \eqref{eq:local_global_recovered_structure_order} shows that $\mathcal P(G)$ belongs to the set in \eqref{eq:smallest_full_local_family}.
If $\mathcal R^{\mathcal T}\subseteq\mathcal P(H)$, the same equivalence gives $G\le H$, and hence $\mathcal P(G)\subseteq\mathcal P(H)$ by monotonicity of $\Glob$.
\end{proof}

The semigroup $\mathcal T$ is therefore recovered from its local specification precisely when it is the greatest semigroup with that specification.
In general, localization can lose information about which compatible laws belong to $\mathcal R^{\mathcal T}$.
The following diffusion example exhibits this loss even for a continuous
linear specification and two Feller semigroups.

\begin{example}[Distinct semigroups with the same local specification]
\label{ex:singular_diffusion_multiple_branches}
\label{ex:strict_branch_canonical_separation}
Let $D=\mathbb R$, $\alpha\in(0,\frac12)$, and
\[
 \sigma_\alpha(x)=\frac{|x|^\alpha}{1+|x|^\alpha},\qquad
 a_\alpha=\sigma_\alpha^2,\qquad
 G_\alpha(x,z,p,H)=\tfrac12a_\alpha(x)H.
\]
This continuous specification belongs to $\mathfrak G_{\mathrm{Lyap}}$:
for $\phi(x)=1+x^2$, $G_\alpha(x,J_x^2\phi)=a_\alpha(x)\le\phi(x)$.
The equation
\begin{equation}
\label{eq:singular_diffusion_example_sde}
 dX_t=\sigma_\alpha(X_t)\,dW_t
\end{equation}
admits an absorbing Feller family $\mathbb Q_x^{\mathrm{abs}}$ and a
regular zero-sojourn Feller family $\mathbb Q_x^0$; see
Girsanov~\cite[Section~3, pp.~327--330]{girsanov1962nonuniqueness}.
Their transition expectations $\mathcal T_t^\bullet f(x)
=\mathbb E^{\mathbb Q_x^\bullet}[f(X_t)]$ define semigroups on
$\USC_b(\mathbb R)$: joint weak continuity gives \ref{V4}, and the
other axioms follow from linearity, monotone convergence, and the Markov
property. It\^o's formula and the uniform short-time oscillation estimate
in Proposition~\ref{prop:short_time_control_bounded_coefficients}\textup{(i)} of the Supplementary Material give, for $f\in C_b^\infty(\mathbb R)$,
\[
 \Delta_h^{\mathcal T^\bullet}f(x)
 =\frac1h\mathbb E^{\mathbb Q_x^\bullet}
       \int_0^h\tfrac12a_\alpha(X_s)f''(X_s)\,ds
 \longrightarrow\tfrac12a_\alpha(x)f''(x)
\]
locally uniformly. In particular, both semigroups have upper generating
function $G_\alpha$.

They are distinct: the absorbing law stays at zero, whereas the
zero-sojourn property
\cite[Definition~5.1 and Theorems~5.4--5.5]{engelbert1985solutions}
and Fubini's theorem give a $t_0>0$ with
$\mathbb Q_0^0(X_{t_0}\ne0)=1$. For $f_+(x)=x^2/(1+x^2)$,
\[
 0=\mathcal T_{t_0}^{\mathrm{abs}}f_+(0)
   <\mathcal T_{t_0}^0f_+(0)
   \le\mathcal V_{t_0}^{G_\alpha}f_+(0).
\]
Together with Corollary~\ref{cor:local_global_structure_order}, these
observations give the explicit relations
\begin{equation}
\label{eq:singular_diffusion_same_specification_strict_domination}
 \begin{gathered}
  G_{\mathcal T^{\mathrm{abs}}}
    =G_{\mathcal T^0}=G_\alpha,\\
  \mathcal T^{\mathrm{abs}}\le\mathcal V^{G_\alpha},
  \qquad
  \mathcal T^{\mathrm{abs}}\ne\mathcal V^{G_\alpha}.
 \end{gathered}
\end{equation}
Thus $(\operatorname{Sem}\circ\Glob\circ\Loc)
(\mathcal T^{\mathrm{abs}})\ne\mathcal T^{\mathrm{abs}}$:
the local specification identifies the canonical greatest semigroup,
but does not recover the original absorbing branch.
\end{example}

\subsection{The effective local specification}
\label{subsec:effective_local_specification}

We now start with a specification $G\in\mathfrak G_{\mathrm{Lyap}}$.
Its full structure $\mathcal P(G)$ generates the semigroup $\mathcal V^G$,
whose local specification is
\[
    G^\downarrow:=G_{\mathcal V^G}
       =(\Loc\circ\operatorname{Sem}\circ\Glob)(G).
\]
The upper generator bound gives $G^\downarrow\le G$.
The next result shows that this reduction preserves all admissible laws
and extracts the least specification defining them.

\begin{theorem}
\label{thm:realizable_specification_reduction}
For every $G\in\mathfrak G_{\mathrm{Lyap}}$, the function $G^\downarrow$ belongs to $\mathfrak G_{\mathrm{Lyap}}$ and satisfies
\begin{equation}
\label{eq:greatest_realizable_minorant}
    G^\downarrow
    =\max\{H\in\Loc(\mathfrak V_{\mathrm{loc}}):H\le G\}.
\end{equation}
Moreover,
\begin{equation}
\label{eq:ambient_equals_own_recovered_structure}
    \mathcal R^{\mathcal V^G}
       =\mathcal P(G)=\mathcal P(G^\downarrow),
\end{equation}
and
\begin{equation}
\label{eq:least_law_equivalent_specification}
    G^\downarrow
    =\min\{H\in\mathfrak G_{\mathrm{Lyap}}:
                    \mathcal P(H)=\mathcal P(G)\}.
\end{equation}
In particular, $G^\downarrow=G$ if and only if $G$ is globally realizable.
\end{theorem}

\begin{proof}
By Lemmas~\ref{prop:ambient_valuation_generator_bounds} and~\ref{prop:upper_domination_jet_structure}, $\mathcal V^G$ belongs to $\mathfrak V_{\mathrm{loc}}$ and satisfies $G_{\mathcal V^G}\le G$ .
Thus $G^\downarrow$ satisfies the same Lyapunov bound as $G$ and
belongs to the set in \eqref{eq:greatest_realizable_minorant}.
For any $H$ in that set, choose $\mathcal T\in\mathfrak V_{\mathrm{loc}}$ with $\Loc(\mathcal T)=H$.
Since $H\le G$, by Corollary~\ref{cor:local_global_structure_order}, we have $\mathcal T\le\mathcal V^G$.
Applying the order-preserving map $\Loc$ yields $H\le G^\downarrow$, proving \eqref{eq:greatest_realizable_minorant}.

The first equality in \eqref{eq:ambient_equals_own_recovered_structure} follows from the inverse formula in Theorem~\ref{thm:globally_dominated_structure_correspondence}, applied to $\mathcal C=\mathcal P(G)$.
Since $G^\downarrow\le G$, monotonicity gives $\mathcal P(G^\downarrow)\subseteq\mathcal P(G)$.
For the reverse inclusion, apply \eqref{eq:local_global_structure_order} with $\mathcal C=\mathcal P(G)$ and local bound $G^\downarrow$, observing that $\Loc(\operatorname{Sem}(\mathcal P(G)))=G^\downarrow$.

It follows that $G^\downarrow$ belongs to the set in \eqref{eq:least_law_equivalent_specification}. 
If $\mathcal P(H)=\mathcal P(G)$, then $\mathcal V^H=\mathcal V^G$ and $G^\downarrow=G_{\mathcal V^H}\le H$, proving the minimality assertion.
Finally, \eqref{eq:greatest_realizable_minorant} gives $G^\downarrow=G$ exactly when $G$ is globally realizable.
\end{proof}

The two extremal descriptions concern different classes:
$G^\downarrow$ is the greatest globally realizable specification below
$G$, and the least specification defining the same admissible laws as
$G$. Any difference between $G$ and $G^\downarrow$ therefore leaves
the structure unchanged. The following example exhibits such a
difference through an additional drift allowance at a single state.

\begin{example}[An unrealizable local specification]
\label{ex:unrealizable_drift_spike}
On $D=\mathbb R$, let
\[
 G(x,z,p,H)=p+\mathbb I_{\{0\}}(x)p^+.
\]
Then $G\in\mathfrak G_{\mathrm{Lyap}}$, with
$\phi(x)=1+x^2$ and $G(x,J_x^2\phi)=2x\le\phi(x)$.
Characteristic recovery
(Proposition~\ref{prop:recovery_aggregation_generator_compatible_characteristics})
shows that every compatible law has zero killing and covariance, and
paths $X_t=x+\int_0^t b_s\,ds$, where $b_s=1$ off zero and
$b_s\in[1,2]$ at zero. Such a path is strictly increasing, so it spends
zero Lebesgue time at zero and $X_t=x+t$. Conversely, this path is
compatible. Thus
\[
 \mathcal P_x(G)=\{\delta_{(x+t)_{t\ge0}}\},\qquad
 \mathcal V_t^Gf(x)=f(x+t),\qquad G^\downarrow(x,z,p,H)=p.
\]
The last function is strictly smaller than $G$ at $x=0$, $p>0$.
By Theorem~\ref{thm:realizable_specification_reduction}, $G$ is not
the upper generating function of any semigroup.
\end{example}

\section{Characterization of the Canonical Semigroup}
\label{sec:exact_generation_comparison}

Section~\ref{sec:local_global_correspondence} identifies
$G^\downarrow$ as the upper generating function of the canonical
semigroup $\mathcal V^G$. We first strengthen this identification:
$G^\downarrow$ also describes its lower half-relaxed infinitesimal
behavior and is therefore an exact generating function of
$\mathcal V^G$. Exact generation by the original specification $G$
then amounts to the equality $G^\downarrow=G$.
We next characterize the semigroup orbits as upper Perron envelopes
and transfer uniqueness of discontinuous viscosity solutions to
semigroup uniqueness. Under comparison, the canonical semigroup also
has the $C_b$-Feller property.

\subsection{Exact generation of the canonical semigroup}
\label{subsec:exact_generation}
\label{subsec:ambient_lower_consistency}

The upper generating function specifies the upper half-relaxed limit
of the difference quotients. We now require the lower half-relaxed
limit to agree with the lower semicontinuous envelope of the same
specification.

\begin{definition}
\label{def:relaxed_infinitesimal_generator}
A local specification $G$ is an \emph{exact generating function} of
$\mathcal T\in\mathfrak V$ if, for every $f\in C_b^\infty(D)$,
\begin{equation}
\label{eq:relaxed_generator_consistency}
 \limsup_{h\searrow0}^{*}\Delta_h^{\mathcal T}f=G(\cdot,J^2f),
 \qquad
 \liminf_{h\searrow0,*}\Delta_h^{\mathcal T}f
       =\bigl(G(\cdot,J^2f)\bigr)_*.
\end{equation}
\end{definition}

For continuous $G(\cdot,J^2f)$ these identities give locally uniform convergence.
The following theorem shows that the canonical semigroup always has this stronger consistency with its effective specification $G^\downarrow$. 
It also identifies the lower limit directly in terms of the original specification $G$.

\begin{theorem}
\label{thm:canonical_exact_reduced_generator}
Let $G\in\mathfrak G_{\mathrm{Lyap}}$.
Then $G^\downarrow$ is an exact generating function of $\mathcal V^G$.
Moreover,
\begin{align}\label{eq:thm:canonical_exact_reduced_generator_1}
    G_*=G^\downarrow_*, \qquad (G_*)^*\le G^\downarrow\le G.
\end{align}
\end{theorem}

\begin{proof}
Fix $f\in C_b^\infty(D)$ and write $q_f:=G(\cdot,J^2f)$ and $q_f^\downarrow:=G^\downarrow(\cdot,J^2f)$.
Set
\[
    L_f:=\liminf_{h\searrow0,*}\Delta_h^{\mathcal V^G}f,
    \qquad
    U_f:=\overline{\mathcal G}^{\mathcal V^G}f=q_f^\downarrow.
\]
By Lemma~\ref{prop:ambient_valuation_generator_bounds}, we have $L_f\le U_f\le q_f$.
The main step is to prove the reverse lower bound $L_f\ge(q_f)_*$.
Fix $x\in D$, choose $B_{2r}(x)\Subset D$, and set $\rho:=\rho_{B_{2r}(x)}^{\,0}$.
We first reduce to compactly supported test functions.
Choose $\widetilde f\in C_c^\infty(D)$ that agrees with $f$ on a neighborhood of $\overline{B_{2r}(x)}$.
The killed terminal functionals associated with $f$ and $\widetilde f$ can differ
only if the path exits the ball while alive before time $h$.
Consequently,
\[
 |\mathcal V_h^Gf(y)-\mathcal V_h^G\widetilde f(y)|\le
 \|f-\widetilde f\|_\infty
 \sup_{\mathbb P\in\mathcal P_y(G)}
 \mathbb P(\rho\le h,\ \rho<\tau_{\mathrm{kill}})
 =o(h),
\]
uniformly for $y\in B_{r/2}(x)$, by Proposition~\ref{prop:short_time_control_bounded_coefficients}\textup{(ii)} of the Supplementary Material.
Since $f=\widetilde f$ near $x$, this implies
$L_f(x)=L_{\widetilde f}(x)$.
Moreover, $q_f=q_{\widetilde f}$ near $x$, so their lower
semicontinuous envelopes agree at $x$.
It therefore suffices to prove the lower bound for
$f\in C_c^\infty(D)$.

For such $f$, the function $q_f$ is bounded on $D$.
Define
\[
    r_m(y):=\inf_{z\in D}\{q_f(z)+m|y-z|\},
    \qquad m\ge1.
\]
Each $r_m$ is bounded and Lipschitz continuous, and
$r_m\uparrow(q_f)_*$ with $r_m\le q_f$.
The time-independent function $u(s,z)=f(z)$ satisfies
\[
    \partial_su-G(z,J_z^2u)+r_m(z)
    =-q_f(z)+r_m(z)\le0.
\]
Thus $u$ is a viscosity subsolution of
\eqref{eq:usc_subsolution_with_source} with source $-r_m$.
For every $h>0$ and $y\in B_{r/2}(x)$,
Theorem~\ref{prop:one_step_usc_subsolution_realization},
applied with $T=h$, provides
$\mathbb P_{h,y,m}\in\mathcal P_y(G)$ such that
\begin{equation}
\label{eq:ambient_source_minorant_realization}
 \mathbb E^{\mathbb P_{h,y,m}}
 \!\left[
 f(X_{h\wedge\rho})
 \mathbb I_{\{h\wedge\rho<\tau_\infty\}}
 \right]-f(y)\ge
 \mathbb E^{\mathbb P_{h,y,m}}
 \!\left[
 \int_0^{h\wedge\rho}
 r_m(X_s)\mathbb I_{\{s<\tau_\infty\}}\,ds
 \right].
\end{equation}

For fixed $m$, the short-time estimates in Proposition~\ref{prop:short_time_control_bounded_coefficients} of the Supplementary Material give
\begin{align*}
 \mathbb E^{\mathbb P}
 \!\left[
 \int_0^{h\wedge\rho}
 r_m(X_s)\mathbb I_{\{s<\tau_\infty\}}\,ds
 \right]
 &=h\,r_m(y)+o(h),\\
 \mathbb E^{\mathbb P}
 \!\left[
 f(X_{h\wedge\rho})
 \mathbb I_{\{h\wedge\rho<\tau_\infty\}}
 \right]
 &=
 \mathbb E^{\mathbb P}
 \!\left[f(X_h)\mathbb I_{\{h<\tau_\infty\}}\right]
 +o(h),
\end{align*}
uniformly for $y\in B_{r/2}(x)$ and $\mathbb P\in\mathcal P_y(G)$.
For the first identity, the Lipschitz continuity of $r_m$ and the local
drift and martingale estimates bound the integrated spatial error by
$O_m(h^{3/2}+h^2)$; killing and live exit contribute
$O(h^2)+h\,o(h)$ to the loss of integration time.
For the second identity, the terminal values can differ only after a live
exit from $B_r(x)$, whose probability is $o(h)$ uniformly over these
initial states and laws.
Applying these estimates to
\eqref{eq:ambient_source_minorant_realization} and using
$\mathbb P_{h,y,m}\in\mathcal P_y(G)$ yields
\[
    \mathcal V_h^Gf(y)-f(y)
    \ge h\,r_m(y)+o(h),
\]
uniformly for $y\in B_{r/2}(x)$.
Dividing by $h$ and taking the lower half-relaxed limit, we obtain $L_f(x)\ge r_m(x)$.
Letting $m\to\infty$ proves $L_f(x)\ge(q_f)_*(x)$.
The localization argument above extends this inequality to every $f\in C_b^\infty(D)$.

We have therefore established $(q_f)_*\le L_f\le U_f\le q_f$.
Since $L_f$ is lower semicontinuous and lies below $q_f$, it also
satisfies $L_f\le(q_f)_*$, and hence $L_f=(q_f)_*$.
Moreover, $(q_f)_*\le q_f^\downarrow\le q_f$ gives $(q_f)_*\le(q_f^\downarrow)_*\le(q_f)_*$.
Thus $L_f=(q_f^\downarrow)_*=(q_f)_*$.
Together with $U_f=q_f^\downarrow$, this proves that $G^\downarrow$ is an exact generating function of $\mathcal V^G$.

Finally, fix $(x,U)\in D\times\mathfrak J$ and choose $f\in C_c^\infty(D)$ with $J_x^2f=U$.
Local uniform Lipschitz continuity of $G$ and $G^\downarrow$ in the jet variable, together with continuity of $y\mapsto J_y^2f$,
gives
\[
    G_*(x,U)
    =(q_f)_*(x)
    =(q_f^\downarrow)_*(x)
    =(G^\downarrow)_*(x,U).
\]
Thus $G_*=(G^\downarrow)_*$.
Since $G^\downarrow$ is upper semicontinuous and $G^\downarrow\le G$, taking upper semicontinuous envelopes yields $(G_*)^*=((G^\downarrow)_*)^*\le G^\downarrow\le G$, which proves \eqref{eq:thm:canonical_exact_reduced_generator_1}.
\end{proof}

Thus the upper generating function recovered in
Section~\ref{sec:local_global_correspondence} is an exact generating
function of the canonical semigroup. The identity in
\eqref{eq:thm:canonical_exact_reduced_generator_1} also shows that replacing $G$
by $G^\downarrow$ preserves the lower infinitesimal specification
along smooth tests. We can now ask when this exact generating function
coincides with the prescribed $G$.

\begin{corollary}
\label{cor:exact_generation}
Let $G\in\mathfrak G_{\mathrm{Lyap}}$. The following are equivalent:
\begin{enumerate}[label=(\roman*),ref=(\roman*)]
\item\label{item:exact_realizability_upper}
$G$ is globally realizable, or equivalently, $G^\downarrow=G$.
\item\label{item:exact_realizability_ambient}
The canonical semigroup $\mathcal V^G$ has exact generating function $G$.
\item\label{item:exact_realizability_some}
There exists $\mathcal T\in\mathfrak V$ having exact generating function $G$.
\end{enumerate}
These conditions hold if the joint envelopes on $D\times\mathfrak J$
satisfy
\begin{equation}
\label{eq:normal_upper_regularity}
    (G_*)^*=G.
\end{equation}
In particular, they hold whenever $G$ is jointly continuous.
\end{corollary}

\begin{proof}
If condition~\ref{item:exact_realizability_upper} holds, then
$G^\downarrow=G$ by
Theorem~\ref{thm:realizable_specification_reduction}.
Theorem~\ref{thm:canonical_exact_reduced_generator} therefore gives
condition~\ref{item:exact_realizability_ambient}.
The implication
\ref{item:exact_realizability_ambient}$\Rightarrow$\ref{item:exact_realizability_some}
is immediate. Conversely, a semigroup with exact generating function
$G$ has upper generating function $G$, so
condition~\ref{item:exact_realizability_some} implies global realizability.

Suppose now that \eqref{eq:normal_upper_regularity} holds.
For $f\in C_b^\infty(D)$, write $q_f=G(\cdot,J^2f)$.
Local uniform Lipschitz continuity in the jet variable and continuity of
$y\mapsto J_y^2f$ give
\[
    (q_f)_*(x)=G_*(x,J_x^2f),\qquad
    ((q_f)_*)^*(x)=(G_*)^*(x,J_x^2f)=q_f(x).
\]
Indeed, on a compact neighborhood the difference between allowing a jet
to approach $J_x^2f$ and using the moving jet $J_y^2f$ is bounded by
$L_K\|U-J_y^2f\|$ and tends to zero. The same estimate applies to $G_*$.
The bounds in \eqref{eq:thm:canonical_exact_reduced_generator_1} therefore give $q_f^\downarrow=q_f$ for every smooth test $f$, and both identities
in \eqref{eq:relaxed_generator_consistency} follow.
\end{proof}

The following example illustrates the role of the $\operatorname{USC}_b$ formulation in this paper. Even when $G$ is linear and continuous, both the class of $C_b$-Feller semigroups with upper generating function $G$ and its subclass with exact generating function $G$ may be nonempty and have no greatest element.

\begin{example}[No greatest \(C_b\)-Feller semigroup]
\label{ex:no_greatest_feller_valuation}
Let \(D=\mathbb R^2\), and define
\begin{equation}\label{eq:no_greatest_feller_specification}
 \begin{aligned}
 G((x,y),z,p,H)&:=b(x,y)p_1+\tfrac12a(x,y)H_{11},\\
 b(x,y)&:=3|x|^{2/3}+3y^2\sqrt[3]{x},\qquad
 a(x,y):=9y^2|x|^{4/3}.
 \end{aligned}
\end{equation}
This specification is continuous and belongs to
\(\mathfrak G_{\mathrm{Lyap}}\). Nevertheless, the nonempty class
\begin{equation}\label{eq:no_greatest_feller_class}
 \left\{
 \mathcal S\in\mathfrak V_{\mathrm{loc}}:
 G_{\mathcal S}=G,\quad
 \mathcal S_t(C_b(D))\subseteq C_b(D)\ \text{for every }t\ge0
 \right\}
\end{equation}
has no greatest element. The same conclusion holds after additionally
requiring
\begin{equation}\label{eq:no_greatest_feller_lu}
 \Delta_h^{\mathcal S}f
 \longrightarrow G(\,\cdot\,,J^2f)
 \quad\text{locally uniformly as }h\downarrow0,
 \qquad f\in C_b^\infty(D).
\end{equation}
Thus the conclusion of
\eqref{eq:greatest_valuation_same_upper_generator}
can fail upon restriction to \(C_b\)-Feller semigroups, even for a
continuous linear specification and with
\eqref{eq:no_greatest_feller_lu}.
\end{example}

\begin{proof}
The smooth proper function \(\phi(x,y):=1+x^2+y^6\) satisfies
\begin{equation}\label{eq:no_greatest_feller_lyapunov}
 G((x,y),J^2_{(x,y)}\phi)
 =6x|x|^{2/3}+15y^2|x|^{4/3}
 \le16\phi(x,y).
\end{equation}
Thus \eqref{eq:no_greatest_feller_specification} belongs to
\(\mathfrak G_{\mathrm{Lyap}}\).

For \(m\ge0\) and an initial state \((x,y)\), let \(W\) be a standard
Brownian motion and define
\begin{equation}\label{eq:no_greatest_feller_clock}
\begin{aligned}
 &Z_s:=\sqrt[3]{x}+s+yW_s,
 \qquad A_s^m:=s+m\ell_s^0(Z),\qquad
 \tau_t^m:=(A^m)^{-1}(t),\\
 &\Xi_t^m:=\bigl((Z_{\tau_t^m})^3,y\bigr),
 \qquad\mathcal T_t^m f(x,y):=\mathbb E[f(\Xi_t^m)],
\end{aligned}
\end{equation}
where \(\ell^0(Z)\) denotes symmetric semimartingale local time.
For each fixed \(y\ne0\), the time-change construction in
\cite[Theorem~2.3]{salins2017markov} gives the Markov property.
For \(y=0\), the local time vanishes and
\(\Xi_t^m=((\sqrt[3]{x}+t)^3,0)\), which also defines a Markov flow.
Moreover, it is straightforward to verify that each $\mathcal T^m$ belongs to $\mathcal V$ with exact generating function $G$ and is $C_b$-Feller.

Suppose that a \(C_b\)-Feller semigroup
\(\mathcal S\in\mathfrak V_{\mathrm{loc}}\) satisfies
\begin{equation}\label{eq:no_greatest_feller_upper_bound}
 G_{\mathcal S}\le G,
 \qquad
 \mathcal T^m\le\mathcal S
 \quad\text{for every }m\ge0.
\end{equation}
By \eqref{eq:local_global_recovered_structure_order}, we have $\mathcal T^0\le\mathcal S\le\mathcal V^G$.
Fix \(t>0\) and \(f(x,y):=e^{-x^2}\).
Starting from \((0,y)\) with \(y\ne0\), positivity of Brownian local
time at its starting point and absolute continuity under a constant
drift give \(\ell_\varepsilon^0(Z)>0\) almost surely for every
\(\varepsilon>0\).
Hence \(\tau_t^m\to0\) almost surely as \(m\to\infty\), and
\begin{equation}\label{eq:no_greatest_feller_sticky_limit}
 \mathcal T_t^m f(0,y)\longrightarrow1,
 \qquad y\ne0.
\end{equation}
Together with \eqref{eq:no_greatest_feller_upper_bound} and
\(0\le\mathcal S_t f\le1\), this implies
\(\mathcal S_t f(0,y)=1\) for \(y\ne0\).
Continuity therefore gives \(\mathcal S_t f(0,0)=1\).

On the other hand,
Proposition~\ref{prop:recovery_aggregation_generator_compatible_characteristics}
and \eqref{eq:no_greatest_feller_specification} imply that every law
in \(\mathcal P_{(x,0)}(G)\), for \(x>0\), has zero killing,
constant second coordinate, and first coordinate satisfying
\[
 \dot u_s=3|u_s|^{2/3},\qquad u_0=x.
\]
Its unique solution is \(u_s=(\sqrt[3]{x}+s)^3\).
Thus $\mathcal T^0\le\mathcal S\le\mathcal V^G$ yields
\begin{equation}\label{eq:no_greatest_feller_positive_ray}
 \mathcal S_t f(x,0)
 =\mathcal T_t^0f(x,0)
 =\mathcal V_t^Gf(x,0)
 =\exp\!\left[-(\sqrt[3]{x}+t)^6\right],
 \qquad x>0.
\end{equation}
Letting \(x\downarrow0\) in
\eqref{eq:no_greatest_feller_positive_ray} gives
\(\mathcal S_t f(0,0)=e^{-t^6}<1\), a contradiction.
Therefore no \(C_b\)-Feller semigroup satisfies
\eqref{eq:no_greatest_feller_upper_bound}.
Since every \(\mathcal T^m\) belongs to the class
\eqref{eq:no_greatest_feller_class} and satisfies
\eqref{eq:no_greatest_feller_lu}, neither class has a greatest element.
\end{proof}

\subsection{Upper Perron characterization}
\label{subsec:valuation_supersolutions}
\label{sec:ambient_perron_comparison}
\label{subsec:maximal_perron_proof}
\label{subsec:upper_perron_feynman_kac}

The preceding subsection identifies the infinitesimal behavior of
$\mathcal V^G$. Its global analytic characterization is a maximality
property: each orbit $V_f^G$ dominates every bounded upper
semicontinuous viscosity subsolution with initial upper bound $f$.
This holds for every $G\in\mathfrak G_{\mathrm{Lyap}}$, including
those for which $G^\downarrow\ne G$.
We write $V_f^G(t,x):=\mathcal V_t^Gf(x)$.

\begin{theorem}
\label{thm:upper_perron_feynman_kac}
Let $G\in\mathfrak G_{\mathrm{Lyap}}$. 
For every $f\in\USC_b(D)$ and $(t,x)\in[0,\infty)\times D$,
\begin{equation}
\label{eq:upper_perron_feynman_kac}
    V_f^G(t,x)
    =\max_{\mathbb P\in\mathcal P_x(G)}
         \mathbb E^{\mathbb P}[f(X_t)\mathbb I_{\{t<\tau_\infty\}}]
    =\max_{u\in\operatorname{Sub}_G^+(f)}u(t,x).
\end{equation}
In particular, $V_f^G$ is the greatest bounded upper semicontinuous
viscosity subsolution of \eqref{eq:mainHJB} subject to the initial upper
bound $u(0,\cdot)\le f$.  It is a bounded viscosity solution in the
discontinuous convention and has initial datum $f$ in the relaxed sense.
\end{theorem}

We prove the theorem after recording the supersolution criterion.
Theorem~\ref{thm:canonical_exact_reduced_generator} gives the required
lower consistency with the original specification $G$, even when
$G^\downarrow\ne G$.
We formulate the criterion for arbitrary semigroups, since it will
also be used in the uniqueness argument.

\begin{lemma}
\label{prop:valuation_supersolution_under_lower_consistency}
Let $G\in\mathfrak G$, and let
$\mathcal T\in\mathfrak V$ satisfy
\begin{equation}
\label{eq:lower_generator_consistency}
    \liminf_{h\searrow0,*}\Delta_h^{\mathcal T}\varphi
    \ge\bigl(G(\cdot,J^2\varphi)\bigr)_*,
    \qquad \varphi\in C_b^\infty(D).
\end{equation}
For $f\in\USC_b(D)$, write $v_f(t,x):=\mathcal T_t f(x)$.
Then $(v_f)_*$ is a viscosity supersolution
of \eqref{eq:mainHJB} and satisfies the relaxed supersolution initial
condition associated with $f$.
If \eqref{eq:upper_generator_domination} also holds, then $v_f$ is a
viscosity solution in the discontinuous convention with initial datum
$f$ in the relaxed sense.
\end{lemma}
\begin{proof}
Let $\psi\in C_b^\infty([0,\infty)\times D)$ touch $(v_f)_*$ globally
from below at $(t_0,x_0)$, with $t_0>0$.
Choose $(t_n,x_n)\to(t_0,x_0)$ such that
$v_f(t_n,x_n)\to(v_f)_*(t_0,x_0)$, and set
\[
    \delta_n:=v_f(t_n,x_n)-\psi(t_n,x_n)\ge0.
\]
Choose $h_n\searrow0$ with $h_n<t_n$ and
$|t_n-t_0|+\delta_n=o(h_n)$, and put
\[
    \varphi=\psi(t_0,\cdot),\qquad
    q_n=\psi(t_n-h_n,\cdot),\qquad
    r=-\partial_t\psi(t_0,\cdot).
\]
The bounded derivatives of $\psi$ give
$\|q_n-\varphi-h_nr\|_\infty=o(h_n)$.
Applying the lower bound in
\eqref{eq:moving_payoff_relaxed_generator_sandwich}, together with
Lemma~\ref{lemma:valuation_initial_continuity} and
\eqref{eq:lower_generator_consistency}, gives
\[
    \liminf_{n\to\infty}\Delta_{h_n}^{\mathcal T}q_n(x_n)
    \ge\bigl(G(\cdot,J^2\varphi)\bigr)_*(x_0)
    \ge G_*(x_0,J_{x_0}^2\varphi).
\]
At the lower contact,
$q_n\le v_f(t_n-h_n,\cdot)$.
Monotonicity and the semigroup property consequently imply
\[
    \Delta_{h_n}^{\mathcal T}q_n(x_n)
    \le\frac{\delta_n}{h_n}
      +\frac{\psi(t_n,x_n)-\psi(t_n-h_n,x_n)}{h_n}.
\]
Taking the lower limit proves
\eqref{eq:viscosity_supersolution_inequality}.
The initial condition follows from
Lemma~\ref{lemma:valuation_initial_continuity}.
When the upper-generator bound also holds,
Lemma~\ref{prop:PDErepn} supplies the subsolution property and
therefore the asserted viscosity-solution statement.
\end{proof}

\begin{proof}[Proof of Theorem~\ref{thm:upper_perron_feynman_kac}]
Terminal attainment follows from
Proposition~\ref{prop:stable_virtual_structure_generates_valuation}.
Lemmas~\ref{prop:ambient_valuation_generator_bounds}
and~\ref{prop:PDErepn} show that
$V_f^G\in\operatorname{Sub}_G^+(f)$.
For any $u\in\operatorname{Sub}_G^+(f)$ and $t>0$,
Corollary~\ref{cor:usc_terminal_realization} provides a law
$\mathbb P\in\mathcal P_x(G)$ such that
\[
    u(t,x)
    \le\mathbb E^{\mathbb P}
       [f(X_t)\mathbb I_{\{t<\tau_\infty\}}]
    \le V_f^G(t,x).
\]
At time zero the same inequality follows from $u(0,\cdot)\le f$.
This proves \eqref{eq:upper_perron_feynman_kac}, with the subsolution
maximum attained by $V_f^G$ itself.

Finally, Theorem~\ref{thm:canonical_exact_reduced_generator} gives
\eqref{eq:lower_generator_consistency} for $\mathcal V^G$.
Lemma~\ref{prop:valuation_supersolution_under_lower_consistency}
therefore supplies the supersolution property and the relaxed
supersolution initial condition. Together with the subsolution
properties above, this proves the viscosity-solution assertion.
\end{proof}

This characterizes the evolution selected by the full structure of admissible laws through its maximality among subsolutions, without requiring
comparison. Exact generation alone need not identify a unique
semigroup, as Example~\ref{ex:singular_diffusion_multiple_branches}
shows even for a continuous specification.
We next derive semigroup uniqueness from viscosity uniqueness and use
comparison to obtain continuity of the canonical evolution.

\subsection{Viscosity uniqueness and the Feller property}
\label{subsec:canonical_collapse_comparison}

A semigroup with exact generating function $G$ is dominated by
$\mathcal V^G$ and satisfies the lower-consistency condition
\eqref{eq:lower_generator_consistency}. Its orbits are therefore
viscosity solutions in the discontinuous convention. To identify the
semigroup, it suffices to assume uniqueness for these solutions with
continuous initial data.

\begin{assume}
\label{assume:discontinuous_viscosity_uniqueness}
For every $T>0$ and $f\in C_b(D)$, there is at most one
$u\in\USC_b((0,T)\times D)$ that is a viscosity solution of
\eqref{eq:mainHJB} in the discontinuous convention and has initial
datum $f$ in the relaxed sense.
\end{assume}

Here $u$ itself is a subsolution and $u_*$ is a supersolution;
continuity of $u$ is not part of the assumption. Uniqueness refers
to equality of these upper semicontinuous representatives on
$(0,T)\times D$. Continuous initial data suffice because semigroups
in $\mathfrak V$ are continuous from above on $\USC_b(D)$.

\begin{theorem}
\label{thm:canonical_uniqueness_under_viscosity_uniqueness}
Let $G\in\mathfrak G_{\mathrm{Lyap}}$ satisfy
Assumption~\ref{assume:discontinuous_viscosity_uniqueness}.
Then $\mathcal V^G$ is the unique semigroup in
$\mathfrak V_{\le}(G)$ satisfying
\eqref{eq:lower_generator_consistency}.
More precisely, every such semigroup $\mathcal T$ obeys
\begin{equation}
\label{eq:canonical_family_collapse_under_comparison}
    \mathcal T=\mathcal V^G,
    \qquad \mathcal R^{\mathcal T}=\mathcal P(G).
\end{equation}
If $G$ is globally realizable, then $\mathcal V^G$ is the unique
semigroup in $\mathfrak V$ having exact generating function $G$.
\end{theorem}

\begin{proof}
Let $\mathcal T\in\mathfrak V_{\le}(G)$ satisfy
\eqref{eq:lower_generator_consistency}.
Theorem~\ref{cor:local_global_domination} gives the upper-generator
bound \eqref{eq:upper_generator_domination}.
Fix $f\in C_b(D)$ and $T>0$.
Lemma~\ref{prop:valuation_supersolution_under_lower_consistency}
shows that $u_f(t,x):=\mathcal T_tf(x)$ is a bounded upper
semicontinuous viscosity solution in the discontinuous convention
on $(0,T)\times D$, with relaxed initial datum $f$.
Theorem~\ref{thm:upper_perron_feynman_kac} gives the same properties
for $V_f^G$. Assumption~\ref{assume:discontinuous_viscosity_uniqueness}
therefore yields $u_f=V_f^G$ on $(0,T)\times D$.
Since $T$ is arbitrary and both semigroups are the identity at time
zero, they agree on $C_b(D)$ at every time.

For $f\in\USC_b(D)$, choose uniformly bounded functions
$f_n\in C_b(D)$ with $f_n\searrow f$.
Applying \ref{V3} to both semigroups extends their equality to
$\USC_b(D)$.
Theorem~\ref{thm:canonical_exact_reduced_generator} shows that
$\mathcal V^G$ itself satisfies the lower-consistency condition,
so the indicated subclass is nonempty.
The identity between structures follows from the inverse formula in
Theorem~\ref{thm:globally_dominated_structure_correspondence}.

If $G$ is globally realizable,
Corollary~\ref{cor:exact_generation} shows that $\mathcal V^G$
has exact generating function $G$.
Every semigroup with exact generating function $G$ satisfies
\eqref{eq:lower_generator_consistency} and belongs to
$\mathfrak V_{\le}(G)$ by
Theorem~\ref{cor:local_global_domination}.
The first assertion therefore proves uniqueness.
\end{proof}

We now impose comparison to obtain continuity of the canonical
semigroup orbits.

\begin{assume}
\label{assume:paraboliccomparison}
For every $T>0$, let $u\in\USC_b([0,T]\times D)$ and
$v\in\LSC_b([0,T]\times D)$ be a viscosity subsolution and a viscosity
supersolution, respectively, of \eqref{eq:mainHJB} on $(0,T)\times D$.
If
\[
 \limsup_{\substack{(s,y)\to(0,x)\\s>0}}u(s,y)
 \le\liminf_{\substack{(s,y)\to(0,x)\\s>0}}v(s,y),\qquad x\in D,
\]
then $u\le v$ on $(0,T)\times D$.
\end{assume}

We call a semigroup $\mathcal T$ \emph{$C_b$-Feller} if
$\mathcal T_t(C_b(D))\subseteq C_b(D)$ for every $t\ge0$ and
$\mathcal T_hf\to f$ locally uniformly on $D$ as $h\searrow0$
for every $f\in C_b(D)$.

\begin{theorem}
\label{thm:canonical_feller_under_comparison}
Let $G\in\mathfrak G_{\mathrm{Lyap}}$ satisfy
Assumption~\ref{assume:paraboliccomparison}.
Then Assumption~\ref{assume:discontinuous_viscosity_uniqueness}
holds, and all conclusions of
Theorem~\ref{thm:canonical_uniqueness_under_viscosity_uniqueness} apply.
Moreover, $\mathcal V^G$ is $C_b$-Feller. In fact, for every
$f\in C_b(D)$,
\begin{equation}
\label{eq:canonical_cb_feller_continuity}
    V_f^G\in C_b([0,\infty)\times D),
    \qquad
    \mathcal V_h^Gf\longrightarrow f
    \quad\text{locally uniformly on }D\text{ as }h\searrow0.
\end{equation}
For every $T>0$, the restriction of $V_f^G$ to $(0,T)\times D$
is the unique bounded viscosity solution in the discontinuous
convention with relaxed initial datum $f$.
In particular, when $G$ is globally realizable, the unique semigroup
with exact generating function $G$ is $C_b$-Feller.
\end{theorem}

\begin{proof}
Fix $f\in C_b(D)$ and write $w=V_f^G$.
By Theorem~\ref{thm:upper_perron_feynman_kac}, $w$ is an upper
semicontinuous subsolution and $w_*$ is a lower semicontinuous
supersolution, both with initial trace $f$.
Extend $w_*$ to time zero by $f$.
The relaxed initial conditions and continuity of $f$ make this
extension lower semicontinuous there.
Comparison on every finite horizon gives
\[
    w\le w_*\le w
    \quad\text{on }(0,\infty)\times D.
\]
Thus $w=w_*$ is jointly continuous at positive times.
Lemma~\ref{lemma:valuation_initial_continuity} supplies locally
uniform convergence to $f$ as time decreases to zero, which also
gives joint continuity at time zero.
Boundedness follows from \ref{V2}.
This proves \eqref{eq:canonical_cb_feller_continuity} and the
$C_b$-Feller property.

Now let $z$ be any bounded viscosity solution in the discontinuous
convention on $(0,T)\times D$ with relaxed initial datum $f$.
For $0<S<T$, extend $z^*$ and $z_*$ to time zero by $f$ and restrict
them to $[0,S]\times D$.
The relaxed initial conditions ensure the required semicontinuity
at time zero.
Comparison of $z^*$ with $w$, and of $w$ with $z_*$, yields
\[
    z^*\le w\le z_*\le z\le z^*
    \quad\text{on }(0,S)\times D.
\]
Since $S<T$ is arbitrary, $z=w$ on $(0,T)\times D$.
This proves the viscosity uniqueness assertion and, in particular,
Assumption~\ref{assume:discontinuous_viscosity_uniqueness}.
The semigroup conclusions now follow from
Theorem~\ref{thm:canonical_uniqueness_under_viscosity_uniqueness}.
\end{proof}

Under viscosity uniqueness, global realizability determines whether a
semigroup with exact generating function $G$ exists. Whenever it does,
that semigroup and its recovered stable structure of laws are
$\mathcal V^G$ and $\mathcal P(G)$, respectively.
Under comparison, this unique semigroup is also $C_b$-Feller.

\section{Conclusion}
\label{sec:conclusion}

We studied the relationships among sublinear sub-Markov semigroups, stable structures of laws, and local specifications.
Under finiteness and locality, the structural conditions on local specifications follow from the semigroup axioms.
With an additional Lyapunov condition, we established a semigroup--structure order isomorphism and an exact correspondence between local upper-generator bounds, semigroup domination, and inclusion of the recovered structures.
We characterized the extremal meaning of the compositions of localization, globalization, and the semigroup map.
We also characterized exact generation and identified the canonical semigroup as an upper Perron evolution. Uniqueness of discontinuous viscosity solutions yields semigroup uniqueness under lower consistency, while comparison additionally gives the $C_b$-Feller property.

\appendix
\renewcommand{\theequation}{\thesection.\arabic{equation}}

\section{Local upper generators}
\label{app:local_upper_generators}
\begin{lemma}
\label{lem:local_second_order_jet_determination}
Under the hypotheses of Proposition~\ref{prop:local_upper_generator_representation},
$\overline{\mathcal G}^{\mathcal T}$ is sublinear on $C_b^\infty(D)$ and
satisfies the local positive maximum principle: if $f$ has a local maximum
at $x$ and $f(x)\ge0$, then
$\overline{\mathcal G}^{\mathcal T}f(x)\le0$.
Moreover,
\begin{align}\label{eq:jet_dependency_overlineA}
    J_x^2f=J_x^2g
    \quad\Longrightarrow\quad
    \overline{\mathcal G}^{\mathcal T}f(x)
    =\overline{\mathcal G}^{\mathcal T}g(x),
    \qquad f,g\in C_b^\infty(D).
\end{align}
\end{lemma}

\begin{proof}
Taking upper half-relaxed limits in the sublinearity inequalities for $\mathcal T_h$ proves sublinearity of $\overline{\mathcal G}^{\mathcal T}$; its values are finite by assumption.
Also $\overline{\mathcal G}^{\mathcal T}0=0$, and $\overline{\mathcal G}^{\mathcal T}f$ is upper semicontinuous for each fixed $f\in C_b^\infty(D)$.

Next we prove the local positive maximum principle.
Suppose that $f\in C_b^\infty(D)$ attains a nonnegative global
maximum $m=f(x)$ at $x$. Set
\[
    Q(z):=\frac{|z-x|^2}{1+|z-x|^2},
    \qquad F_\varepsilon:=f-\varepsilon Q.
\]
Fix $\varepsilon>0$. Since $\nabla f(x)=0$ and
$\nabla^2f(x)\le0$, we have $\nabla F_\varepsilon(x)=0$ and $\nabla^2F_\varepsilon(x)\le-2\varepsilon I_d$.
By continuity, we may choose $r>0$ such that $\overline{B_{2r}(x)}\subset D$, and $\nabla^2F_\varepsilon(z)\le-\varepsilon I_d$ for all $z\in B_{2r}(x)$.
Moreover, since $f\le m$ on $D$,
\[
    F_\varepsilon(z)\le m-\delta
    \qquad\text{for }z\in D\setminus B_r(x),
    \qquad
    \delta:=\frac{\varepsilon r^2}{1+r^2}>0.
\]
Choose $\chi\in C_c^\infty(D)$ equal to one on a neighborhood of $\overline{B_{2r}(x)}$, and define
\[
    \ell_i(z):=\chi(z)(z_i-x_i),
    \qquad
    \ell:=(\ell_1,\ldots,\ell_d),
    \qquad
    L:=\sup_{z\in D}|\ell(z)|<\infty.
\]
For $y$ near $x$, define $F_{\varepsilon,y}(z):=F_\varepsilon(z)-\nabla F_\varepsilon(y)\cdot\ell(z)$.
We claim that $F_{\varepsilon,y}$ attains its global maximum at $y$, with value $F_{\varepsilon,y}(y)\ge m\ge0$.
Since $\nabla F_\varepsilon(y)\to\nabla F_\varepsilon(x)=0$ as $y\to x$, choose $\eta\in(0,r)$ such that
$|\nabla F_\varepsilon(y)|\le\delta/(2L)$ for every $y\in B_\eta(x)$.
For such $y$, the function $F_{\varepsilon,y}$ satisfies
\[
    \nabla F_{\varepsilon,y}(y)=0,
    \qquad
    \nabla^2F_{\varepsilon,y}(z)\le-\varepsilon I_d
    \quad\text{on }B_{2r}(x).
\]
Consequently, for $z\in B_{2r}(x)$,
\[
    F_{\varepsilon,y}(z)
    \le F_{\varepsilon,y}(y)
       -\frac{\varepsilon}{2}|z-y|^2.
\]
In particular, taking $z=x$, we have $F_{\varepsilon,y}(y)\ge F_{\varepsilon,y}(x)=m$.
On the other hand, for $z\in D\setminus B_r(x)$,
\[
    F_{\varepsilon,y}(z)
    \le m-\delta+|\nabla F_\varepsilon(y)|L
    \le m-\frac{\delta}{2}.
\]
These inequalities show that $F_{\varepsilon,y}$ attains its
global maximum at $y$, with value
$m_y:=F_{\varepsilon,y}(y)\ge m\ge0$.

By monotonicity and \ref{V2}, applied to the constant function $m_y$,
\[
    \mathcal T_hF_{\varepsilon,y}(y)
    \le \mathcal T_h(m_y\mathbf 1)(y)
    \le m_y,
    \qquad h>0.
\]
Thus $\Delta_h^{\mathcal T}F_{\varepsilon,y}(y)\le0$.
Since $\overline{\mathcal G}^{\mathcal T}(\pm\ell_i)(x)$ is finite
for every $i$, the definition of the joint upper limit yields
constants $M\ge0$, $h_0>0$, and $\rho\in(0,\eta)$ such that
\[
    \Delta_h^{\mathcal T}(\pm\ell_i)(y)\le M
    \qquad
    \text{for }0<h<h_0,\quad y\in B_\rho(x),\quad 1\le i\le d.
\]
Writing $a_i^\pm(y):=\max\{\pm \partial_iF_\varepsilon(y),0\}$, sublinearity gives
\[
    \Delta_h^{\mathcal T}F_\varepsilon(y)
    \le
    \Delta_h^{\mathcal T}F_{\varepsilon,y}(y)
    +\sum_{i=1}^d
      \left[
        a_i^+(y)\Delta_h^{\mathcal T}\ell_i(y)
        +a_i^-(y)\Delta_h^{\mathcal T}(-\ell_i)(y)
      \right]
    \le M\sum_{i=1}^d|\partial_iF_\varepsilon(y)|.
\]
Taking the joint upper limit as $h\searrow0$ and $y\to x$,
and using $\nabla F_\varepsilon(y)\to0$, proves
$\overline{\mathcal G}^{\mathcal T}F_\varepsilon(x)\le0$.
Finally, sublinearity of the upper generator yields
\[
    \overline{\mathcal G}^{\mathcal T}f(x)
    \le
    \overline{\mathcal G}^{\mathcal T}F_\varepsilon(x)
    +\varepsilon\overline{\mathcal G}^{\mathcal T}Q(x)
    \le
    \varepsilon\overline{\mathcal G}^{\mathcal T}Q(x).
\]
Since $\overline{\mathcal G}^{\mathcal T}Q(x)$ is finite, letting
$\varepsilon\searrow0$ gives
$\overline{\mathcal G}^{\mathcal T}f(x)\le0$.
If $f$ has only a local maximum at $x$, choose $\overline f\in C_b^\infty(D)$ equal to $f$ near $x$ and attaining its global maximum $f(x)$ at $x$. Locality then gives $\overline{\mathcal G}^{\mathcal T}f(x)=\overline{\mathcal G}^{\mathcal T}\overline f(x)\le0$.
This proves the local positive maximum principle.

Now suppose $J_x^2f=J_x^2g$, and put $r=f-g$.
Choose $q\in C_b^\infty(D)$ equal to $|y-x|^2$ near $x$.
Taylor's formula gives $r(y)=o(|y-x|^2)$, so both
$r-\varepsilon q$ and $-r-\varepsilon q$ have a local maximum equal to
zero at $x$. The local maximum principle and sublinearity yield
\[
    \overline{\mathcal G}^{\mathcal T}(\pm r)(x)
    \le
    \overline{\mathcal G}^{\mathcal T}(\pm r-\varepsilon q)(x)
    +\varepsilon\overline{\mathcal G}^{\mathcal T}q(x)\le
    \varepsilon\overline{\mathcal G}^{\mathcal T}q(x).
\]
Letting $\varepsilon\searrow0$ and using $\overline{\mathcal G}^{\mathcal T}r(x)+\overline{\mathcal G}^{\mathcal T}(-r)(x)\ge0$ shows that both values vanish.
Sublinearity applied to $f=g+r$ and $g=f-r$ now proves \eqref{eq:jet_dependency_overlineA}.
\end{proof}

\begin{proof}[Proof of Proposition~\ref{prop:local_upper_generator_representation}]
Fix $x\in D$ and choose $\chi_x\in C_c^\infty(D)$ equal to one on a neighborhood of $x$. 
For $Z=(z,p,H)\in\mathfrak J$, define
\[
    f_{x,Z}(y)
    :=\chi_x(y)\left(
        z+p\cdot(y-x)
        +\frac12(y-x)^\top H(y-x)
    \right),
    \qquad y\in D.
\]
Then $f_{x,Z}\in C_c^\infty(D)$ and $J_x^2f_{x,Z}=Z$.
Set $G(x,Z):=\overline{\mathcal G}^{\mathcal T}f_{x,Z}(x)$.
Finiteness of the upper generator on $C_b^\infty(D)$ implies
that $G$ is finite.
Moreover, $G$ satisfies \eqref{eq:thm:representationgenerator_eq1} by Lemma~\ref{lem:local_second_order_jet_determination}.
Uniqueness of $G$ is also immediate from \eqref{eq:thm:representationgenerator_eq1}.

Next, we verify \ref{item:G1}.
For each fixed $x\in D$, the map $Z\mapsto f_{x,Z}$ is linear, so sublinearity of $\overline{\mathcal G}^{\mathcal T}$ implies sublinearity of $G(x,\cdot)$.
We first establish the locally uniform jet-Lipschitz bound \eqref{eq:local_uniform_jet_lipschitz_G}.
Fix a compact set $K\Subset D$, let $N:=\dim\mathfrak J$, and choose $\chi\in C_c^\infty(D)$ equal to one on a neighborhood of $K$.
For a basis $P_1,\ldots,P_N$ of the polynomials of degree at most two, set $\psi_i:=\chi P_i$.
For every $y\in K$, the jet map on this polynomial space is an isomorphism, so $\{J_y^2\psi_i\}_{i=1}^N$ is a basis of $\mathfrak J$.
Since these basis vectors depend continuously on $y$, there exist unique functions $a_i:K\times\mathfrak J\to\mathbb R$ which are jointly continuous and linear in the second variable, with
\[
    Z=\sum_{i=1}^N a_i(y,Z)J_y^2\psi_i,
    \qquad y\in K,\quad Z\in\mathfrak J.
\]
By jointly continuity of the functions $a_i$ and their linearity in the second variable, there exists a constant $C_K>0$ such that
\[
    \sum_{i=1}^N|a_i(y,Z)|\le C_K\|Z\|,
    \qquad y\in K,\quad Z\in\mathfrak J.
\]
Set
\[
    M_K:=
    \max_{1\le i\le N}\sup_{y\in K}
    \max\left\{
        \overline{\mathcal G}^{\mathcal T}\psi_i(y),
        \overline{\mathcal G}^{\mathcal T}(-\psi_i)(y),
        0
    \right\}.
\]
Each $\overline{\mathcal G}^{\mathcal T}(\pm\psi_i)$ is a finite upper
semicontinuous function, so $M_K<\infty$. Sublinearity gives
\begin{equation}
\label{eq:GleqLKZ}
\begin{aligned}
    G(y,Z)
    &\le
    \sum_{i=1}^N\left[
        a_i(y,Z)^+\overline{\mathcal G}^{\mathcal T}\psi_i(y)
        +a_i(y,Z)^-\overline{\mathcal G}^{\mathcal T}(-\psi_i)(y)
    \right]\le M_KC_K\|Z\|.
\end{aligned}
\end{equation}
Applying the same estimate to $-Z$ and using
$0\le G(y,Z)+G(y,-Z)$ yields $|G(y,Z)|\le M_KC_K\|Z\|$.
Finally,
\[
    G(y,Z_1)-G(y,Z_2)
    \le G(y,Z_1-Z_2)
    \le M_KC_K\|Z_1-Z_2\|.
\]
Interchanging $Z_1$ and $Z_2$ proves \eqref{eq:local_uniform_jet_lipschitz_G} with $L_K:=M_KC_K$.

We now prove joint upper semicontinuity of $G$.
Let $(x_n,Z_n)\to(x,Z)$ in $D\times\mathfrak J$.
Choose a compact neighborhood $K\Subset D$ of $x$ containing $x_n$ for all sufficiently large $n$, and fix $f\in C_c^\infty(D)$ with $J_x^2f=Z$.
By \eqref{eq:thm:representationgenerator_eq1} and \eqref{eq:local_uniform_jet_lipschitz_G},
\[
    G(x_n,Z_n)
    \le G(x_n,J_{x_n}^2f)
       +L_K\|Z_n-J_{x_n}^2f\|
    =
    \overline{\mathcal G}^{\mathcal T}f(x_n)
       +L_K\|Z_n-J_{x_n}^2f\|.
\]
Since $Z_n\to Z=J_x^2f$ and $J_{x_n}^2f\to J_x^2f$, the local uniform jet-Lipschitz bound and the upper semicontinuity of $\overline{\mathcal G}^{\mathcal T}f$ yield
\[
    \limsup_{n\to\infty}G(x_n,Z_n)
    \le \limsup_{n\to\infty}
        \overline{\mathcal G}^{\mathcal T}f(x_n)
    \le \overline{\mathcal G}^{\mathcal T}f(x)
    =G(x,Z).
\]
This proves joint upper semicontinuity and completes the
verification of \ref{item:G1}.

If $H_1\ge H_2$, realize $(0,0,H_2-H_1)$ by a smooth function having a local maximum equal to zero at $x$.
The local maximum principle gives $G(x,0,0,H_2-H_1)\le0$, and hence
\[
    G(x,z,p,H_2)
    \le G(x,z,p,H_1)+G(x,0,0,H_2-H_1)
    \le G(x,z,p,H_1).
\]
This proves \ref{item:G2}.
For $z_1\ge z_2$, the nonnegative constant test gives $G(x,z_1-z_2,0,0)\le0$. Sublinearity then yields
\[
    G(x,z_1,p,H)
    \le G(x,z_2,p,H)+G(x,z_1-z_2,0,0)
    \le G(x,z_2,p,H),
\]
proving \ref{item:G3}.
\end{proof}

\section{Recovery and aggregation of local characteristics}
\label{app:characteristic_recovery}

\begin{proof}[Proof of Proposition~\ref{prop:recovery_aggregation_generator_compatible_characteristics}]
All martingale and supermartingale properties below extend to the
usual augmentation $\widetilde{\mathbb F}^{\,\mathbb P}$, since the
processes are c\`adl\`ag. We take canonical decompositions, compensators,
and predictable densities with respect to this filtration.
Choose the least \(n_0\) such that
\(\widetilde\omega([0,t])\subset D_{n_0}\), and fix \(n\ge n_0\).
For $s\ge0$, write
\[
 N_s^{n,t}:=N_{[s]_{t,\tau_n}}-N_t.
\]
Since the prescribed history is alive at time $t$, we have $N_t=0$
and $\tau_n>t$, $\mathbb P$-a.s.
Let \(\chi_n\in C_c^\infty(D)\) equal one on a neighborhood of
\(\overline D_n\), and put
\[
 q_i^n(x):=\chi_n(x)x_i,
 \qquad
 q_{ij}^n(x):=\chi_n(x)x_ix_j,
 \qquad 1\le i,j\le d.
\]

\smallskip
\noindent\emph{Step 1: finite-test drift recovery.}
Let
\[
    \mathscr H_n:=\{1\}\cup\{q_i^n:1\le i\le d\}
       \cup\{q_{ij}^n:1\le i\le j\le d\}.
\]
For \(h\in\mathscr H_n\) and $s\ge t$, define $Y_s^{h,n}:=\mathbb I_{\{s\wedge\tau_n<\tau_\infty\}}h\bigl(\overline X_{s\wedge\tau_n}\bigr)$.
By the definition of \(\widetilde M^{h,n}\), the process \(Y^{h,n}-\widetilde M^{h,n}\) has predictable finite variation on every finite horizon. 
Moreover, since \(\mathbb P\in\mathcal P_{t,\widetilde\omega}(G)\), both \(\widetilde M^{h,n}\) and \(\widetilde M^{-h,n}\) are \(\mathbb P\)-supermartingales.
It follows that \(Y^{h,n}\) is a special semimartingale. 
Write its canonical decomposition as $Y^{h,n}=Y_t^{h,n}+M+V$.
The canonical finite-variation parts of supermartingales \(\widetilde M^{h,n}\) and \(\widetilde M^{-h,n}\) are, respectively,
\[
    V-
    \int_t^{\cdot\wedge\tau_n}
    G\bigl(\overline X_r,J^2_{\overline X_r}h\bigr)\mathbb I_{\{r<\tau_\infty\}}\,dr, \qquad -V-
    \int_t^{\cdot\wedge\tau_n}
    G\bigl(\overline X_r,-J^2_{\overline X_r}h\bigr)\mathbb I_{\{r<\tau_\infty\}}\,dr.
\]
Since both processes must be nonincreasing, we obtain, in the sense of signed random measures on \(t<r\le\tau_n\),
\[
    -G\bigl(\overline X_r,-J^2_{\overline X_r}h\bigr)\mathbb I_{\{r<\tau_\infty\}}\,dr
    \le dV_r
    \le
    G\bigl(\overline X_r,J^2_{\overline X_r}h\bigr)\mathbb I_{\{r<\tau_\infty\}}\,dr.
\]
In particular, $V$ has no atom at $\tau_n$.
Since $V$ is stopped at $\tau_n$, it is \(\mathbb P\)-a.s. absolutely
continuous with respect to Lebesgue measure and therefore admits a
predictable density \(\alpha^{h,n}\).
In particular,
\begin{equation}
\label{eq:second_stage_finite_test_decomposition}
\left(
Y_s^{h,n}-Y_t^{h,n}
-
\int_t^{s\wedge\tau_n}\alpha_r^{h,n}\,dr
\right)_{s\ge t}
\end{equation}
is a \(\mathbb P\)-local martingale.
Moreover,
\begin{align}\label{eq:alpha_bounded}
-G\bigl(\overline X_r,-J^2_{\overline X_r}h\bigr)\mathbb I_{\{r<\tau_\infty\}}
\le
\alpha_r^{h,n}
\le
G\bigl(\overline X_r,J^2_{\overline X_r}h\bigr)\mathbb I_{\{r<\tau_\infty\}}
\end{align}
for \(dr\otimes d\mathbb P\)-a.e. \((r,\widetilde\eta)\) on \(\{(r,\widetilde\eta):t<r<\tau_n(\widetilde\eta)\}\).
Choose bounded predictable versions of these finitely many densities,
equal to zero outside $\{t<r\le\tau_n\}$, and use the same version for
$q_{ij}^n=q_{ji}^n$.
Write
\[
\alpha^n:=\alpha^{1,n},\qquad
\alpha^{i,n}:=\alpha^{q_i^n,n},\qquad
\alpha^{ij,n}:=\alpha^{q_{ij}^n,n},
\qquad 1\le i,j\le d.
\]
On the predictable set \(\{t<r\le\tau_n\}\), define
\begin{align*}
    &c_r^n
    :=\alpha_r^n, \qquad 
    b_r^{n,i}
    :=\alpha_r^{i,n}-c_r^n\overline X_r^i,\\
    &a_r^{n,ij}
    :=
    \alpha_r^{ij,n}
    -b_r^{n,i}\overline X_r^j
    -b_r^{n,j}\overline X_r^i
    -c_r^n\overline X_r^i\overline X_r^j,
    \qquad 1\le i,j\le d,
\end{align*}
and extend \(\beta^n:=(c^n,b^n,a^n)\) by zero outside this set.
Then \(\beta^n\) is bounded and predictable.

Since \(\mathbb I_{\{s\wedge\tau_n<\tau_\infty\}}=1-N_s^{n,t}\), applying \eqref{eq:second_stage_finite_test_decomposition} with \(h=1\) shows that
\begin{equation}
\label{eq:K_decomposition}
    K_s^n
    :=
    N_s^{n,t}
    -
    \int_t^{s\wedge\tau_n}(-c_r^n)\,dr
\end{equation}
is a local martingale. 
Since \(N^{n,t}\) is bounded and increasing, uniqueness of the canonical decomposition implies that \(\int_t^{\cdot\wedge\tau_n}(-c_r^n)\,dr\) is its predictable compensator. 
Thus $-c^n$ is the density of the stopped killing compensator; the
minus sign reflects the decrease of the survival indicator at killing.
Consequently, \(c_r^n\le0\) for \(dr\otimes d\mathbb P\)-a.e. \((r,\widetilde\eta)\) on \(\{(r,\widetilde\eta):t<r<\tau_n(\widetilde\eta)\}\).
Moreover, since \(\chi_n\equiv1\) on a neighborhood of
\(\overline D_n\), we have, for every \(s\ge t\),
\begin{equation}
\label{eq:overlineX_decomposition}
\begin{aligned}
    \overline X_{s\wedge\tau_n}^{\,i}
    &=
    Y_s^{q_i^n,n}
    +
    \int_t^{s\wedge\tau_n}\overline X_r^i\,dN_r^{n,t},\\
    \overline X_{s\wedge\tau_n}^{\,i}
    \overline X_{s\wedge\tau_n}^{\,j}
    &=
    Y_s^{q_{ij}^n,n}
    +
    \int_t^{s\wedge\tau_n}
        \overline X_r^i\overline X_r^j\,dN_r^{n,t}.
\end{aligned}
\end{equation}
The added integrals restore the values removed from the killed
coordinates at the killing jump. Their compensator contributions
therefore remove the killing terms $c^n\overline X^i$ and
$c^n\overline X^i\overline X^j$ from the recovered coordinate drifts.
Combining the first identity in \eqref{eq:overlineX_decomposition} with
\eqref{eq:second_stage_finite_test_decomposition} for \(h=q_i^n\) and
with \eqref{eq:K_decomposition}, componentwise, yields an
\(\mathbb R^d\)-valued continuous local martingale \(M^n\), null on
\([0,t]\), such that
\begin{equation}
\label{eq:second_stage_recovered_characteristics}
    \overline X_{s\wedge\tau_n}
    =
    \widetilde\omega(t)
    +
    \int_t^{s\wedge\tau_n}b_r^n\,dr
    +
    M_{s\wedge\tau_n}^n.
\end{equation}
Applying the product formula to
\eqref{eq:second_stage_recovered_characteristics} and comparing its
predictable finite-variation part with that of the second identity in
\eqref{eq:overlineX_decomposition}, we obtain
\[
    \bigl\langle M^{n,i},M^{n,j}\bigr\rangle_{s\wedge\tau_n}
    =
    \int_t^{s\wedge\tau_n}a_r^{n,ij}\,dr,
    \qquad 1\le i,j\le d.
\]
Thus $b^n$ is the drift density of the stopped continuous process
$\overline X$, and the matrix-valued measure \(a_r^n\,dr\) is the
predictable quadratic-covariation measure of its local martingale part.
Consequently, \(a_r^n\in\mathbb S^+(d)\) for \(dr\otimes d\mathbb P\)-a.e. \((r,\widetilde\eta)\) on \(\{(r,\widetilde\eta):t<r<\tau_n(\widetilde\eta)\}\). 
Moreover, the stopped local martingale \(M^n_{\cdot\wedge\tau_n}\) is square-integrable since $a^n$ is bounded, and therefore a true martingale, on every finite horizon.

\smallskip
\noindent\emph{Step 2: smooth tests and support membership.}
For arbitrary \(f\in C_b^\infty(D)\), It\^o's formula applied to
\eqref{eq:second_stage_recovered_characteristics}, followed by integration
by parts with \(\mathbb I_{\{\cdot\wedge\tau_n<\tau_\infty\}}\), shows that
\begin{equation}
\label{eq:second_stage_linear_test_martingale}
\left(
\mathbb I_{\{s\wedge\tau_n<\tau_\infty\}}f(\overline X_{s\wedge\tau_n})
-\mathbb I_{\{t<\tau_\infty\}} f(\widetilde\omega(t))
-
\int_t^{s\wedge\tau_n}
\ell_{\beta_r^n}\bigl(J^2_{\overline X_r}f\bigr)\mathbb I_{\{r<\tau_\infty\}}\,dr
\right)_{s\ge t}
\end{equation}
is a martingale. The same process with
\(\ell_{\beta_r^n}\) replaced by \(G(\overline X_r,\cdot)\) is a
supermartingale because
\(\mathbb P\in\mathcal P_{t,\widetilde\omega}(G)\). Their difference,
\[
    \left(
    \int_t^{s\wedge\tau_n}
    \left[
      \ell_{\beta_r^n}\bigl(J^2_{\overline X_r}f\bigr)
      -G\bigl(\overline X_r,J^2_{\overline X_r}f\bigr)
    \right]\mathbb I_{\{r<\tau_\infty\}}\,dr
    \right)_{s\ge t},
\]
is therefore a predictable finite-variation supermartingale and hence is
nonincreasing. Since \(\mathbb I_{\{r<\tau_\infty\}}=1\) on \(\{t<r<\tau_n\}\), it follows that
\begin{align}
\label{eq:second_stage_smooth_test_domination}
    \ell_{\beta_r^n}\bigl(J^2_{\overline X_r}f\bigr)
    \le
    G\bigl(\overline X_r,J^2_{\overline X_r}f\bigr).
\end{align}
for \(dr\otimes d\mathbb P\)-a.e. \((r,\widetilde\eta)\) on \(\{(r,\widetilde\eta):t<r<\tau_n(\widetilde\eta)\}\). 

By Proposition~\ref{lem:countable_smooth_determining_class} in the Supplementary Material, there exists a countable determining class \(\mathscr F_{\mathrm{sm}}\subset C_b^\infty(D)\) with the following pointwise jet-density property: for every \((x,U)\in D\times\mathfrak J\), there exists a sequence \((g_k)\subset\mathscr F_{\mathrm{sm}}\) such that \(J_x^2g_k\to U\).
Applying \eqref{eq:second_stage_smooth_test_domination} to each \(g\in\mathscr F_{\mathrm{sm}}\) and using the countability of \(\mathscr F_{\mathrm{sm}}\), we obtain a single \(dr\otimes d\mathbb P\)-null set \(\mathfrak E_n\) such that
\begin{align}\label{eq:determiningclass_ell_le_G}
    \ell_{\beta_r^n(\widetilde\eta)}
    \bigl(J^2_{\overline X_r(\widetilde\eta)}g\bigr)
    \le
    G\bigl(
        \overline X_r(\widetilde\eta),
        J^2_{\overline X_r(\widetilde\eta)}g
    \bigr)
    \qquad
    \text{for every }g\in\mathscr F_{\mathrm{sm}}
\end{align}
whenever \((r,\widetilde\eta)\notin\mathfrak E_n\) and
\(t<r<\tau_n(\widetilde\eta)\).
Fix any such \((r,\widetilde\eta)\), and set
\[
    (x,V)
    :=
    \bigl(\overline X_r(\widetilde\eta),
           \beta_r^n(\widetilde\eta)\bigr).
\]
Given \(U\in\mathfrak J\), choose
\((g_k)\subset\mathscr F_{\mathrm{sm}}\) such that \(J_x^2g_k\to U\).
Then, by \eqref{eq:determiningclass_ell_le_G} and the continuity of \(\ell_V\) and \(G(x,\cdot)\), we have
\[
    \ell_V(U)
    =\lim_{k\to\infty}\ell_V(J_x^2g_k)
    \le
    \lim_{k\to\infty}G(x,J_x^2g_k)
    =G(x,U)
    \qquad\text{for every }U\in\mathfrak J.
\]
Hence, by \eqref{def:supportset}, \(V\in\mathfrak A_G(x)\). Since
\((r,\widetilde\eta)\) was arbitrary outside \(\mathfrak E_n\), it follows
that
\begin{equation}
\label{eq:second_stage_local_support_membership}
    \beta_r^n\in\mathfrak A_G(\overline X_r)
    \quad
    dr\otimes d\mathbb P\text{-a.e. $(r,\widetilde\eta)$ on }\{(r,\widetilde\eta):t<r<\tau_n(\widetilde\eta)\}.
\end{equation}

\smallskip
\noindent\emph{Step 3: consistency and aggregation.}
If \(m\ge n\ge n_0\), uniqueness of the compensator of \(N^{n,t}\), the canonical decomposition of \(\overline X_{\cdot\wedge\tau_n}\), and its predictable quadratic covariation imply
\begin{align}
\label{eq:beta_consistency}
    \beta^m=\beta^n
    \quad
    dr\otimes d\mathbb P\text{-a.e. on }
    \{(r,\widetilde\eta):t<r<\tau_n(\widetilde\eta)\}.
\end{align}
For \(n\ge n_0\), set $E_n:=\bigl\{(r,\widetilde\eta):t<r\le\tau_n(\widetilde\eta)\bigr\}$ and $E_{n_0-1}:=\varnothing$.
Each \(E_n\) is \(\widetilde{\mathbb F}^{\,\mathbb P}\)-predictable.
Since the graph of \(\tau_n\) is \(dr\otimes d\mathbb P\)-null,
\eqref{eq:beta_consistency} is unchanged if
\(\{t<r<\tau_n\}\) is replaced by \(E_n\).
Define on \(\widetilde\Omega\)
\[
    \overline\beta^{\mathbb P}
    :=
    \sum_{n=n_0}^{\infty}
    \beta^n\mathbb I_{E_n\setminus E_{n-1}}.
\]
The sets in this sum are disjoint, so
\(\overline\beta^{\mathbb P}\) is a well-defined \(\widetilde{\mathbb F}^{\,\mathbb P}\)-predictable process.
Moreover, by \eqref{eq:beta_consistency}, $\overline\beta^{\mathbb P}=\beta^n$, $dr\otimes d\mathbb P$-a.e. on $E_n$ for every $n\ge n_0$.
Applying \cite[Lemma~7, p.~399]{dellacherie1982probabilities}, there exists a raw \(\widetilde{\mathbb F}\)-predictable process \(\breve\beta^{\mathbb P}\) that is \(\mathbb P\)-indistinguishable from \(\overline\beta^{\mathbb P}\).
Consequently,
\begin{align}
\label{eq:raw_virtual_beta_aggregation}
    \breve\beta^{\mathbb P}
    =
    \beta^n
    \quad
    dr\otimes d\mathbb P\text{-a.e. on }E_n
\end{align}
for every \(n\ge n_0\).

By \eqref{eq:raw_virtual_beta_aggregation} and
\eqref{eq:second_stage_local_support_membership},
\begin{equation}
\label{eq:aggregated_local_support_membership}
 \breve\beta_r^{\mathbb P}(\widetilde\eta)
 \in\mathfrak A_G(X_r(\widetilde\eta)),
 \qquad dr\otimes d\mathbb P\text{-a.e. on }
 \{(r,\widetilde\eta):t<r<\tau_\infty(\widetilde\eta)\}.
\end{equation}
Choose a Borel selector $v_0:D\to\mathfrak K$ with
$v_0(x)\in\mathfrak A_G(x)$ for every $x\in D$, and define directly on
$\widetilde\Omega$
\[
 \beta_r^{\mathbb P}(\widetilde\eta):=
 \begin{cases}
 \breve\beta_r^{\mathbb P}(\widetilde\eta),
 &r<\tau_\infty(\widetilde\eta),
   \breve\beta_r^{\mathbb P}(\widetilde\eta)
       \in\mathfrak A_G(X_r(\widetilde\eta)),\\[1mm]
 v_0(X_r(\widetilde\eta)),
 &r<\tau_\infty(\widetilde\eta)\text{ otherwise},\\[1mm]
 0,&r\ge\tau_\infty(\widetilde\eta).
 \end{cases}
\]
The Borel graph and local boundedness of $\mathfrak A_G$ imply that
$\beta^{\mathbb P}$ is raw progressively measurable and locally bounded.
Thus $\beta^{\mathbb P}\in\mathfrak B_{\mathrm{gc}}(G)$, and
\eqref{eq:raw_virtual_beta_aggregation} gives
\begin{equation}
\label{eq:aggregate_relation_beta}
 \beta^{\mathbb P}=\beta^n,
 \qquad dr\otimes d\mathbb P\text{-a.e. on }
 \{(r,\widetilde\eta):t<r<\tau_n(\widetilde\eta)\}.
\end{equation}
Consequently, \eqref{eq:second_stage_linear_test_martingale} implies that
$\widetilde M^{f,n,\beta^{\mathbb P}}$ is a $\mathbb P$-martingale
for every $n\ge n_0$.
For $n<n_0$, the prescribed history gives $\tau_n\le t$, so
$[s]_{t,\tau_n}=t$ and the same process is constant on $[t,\infty)$.
Hence
$\mathbb P\in\widetilde{\mathcal P}_{t,\widetilde\omega}
(L^{\beta^{\mathbb P}})$, establishing the asserted existence of a
generator-compatible coefficient field.

The same aggregation identity transfers the recovered characteristics
to $\beta^{\mathbb P}$. Indeed, for $n\ge n_0$, setting
$M_s^{\mathbb P,n}:=M_{s\wedge\tau_n}^n$ in
\eqref{eq:second_stage_recovered_characteristics} gives the decomposition
and quadratic covariation in
part~\ref{prop:recovery_continuous_characteristics}, with $b^n$ and $a^n$
replaced by $b^{\mathbb P}$ and $a^{\mathbb P}$.
Uniqueness of the stopped compensators and continuous semimartingale
characteristics determines these densities only
$dr\otimes d\mathbb P$-a.e. on $\{t<r<\tau_\infty\}$.
The choice of a generator-compatible field away from this set is not
asserted to be unique.

\smallskip
\noindent\emph{Step 4: the global killing compensator.}
By \eqref{eq:aggregate_relation_beta}, the compensator in
\eqref{eq:K_decomposition}, extended by zero on $[0,t]$, is
\[
 A_s^{\mathbb P,n}
 :=\int_t^{[s]_{t,\tau_n}}
       k^{\beta^{\mathbb P}}(r,X)\mathbb I_{\{r<\tau_\infty\}}\,dr,
 \qquad s\ge0.
\]
For each $n$, local boundedness makes
$N^{n,t}-A^{\mathbb P,n}$ a true martingale on every finite horizon.
In particular,
\[
 \mathbb E^{\mathbb P}[A_s^{\mathbb P,n}]
 =\mathbb E^{\mathbb P}[N_s^{n,t}]\le1.
\]
As $n\to\infty$, we have $\tau_n\uparrow\tau_\infty$ and
\[
 A_s^{\mathbb P,n}\uparrow
 A_s^{\mathbb P}
 :=\int_t^{s\vee t}k^{\beta^{\mathbb P}}(r,X)\mathbb I_{\{r<\tau_\infty\}}\,dr.
\]
Moreover, $N_s^{n,t}\to N_s$ pathwise: every finite killing jump is
captured once $D_n$ contains the path up to its pre-killing limit,
whereas on paths with continuous explosion both indicators vanish.
The expectation bound and monotone convergence give
$A_s^{\mathbb P}<\infty$ a.s. and convergence in $L^1(\mathbb P)$;
the bounded indicators also converge in $L^1(\mathbb P)$.
For $0\le r\le s$ and bounded $\widetilde{\mathcal F}_r$-measurable $F$,
pass to the limit in
$\mathbb E^{\mathbb P}[F(N_s^{n,t}-A_s^{\mathbb P,n})]
=\mathbb E^{\mathbb P}[F(N_r^{n,t}-A_r^{\mathbb P,n})]$.
This proves that $N-A^{\mathbb P}$ is a $\mathbb P$-martingale on
every finite horizon.
The process $A^{\mathbb P}$ is continuous, adapted, and increasing,
so it is the predictable compensator of $N$.
Consequently, for every deterministic $a>0$,
\[
 \mathbb P(\tau_{\mathrm{kill}}=a)
 =\mathbb E^{\mathbb P}[\Delta N_a]
 =\mathbb E^{\mathbb P}[\Delta A_a^{\mathbb P}]=0.
\]
There is also no atom at $a=0$ because the prescribed history is alive
at time $t\ge0$. This proves
part~\ref{prop:recovery_killing_compensator} without a Lyapunov assumption
and completes the proof.
\end{proof}

\section{Proof of the Subsolution Realization Theorem}
\label{sec:stochastic_realization_perron}
\label{sec:usc_realization_proofs}
\label{subsec:usc_stochastic_realization}
This appendix proves Theorem~\ref{prop:one_step_usc_subsolution_realization}.
We first prepare an outer approximation of the local specification and
a diagonal stochastic approximation of the subsolution. These yield the
stopped realization inequality with a source term.

\subsection{Outer approximation of the local specification}
We first show that every jointly continuous generating function admits
a nested family of outer approximations.
In the rest of the subsection, we fix the Lyapunov function $\phi:D\to[1,\infty)$ and denote
\[
\kappa_\phi(x):=\frac{\phi(x)}{1+\|J_x^2\phi\|}, \qquad x\in D.
\]
Let \(d_{\mathfrak J}\) denote the distance on \(\mathfrak J\) induced by \(\|\cdot\|\), and let \(d_{H,\mathfrak J}\) denote the corresponding Hausdorff distance between nonempty compact subsets of \(\mathfrak J\).
For a jointly continuous \(G\in\mathfrak G\) and \(\varepsilon>0\), define
\begin{equation}
\label{eq:def_outer_support_regularization}
    \mathfrak A_G^\varepsilon(x)
    :=
    \Bigl\{
        V\in\mathfrak K:
        d_{\mathfrak J}\bigl(V,\mathfrak A_G(x)\bigr)
        \le\varepsilon\kappa_\phi(x)
    \Bigr\}, \qquad
    G^\varepsilon(x,U)
    :=
    \max_{V\in\mathfrak A_G^\varepsilon(x)}\ell_V(U),
\end{equation}

\begin{lemma}
\label{lem:monotone_outer_generator_regularization}
Let \(G\in\mathfrak G_{\mathrm{Lyap}}\) be jointly continuous.
Then the following statements hold.
\begin{enumerate}[label=(\roman*), ref=(\roman*)]
\item\label{lem:outer_generator_properties}
Each \(G^\varepsilon\) belongs to \(\mathfrak G\) and is jointly continuous.
Moreover,
\begin{equation}
\label{eq:outer_generator_monotonicity}
    \mathfrak A_G(x)
    \subseteq\mathfrak A_G^{\varepsilon'}(x)
    \subseteq\mathfrak A_G^\varepsilon(x),
    \qquad
    G(x,U)\le G^{\varepsilon'}(x,U)\le G^\varepsilon(x,U)
\end{equation}
whenever \(0<\varepsilon'\le\varepsilon\).

\item\label{lem:outer_generator_error}
For every \((x,U)\in D\times\mathfrak J\),
\begin{equation}
\label{eq:outer_generator_error_bound}
    0\le G^\varepsilon(x,U)-G(x,U)
    \le\varepsilon\kappa_\phi(x)\|U\|.
\end{equation}
In particular, \(G^\varepsilon\searrow G\) locally uniformly on
\(D\times\{U\in\mathfrak J:\|U\|\le1\}\) as \(\varepsilon\searrow0\), and
\begin{equation}
\label{eq:outer_generator_lyapunov_bound}
    G^\varepsilon\bigl(x,J_x^2\phi\bigr)
    \le(C_\phi+\varepsilon)\phi(x).
\end{equation}
\end{enumerate}
\end{lemma}

\begin{proof}
Joint continuity of \(G\) implies that, whenever \(x_j\to x\),
\[
    \sup_{\|U\|\le1}
    \bigl|G(x_j,U)-G(x,U)\bigr|
    \longrightarrow0.
\]
Since \(G(x,\cdot)\) is the support function of \(\mathfrak A_G(x)\),
the support-function formula for the Hausdorff distance
\cite[Theorem~1.8.11]{schneider2014convex} shows that
\(x\mapsto\mathfrak A_G(x)\) is continuous with respect to
\(d_{H,\mathfrak J}\).
Moreover, for every \(x,y\in D\),
\begin{align}
    d_{H,\mathfrak J}\bigl(
        \mathfrak A_G^\varepsilon(x),
        \mathfrak A_G^\varepsilon(y)
    \bigr)
    \le
    d_{H,\mathfrak J}\bigl(
        \mathfrak A_G(x),
        \mathfrak A_G(y)
    \bigr)
    +
    \varepsilon
    \bigl|\kappa_\phi(x)-\kappa_\phi(y)\bigr|.
\end{align}
Consequently, the Hausdorff continuity of \(\mathfrak A_G\), together
with the continuity of \(\kappa_\phi\), implies that
\(\mathfrak A_G^\varepsilon\) is Hausdorff continuous.
Hence \(G^\varepsilon\) is jointly continuous, and
\(G^\varepsilon(x,\cdot)\) is sublinear with support correspondence
\(\mathfrak A_G^\varepsilon\).
Since this support is contained in \(\mathfrak K\),
Conditions~\ref{item:G2} and \ref{item:G3} follow from the same
support-function argument as in
Lemma~\ref{lem:support_correspondence_geometry}.
This proves part~\ref{lem:outer_generator_properties}.

For \(V\in\mathfrak A_G^\varepsilon(x)\), choose
\(W\in\mathfrak A_G(x)\) such that
\(\|V-W\|\le\varepsilon\kappa_\phi(x)\). Then
\[
    \ell_V(U)
    \le
    \ell_W(U)+\varepsilon\kappa_\phi(x)\|U\|
    \le
    G(x,U)+\varepsilon\kappa_\phi(x)\|U\|.
\]
Taking the maximum over \(V\) and using
\(\mathfrak A_G(x)\subseteq\mathfrak A_G^\varepsilon(x)\) proves
\eqref{eq:outer_generator_error_bound}.
Since
\(\kappa_\phi(x)\|J_x^2\phi\|\le\phi(x)\),
\eqref{eq:outer_generator_lyapunov_bound} follows from
Assumption~\ref{assume:lyapunov}.
This completes the proof.
\end{proof}

\subsection{Diagonal stochastic approximation}
The following lemma, which is the main step in the proof of Theorem~\ref{prop:one_step_usc_subsolution_realization}, simultaneously approximates the subsolution \(u\) and the source term \(q\) by smooth functions with vanishing integrated residual and continuous source terms, respectively, while preserving the required realization properties and yielding a realization inequality.

For a locally bounded Borel map $\zeta=(c,b,a):[0,\infty)\times D\to\mathfrak K$, we denote by \(\beta^\zeta\) the induced Markovian coefficient field on \(\widetilde\Omega\), defined by
\[
    \beta^\zeta(t,\omega)
    :=
    \begin{cases}
        \zeta\bigl(t,X_t(\omega)\bigr),
        & t<\tau_\infty(\omega),\\
        0,
        & t\ge\tau_\infty(\omega).
    \end{cases}
\]
After setting \(\zeta(t,\triangle):=0\), this can equivalently be written as $\beta^\zeta_t=\zeta(t,X_t)\mathbb I_{\{t<\tau_\infty\}}$.

\begin{lemma}
\label{lem:usc_diagonal_stochastic_approximation}
Let \(G\in\mathfrak G_{\mathrm{Lyap}}\).
Let \(T>0\), and let \(u,q\in\USC_b([0,T+1]\times D)\).
Suppose that \(u\) is a viscosity subsolution of
\eqref{eq:usc_subsolution_with_source} on \((0,T+1)\times D\).
Fix \(m\ge1\) and \(x\in D_m\).
Then there exist
\[
    t_n\searrow0,
    \qquad
    v_n\in C^\infty([t_n,T]\times\overline D_{m+1}),
    \qquad
    q_n\in
    C([t_n,T]\times\overline D_{m+1}),
\]
and laws \(\mathbb P_n\in\mathfrak M\) with the following properties.

\begin{enumerate}[label=(\roman*), ref=(\roman*)]
\item\label{lem:usc_diagonal_compactness}
The sequence \((\mathbb P_n)_{n\ge1}\) is relatively weakly compact,
and every weak cluster point belongs to
\(\mathcal P_x(G)\).

\item\label{lem:usc_diagonal_one_step_inequality}
For every open set \(O\) satisfying \(x\in O\Subset D_m\), define $\sigma_n^{O}:=(T-t_n)\wedge\rho_{O}^{\,0}$.
Then
\begin{equation}
\label{eq:usc_diagonal_one_step_inequality}
\begin{aligned}
    v_n(T,x)
    &\le
    \mathbb E^{\mathbb P_n}\!\Bigg[
        v_n\bigl(T-\sigma_n^{O},X_{\sigma_n^{O}}\bigr)
        \mathbb I_{\{\sigma_n^{O}<\tau_\infty\}}\\
    &\hspace{3em}+\int_0^{\sigma_n^{O}}
            q_n(T-s,X_s)\mathbb I_{\{s<\tau_\infty\}}\,ds
    \Bigg]
    +\frac{1+T}{n}.
\end{aligned}
\end{equation}

\item\label{lem:usc_diagonal_value_convergence}
One has
\begin{equation}
\label{eq:usc_diagonal_prescribed_point_convergence}
    v_n(T,x)\longrightarrow u(T,x),
    \qquad
    \sup_{n\ge1}
    \|v_n\|_{L^\infty([t_n,T]\times D_{m+1})}
    \le\|u\|_\infty,
\end{equation}
and
\begin{equation}
\label{eq:usc_diagonal_value_upper_limit}
    \limsup_{n\to\infty}^{*}v_n
    \le u
    \qquad\text{on }[0,T]\times D_m.
\end{equation}

\item\label{lem:usc_diagonal_source_convergence}
The functions \(q_n\) satisfy
\begin{equation}
\label{eq:usc_diagonal_source_bound}
    \sup_{n\ge1}
    \|q_n\|_{L^\infty(
        [t_n,T]\times\overline D_m)}
    <\infty,
\end{equation}
and
\begin{equation}
\label{eq:usc_diagonal_source_upper_limit}
    \limsup_{n\to\infty}^{*}q_n
    \le q
    \qquad\text{on }[0,T]\times D_m.
\end{equation}
\end{enumerate}
\end{lemma}

\begin{proof}
Set $E_T:=[0,T+1]\times D$ and $d\bigl((s,y),(t,z)\bigr):=1\wedge\bigl(|s-t|+|y-z|\bigr)$.
We first construct a sequence \(\bigl(q^{(n)},G_n\bigr)_{n\ge1}\) of continuous outer approximations to \((q,G)\).
For each \(n\ge1\), define
\[
    q^{(n)}(\xi)
    :=
    \sup_{\eta\in E_T}
    \bigl\{q(\eta)-n d(\xi,\eta)\bigr\},
    \qquad \xi\in E_T.
\]
Then \(q^{(n)}\in C_b(E_T)\), \(q^{(n)}\searrow q\) pointwise on $E_T$, and $\sup_{n\ge1}\|q^{(n)}\|_\infty\le\|q\|_\infty$.
On the other hand, since \(G\in\mathfrak G\), it admits a sequence \((G_n)_{n\ge1}\) of continuous outer approximations to \(G\); see Proposition~\ref{prop:continuous_outer_generator_approximation} in the Supplementary Material.
Each \(G_n\) belongs to \(\mathfrak G\) and is jointly continuous. 
Moreover,
\begin{align}\label{eq:G_n_lyapunov_bound}
    G_n\ge G,
    \qquad
    G_n\bigl(x,J_x^2\phi\bigr)
    \le (C_\phi+n^{-1})\phi(x),
    \qquad x\in D,
\end{align}
and
\begin{equation}
\label{eq:G_n_upperconsistency_G}
    \limsup_{n\to\infty}^{*}G_n
    \le G
    \qquad\text{on }D\times\mathfrak J.
\end{equation}
Since $q^{(n)}\ge q$ and $G_n\ge G$, the function $u$ is a viscosity subsolution of
\begin{align}\label{eq:viscositysubsolution_hjb_widehatG}
    \partial_tu
    -G_n\bigl(x,J_x^2u(t,\cdot)\bigr)
    -q^{(n)}(t,x)
    =0
\end{align}
on \((0,T+1)\times D\).

Fix \(n\ge1\).  For \(\xi\in E_T\), set
\[
    u^{\langle n\rangle}(\xi)
    :=
    \sup\bigl\{
        u(\eta):
        d(\xi,\eta)\le n^{-1}
    \bigr\}.
\]
Apply Propositions~\ref{prop:usc_smooth_elliptic_approximation} and~\ref{prop:usc_outer_compatible_coefficients} in the Supplementary Material with data \((G_n,q^{(n)},u)\) and
Lyapunov pair \((\phi,C_\phi+n^{-1})\).  Then there exist $\lambda_{n,j}>0$, $t_{n,j},\varepsilon_{n,j}\searrow0$, together with
\[
\begin{aligned}
    &v_{n,j}
    \in
    C^\infty\bigl(
        [t_{n,j},T]\times\overline D_{m+1}
    \bigr), \qquad q_{n,j},e_{n,j}
    \in
    C\bigl(
        [t_{n,j},T]\times\overline D_{m+1}
    \bigr),
    \qquad e_{n,j}\ge0,\\
    &\zeta_{n,j}
    \in
    C\bigl(
        [0,\infty)\times D;\mathfrak K
    \bigr),
    \qquad \mathbb P_{n,j}
    \in
    \widetilde{\mathcal P}_{x}
        \bigl(L^{\beta^{\zeta_{n,j}}}\bigr),
\end{aligned}
\]
such that the following properties hold:
\begin{enumerate}[
    label=\textup{(\roman*)},
    ref=\textup{(\roman*)},
    leftmargin=*
]

\item \emph{Approximation of \(u\).}
The smooth functions \(v_{n,j}\) converge to \(u\) in the
upper-relaxed sense:
\begin{equation}
\label{eq:fixed_n_value_upper_relaxed_convergence}
    \limsup_{j\to\infty}^{*}v_{n,j}
    \le u
    \qquad\text{on }[0,T]\times D_{m+1}.
\end{equation}
Moreover, $v_{n,j}(t,y)\to u(t,y)$ for every $(t,y)\in(0,T]\times D_{m+1}$, and
\begin{equation}
\label{eq:fixed_n_value_uniform_bound}
    \sup_{j\ge1}
    \left\|
        v_{n,j}
    \right\|_{L^\infty(
        [t_{n,j},T]\times\overline D_{m+1}
    )}
    \le
    \|u\|_\infty.
\end{equation}
By \eqref{eq:fixed_n_value_upper_relaxed_convergence} and the
compactness of \([0,T]\times\overline D_m\), there exists
\(j_n\ge1\) such that, for every \(j\ge j_n\),
\begin{equation}
\label{eq:fixed_n_value_outer_bound}
    v_{n,j}(t,y)
    \le
    u^{\langle n\rangle}(t,y)+n^{-1},
    \qquad
    (t,y)\in
    [t_{n,j},T]\times\overline D_m.
\end{equation}

\item \emph{Approximation of the source term.}
The continuous functions \(q_{n,j}\) satisfy
\begin{equation}
\label{eq:fixed_n_source_convergence}
    \sup_{(t,y)\in
        [t_{n,j},T]\times\overline D_m}
    \bigl|
        q_{n,j}(t,y)-q^{(n)}(t,y)
    \bigr|
    \longrightarrow0.
\end{equation}

\item \emph{Compatibility and bounds of the coefficient fields.}
For every \((s,y)\in[0,\infty)\times D\),
\begin{equation}
\label{eq:fixed_n_outer_support_membership}
    \zeta_{n,j}(s,y)
    \in
    \mathfrak A_{G_n}^{\varepsilon_{n,j}}(y).
\end{equation}
Consequently, $\mathbb P_{n,j}\in\widetilde{\mathcal P}_{x}\bigl(L^{\beta^{\zeta_{n,j}}}\bigr)\subseteq\mathcal P_{x}\bigl((G_n)^{\varepsilon_{n,j}}\bigr)$.
Writing $\zeta_{n,j}=(c_{n,j},b_{n,j},a_{n,j})$, the covariance coordinates satisfy
\begin{equation}
\label{eq:fixed_n_covariance_floor}
    a_{n,j}(s,y)
    \ge
    \lambda_{n,j}I_d,
    \qquad
    (s,y)\in[0,T]\times\overline D_m.
\end{equation}
By \eqref{eq:fixed_n_outer_support_membership},
\(\varepsilon_{n,j}\le\varepsilon_{n,1}\), and
Lemma~\ref{lem:monotone_outer_generator_regularization}, these fields
are spatially locally bounded uniformly in \(j\) and time, and satisfy
\[
    \ell_{\zeta_{n,j}(s,y)}(J_y^2\phi)
    \le (C_\phi+n^{-1}+\varepsilon_{n,j})\phi(y).
\]

\item \emph{Linear residual inequality.}
For every $(s,y)\in[0,T-t_{n,j}]\times\overline D_m$, one has
\begin{align}
\label{eq:fixed_n_linear_residual_inequality}
    \partial_t v_{n,j}(T-s,y)
    -
    \ell_{\zeta_{n,j}(s,y)}
        \bigl(
            J_y^2v_{n,j}(T-s,\cdot)
        \bigr)
    -
    q_{n,j}(T-s,y)
    \le
    e_{n,j}(T-s,y)+j^{-1}.
\end{align}

\item \emph{Vanishing of the integrated residual.}
Let \(N_{n,j}\) denote the constant \(N_j\) in
Proposition~\ref{prop:usc_smooth_elliptic_approximation}(ii), applied with index \(m+1\) and data \((G_n,q^{(n)},u)\).
The coefficient bounds in Proposition~\ref{prop:usc_outer_compatible_coefficients}(iii) allow Lemma~\ref{lem:usc_localized_krylov_bound} of the Supplementary Material to be applied to
\(g(s,y)=e_{n,j}(T-s,y)\mathbb I_{[0,T-t_{n,j}]}(s)\), extended by zero.
Consequently,
\begin{equation}
\label{eq:fixed_n_integrated_residual_vanishing}
    \begin{aligned}
    \mathbb E^{\mathbb P_{n,j}}\!\left[
        \int_0^{(T-t_{n,j})\wedge\tau_m}
            e_{n,j}(T-s,X_s)\mathbb I_{\{s<\tau_\infty\}}\,ds
    \right]
    &\le N_{n,j}
    \|e_{n,j}\|_{L^{d+1}([t_{n,j},T]\times D_{m+1})}
    \\
    &\longrightarrow0
    \end{aligned}
\end{equation}
as $j\to\infty$.
\end{enumerate}

We may therefore choose \(j(n)\ge n\) recursively sufficiently large
such that, upon setting
\[
\begin{aligned}
    t_n&:=t_{n,j(n)},&
    \varepsilon_n&:=\varepsilon_{n,j(n)},&
    \lambda_n&:=\lambda_{n,j(n)},\\
    v_n&:=v_{n,j(n)},&
    q_n&:=q_{n,j(n)},&
    e_n&:=e_{n,j(n)},\\
    \zeta_n&:=\zeta_{n,j(n)},&
    \mathbb P_n&:=\mathbb P_{n,j(n)},
\end{aligned}
\]
writing $\zeta_n=(c_n,b_n,a_n)$, the following properties hold:
\begin{gather}
    t_1\le\min\{1,T/2\},
    \qquad
    t_n
    \le
    \min\left\{
        n^{-1},
        \tfrac12 t_{n-1}
    \right\},
    \quad n\ge2,
    \qquad
    \varepsilon_n\le n^{-1},
    \label{eq:usc_diagonal_parameter_choice}\\
    \begin{aligned}
    &\zeta_n(s,y)
    \in
    \mathfrak A_{G_n}^{\varepsilon_n}(y),
    \qquad
    (s,y)\in[0,\infty)\times D,\\
    &\mathbb P_n
    \in
    \widetilde{\mathcal P}_{x}
        \bigl(L^{\beta^{\zeta_n}}\bigr)
    \subseteq
    \mathcal P_x\bigl((G_n)^{\varepsilon_n}\bigr),
    \end{aligned}
    \label{eq:usc_diagonal_lyapunov_choice}\\
    \bigl|v_n(T,x)-u(T,x)\bigr|
    \le n^{-1},
    \qquad
    \sup_{(t,y)\in[t_n,T]\times\overline D_m}
        \bigl|q_n(t,y)-q^{(n)}(t,y)\bigr|
    \le n^{-1},
    \label{eq:usc_diagonal_data_choice}\\
    \begin{aligned}
        &\partial_t v_n(T-s,y)
        -
        \ell_{\zeta_n(s,y)}
            \bigl(
                J_y^2v_n(T-s,\cdot)
            \bigr)
        -
        q_n(T-s,y)                                      \\
        &\hspace{12em}
        \le
        e_n(T-s,y)+n^{-1},
        \qquad
        (s,y)\in[0,T-t_n]\times\overline D_m,
    \end{aligned}
    \label{eq:usc_diagonal_linear_residual_choice}\\
    \mathbb E^{\mathbb P_n}\!\left[
        \int_0^{(T-t_n)\wedge\tau_m}
            e_n(T-s,X_s)\mathbb I_{\{s<\tau_\infty\}}\,ds
    \right]
    \le n^{-1},
    \label{eq:usc_diagonal_residual_choice}\\
    \|v_n\|_{L^\infty([t_n,T]\times D_{m+1})}
    \le\|u\|_\infty, \qquad
    v_n(t,y)
    \le
    u^{\langle n\rangle}(t,y)+n^{-1},
    \qquad
    (t,y)\in
    [t_n,T]\times\overline D_m.
    \label{eq:usc_diagonal_value_choice}
\end{gather}
In particular, \(t_n\searrow0\).

We next prove part~\ref{lem:usc_diagonal_compactness}.
For each \(n\ge1\), set $H_n:=(G_n)^{\varepsilon_n}$ and $\mathfrak B_n(y):=\mathfrak A_{G_n}^{\varepsilon_n}(y)$.
Then by \eqref{eq:usc_diagonal_lyapunov_choice}, $\mathbb P_n\in\mathcal P_x(H_n)$.
Moreover, by Lemma~\ref{lem:monotone_outer_generator_regularization}, \(H_n\) is a jointly continuous generating function with compact convex
support correspondence \(\mathfrak B_n\), and
\begin{equation}
\label{eq:H_n_G_n_error}
    0
    \le
    H_n(y,U)-G_n(y,U)
    \le
    \varepsilon_n\kappa_\phi(y)\|U\|,
    \qquad
    (y,U)\in D\times\mathfrak J.
\end{equation}

The locally uniform support bound in
Proposition~\ref{prop:continuous_outer_generator_approximation}(ii) of the Supplementary Material, together with
Lemma~\ref{lem:monotone_outer_generator_regularization} and
\(\varepsilon_n\le1\), implies that
\((\mathfrak B_n)_{n\ge1}\) is locally uniformly bounded.
Since
\(\zeta_n(s,y)\in\mathfrak B_n(y)\), the coefficient fields
\((\beta^{\zeta_n})_{n\ge1}\) are therefore uniformly bounded before $\tau_k$ for any $k\ge1$. Furthermore, by \eqref{eq:G_n_lyapunov_bound}, \eqref{eq:usc_diagonal_lyapunov_choice}, \eqref{eq:H_n_G_n_error}, and
\(\varepsilon_n\le n^{-1}\), one has, on
\(\{s<\tau_\infty\}\),
\[
\begin{aligned}
    L^{\beta^{\zeta_n}}
        \bigl(s,\omega,J_{X_s}^2\phi\bigr)
    &=
    \ell_{\zeta_n(s,X_s)}
        \bigl(J_{X_s}^2\phi\bigr)                                     \le
    H_n\bigl(X_s,J_{X_s}^2\phi\bigr)                                   \\
    &\le
    \bigl(C_\phi+n^{-1}+\varepsilon_n\bigr)\phi(X_s)
    \le
    (C_\phi+2)\phi(X_s).
\end{aligned}
\]
Thus by Proposition~\ref{prop:compactness_criterion} in the Supplementary Material, \((\mathbb P_n)_{n\ge1}\) is tight.

Let \(\mathbb P\) be a weak cluster point of
\((\mathbb P_n)_{n\ge1}\). Passing to a subsequence, we may assume that
\(\mathbb P_n\Rightarrow\mathbb P\). We show that
\(\mathbb P\in\mathcal P_x(G)\). By
Definition~\ref{def:G_supermartingale_problem}, it suffices to verify
that \(\mathbb P\) solves the generalized \(G\)-supermartingale problem
starting from \((0,x)\).

For every \(n\), \(\mathbb P_n\in\mathcal P_x(H_n)\), and hence \(\mathbb P_n(X_0=x)=1\). 
Since the time-zero evaluation map is continuous on the canonical path space, weak convergence gives \(\mathbb P(X_0=x)=1\).
On the other hand, Lemma~\ref{lem:monotone_outer_generator_regularization}, \eqref{eq:H_n_G_n_error}, and \eqref{eq:G_n_upperconsistency_G} imply
\begin{equation}
\label{eq:H_n_upperconsistency_G}
    \limsup_{n\to\infty}^{*}H_n
    \le G
    \qquad\text{on }D\times\mathfrak J.
\end{equation}
Indeed, if \((y_n,U_n)\to(y,U)\), then
\[
    \limsup_{n\to\infty}H_n(y_n,U_n)
    \le
    \limsup_{n\to\infty}
    \Bigl[
        G_n(y_n,U_n)
        +
        \varepsilon_n\kappa_\phi(y_n)\|U_n\|
    \Bigr]\le G(y,U),
\]
because the second term converges to zero.
Fix \(R>0\), \(\ell\ge1\), and \(f\in C_b^\infty(D)\), and define
\[
    q_n^f(s,y)
    :=H_n\bigl(y,J_y^2f\bigr),
    \qquad
    q^f(s,y)
    :=G\bigl(y,J_y^2f\bigr),
    \qquad
    (s,y)\in[0,R]\times D.
\]
Since \(y\mapsto J_y^2f\) is continuous,
\eqref{eq:H_n_upperconsistency_G} yields
\begin{equation}
\label{eq:H_n_test_source_upperconsistency_G}
    \limsup_{n\to\infty}^{*}q_n^f
    \le q^f
    \qquad\text{on }[0,R]\times D.
\end{equation}
The locally uniform boundedness of \((\mathfrak B_n)_{n\ge1}\) also gives
\[
    \sup_{n\ge1}
    \|q_n^f\|_{L^\infty([0,R]\times\overline D_\ell)}
    <\infty.
\]
Set \(\kappa_r^n:=k^{\beta^{\zeta_n}}(r,X)\).
The local-characteristic identification in the proof of
Proposition~\ref{prop:recovery_aggregation_generator_compatible_characteristics}
applies to the fields $\beta^{\zeta_n}$. Thus the processes
\[
    N_s-\int_0^s\kappa_r^n\mathbb I_{\{r<\tau_\infty\}}\,dr,\qquad s\ge0,
\]
are \(\mathbb P_n\)-local martingales. The local coefficient bounds
above also show that the densities \(\kappa^n\mathbb I_{\{\cdot<\tau_\infty\}}\) are uniformly
essentially bounded on finite time intervals while
\(\overline X\) remains in a fixed compact subset of \(D\).
Since \(\mathbb P_n\in\mathcal P_x(H_n)\), the process
\[
    s\longmapsto
    f\bigl(X_{[s]_{0,\tau_\ell}}\bigr)
    \mathbb I_{\{[s]_{0,\tau_\ell}<\tau_\infty\}}
    -
    \int_0^{[s]_{0,\tau_\ell}}
        H_n\bigl(X_r,J_{X_r}^2f\bigr)\mathbb I_{\{r<\tau_\infty\}}\,dr,
    \qquad s\ge0,
\]
is a \(\mathbb P_n\)-supermartingale.
Hence, by \eqref{eq:H_n_test_source_upperconsistency_G} and the
moving-source assertion of Proposition~\ref{prop:weak_passage_martingale_relations} in the Supplementary Material,
applied with the preceding \(\kappa^n\),
$\widetilde M^{f,\ell}$ is a \(\mathbb P\)-supermartingale on
\([0,R]\).
Since \(R>0\), \(\ell\ge1\), and \(f\in C_b^\infty(D)\) were arbitrary, \(\mathbb P\) solves the generalized \(G\)-supermartingale problem.
Therefore, Definition~\ref{def:G_supermartingale_problem} gives
$\mathbb P\in\mathcal P_x(G)$.
This proves part~\ref{lem:usc_diagonal_compactness}.

Fix now an open set \(O\) with \(x\in O\Subset D_m\).
Then \(\sigma_n^O\le(T-t_n)\wedge\tau_m\).
Using a smooth bounded extension of \((s,y)\mapsto v_n(T-s,y)\) from a neighborhood of \([0,T-t_n]\times\overline D_m\), by Proposition~\ref{prop:space_time_linear_martingale_extension}, we have
\begin{align*}
    v_n(T,x)
    =
    \mathbb E^{\mathbb P_n}\!\Bigg[
        &v_n\bigl(T-\sigma_n^O,X_{\sigma_n^O}\bigr)
        \mathbb I_{\{\sigma_n^O<\tau_\infty\}}\\
        &+\int_0^{\sigma_n^O}
        \Bigl(
            \partial_tv_n(T-s,X_s)
            -\ell_{\zeta_n(s,X_s)}
                \bigl(J_{X_s}^2v_n(T-s,\cdot)\bigr)
        \Bigr)\mathbb I_{\{s<\tau_\infty\}}\,ds
    \Bigg].
\end{align*}
Combining the preceding identity with \eqref{eq:usc_diagonal_linear_residual_choice} and \eqref{eq:usc_diagonal_residual_choice} yields part~\ref{lem:usc_diagonal_one_step_inequality}.

The first assertion of part~\ref{lem:usc_diagonal_value_convergence} follows from \eqref{eq:usc_diagonal_data_choice} and \eqref{eq:usc_diagonal_value_choice}.
Let
\[
    (s_n,y_n)
    \in[t_n,T]\times\overline D_m,
    \qquad
    (s_n,y_n)\to(t,y).
\]
Choose \(\xi_n\in E_T\) such that
\[
    d\bigl(\xi_n,(s_n,y_n)\bigr)\le n^{-1},
    \qquad
    u^{\langle n\rangle}(s_n,y_n)
    \le u(\xi_n)+n^{-1}.
\]
Then \(\xi_n\to(t,y)\), and
\eqref{eq:usc_diagonal_value_choice} gives
\[
    v_n(s_n,y_n)
    \le u(\xi_n)+2n^{-1}.
\]
Upper semicontinuity of \(u\) proves
\eqref{eq:usc_diagonal_value_upper_limit}.

Finally, we verify part~ \ref{lem:usc_diagonal_source_convergence}.
By \eqref{eq:usc_diagonal_data_choice}, we have
\[
    \sup_{n\ge1}
    \|q_n\|_{L^\infty(
        [t_n,T]\times\overline D_m)}
    \le\|q\|_\infty+1,
\]
which proves \eqref{eq:usc_diagonal_source_bound}.
Since \(q^{(n)}\searrow q\) and each \(q^{(n)}\) is continuous,
the monotone upper-half-relaxed-limit argument used to verify \ref{V3}
in the proof of
Proposition~\ref{prop:stable_virtual_structure_generates_valuation}
gives \(\limsup_{n\to\infty}^{*}q^{(n)}\le q\) on \(E_T\).
Together with the uniform source error in
\eqref{eq:usc_diagonal_data_choice}, this proves
\eqref{eq:usc_diagonal_source_upper_limit}.
This proves part~\ref{lem:usc_diagonal_source_convergence} and completes the
proof.
\end{proof}

\subsection{Stopped realization}

Now we prove Theorem~\ref{prop:one_step_usc_subsolution_realization}.
For \(h>0\), an open set \(O\Subset D\), \(a\in[0,h)\), and
\(\widetilde\omega\in\widetilde\Omega\), set $\sigma_{h,O}^{a}(\widetilde\omega):=(h-a)\wedge\rho_O^{\,0}(\widetilde\omega)$.
Define the stopped-endpoint map $\mathsf{End}_{h,O}:\widetilde\Omega\times[0,h)\to[0,h]\times\widehat D\times\{0,1\}$ by
\[
\begin{aligned}
    \mathsf{End}_{h,O}(\widetilde\omega,a)
    :=
    \Bigl(
        \sigma_{h,O}^{a}(\widetilde\omega),
        X_{\sigma_{h,O}^{a}(\widetilde\omega)}(\widetilde\omega),
        \mathbb I_{\{\sigma_{h,O}^{a}(\widetilde\omega)<\tau_\infty(\widetilde\omega)\}}
    \Bigr).
\end{aligned}
\]
Given \(\mathbb P\in\mathfrak M\) and \(h>0\), a bounded open set
\(O\Subset D\) is called \emph{\((\mathbb P,h)\)-regular} if
\(\mathsf{End}_{h,O}\) is continuous at \((\widetilde\omega,0)\)
for \(\mathbb P\)-almost every \(\widetilde\omega\), where
\(\widetilde\Omega\times[0,h)\) carries the product of the \(J_1\)
and Euclidean topologies. An increasing sequence of open sets
\((O_k)_{k\ge1}\) is a \emph{\((\mathbb P,h)\)-regular exhaustion}
of \(O\Subset D\) if
\[
    O_k\Subset O_{k+1}\Subset O,
    \qquad \bigcup_{k\ge1}O_k=O,
\]
and every \(O_k\) is \((\mathbb P,h)\)-regular.

\begin{proof}[Proof of Theorem~\ref{prop:one_step_usc_subsolution_realization}]
Replacing \(u(s,y)\) and \(q(s,y)\) by
\[
    \overline u(s,y):=u(T-h+s,y),
    \qquad
    \overline q(s,y):=q(T-h+s,y),
    \qquad 0\le s\le h+1,
\]
we may assume that \(T=h\). Choose \(m\ge1\) such that \(\overline O\subset D_m\), and apply
Lemma~\ref{lem:usc_diagonal_stochastic_approximation} with $T=h$.
After passing to a subsequence, we may assume that $\mathbb P_n\Rightarrow\mathbb P\in\mathcal P_x(G)$.
By Proposition~\ref{prop:recovery_aggregation_generator_compatible_characteristics}%
\ref{prop:recovery_killing_compensator},
\(\mathbb P(\tau_{\mathrm{kill}}=h)=0\). 
By Proposition~\ref{lem:regular_inner_exhaustion} in the
Supplementary Material, there exists a \((\mathbb P,h)\)-regular exhaustion
\[
    x\in O_1,
    \qquad
    O_k\Subset O_{k+1}\Subset O,
    \qquad
    O_k\uparrow O.
\]
For \(k,n\ge1\), set $\sigma_n^k:=(h-t_n)\wedge\rho_{O_k}^{\,0}$ and $\sigma^k:=h\wedge\rho_{O_k}^{\,0}$.
By
Lemma~\ref{lem:usc_diagonal_stochastic_approximation}%
\ref{lem:usc_diagonal_one_step_inequality},
\begin{equation}
\label{eq:usc_diagonal_approximate_realization}
\begin{aligned}
    v_n(h,x)
    &\le
    \mathbb E^{\mathbb P_n}\!\Bigg[
        v_n\bigl(h-\sigma_n^k,X_{\sigma_n^k}\bigr)
        \mathbb I_{\{\sigma_n^k<\tau_\infty\}}\\
    &\hspace{3em}+\int_0^{\sigma_n^k}
            q_n(h-s,X_s)\mathbb I_{\{s<\tau_\infty\}}\,ds
    \Bigg]
    +\frac{h+1}{n}.
\end{aligned}
\end{equation}
For each fixed \(k\), Proposition~\ref{lem:regular_stopped_endpoint_passage} in the Supplementary Material, together with
\eqref{eq:usc_diagonal_prescribed_point_convergence},
\eqref{eq:usc_diagonal_value_upper_limit}, and
\eqref{eq:usc_diagonal_source_upper_limit}, allow us to let
\(n\to\infty\) in \eqref{eq:usc_diagonal_approximate_realization}. We obtain
\begin{align}
\label{eq:usc_diagonal_regular_domain_limit}
    u(h,x)
    \le
    \mathbb E^{\mathbb P}\!\Bigg[
        u\bigl(h-\sigma^k,X_{\sigma^k}\bigr)\mathbb I_{\{\sigma^k<\tau_\infty\}}+\int_0^{\sigma^k}q(h-s,X_s)\mathbb I_{\{s<\tau_\infty\}}\,ds
    \Bigg].
\end{align}

By Proposition~\ref{lem:regular_inner_exhaustion} in the Supplementary Material,
\[
    \rho_{O_k}^{\,0}\uparrow\rho_O^{\,0},
    \qquad \mathbb P\text{-almost surely}.
\]
The live exit time \(\widehat\tau_O^{\,0}\) is predictable, whereas
\(\tau_{\mathrm{kill}}\) is totally inaccessible under \(\mathbb P\).
Consequently,
\[
    \mathbb P\bigl(
        \widehat\tau_O^{\,0}
        =\tau_{\mathrm{kill}}<\infty
    \bigr)=0,
\]
and hence
\[
    \bigl(\sigma^k,X_{\sigma^k},\mathbb I_{\{\sigma^k<\tau_\infty\}}\bigr)
    \longrightarrow
    \bigl(
        h\wedge\rho_O^{\,0},
        X_{h\wedge\rho_O^{\,0}},
        \mathbb I_{\{h\wedge\rho_O^{\,0}<\tau_\infty\}}
    \bigr)
\]
\(\mathbb P\)-almost surely. Since \(u\) is bounded and upper
semicontinuous, Fatou lemma handles the terminal term in
\eqref{eq:usc_diagonal_regular_domain_limit}. Since \(q\) is bounded and
\(\sigma^k\uparrow h\wedge\rho_O^{\,0}\), bounded convergence handles the
source term. Letting \(k\to\infty\) gives
\[
\begin{aligned}
    u(h,x)
    \le
    \mathbb E^{\mathbb P}\!\Bigg[
        &u\bigl(
            h-(h\wedge\rho_O^{\,0}),
            X_{h\wedge\rho_O^{\,0}}
        \bigr)
        \mathbb I_{\{h\wedge\rho_O^{\,0}<\tau_\infty\}}
        +\int_0^{h\wedge\rho_O^{\,0}}
            q(h-s,X_s)\mathbb I_{\{s<\tau_\infty\}}\,ds
    \Bigg].
\end{aligned}
\]
Undoing the time shift and taking
\(\mathbb P^{T,h,O,x}:=\mathbb P\) proves
\eqref{eq:one_step_usc_realization}.
\end{proof}

\noeqref{eq:initialconditionP}


\section*{Statements and Declarations}

\subsection*{Conflict of interest}
The authors have no conflicts of interest to declare.

\subsection*{Author contributions}
All authors equally contributed to the writing and to the review of the entire
manuscript.

\subsection*{Data availability}
No datasets were generated or analysed during this study.

\subsection*{Funding}
Hyungbin Park was supported by National Research Foundation of Korea grants funded by the Ministry of Science and ICT (2022R1A5A6000840, RS-2026-25488333).
Financial support from the Institute for Research in Finance and Economics of Seoul National University is gratefully acknowledged. David Criens gratefully acknowledges support from the Freiburg Center for Data Analysis, Modeling and AI.

\section*{Supplementary Material}
\noindent\textbf{Supplement to ``Sublinear Markov Semigroups Beyond Uniqueness''}
\par\noindent{The supplement contains Appendices D-F: measurable selection and countable determining classes; probabilistic estimates, tightness, and weak stability; and approximation and weak existence for one-step realization.}


\startcombinedsupplement

\title{Supplement to ``Sublinear Markov Semigroups Beyond Uniqueness''}

\author{
Jongjin Park\thanks{\raggedright Research Institute of Mathematics, Seoul National University, Seoul, Republic of Korea. Email: \texttt{pjj4230@snu.ac.kr}.}
\and
David Criens\thanks{\raggedright Department of Mathematical Stochastics, University of Freiburg, Freiburg, Germany. Email: \texttt{david.criens@stochastik.uni-freiburg.de}.}
\and
Hyungbin Park\thanks{\raggedright Department of Mathematical Sciences, Seoul National University, Seoul, Republic of Korea. Emails: \texttt{hyungbin@snu.ac.kr}, \texttt{hyungbin2015@gmail.com}.}
}
\date{}
\maketitle

\begin{abstract}
This supplement collects the auxiliary results used in the main article.
Appendix D provides measurable-selection and countable-determination tools.
Appendix E establishes local characteristics, probabilistic estimates,
tightness, and weak stability on the canonical path space.
Appendix F develops the approximation and weak-existence results used for
one-step realization. The results are arranged so that their proofs use
only earlier auxiliary results and the indicated results of the main article.
\end{abstract}

\noindent\textbf{2020 Mathematics Subject Classification.} Primary: 60J25, 60G44. Secondary: 47H20, 49L25, 91G80.\par\smallskip

\noindent\textbf{Keywords.} sublinear sub-Markov semigroup, stable structure of laws, killing, weak convergence, compactness.\par\medskip

\setcounter{secnumdepth}{3}
\setcounter{tocdepth}{2}

\tableofcontents

$ $

This supplement contains Appendices D--F, continuing the numbering after Appendices A--C of the main article.
References to Appendices A--C of the main article are identified explicitly;
unqualified appendix references refer to this supplement.
Unless redefined, the notation
$D,D_m,\widetilde\Omega,\mathfrak M,G$ and $\mathfrak A_G$
is inherited from the main text.
Here $G$ satisfies \textup{(G1)}--\textup{(G3)} and is not assumed to be
an exactly generated coefficient function.  The half-relaxed-limit
conventions are given in Section~\ref{sec:valuation_uncertainty} and the support correspondence in Subsection~\ref{subsec:generating_function} of the main text.
The hypotheses needed for each auxiliary result are stated explicitly.

\clearpage
\appendix
\renewcommand{\theequation}{\thesection.\arabic{equation}}
\setcounter{section}{3}

\section{Measurable Selection and Countable Determining Classes}
\label{sec:mixedtopology}

We collect the measurable-selection and countable-determination results used
throughout the main article. The first result is the standard selection
theorem for compact-valued Borel correspondences.

\begin{proposition}
\label{lem:measurable_model_selectors}
Let \(E\) be a standard Borel space, let \(Y\) be Polish, and let \(\Gamma:E\rightrightarrows Y\) have nonempty compact values and Borel graph.
Then \(\Gamma\) admits a Borel selector.
\end{proposition}

\begin{proof}
This is the Arsenin--Kunugui uniformization theorem \cite[Theorem~18.18]{supp:kechris1995classical}.
\end{proof}

The next two results provide countable smooth tests and cylinder multipliers
for verifying that \(\widetilde M^{f,n,\beta}\) is a martingale and that
\(\widetilde M^{f,n}\) and \(Y^{\mathcal T;R,g}\) are supermartingales
under the law in question.

\begin{proposition}
\label{lem:countable_smooth_determining_class}
There exists a countable class
\(\mathscr F_{\mathrm{sm}}\subset C_c^\infty(D)\) such that, for every
\(n\ge1\) and \(f\in C_b^\infty(D)\), some
\((f_j)_{j\ge1}\subset\mathscr F_{\mathrm{sm}}\) satisfies
\[
    (f_j,\nabla f_j,\nabla^2f_j)
    \longrightarrow
    (f,\nabla f,\nabla^2f)
    \quad\text{uniformly on }\overline D_n.
\]
In particular,
\(\{J_x^2f:f\in\mathscr F_{\mathrm{sm}}\}\) is dense in
\(\mathfrak J\) for every \(x\in D\).
\end{proposition}

\begin{proof}
For each \(n\ge1\), define
\[
    \|h\|_{*,n}
    :=
    \|h\|_\infty
    +
    \|\nabla h\|_{L^\infty(\overline D_n)}
    +
    \|\nabla^2h\|_{L^\infty(\overline D_n)},
    \qquad h\in C_c^\infty(D).
\]
The space \(C_c^\infty(D)\), endowed with \(\|\cdot\|_{*,n}\), is
separable.  Indeed, the map
\[
    h
    \longmapsto
    \bigl(
        h,\nabla h|_{\overline D_n},
        \nabla^2h|_{\overline D_n}
    \bigr)
\]
embeds it isometrically into a finite product of the separable spaces
\(C_c(D)\) and \(C(\overline D_n)\).  Hence, for every \(n\ge1\), we
may choose a countable family
\(\mathscr A_n\subset C_c^\infty(D)\) that is dense with respect to
\(\|\cdot\|_{*,n}\).  Set
\[
    \mathscr F_{\mathrm{sm}}
    :=
    \bigcup_{n\ge1}\mathscr A_n.
\]
Then \(\mathscr F_{\mathrm{sm}}\) is countable.

Fix \(n\ge1\) and \(f\in C_b^\infty(D)\).  Choose
\(\chi\in C_c^\infty(D)\) such that \(\chi=1\) on a neighborhood of
\(\overline D_n\).  Since \(\chi f\in C_c^\infty(D)\), there exists
\((f_j)_{j\ge1}\subset\mathscr A_n\) such that
\(\|f_j-\chi f\|_{*,n}\to0\).  Because \(\chi f=f\) on a
neighborhood of \(\overline D_n\), this proves the asserted local
\(C^2\)-approximation.

Finally, for every \(x\in D\) and \(U\in\mathfrak J\), there exists
\(f\in C_b^\infty(D)\) such that \(J_x^2f=U\).  Applying the preceding
approximation on a localization domain containing \(x\) proves the
pointwise jet-density assertion.
\end{proof}

\begin{proposition}
\label{lem:countable_determining_cylinder_classes}
For every \(r\ge0\), there exists a countable class
\(\mathscr H_{r,0}\) of bounded nonnegative
\(\widetilde{\mathcal F}_r\)-measurable cylinder functions with the
following property.  Let \(\mathbb P\) be a probability measure on
\(\widetilde\Omega\), and let \(U,V\in L^1(\mathbb P)\).  If
\begin{align}
\label{eq:countable_indicator_determining_supermtgineq}
    \mathbb E^{\mathbb P}[F U]
    \le
    \mathbb E^{\mathbb P}[F V]
    \qquad
    \text{for every }F\in\mathscr H_{r,0},
\end{align}
then \eqref{eq:countable_indicator_determining_supermtgineq} holds for every
bounded nonnegative \(\widetilde{\mathcal F}_r\)-measurable random variable \(F\).
\end{proposition}

\begin{proof}
Since \(\widehat D\) is compact and metrizable, there exists a countable
algebra \(\mathscr B_0\subset\mathcal B(\widehat D)\) such that
\(\sigma(\mathscr B_0)=\mathcal B(\widehat D)\).  For \(r\ge0\), set
\(\mathbb T_{r,0}:=(\mathbb Q_+\cap[0,r))\cup\{r\}\), and let
\(\mathscr A_{r,0}\) be the algebra of subsets of
\(\widetilde\Omega\) generated by
\[
    \Bigl\{
        \{X_q\in B\}:
        q\in\mathbb T_{r,0},\ 
        B\in\mathscr B_0
    \Bigr\}.
\]
Define
\(\mathscr H_{r,0}:=\{\mathbb I_A:A\in\mathscr A_{r,0}\}\).  This
class is countable and consists of bounded nonnegative
\(\widetilde{\mathcal F}_r\)-measurable cylinder functions.

We first record that
\begin{equation}
\label{eq:countable_cylinder_generation}
    \sigma(\mathscr H_{r,0})=\widetilde{\mathcal F}_r.
\end{equation}
Indeed, the inclusion from left to right is immediate.  Conversely, each
\(X_q\), \(q\in\mathbb T_{r,0}\), is
\(\sigma(\mathscr A_{r,0})\)-measurable.  If \(0\le u<r\), choose
\((q_n)_{n\ge1}\subset\mathbb Q_+\cap(u,r)\) with \(q_n\downarrow u\).
Right continuity gives \(X_{q_n}\to X_u\), so \(X_u\) is also
\(\sigma(\mathscr A_{r,0})\)-measurable.  Since
\(r\in\mathbb T_{r,0}\), this proves
\eqref{eq:countable_cylinder_generation}.

Put \(Z:=U-V\) and define
\[
    \mathscr D
    :=
    \left\{
        A\in\widetilde{\mathcal F}_r:
        \mathbb E^{\mathbb P}[\mathbb I_A Z]\le0
    \right\}.
\]
Since \(Z\in L^1(\mathbb P)\), the finite signed measure
\(A\mapsto\mathbb E^{\mathbb P}[\mathbb I_A Z]\) is continuous along
increasing and decreasing sequences.  Thus \(\mathscr D\) is a monotone
class.  By
\eqref{eq:countable_indicator_determining_supermtgineq},
\(\mathscr A_{r,0}\subseteq\mathscr D\), and the monotone class theorem,
together with \eqref{eq:countable_cylinder_generation}, gives
\(\widetilde{\mathcal F}_r\subseteq\mathscr D\).  Approximation by
bounded nonnegative simple functions proves
\eqref{eq:countable_indicator_determining_supermtgineq} for all bounded
nonnegative \(\widetilde{\mathcal F}_r\)-measurable \(F\).
\end{proof}

\section{Probabilistic Estimates, Tightness, and Weak Stability}
\label{app:virtual_path_weak_convergence}

We first identify the local characteristics of linear martingale-problem solutions and derive short-time and Lyapunov estimates.
These estimates provide tightness. 
We then show that stopped compensated processes remain (super)martingales
under weak convergence and pass expectation bounds for killed functionals to weak
limits, followed by closedness and countable determination of the inequalities
\eqref{eq:onestep-supermtgrelation_Ciota} within the admissible family.

\subsection{Local Characteristics and Space-Time Tests}\label{subsec:gmp_basic_properties}

\begin{lemma}
\label{prop:localized_characteristics_linear_gmp}
Fix a coefficient field \(\beta=(c,b,a)\), set \(k:=k^\beta=-c\), and let
\(\mathbb P\in\widetilde{\mathcal P}_{t,\widetilde\omega}(L^\beta)\), where
\(t<\tau_\infty(\widetilde\omega)\). Let \(O\Subset D\) be open. Then, with
respect to \(\widetilde{\mathbb F}^{\,\mathbb P}\), for every \(s\ge t\),
\begin{align}
\label{eq:lem:localized_characteristics_linear_gmp_1}
\overline X_{s\wedge\rho_O^{\,t}}
&=
\widetilde\omega(t)
+
\int_t^{s\wedge\rho_O^{\,t}}b_r\,dr
+
M_{s\wedge\rho_O^{\,t}},
\end{align}
where \(M\) is a continuous local martingale with \(M_t=0\) and
\begin{align}
\label{eq:lem:localized_characteristics_linear_gmp_2}
\langle M^i,M^j\rangle_{s\wedge\rho_O^{\,t}}
&=
\int_t^{s\wedge\rho_O^{\,t}}a_r^{ij}\,dr,
\qquad 1\le i,j\le d.
\end{align}
Moreover,
\begin{equation}
\label{eq:localized_killing_martingale_proof}
\mathbb I_{\{t<\tau_{\mathrm{kill}}\le s\wedge\widehat\tau_O^{\,t}\}}
-
\int_t^{s\wedge\widehat\tau_O^{\,t}}k_r\mathbb I_{\{r<\tau_\infty\}}\,dr,
\qquad s\ge t,
\end{equation}
is a local martingale.
\end{lemma}

\begin{proof}
The assertion is trivial if \(\widetilde\omega(t)\notin O\).
Otherwise, choose \(n\) such that
\(\widetilde\omega([0,t])\cup\overline O\subset D_n\).
Before \(\widehat\tau_O^{\,t}\), continuous explosion is impossible, and
hence \(\mathbb I_{\{s<\tau_\infty\}}=1-N_s^t\) for \(t\le s<\widehat\tau_O^{\,t}\). Since
\(\mathbb P\in\widetilde{\mathcal P}_{t,\widetilde\omega}(L^\beta)\), the
process \(\widetilde M^{1,n,\beta}\) in
equation~\eqref{eq:defM_f,nlinearmtg} in the main text, stopped at \(\widehat\tau_O^{\,t}\), is a
\(\mathbb P\)-martingale. Using \(c=-k\) and this identity gives
\eqref{eq:localized_killing_martingale_proof}.

Now let \(f\in C_c^\infty(\mathbb R^d)\), and choose a cutoff
\(\chi\in C_c^\infty(D)\) that is equal to one on a neighborhood of
\(\overline O\). The process \(\widetilde M^{\chi f,n,\beta}\), stopped at
\(\widehat\tau_O^{\,t}\), is again a local martingale. Since
\(\chi\equiv1\) on a neighborhood of \(\overline O\), using
\eqref{eq:localized_killing_martingale_proof} to compensate its single
killing jump cancels the zeroth-order terms and shows that
\begin{equation}
\label{eq:localized_recovered_martingale_problem}
 f(\overline X_{s\wedge\rho_O^{\,t}})-f(\widetilde\omega(t))
 -\int_t^{s\wedge\rho_O^{\,t}}
 \left[
 b_r\cdot\nabla f(\overline X_r)
 +\frac12\operatorname{tr}
 \bigl(a_r\nabla^2f(\overline X_r)\bigr)
 \right]dr
\end{equation}
is a local martingale. The standard martingale-problem characterization
of semimartingale characteristics
\cite[Theorem~II.2.42]{supp:jacod2013limit}, applied to
\eqref{eq:localized_recovered_martingale_problem}, yields the continuous
characteristics
\[
\left(\int_0^\cdot b_r\,dr,\int_0^\cdot a_r\,dr,0\right),
\]
which is equivalent to
\eqref{eq:lem:localized_characteristics_linear_gmp_1}--
\eqref{eq:lem:localized_characteristics_linear_gmp_2}.
\end{proof}

\begin{proposition}
\label{prop:space_time_linear_martingale_extension}
Fix \((t,\widetilde\omega)\in[0,\infty)\times\widetilde\Omega\),
\(u\in C_b^\infty([0,\infty)\times D)\), and \(n\ge1\).
If \(\mathbb P\in
\widetilde{\mathcal P}_{t,\widetilde\omega}(L^\beta)\), then
\begin{equation}
\label{eq:defM_u_n_linear}
\begin{aligned}
\widetilde M_s^{u,n,\beta}
:={}&
 u([s]_{t,\tau_n},X_{[s]_{t,\tau_n}})
 \mathbb I_{\{\tau_\infty>[s]_{t,\tau_n}\}}\\
&-
\int_t^{[s]_{t,\tau_n}}
\left[
\partial_r u(r,X_r)
+L^\beta(r,X,J^2_{X_r}u(r,\cdot))
\right]
\mathbb I_{\{\tau_\infty>r\}}\,dr.
\end{aligned}
\end{equation}
is a \(\mathbb P\)-martingale.
\end{proposition}

\begin{proof}
It remains only to consider the case in which both
\(t<\tau_\infty(\widetilde\omega)\) and
\(\tau_n(\widetilde\omega)>t\) hold, since the assertion is immediate
whenever either condition fails.  In this case, the assertion follows by Lemma~\ref{prop:localized_characteristics_linear_gmp} with It\^o's formula applied to the process $u([s]_{t,\tau_n},X_{[s]_{t,\tau_n}})\mathbb I_{\{\tau_\infty>[s]_{t,\tau_n}\}}$.
\end{proof}

\subsection{Short-Time Estimates, Lyapunov Bounds, and Tightness}\label{subsec:small_time_estimates}

\label{subsec:nonexplosive_criterion}

\begin{proposition}
\label{prop:short_time_control_bounded_coefficients}
\begin{enumerate}[label=(\roman*), ref=(\roman*)]
\item\label{prop:short_time_control_deterministic}
Let \(\beta\) be a coefficient field and let
\[
    \mathbb P\in
    \widetilde{\mathcal P}_{t,\widetilde\omega}(L^\beta),
    \qquad
    t<\tau_\infty(\widetilde\omega).
\]
Let \(O\Subset D\) contain \(x:=\widetilde\omega(t)\).  Fix
\(T>t\) and suppose that, for some \(K<\infty\),
\begin{equation}
\label{eq:short_time_local_coefficient_bound}
    |c_s|+|b_s|+\|a_s\|\le K
\end{equation}
for \(ds\otimes d\mathbb P\)-a.e. \((s,\eta)\) satisfying
\(t<s\le T\) and \(s<\rho_O^{\,t}(\eta)\).  
Then there exists a constant $c_d>0$, depending only on $d$, such that for every $0<h\le T-t$ and $\varepsilon>0$,
\begin{align}
\mathbb P\bigl(t<\tau_{\mathrm{kill}}
\le(t+h)\wedge\widehat\tau_O^{\,t}\bigr)
&\le Kh,
\label{eq:short_time_killing_bound}
\\
\mathbb P\left(
\sup_{t\le s\le t+h}
|\overline X_{s\wedge\rho_O^{\,t}}-x|>\varepsilon
\right)
&\le
\frac{c_d(K^2h^2+Kh)}{\varepsilon^2}.
\label{eq:short_time_oscillation_bound}
\end{align}

\item\label{prop:short_time_control_exponential_exit}
Let \(G\) satisfy \textup{(G1)}--\textup{(G3)} and let
\(\overline{B_r(x_0)}\subset D\).  Then, as \(h\downarrow0\),
\[
\sup_{y\in B_{r/2}(x_0)}
\sup_{\mathbb P\in\mathcal P_y(G)}
\mathbb P\bigl(
\rho_{B_r(x_0)}^{\,0}\le h,
\ \rho_{B_r(x_0)}^{\,0}<\tau_{\mathrm{kill}}
\bigr)=o(h).
\]
\end{enumerate}
\end{proposition}

\begin{proof}
We first establish the conditional estimates used to prove
part~\ref{prop:short_time_control_deterministic}.  Under its assumptions,
let \(\sigma\) and \(\theta\) be bounded
\(\widetilde{\mathbb F}^{\,\mathbb P}\)-stopping times such that
\[
    t\le\sigma\le\theta\le T,
    \qquad
    \theta\le(\sigma+\delta)\wedge T
\]
for some \(\delta\ge0\).  Set
\[
    N_s^O
    :=\mathbb I_{\{t<\tau_{\mathrm{kill}}
                       \le s\wedge\widehat\tau_O^{\,t}\}},
    \qquad
    I_s^O
    :=\int_t^{s\wedge\widehat\tau_O^{\,t}}
       (-c_r)\mathbb I_{\{r<\tau_\infty\}}\,dr.
\]
By Lemma~\ref{prop:localized_characteristics_linear_gmp},
\(N^O-I^O\), stopped at \(T\), is a true martingale.  Optional sampling,
\(-c\le K\), and \(\theta-\sigma\le\delta\) give
\begin{equation}
\label{eq:conditional_short_time_killing_bound}
\mathbb P\left(
   \sigma<\tau_{\mathrm{kill}}
   \le\theta\wedge\widehat\tau_O^{\,t}
   \,\middle|\,
   \widetilde{\mathcal F}^{\,\mathbb P}_\sigma
\right)
\le K\delta.
\end{equation}
The same lemma yields, before \(\rho_O^{\,t}\),
\[
    \overline X_{s\wedge\rho_O^{\,t}}
    =x+\int_t^{s\wedge\rho_O^{\,t}}b_r\,dr
      +M_{s\wedge\rho_O^{\,t}},
    \qquad
    d\langle M^i,M^j\rangle_r
    =a_r^{ij}\mathbb I_{\{r<\rho_O^{\,t}\}}\,dr.
\]
After stopping at \(T\), the martingale increments are square integrable.
The drift increment is bounded by \(K\delta\), while conditional BDG gives
\[
\mathbb E^{\mathbb P}\!\left[
 \sup_{\sigma\le r\le\theta}
 |M_{r\wedge\rho_O^{\,t}}-M_{\sigma\wedge\rho_O^{\,t}}|^2
 \,\middle|\,\widetilde{\mathcal F}^{\,\mathbb P}_\sigma
\right]
\le c_dK\delta.
\]
Increasing \(c_d\) if necessary, we obtain
\begin{equation}
\label{eq:conditional_short_time_oscillation_bound}
\mathbb E^{\mathbb P}\!\left[
   \sup_{\sigma\le r\le\theta}
   \left|
      \overline X_{r\wedge\rho_O^{\,t}}
      -\overline X_{\sigma\wedge\rho_O^{\,t}}
   \right|^2
   \,\middle|\,
   \widetilde{\mathcal F}^{\,\mathbb P}_\sigma
\right]
\le c_d(K^2\delta^2+K\delta).
\end{equation}
Taking \(\sigma=t\), \(\theta=t+h\), and applying Markov's inequality
proves \eqref{eq:short_time_killing_bound}--
\eqref{eq:short_time_oscillation_bound}.

For part~\ref{prop:short_time_control_exponential_exit},
local boundedness of \(\mathfrak A_G\) and compactness of
\(\overline{B_r(x_0)}\) give constants \(K_r,M_r>0\), depending only on
\(x_0\) and \(r\), such that
\begin{align}\label{eq:b,a_union_bound}
    |b|\le K_r,
    \qquad
    a\preceq M_r^2I
\end{align}
before exit from \(B_r(x_0)\), uniformly over generator-compatible laws.
For
\(y\in B_{r/2}(x_0)\), put \(\rho:=\rho_{B_r(x_0)}^{\,0}\).  The
localized characteristic decomposition gives
\[
    \overline X_{s\wedge\rho}-y
    =\int_0^{s\wedge\rho}b_\xi\,d\xi+M_{s\wedge\rho},
    \qquad
    d\langle M^i\rangle_\xi
    \le M_r^2\mathbb I_{\{\xi<\rho\}}\,d\xi.
\]
If \(h<r/(4K_r)\), live exit before \(h\) forces
\(\sup_{0\le s\le h}|M_{s\wedge\rho}|\ge r/4\).
By the coordinatewise
Bernstein maximal inequality
\cite[Eq.~(1.5)]{supp:dzhaparidze2001bernstein} and \eqref{eq:b,a_union_bound},
\begin{equation}
\label{eq:prop:small_time_exit_1}
\mathbb P\left(
\sup_{0\le s\le h}|M_{s\wedge\rho}|\ge a
\right)
\le
C_d\exp\!\left(-\frac{a^2}{2dM_r^2h}\right),
\qquad a>0.
\end{equation}
Taking \(a=r/4\) gives, for \(0<h<r/(4K_r)\),
\begin{equation}
\label{eq:small_time_exit_bound}
\sup_{y\in B_{r/2}(x_0)}
\sup_{\mathbb P\in\mathcal P_y(G)}
\mathbb P\bigl(
\rho_{B_r(x_0)}^{\,0}\le h,
\ \rho_{B_r(x_0)}^{\,0}<\tau_{\mathrm{kill}}
\bigr)
\le
C_d\exp\!\left(-\frac{r^2}{32dM_r^2h}\right).
\end{equation}
The assertion follows from \(h^{-1}e^{-c/h}\to0\).
\end{proof}

\begin{proposition}
\label{prop:lyap_tail_bounds}
Let \(\beta=(c,b,a)\) be a coefficient field. Suppose that a function \(\phi\in C^2(D)\), \(\phi\ge1\), and a constant
\(C_\phi\ge0\) satisfy
\[
    m_n^\phi:=\inf_{y\in D\setminus D_n}\phi(y)
    \longrightarrow\infty,
    \qquad
    L^\beta\bigl(s,\eta,J^2_{\eta(s)}\phi\bigr)
    \le C_\phi\phi(\eta(s))
    \quad\text{for }s<\tau_\infty(\eta).
\]
Fix \((t,\widetilde\omega)\in[0,\infty)\times\widetilde\Omega\)
with \(t<\tau_\infty(\widetilde\omega)\), set
\(x:=\widetilde\omega(t)\), and choose \(n_0>t\) such that
\(\widetilde\omega([0,t])\subset D_{n_0}\).
Then, for every
\(\mathbb P\in\widetilde{\mathcal P}_{t,\widetilde\omega}(L^\beta)\),
\(T\ge t\), and \(n\ge n_0\),
\begin{equation}
\label{eq:survival_lower_bound}
    \mathbb P\bigl(\tau_n\le T,\ \tau_n<\tau_{\mathrm{kill}}\bigr)
    \le
    \frac{e^{C_\phi(T-t)}\phi(x)}{m_n^\phi}.
\end{equation}
Consequently, \(\mathbb P(\tau_{\mathrm{exp}}=\infty)=1\) and
\(\tau_\infty=\tau_{\mathrm{kill}}\), \(\mathbb P\)-almost surely.
\end{proposition}

\begin{proof}
Fix \(n\ge n_0\). Choose \(\phi_j\in C_c^\infty(D)\) such that
\(\phi_j\to\phi\) in \(C^2(\overline D_n)\). Each stopped process
\(\widetilde M^{\phi_j,n,\beta}\) is a \(\mathbb P\)-martingale with
respect to \(\widetilde{\mathbb F}\). The \(C^2\)-convergence and local
boundedness of the coefficients permit passage to the limit on each finite
time interval, so \(\widetilde M^{\phi,n,\beta}\), defined by the same
formula with \(\phi\) in place of \(\phi_j\), is also a
\(\mathbb P\)-martingale.

The Lyapunov inequality and integration by parts show that the
nonnegative process
\[
    \Xi_s^n
    :=e^{-C_\phi([s]_{t,\tau_n}-t)}
      \phi(X_{[s]_{t,\tau_n}})
      \mathbb I_{\{[s]_{t,\tau_n}<\tau_\infty\}},
    \qquad s\ge t,
\]
is a \(\mathbb P\)-supermartingale with respect to
\(\widetilde{\mathbb F}\), with \(\Xi_t^n=\phi(x)\).
On \(\{\tau_n\le T,\ \tau_n<\tau_{\mathrm{kill}}\}\), the exit
occurs while the path is alive, \(X_{\tau_n}\in D\setminus D_n\),
and
\[
    \Xi_T^n\ge e^{-C_\phi(T-t)}m_n^\phi.
\]
Taking expectations proves \eqref{eq:survival_lower_bound}.

On \(\{\tau_{\mathrm{exp}}\le T\}\), the path exits every sufficiently
large \(D_n\) before time \(T\) and before killing. Hence
\eqref{eq:survival_lower_bound} and \(m_n^\phi\to\infty\) imply
\(\mathbb P(\tau_{\mathrm{exp}}\le T)=0\).
Letting \(T\to\infty\) excludes finite continuous explosion. Finally,
\(\tau_\infty=\tau_{\mathrm{kill}}\wedge\tau_{\mathrm{exp}}\) gives
the asserted equality of the lifetime and the killing time.
\end{proof}

The preceding exit estimate and local characteristic bounds yield the following tightness criterion.

\begin{proposition}
\label{prop:compactness_criterion}
Let
\[
    \mathbb P_j\in
    \widetilde{\mathcal P}_{t_j,x_j}(L^{\beta_j}),
    \qquad (t_j,x_j)\in K,\quad j\ge1,
\]
where \(K\subset[0,\infty)\times D\) is compact.
Suppose that, for every \(n\ge1\) and \(T<\infty\), there is
\(K_{n,T}<\infty\) such that
\[
    \|\beta_j(s,\omega)\|\le K_{n,T}
    \quad\text{on }\{s\le T,\ s<\tau_n(\omega)\},
    \qquad j\ge1.
\]
Suppose also that a function \(\phi\in C^2(D)\), \(\phi\ge1\), and \(C\ge0\) satisfy
\[
    m_n^\phi:=\inf_{y\in D\setminus D_n}\phi(y)
    \longrightarrow\infty,
    \qquad
    L^{\beta_j}\bigl(s,\omega,J^2_{\omega(s)}\phi\bigr)
    \le C\phi(\omega(s))
    \quad\text{on }\{s<\tau_\infty(\omega)\},
    \quad j\ge1.
\]
Then \((\mathbb P_j)_{j\ge1}\) is tight on \(\widetilde\Omega\)
in the \(J_1\) topology.
\end{proposition}

\begin{proof}
Compactness of \(K\) implies that the constant initial histories form a
precompact family and that
\[
    M_K:=\sup\{\phi(x):(t,x)\in K\}<\infty.
\]
Fix a finite horizon \(T\), and choose \(n_0\) large enough that all
initial states lie in \(D_{n_0}\) and \(n_0>\sup_jt_j\).

For \(n\ge n_0\), consider the \(\mathbb R^{d+1}\)-valued process
\[
    Y_s^n
    :=
    \bigl(\overline X_{s\wedge\tau_n},
          N_{s\wedge\tau_n}\bigr),
    \qquad 0\le s\le T.
\]
Write \(\beta_j=(c_j,b_j,a_j)\), \(k_j:=-c_j\), and \(e:=(0,\ldots,0,1)\in\mathbb R^{d+1}\).
With a continuous truncation function equal to the identity on a neighborhood of the closed unit ball,
Lemma~\ref{prop:localized_characteristics_linear_gmp} gives the differential characteristics
\[
    \mathbb I_{\{t_j<s<\tau_n\}}
    \left(
        \begin{pmatrix}b_j\\k_j\end{pmatrix},
        \begin{pmatrix}a_j&0\\0&0\end{pmatrix},
        k_j\delta_e
    \right).
\]
Thus killing contributes a jump of size \(e\), whose intensity
is bounded by the local coefficient bound. In particular,
\[
    \int_{\mathbb R^{d+1}}
        (|z|^2\wedge|z|)\,k_j\delta_e(dz)=k_j,
\]
so Condition~(B) of \cite{supp:liu2019compactness} holds uniformly
in \(j\).
Since \(Y_0^n=(x_j,0)\), compactness of the initial states and
\cite[Corollary~3.22]{supp:liu2019compactness}
give \(J_1\)-tightness of the joint laws.

The map sending \((y,0)\) to \(y\) and \((y,1)\) to
\(\triangle\) is continuous on
\(\overline D_n\times\{0,1\}\), and therefore induces a
continuous map on the corresponding \(J_1\) path spaces.
Its application to \(Y^n\) recovers \(X_{\cdot\wedge\tau_n}\).
Moreover, continuity of the first coordinate and monotonicity
of the binary second coordinate are preserved under
\(J_1\) convergence.
Hence the laws of $X_{\cdot\wedge\tau_n}$ under $\mathbb P_j$ are tight, and every weak limit is concentrated on paths that are continuous while alive and absorbed after killing.

For \(t_j\le T\), Proposition~\ref{prop:lyap_tail_bounds} gives
\[
    \mathbb P_j(\tau_n\le T,\ \tau_n<\tau_{\mathrm{kill}})
    \le \frac{e^{CT}M_K}{m_n^\phi}.
\]
For \(t_j>T\), the left-hand side is zero because the path is constant
on \([0,T]\).  Since \(m_n^\phi\to\infty\), the probability of a live exit tends to zero uniformly in \(j\).
The standard localization argument therefore removes the stopping; see
the proof of
\cite[Theorem~2.3, Sections~3.2.1 and~3.2.4]{supp:criens2020existence}.
The same vanishing exit probability preserves the required path
properties in every cluster law.  Diagonalizing over integer horizons
proves tightness on \(\widetilde\Omega\).
\end{proof}

\subsection{Weak Passage of Martingale and Expectation Relations}\label{subsec:weak_stability_compensated_relations}

\begin{lemma}
\label{lem:skorokhod_weak_passage}
Let \(E\) be a Polish space, and let \(\mu_j,\mu\in\mathcal P(E)\) satisfy \(\mu_j\Rightarrow\mu\). 
Let \(C_j,C\in\mathcal B(E)\) satisfy \(\mu_j(C_j)=\mu(C)=1\) for all \(j\ge1\), and let \(Z_j,Z:E\to\mathbb R\) be Borel functions satisfying
\begin{align}
\label{eq:uniformbounded_ZjZ}
    \sup_{j\ge1}\|Z_j\|_\infty
    \vee
    \|Z\|_\infty
    <\infty.
\end{align}
Then the following statements hold:

\begin{enumerate}[label=(\roman*), ref=(\roman*)]
\item\label{lem:skorokhod_weak_passage_upper}
Suppose that, with the upper half-relaxed limit taken in \(E\) for the
restrictions to the varying domains \(C_j\),
\[
    \limsup_{j\to\infty}^{*}\bigl(Z_j|_{C_j}\bigr)
    \le Z
    \qquad\text{on }C.
\]
Then
\[
    \limsup_{j\to\infty}
    \int_E Z_j\,d\mu_j
    \le
    \int_E Z\,d\mu.
\]

\item\label{lem:skorokhod_weak_passage_exact}
In particular, suppose that
\[
    \limsup_{j\to\infty}^{*}\bigl(Z_j|_{C_j}\bigr)
    \le Z,
    \qquad
    \liminf_{j\to\infty,*}\bigl(Z_j|_{C_j}\bigr)
    \ge Z
    \qquad\text{on }C.
\]
Then
\[
    \int_E Z_j\,d\mu_j
    \longrightarrow
    \int_E Z\,d\mu.
\]
\end{enumerate}
\end{lemma}

\begin{proof}
By Skorokhod's representation theorem
\cite[Theorem~6.7]{supp:billingsley2013convergence}, realize the laws
\(\mu_j,\mu\) by random variables \(Y_j,Y\) on a common probability
space such that \(Y_j\to Y\) almost surely.  Outside a single null
set, we also have \(Y_j\in C_j\) for every \(j\) and \(Y\in C\).
The hypothesis in part~\ref{lem:skorokhod_weak_passage_upper} therefore
gives \(\limsup_j Z_j(Y_j)\le Z(Y)\) almost surely.
The uniform bound and the Fatou inequality yield
\[
    \limsup_{j\to\infty}\int_E Z_j\,d\mu_j
    \le \mathbb E\!\left[\limsup_{j\to\infty}Z_j(Y_j)\right]
    \le \mathbb E[Z(Y)]
    =\int_E Z\,d\mu.
\]
Applying this result to both \((Z_j,Z)\) and \((-Z_j,-Z)\)
proves part~\ref{lem:skorokhod_weak_passage_exact}.
\end{proof}

The compactly supported mark in the next lemma separates genuine killing
inside the state space from continuous explosion at the cemetery boundary.
This localization is essential because the unmarked killing indicator is not
continuous along sequences whose terminal jumps escape spatially and converge
to continuous explosion.

\begin{lemma}
\label{lem:weak_stability_locally_bounded_killing_compensators}
Let \(\mathbb P_j\Rightarrow\mathbb P\) weakly in \(\mathfrak M\), and
assume
\[
    \mathbb P_j(N_0=0)=\mathbb P(N_0=0)=1,
    \qquad j\ge1.
\]
For every \(j\ge1\), let \(\kappa^j\) be a nonnegative
\(\widetilde{\mathbb F}\)-progressively measurable process such that
\begin{equation}
\label{eq:weak_killing_compensator_sequence}
    N_s-
    \int_0^s\kappa_r^j\mathbb I_{\{r<\tau_\infty\}}\,dr,
    \qquad s\ge0,
\end{equation}
is a \(\mathbb P_j\)-local martingale with respect to
\(\widetilde{\mathbb F}^{\,\mathbb P_j}\).  Suppose moreover that, for every
\(H>0\) and compact set \(K\Subset D\),
\begin{equation}
\label{eq:weak_killing_compensator_local_bound}
    C_{H,K}
    :=
    \sup_{j\ge1}
    \left\|
        \kappa^j\mathbb I_{\{\cdot<\tau_\infty\}}
        \mathbb I_{\{\overline X\in K\}}
    \right\|_{L^\infty(
        [0,H]\times\widetilde\Omega,
        dr\otimes d\mathbb P_j
    )}
    <\infty.
\end{equation}
Then
\begin{equation}
\label{eq:weak_limit_fixed_killing_avoidance}
    \mathbb P(\tau_{\mathrm{kill}}=h)=0,
    \qquad h>0,
\end{equation}
and, for every \(a\ge0\) and open set \(O\Subset D\),
\begin{equation}
\label{eq:weak_limit_exit_killing_avoidance}
    \mathbb P\bigl(
        \tau_{\mathrm{kill}}=\widehat\tau_O^{\,a}<\infty,
        \ \overline X_a\in O
    \bigr)=0.
\end{equation}
\end{lemma}

\begin{proof}
For \(\chi\in C_c(D;[0,1])\), extend \(\chi\) continuously to
\(\widehat D\) by \(\chi(\triangle)=0\), and define
\begin{equation}
\label{eq:marked_killing_process}
    N_s^\chi
    :=
    \int_{(0,s]}\chi(\overline X_r)\,dN_r
    =
    \chi(\overline X_{\tau_{\mathrm{kill}}})
    \mathbb I_{\{\tau_{\mathrm{kill}}\le s\}},
    \qquad s\ge0,
\end{equation}
where the last expression is zero on
\(\{\tau_{\mathrm{kill}}=\infty\}\).  We show that \(N^\chi\) has under
\(\mathbb P\) a predictable compensator of the form
\[
    A_s^\chi=\int_0^s\alpha_r^\chi\,dr,
    \qquad s\ge0,
\]
where \(\alpha^\chi\) is nonnegative and locally bounded.  In particular,
for every \(\widetilde{\mathbb F}^{\,\mathbb P}\)-predictable stopping time
\(\sigma\),
\begin{equation}
\label{eq:predictable_marked_killing_avoidance}
    \mathbb P\bigl(
        0<\sigma=\tau_{\mathrm{kill}}<\infty,
        \ \chi(\overline X_\sigma)>0
    \bigr)=0.
\end{equation}

We show that $\Psi_\chi:\widetilde\Omega\to D([0,\infty),[0,1])$, defined by $\Psi_\chi(\widetilde\eta):=N^\chi(\widetilde\eta)$, is continuous at every $\widetilde\eta$ with $\tau_{\mathrm{kill}}(\widetilde\eta)>0$.
Suppose that $\widetilde\eta_j\to\widetilde\eta$ in $J_1$, and fix a continuity time $H>0$ of $\widetilde\eta$.
By the definition of $J_1$ convergence, there are increasing homeomorphisms $\lambda_j:[0,H]\to[0,H]$ such that $\lambda_j\to\mathrm{id}$ and $\widetilde\eta_j\circ\lambda_j\to\widetilde\eta$ uniformly, the latter in the metric of $\widehat D$.
Consequently, $h_j:=\chi\circ\widetilde\eta_j\circ\lambda_j$ converges uniformly to $h:=\chi\circ\widetilde\eta$.
Each of these functions is continuous except possibly for one downward jump at killing.
The size of this decrease equals the corresponding increase of $N^\chi$.

Set $\varepsilon_j:=\|h_j-h\|_\infty$.
Taking left limits gives $|\Delta h_j(r)-\Delta h(r)|\le2\varepsilon_j$ for every $r\in(0,H]$, where $\Delta h(r):=h(r)-h(r-)$.
If $h$ is continuous, the jump size of
$N^\chi(\widetilde\eta_j)\circ\lambda_j$ is therefore at most
$2\varepsilon_j$, while $N^\chi(\widetilde\eta)$ is zero
on $[0,H]$.
If $h$ has a nonzero jump at $\zeta\in(0,H)$, the same
inequality implies that $h_j$ also jumps at $\zeta$
for all sufficiently large $j$, and the jump sizes converge.
In either case,
$N^\chi(\widetilde\eta_j)\circ\lambda_j
\to N^\chi(\widetilde\eta)$ uniformly on $[0,H]$.
Since such continuity times $H$ can be chosen arbitrarily
large, this proves the claimed $J_1$ continuity.

Since \(\chi(\overline X)\) is bounded and predictable, stochastic
integration in \eqref{eq:weak_killing_compensator_sequence} shows that
\begin{equation}
\label{eq:marked_killing_compensator_under_Pj}
    N_s^\chi-
    \int_0^s
        \chi(\overline X_r)\kappa_r^j\mathbb I_{\{r<\tau_\infty\}}\,dr,
    \qquad s\ge0,
\end{equation}
is a \(\mathbb P_j\)-local martingale.  On each finite horizon it is a
true martingale, by
\eqref{eq:weak_killing_compensator_local_bound} and the boundedness of
\(N^\chi\).
Choose \(H>0\) which is a \(\mathbb P\)-almost-sure continuity time of
both \(X\) and \(N^\chi\).  
Put \(K:=\operatorname{supp}\chi\), and define finite measures on \([0,H]\times\Omega^{\mathrm{cad}}\) by
\begin{equation}
\label{eq:marked_compensator_measures}
    \Lambda_j(B)
    :=
    \mathbb E^{\mathbb P_j}\!\left[
        \int_0^H
            \mathbb I_B(r,X)
            \chi(\overline X_r)\kappa_r^j\mathbb I_{\{r<\tau_\infty\}}\,dr
    \right].
\end{equation}
Then
\begin{equation}
\label{eq:marked_compensator_measure_domination}
    \Lambda_j
    \le
    C_{H,K}\,(dr\otimes\mathbb P_j).
\end{equation}
The weak convergence of \(\mathbb P_j\) makes
\((\Lambda_j)_{j\ge1}\) tight.  Passing to a subsequence, we may assume
that \(\Lambda_j\Rightarrow\Lambda\).  Since
\(dr\otimes\mathbb P_j\Rightarrow dr\otimes\mathbb P\),
\eqref{eq:marked_compensator_measure_domination} gives
\[
    \Lambda\le C_{H,K}\,(dr\otimes\mathbb P).
\]
Hence there is a Borel function
\(\alpha^{\chi,H}:[0,H]\times\Omega^{\mathrm{cad}}
\to[0,C_{H,K}]\) such that
\begin{equation}
\label{eq:marked_compensator_limit_density}
    \Lambda(dr,d\widetilde\eta)
    =
    \alpha_r^{\chi,H}(\widetilde\eta)
    \,dr\,\mathbb P(d\widetilde\eta).
\end{equation}

Choose a countable dense set \(\mathbb T_H\subset[0,H]\), containing \(0\) and \(H\), every element of which is a \(\mathbb P\)-almost-sure continuity time of both \(X\) and \(N^\chi\). 
Let \(0\le s_1<s_2\le H\) belong to \(\mathbb T_H\), and let \(F\) be a bounded continuous cylinder function with observation times in \(\mathbb T_H\cap[0,s_1]\). 
From
\eqref{eq:marked_killing_compensator_under_Pj},
\begin{equation}
\label{eq:marked_compensator_martingale_identity_j}
    \mathbb E^{\mathbb P_j}\!\left[
        F\bigl(N_{s_2}^\chi-N_{s_1}^\chi\bigr)
    \right]
    =
    \int_{(s_1,s_2]\times\Omega^{\mathrm{cad}}}
        F(\widetilde\eta)\,
        \Lambda_j(dr,d\widetilde\eta).
\end{equation}
The pathwise continuity of \(\Psi_\chi\) and the choice of
\(\mathbb T_H\) pass the left-hand side to the limit.  On the right-hand
side, the function
\[
    (r,\widetilde\eta)
    \longmapsto
    \mathbb I_{(s_1,s_2]}(r)F(\widetilde\eta)
\]
is \(\Lambda\)-almost surely continuous: the time boundaries are
\(\Lambda\)-null by \eqref{eq:marked_compensator_limit_density}, and the
cylinder discontinuity set is \(\mathbb P\)-null.  Hence
\eqref{eq:marked_compensator_martingale_identity_j} and
\eqref{eq:marked_compensator_limit_density} yield
\begin{equation}
\label{eq:marked_compensator_martingale_identity_limit}
    \mathbb E^{\mathbb P}\!\left[
        F\bigl(N_{s_2}^\chi-N_{s_1}^\chi\bigr)
    \right]
    =
    \mathbb E^{\mathbb P}\!\left[
        F\int_{s_1}^{s_2}\alpha_r^{\chi,H}\,dr
    \right].
\end{equation}

Replace \(\alpha^{\chi,H}\) by its predictable projection under the
\(\mathbb P\)-usual augmentation.  This preserves
\eqref{eq:marked_compensator_martingale_identity_limit} and the bounds
\(0\le\alpha^{\chi,H}\le C_{H,K}\).  A monotone-class argument, followed
by right-continuous extension from \(\mathbb T_H\), now shows that
\[
    N_s^\chi-
    \int_0^s\alpha_r^{\chi,H}\,dr,
    \qquad 0\le s\le H,
\]
is a \(\mathbb P\)-martingale.  Taking an increasing sequence of such
continuity horizons \(H\uparrow\infty\), uniqueness of predictable
compensators makes these constructions consistent.  They therefore define
a globally predictable density \(\alpha^\chi\), which is bounded on every
compact time interval.
The compensator \(A^\chi\) is continuous.  
Hence $\mathbb E^{\mathbb P}\!\left[\Delta N_\sigma^\chi\mathbb I_{\{0<\sigma<\infty\}}\right]=0$ for every predictable stopping time \(\sigma\). 
Since \(N^\chi\) is nondecreasing, this proves \eqref{eq:predictable_marked_killing_avoidance}.

Choose \(\chi_m\in C_c(D;[0,1])\) equal to one on \(\overline D_m\).  On \(\{\tau_{\mathrm{kill}}<\infty\}\), the
pre-killing state belongs to \(D\), and hence
\(\chi_m(\overline X_{\tau_{\mathrm{kill}}})=1\) for some \(m\).
Applying \eqref{eq:predictable_marked_killing_avoidance} to the
deterministic stopping time \(h>0\), and then taking the countable union over
\(m\), proves \eqref{eq:weak_limit_fixed_killing_avoidance}.

Finally, set \(\sigma:=\widehat\tau_O^{\,a}\).  The continuity of
\(\overline X\) makes \(\sigma\) predictable, by the standard announcing
sequence obtained from the distance of \(\overline X\) to \(D\setminus O\).
On
\(\{\overline X_a\in O,\ \sigma<\infty\}\), one has
\(\overline X_\sigma\in\partial O\).  Choose
\(\chi\in C_c(D;[0,1])\) equal to one on \(\overline O\), and apply
\eqref{eq:predictable_marked_killing_avoidance}.  This proves
\eqref{eq:weak_limit_exit_killing_avoidance} and completes the proof.
\end{proof}

We recall the definition of \((\mathbb P,h)\)-regularity from the main text.
For \(\mathbb P\in\mathfrak M\), \(h>0\), an open set \(O\Subset D\), and \(a\in[0,h)\), set \(\sigma_{h,O}^{a}:=(h-a)\wedge\rho_O^{\,0}\).
The set \(O\) is called \emph{\((\mathbb P,h)\)-regular} if the map \(\mathsf{End}_{h,O}:\widetilde\Omega\times[0,h)\to[0,h]\times\widehat D\times\{0,1\}\), defined by
\[
    \mathsf{End}_{h,O}(\widetilde\omega,a)
    :=\Bigl(
        \sigma_{h,O}^{a}(\widetilde\omega),
        X_{\sigma_{h,O}^{a}(\widetilde\omega)}(\widetilde\omega),
        \mathbb I_{\{\sigma_{h,O}^{a}(\widetilde\omega)
                     <\tau_\infty(\widetilde\omega)\}}
    \Bigr),
\]
is continuous at \((\widetilde\omega,0)\) for
\(\mathbb P\)-almost every \(\widetilde\omega\), where the domain carries
the product of the relative \(J_1\) and Euclidean topologies.

\begin{proposition}
\label{lem:regular_inner_exhaustion}
Let \(\mathbb P\in\mathfrak M\) satisfy
\(\mathbb P(X_0\in D)=1\), let
\(\mathscr H\subset(0,\infty)\) be finite, and assume that
\(\mathbb P(\tau_{\mathrm{kill}}=h)=0\) for every
\(h\in\mathscr H\).
If \(K\Subset O\Subset D\), with \(O\) open, then there exists an increasing sequence of open sets \((O_k)_{k\ge1}\) such that \(K\Subset O_1\) and \((O_k)_{k\ge1}\) is a \((\mathbb P,h)\)-regular exhaustion of \(O\) for every \(h\in\mathscr H\).
Moreover,
\begin{equation}
\label{eq:regular_inner_exit_convergence}
    \rho_{O_k}^{\,0}\uparrow\rho_O^{\,0},
    \qquad \mathbb P\text{-almost surely}.
\end{equation}
\end{proposition}

\begin{proof}
By \cite[Proposition~2.28]{supp:lee2013introduction}, we can choose a smooth exhaustion function \(\psi:O\to[0,\infty)\) satisfying
\[
    O_r:=\{y\in O:\psi(y)<r\}\Subset O,
    \qquad r>0.
\]
We first select levels at which exit times are continuous.
For each path \(\widetilde\omega\), the map
\(r\mapsto\rho_{O_r}^{\,0}(\widetilde\omega)\) is nondecreasing
and has at most countably many discontinuities, with values in
\([0,\infty]\).  If \(\widetilde\omega_j\to\widetilde\omega\)
in \(J_1\), the compact inclusions between the sublevel sets and
the time-change characterization of \(J_1\) convergence give, for
\(0<\varepsilon<r\),
\[
    \rho_{O_{r-\varepsilon}}^{\,0}(\widetilde\omega)
    \le\liminf_{j\to\infty}\rho_{O_r}^{\,0}(\widetilde\omega_j)
    \le\limsup_{j\to\infty}\rho_{O_r}^{\,0}(\widetilde\omega_j)
    \le\rho_{O_{r+\varepsilon}}^{\,0}(\widetilde\omega).
\]
Indeed, before exit from \(O_{r-\varepsilon}\), the time-changed
approximating paths remain in \(O_r\); at a finite exit from
\(O_{r+\varepsilon}\), they lie outside \(O_r\).
Letting \(\varepsilon\searrow0\) proves continuity of the exit-time
map at every continuity level of the preceding nondecreasing map.
Fubini's theorem therefore gives a full-Lebesgue-measure set of levels
at which this continuity holds \(\mathbb P\)-almost surely.

We further discard the levels for which either
\begin{align}\label{eq:boundarymeasurenotzero}
    \mathbb P(X_0\in\partial O_r)>0
    \quad\text{or}\quad
    \mathbb P\bigl(\tau_{\mathrm{kill}}<\infty,
        X_{\tau_{\mathrm{kill}}-}\in\partial O_r\bigr)>0.
\end{align}
Each boundary condition singles out at most one level for each path,
so Fubini's theorem again removes only a Lebesgue-null set.
For every \(h\in\mathscr H\), we also discard levels with
\(\mathbb P(\rho_{O_r}^{\,0}=h)>0\).
Since \(X\) is \(\mathbb P\)-almost surely continuous at \(h\),
the equality \(\rho_{O_r}^{\,0}=h>0\) can hold at at most one
level for almost every path.  Thus the remaining set of levels still
has full Lebesgue measure.

Let \(\mathscr L\) be the set of levels \(r>0\) such that \(\rho_{O_r}^{\,0}:\widetilde\Omega\to[0,\infty]\) is \(\mathbb P\)-almost surely continuous in the relative \(J_1\) topology, both probabilities in \eqref{eq:boundarymeasurenotzero} are zero, and \(\mathbb P(\rho_{O_r}^{\,0}=h')=0\) for every \(h'\in\mathscr H\).
The preceding arguments show that \(\mathscr L\) has full Lebesgue measure.
We prove that \(O_r\) is \((\mathbb P,h)\)-regular for every  \(r\in\mathscr L\) and \(h\in\mathscr H\).

Fix \(r\in\mathscr L\) and \(h\in\mathscr H\).
By the definition of \(\mathscr L\) and \(\mathbb P(X_0\in D)=1\), there exists a set \(\widetilde\Omega_{h,r}\) of \(\mathbb P\)-measure one such that, for every \(\widetilde\omega\in\widetilde\Omega_{h,r}\), the map \(\rho_{O_r}^{\,0}\) is continuous at \(\widetilde\omega\),
\(X_0(\widetilde\omega)\in D\setminus\partial O_r\),
\(\rho_{O_r}^{\,0}(\widetilde\omega)\ne h\), and
\(X_{\tau_{\mathrm{kill}}-}(\widetilde\omega)\notin\partial O_r\)
whenever \(\tau_{\mathrm{kill}}(\widetilde\omega)<\infty\).
Fix \(\widetilde\omega\in\widetilde\Omega_{h,r}\), and let
\(\widetilde\omega_j\to\widetilde\omega\) in \(J_1\) and
\(a_j\to0\) in \([0,h)\).
Set
\[
    \rho_j:=\rho_{O_r}^{\,0}(\widetilde\omega_j),
    \qquad\rho:=\rho_{O_r}^{\,0}(\widetilde\omega),\qquad
    \sigma_j:=(h-a_j)\wedge\rho_j,
    \qquad\sigma:=h\wedge\rho.
\]
Continuity of \(\rho_{O_r}^{\,0}\) at \(\widetilde\omega\)
gives \(\rho_j\to\rho\), and hence \(\sigma_j\to\sigma\).

Suppose first that \(X_\sigma(\widetilde\omega)\in D\).
Then \(\widetilde\omega\) is continuous at \(\sigma\), with
right continuity understood when \(\sigma=0\).
Thus \(J_1\) convergence and \(\sigma_j\to\sigma\) give
\(X_{\sigma_j}(\widetilde\omega_j)
\to X_\sigma(\widetilde\omega)\).
Since \(D\) is open in \(\widehat D\), we have
\(X_{\sigma_j}(\widetilde\omega_j)\in D\) for all sufficiently
large \(j\). Consequently,
\(\mathbb I_{\{\sigma_j<\tau_\infty(\widetilde\omega_j)\}}
=\mathbb I_{\{\sigma<\tau_\infty(\widetilde\omega)\}}=1\)
for all sufficiently large \(j\).

Suppose now that \(X_\sigma(\widetilde\omega)=\triangle\).
Since \(X_0(\widetilde\omega)\in D\) and \(\rho\ne h\), we have
\(0<\sigma=\rho=\tau_\infty(\widetilde\omega)<h\).
Moreover, \(\widetilde\omega([0,\rho))\subset O_r\Subset D\),
so \(\widetilde\omega\) cannot reach \(\triangle\) continuously
at \(\rho\). Hence
\(\zeta:=\tau_{\mathrm{kill}}(\widetilde\omega)=\rho\).
The choice of \(\widetilde\Omega_{h,r}\) gives
\(\widetilde\omega(\zeta-)
\in\overline O_r\setminus\partial O_r=O_r\).
Therefore,
\(C:=\widetilde\omega([0,\zeta))
\cup\{\widetilde\omega(\zeta-)\}\)
is a compact subset of \(O_r\).
Since \(\widetilde\omega\) is continuous at \(h>\zeta\),
there are increasing homeomorphisms
\(\lambda_j:[0,h]\to[0,h]\) such that
\(\lambda_j\to\mathrm{id}\) uniformly and
\(\widetilde\omega_j\circ\lambda_j\to\widetilde\omega\)
uniformly in the metric of \(\widehat D\).
Because \(d_{\widehat D}(C,\widehat D\setminus O_r)>0\),
we have \(\widetilde\omega_j(\lambda_j(u))\in O_r\) for every
\(0\le u<\zeta\), for all sufficiently large \(j\).
Uniform convergence, also after taking left limits, gives
\[
    \widetilde\omega_j(\lambda_j(\zeta)-)
    \longrightarrow\widetilde\omega(\zeta-)\in O_r,
    \qquad
    \widetilde\omega_j(\lambda_j(\zeta))
    \longrightarrow\triangle.
\]
The two limits are distinct, so \(\widetilde\omega_j\) has a
jump at \(\lambda_j(\zeta)\) for all sufficiently large \(j\).
By the definition of \(\widetilde\Omega\), this jump occurs
at \(\tau_{\mathrm{kill}}(\widetilde\omega_j)\).
Together with
\(\widetilde\omega_j([0,\lambda_j(\zeta)))\subset O_r\),
this yields
\(\rho_j=\tau_{\mathrm{kill}}(\widetilde\omega_j)
=\lambda_j(\zeta)<h-a_j\)
for all sufficiently large \(j\).
Thus \(\sigma_j=\rho_j\),
\(X_{\sigma_j}(\widetilde\omega_j)=\triangle\), and
\(\mathbb I_{\{\sigma_j<\tau_\infty(\widetilde\omega_j)\}}=0\)
for all sufficiently large \(j\).

In both cases, all three coordinates converge, so
\(\mathsf{End}_{h,O_r}(\widetilde\omega_j,a_j)
\to\mathsf{End}_{h,O_r}(\widetilde\omega,0)\).
Hence \(\mathsf{End}_{h,O_r}\) is continuous at every point
of \(\widetilde\Omega_{h,r}\times\{0\}\), relative to
\(\widetilde\Omega\times[0,h)\).
Since \(\mathbb P(\widetilde\Omega_{h,r})=1\),
\(O_r\) is \((\mathbb P,h)\)-regular.

Since \(\mathscr L\) has full Lebesgue measure, we can choose a sequence \((r_k)_{k\ge1}\subset\mathscr L\) satisfying
\[
    \sup_{y\in K}\psi(y)<r_1<r_2<\cdots\uparrow\infty.
\]
Setting \(O_k:=O_{r_k}\), we obtain a \((\mathbb P,h)\)-regular exhaustion of \(O\) for every \(h\in\mathscr H\), with \(K\Subset O_1\).
Finally, \(\rho_{O_k}^{\,0}\le\rho_O^{\,0}\).
If \(s<\rho_O^{\,0}\), the path range up to time \(s\) is a compact subset of \(O\), and hence is contained in \(O_k\) for all sufficiently large \(k\).
Therefore \(\rho_{O_k}^{\,0}>s\) eventually, which proves \eqref{eq:regular_inner_exit_convergence}.
\end{proof}

\begin{proposition}
\label{lem:regular_stopped_endpoint_passage}
Fix \(h>0\), an open set \(O\Subset D\), and \(x\in O\).  Let
\(\mathbb P_j\Rightarrow\mathbb P\) in \(\mathfrak M\), with
\(\mathbb P_j(X_0=x)=1\) for every \(j\).
Assume that \(\mathbb P(\tau_{\mathrm{kill}}=h)=0\), and let
\(a_j\in[0,h)\) decrease to zero.
Suppose that \(v_j,r_j\in C([a_j,h]\times\overline O)\) are
uniformly bounded and, for some
\(w,r\in\USC_b([0,h]\times O)\), satisfy
\begin{equation}
\label{eq:regular_stopped_value_source_upper_limits}
    \limsup_{j\to\infty}^{*}v_j\le w,
    \qquad
    \limsup_{j\to\infty}^{*}r_j\le r
    \qquad\text{on }[0,h]\times O.
\end{equation}
Continuity of evaluation at time zero gives
\(\mathbb P(X_0=x)=1\).  Thus
Proposition~\ref{lem:regular_inner_exhaustion}, applied with
\(K=\{x\}\) and \(\mathscr H=\{h\}\), supplies a
\((\mathbb P,h)\)-regular exhaustion \((O_k)_{k\ge1}\) of
\(O\) with \(x\in O_1\).  Fix such an exhaustion and set
\[
    \sigma_j^k:=(h-a_j)\wedge\rho_{O_k}^{\,0},
    \qquad
    \sigma^k:=h\wedge\rho_{O_k}^{\,0}.
\]
Then, for every \(k\ge1\),
\begin{equation}
\label{eq:regular_stopped_endpoint_passage}
\begin{aligned}
    \limsup_{j\to\infty}
    \mathbb E^{\mathbb P_j}\!\Bigg[
        &v_j(h-\sigma_j^k,X_{\sigma_j^k})
            \mathbb I_{\{\sigma_j^k<\tau_\infty\}}
        +\int_0^{\sigma_j^k}
            r_j(h-s,X_s)\mathbb I_{\{s<\tau_\infty\}}\,ds
    \Bigg]\\
    \le
    \mathbb E^{\mathbb P}\!\Bigg[
        &w(h-\sigma^k,X_{\sigma^k})
            \mathbb I_{\{\sigma^k<\tau_\infty\}}
        +\int_0^{\sigma^k}
            r(h-s,X_s)\mathbb I_{\{s<\tau_\infty\}}\,ds
    \Bigg].
\end{aligned}
\end{equation}
\end{proposition}

\begin{proof}
Fix \(k\).  Use a Skorokhod representation to realize
\(\mathbb P_j,\mathbb P\) by paths
\(\widetilde\omega_j\to\widetilde\omega\) almost surely in the
\(J_1\) topology, retaining the canonical notation for their coordinates.
The joint continuity in the definition of \((\mathbb P,h)\)-regularity,
applied at \((\widetilde\omega,0)\) along
\((\widetilde\omega_j,a_j)\), gives, almost surely,
\[
    \bigl(\sigma_j^k,X_{\sigma_j^k},
        \mathbb I_{\{\sigma_j^k<\tau_\infty\}}\bigr)
    \longrightarrow
    \bigl(\sigma^k,X_{\sigma^k},
        \mathbb I_{\{\sigma^k<\tau_\infty\}}\bigr).
\]
Moreover, \(\mathbb I_{\{s<\sigma_j^k\}}\to\mathbb I_{\{s<\sigma^k\}}\) for Lebesgue-almost every \(s\in[0,h]\), and \(X_s(\widetilde\omega_j)\to X_s(\widetilde\omega)\in O_k\) whenever \(s<\sigma^k\).
Combining these facts with the preceding convergence and \eqref{eq:regular_stopped_value_source_upper_limits}, uniform boundedness allows us to apply Fatou lemma, first to the time integrals, with the integrands extended by zero outside \([0,\sigma_j^k]\), and then to expectations.
This yields \eqref{eq:regular_stopped_endpoint_passage}.
\end{proof}

For \(a\ge0\), a stopping time \(\rho\), and Borel functions
\(v,q:[0,\infty)\times D\to\mathbb R\) for which the following expression is
well defined, set
\begin{equation}
\label{eq:abstract_compensated_process_notation}
\mathcal Z_s^{a,\rho}[v,q]
:=
 v\bigl([s]_{a,\rho},X_{[s]_{a,\rho}}\bigr)
 \mathbb I_{\{[s]_{a,\rho}<\tau_\infty\}}
 -
 \int_a^{[s]_{a,\rho}}
 q(r,X_r)\mathbb I_{\{r<\tau_\infty\}}\,dr,
 \qquad s\ge0.
\end{equation}

\begin{proposition}
\label{prop:weak_passage_martingale_relations}
Let $(t_j,x_j)\to(t,x)$ in $[0,\infty)\times D$, and let
$\mathbb P_j\Rightarrow\mathbb P$ in $\mathfrak M$, with
\[
    \mathbb P_j(X_s=x_j\text{ for every }0\le s\le t_j)=1.
\]
For each $j$, suppose that there exists a nonnegative
$\widetilde{\mathbb F}$-progressively measurable process $\kappa^j$
such that
\begin{equation}
\label{eq:weak_passage_killing_compensator}
    N_s-\int_0^s
    \kappa_r^j\mathbb I_{\{r<\tau_\infty\}}\,dr,
    \qquad s\ge0,
\end{equation}
is a $\mathbb P_j$-local martingale.
Assume also that, for every $H>0$ and compact $K\Subset D$,
\begin{equation}
\label{eq:weak_passage_killing_density_bound}
    \sup_{j\ge1}
    \left\|
        \kappa^j
        \mathbb I_{\{\cdot<\tau_\infty\}}
        \mathbb I_{\{\overline X\in K\}}
    \right\|_{L^\infty(
        [0,H]\times\widetilde\Omega,\,
        dr\otimes d\mathbb P_j
    )}
    <\infty.
\end{equation}
Then the following assertions hold.
\begin{enumerate}[label=(\roman*), ref=(\roman*)]
\item
\label{item:weak_passage_initial_history}
We have
\begin{equation}
\label{eq:weak_passage_limit_initial_history}
    \mathbb P(X_s=x\text{ for every }0\le s\le t)=1.
\end{equation}

\item
\label{item:weak_passage_supermartingale}
Fix $T>t$, $n\ge1$, and $f\in C_b^\infty(D)$, and set $K_n:=[0,T]\times\overline D_n$.
Let $q_j:[0,T]\times D\to\mathbb R$ be Borel functions and
$q\in\USC([0,T]\times D)$ satisfying
\begin{align}
    &\sup_{j\ge1}\sup_{K_n}|q_j|+\sup_{K_n}|q|
    <\infty,
    \label{eq:weak_stability_moving_source_bound}
    \\
    &\limsup_{j\to\infty}^{*}q_j
    \le q
    \qquad\text{on }K_n,
    \label{eq:weak_stability_moving_source_upper_limit}
\end{align}
where the half-relaxed limit is taken relative to $K_n$.
In the notation of
\eqref{eq:abstract_compensated_process_notation},
regard $f$ as independent of time.
If $\mathcal Z^{t_j,\tau_n}[f,q_j]$ is a
$\mathbb P_j$-supermartingale on $[0,T]$ for all sufficiently
large $j$, then $\mathcal Z^{t,\tau_n}[f,q]$ is a
$\mathbb P$-supermartingale on $[0,T]$.

\item
\label{item:weak_passage_martingale}
In the setting of \ref{item:weak_passage_supermartingale},
assume additionally that $q\in C([0,T]\times D)$ and
\begin{equation}
\label{eq:weak_stability_moving_source_lower_limit}
    \liminf_{j\to\infty,*}q_j
    \ge q
    \qquad\text{on }K_n,
\end{equation}
with the half-relaxed limit again taken relative to $K_n$.
If $\mathcal Z^{t_j,\tau_n}[f,q_j]$ is a
$\mathbb P_j$-martingale on $[0,T]$ for all sufficiently
large $j$, then $\mathcal Z^{t,\tau_n}[f,q]$ is a
$\mathbb P$-martingale on $[0,T]$.
\end{enumerate}
\end{proposition}

\begin{proof}
We first prove \ref{item:weak_passage_initial_history}.
Since \(\omega\mapsto X_0(\omega)\) is continuous in the \(J_1\) topology, weak convergence gives \(\mathbb P(X_0=x)=1\).
In particular, $\mathbb P_j(N_0=0)=\mathbb P(N_0=0)=1$.
Lemma~\ref{lem:weak_stability_locally_bounded_killing_compensators} therefore gives
\begin{equation}
\label{eq:weak_passage_limit_killing_avoidance}
    \mathbb P(\tau_{\mathrm{kill}}=h)=0,
    \qquad h>0,
\end{equation}
and, for every open $O\Subset D$,
\begin{equation}
\label{eq:weak_passage_limit_exit_killing_avoidance}
    \mathbb P\bigl(
        \tau_{\mathrm{kill}}=\widehat\tau_O^{\,0}<\infty,
        \ \overline X_0\in O
    \bigr)=0.
\end{equation}
On a Skorokhod representation of \(\mathbb P_j\Rightarrow\mathbb P\), let \(\widetilde\omega_j\to\widetilde\omega\) almost surely in \(J_1\), with \(X_u(\widetilde\omega_j)=x_j\) for every \(0\le u\le t_j\).
For almost every realization, choose \(H>t\) at which \(\widetilde\omega\) is continuous.
There exist increasing homeomorphisms
\(\lambda_j:[0,H]\to[0,H]\) such that
\[
    \sup_{0\le u\le H}|\lambda_j(u)-u|\longrightarrow0,
    \qquad
    \sup_{0\le u\le H}
    d_{\widehat D}\!\left(
        X_{\lambda_j(u)}(\widetilde\omega_j),
        X_u(\widetilde\omega)
    \right)\longrightarrow0.
\]
For every rational \(s<t\), we have \(\lambda_j(s)<t_j\) for all sufficiently large \(j\), since \(\lambda_j(s)\to s<t=\lim_j t_j\).
Consequently, \(X_s(\widetilde\omega)=\lim_{j\to\infty}X_{\lambda_j(s)}(\widetilde\omega_j)=\lim_{j\to\infty}x_j=x\).
Right continuity of \(X\) then yields \(X_s=x\) for every \(s\in[0,t)\).
When $t>0$, its left limit at $t$ is $x\in D$; the only possible
discontinuity at $t$ is consequently a killing jump, which is excluded
by \eqref{eq:weak_passage_limit_killing_avoidance}.
Together with $X_0=x$, this proves
\ref{item:weak_passage_initial_history}.

We next prove \ref{item:weak_passage_supermartingale}.
If $x\notin D_n$, then $\tau_n=0$ under $\mathbb P$, and
$\mathcal Z^{t,\tau_n}[f,q]$ is constant.
Suppose henceforth that $x\in D_n$.
Fix $t<s_1<s_2\le T$ and a bounded nonnegative continuous cylinder
function $F$ whose observation times are at most $s_1$.
All its observation times are $\mathbb P$-almost-sure continuity
times of $X$, by \eqref{eq:weak_passage_limit_killing_avoidance}.
We have \(\mathbb P(X_0=x)=1\), with \(x\in D\), and \eqref{eq:weak_passage_limit_killing_avoidance} gives \(\mathbb P(\tau_{\mathrm{kill}}=s_1)=0\) and \(\mathbb P(\tau_{\mathrm{kill}}=s_2)=0\).
Thus Proposition~\ref{lem:regular_inner_exhaustion} applies with \(K=\{x\}\), \(O=D_n\), and \(\mathscr H=\{s_1,s_2\}\).
It yields open sets $x\in O_k\Subset D_n$, increasing to $D_n$,
regular for both horizons. Set
\[
    \rho_k:=\rho_{O_k}^{\,0},
    \qquad
    a_i^k:=s_i\wedge\rho_k,\quad i=1,2.
\]
For fixed $k$ and all sufficiently large $j$, one has
$x_j\in O_k$ and $t_j<s_1$.
The constant initial histories give
$\rho_k=\rho_{O_k}^{\,t_j}>t_j$, $\mathbb P_j$-almost surely,
and $\rho_k\le\tau_n$.
Stopping the $\mathbb P_j$-supermartingale
$\mathcal Z^{t_j,\tau_n}[f,q_j]$ at $\rho_k$ therefore gives
\begin{equation}
\label{eq:weak_stability_regular_increment}
    \mathbb E^{\mathbb P_j}[\Xi_j^k]\le0,
\end{equation}
where
\[
    \Xi_j^k:=F\bigg(
        f(X_{a_2^k})\mathbb I_{\{a_2^k<\tau_\infty\}}
        -f(X_{a_1^k})\mathbb I_{\{a_1^k<\tau_\infty\}}
        -\int_{a_1^k}^{a_2^k}
            q_j(r,X_r)\mathbb I_{\{r<\tau_\infty\}}\,dr
    \bigg).
\]
Let $\Xi^k$ denote the same expression with $q$ in place of $q_j$.
We claim that
\begin{equation}
\label{eq:weak_stability_regular_increment_limit}
    \mathbb E^{\mathbb P}[\Xi^k]
    \le\liminf_{j\to\infty}\mathbb E^{\mathbb P_j}[\Xi_j^k]
    \le0.
\end{equation}
Choose a Borel set \(C\subset\widetilde\Omega\) with \(\mathbb P(C)=1\) such that, for every \(\widetilde\omega\in C\), \(F\) is continuous at \(\widetilde\omega\) and \(\mathsf{End}_{s_i,O_k}\) is continuous at \((\widetilde\omega,0)\) for \(i=1,2\).
Such a set exists by \eqref{eq:weak_passage_limit_killing_avoidance} and the \((\mathbb P,s_i)\)-regularity of \(O_k\).
Let \(\widetilde\omega_j\in\widetilde\Omega\) and \(\widetilde\omega\in C\) satisfy
\(\widetilde\omega_j\to\widetilde\omega\) in \(J_1\).
Write \(a_{j,i}:=a_i^k(\widetilde\omega_j)\) and
\(a_i:=a_i^k(\widetilde\omega)\), \(i=1,2\).
The definition of \((\mathbb P,s_i)\)-regularity and
continuity of \(f\) give $a_{j,i}\to a_i$ and
\[
    f(X_{a_{j,i}}(\widetilde\omega_j))
    \mathbb I_{\{a_{j,i}<\tau_\infty(\widetilde\omega_j)\}}
    \longrightarrow
    f(X_{a_i}(\widetilde\omega))
    \mathbb I_{\{a_i<\tau_\infty(\widetilde\omega)\}},
    \qquad i=1,2.
\]
We also have \(F(\widetilde\omega_j)\to F(\widetilde\omega)\).

For \(a_1<r<a_2\), we have
\(r<\rho_k(\widetilde\omega)\), so
\(X_r(\widetilde\omega)\in O_k\) and
\(\widetilde\omega\) is continuous at \(r\).
Consequently,
\(X_r(\widetilde\omega_j)\to X_r(\widetilde\omega)\),
and \(a_{j,1}<r<a_{j,2}\) for all sufficiently large \(j\).
In particular, \(X_r(\widetilde\omega_j)\in O_k\) for such \(j\),
and \eqref{eq:weak_stability_moving_source_upper_limit} yields
\(\limsup_{j\to\infty}q_j(r,X_r(\widetilde\omega_j))
\le q(r,X_r(\widetilde\omega))\).
For \(r\notin[a_1,a_2]\),
\(\mathbb I_{\{a_{j,1}<r<a_{j,2}\}}=0\) eventually.
Thus, extending the integrands by zero outside their integration
intervals, \eqref{eq:weak_stability_moving_source_bound} and Fatou lemma give
\[
    \limsup_{j\to\infty}
    \int_{a_{j,1}}^{a_{j,2}}
        q_j(r,X_r(\widetilde\omega_j))
        \mathbb I_{\{r<\tau_\infty(\widetilde\omega_j)\}}\,dr\le
    \int_{a_1}^{a_2}
        q(r,X_r(\widetilde\omega))
        \mathbb I_{\{r<\tau_\infty(\widetilde\omega)\}}\,dr.
\]
Since \(F(\widetilde\omega_j)\to F(\widetilde\omega)\ge0\),
the preceding limits imply
\(\Xi^k(\widetilde\omega)
\le\liminf_{j\to\infty}\Xi_j^k(\widetilde\omega_j)\).
The same argument applies to every subsequence.
Extend \(\Xi_j^k\) and \(\Xi^k\) by zero outside
\(\widetilde\Omega\).
By \eqref{eq:weak_stability_moving_source_bound} and boundedness
of \(F\) and \(f\), these functions are uniformly bounded.
Lemma~\ref{lem:skorokhod_weak_passage}, applied on
\(\Omega^{\mathrm{cad}}\) with \(C_j=\widetilde\Omega\),
\(Z_j=-\Xi_j^k\), and \(Z=-\Xi^k\), therefore gives
\(\mathbb E^{\mathbb P}[\Xi^k]
\le\liminf_{j\to\infty}\mathbb E^{\mathbb P_j}[\Xi_j^k]\).
Together with \eqref{eq:weak_stability_regular_increment},
this proves \eqref{eq:weak_stability_regular_increment_limit}.

We next pass to the limit as \(k\to\infty\) in \eqref{eq:weak_stability_regular_increment_limit} to obtain
\begin{align}\label{eq:assertedineqaaaa}
    \mathbb E^{\mathbb P}\!\left[
        F\left(
            \mathcal Z_{s_2}^{t,\tau_n}[f,q]
            -\mathcal Z_{s_1}^{t,\tau_n}[f,q]
        \right)
    \right]\le0.
\end{align}
By Proposition~\ref{lem:regular_inner_exhaustion},
\(\rho_k\uparrow\tau_n\), and hence
\(a_i^k\uparrow s_i\wedge\tau_n\), \(\mathbb P\)-almost surely,
for \(i=1,2\).
On \(\{s_i\wedge\tau_n<\tau_\infty\}\), continuity of \(X\)
gives \(X_{a_i^k}\to X_{s_i\wedge\tau_n}\), with
\(a_i^k<\tau_\infty\) for every \(k\).
On \(\{s_i\wedge\tau_n=\tau_\infty\}\), since
\(D_n\Subset D\), we have
\(\tau_n=\tau_{\mathrm{kill}}=\tau_\infty\le s_i\).
Moreover, \eqref{eq:weak_passage_limit_exit_killing_avoidance}
implies
\(\tau_{\mathrm{kill}}<\widehat\tau_{D_n}^{\,0}\)
almost surely on this event.
Thus the compact set
\(\{\overline X_u:0\le u\le\tau_{\mathrm{kill}}\}\subset D_n\)
is contained in \(O_k\) for all sufficiently large \(k\).
For such \(k\), \(\rho_k=\tau_{\mathrm{kill}}\) and
\(a_i^k=\tau_\infty\).
Consequently,
\[
    f(X_{a_i^k})\mathbb I_{\{a_i^k<\tau_\infty\}}
    \longrightarrow
    f(X_{s_i\wedge\tau_n})
        \mathbb I_{\{s_i\wedge\tau_n<\tau_\infty\}},
    \qquad i=1,2,
    \quad \mathbb P\text{-almost surely}.
\]
Together with \(a_i^k\to s_i\wedge\tau_n\),
the bound \eqref{eq:weak_stability_moving_source_bound}
and boundedness of \(f\) and \(F\) allow bounded convergence
in the time integrals and then in expectations.
Passing to the limit in
\(\mathbb E^{\mathbb P}[\Xi^k]\le0\) therefore proves \eqref{eq:assertedineqaaaa}.
A monotone-class argument extends
$\mathbb E^{\mathbb P}[F(\mathcal Z_{s_2}^{t,\tau_n}[f,q]
-\mathcal Z_{s_1}^{t,\tau_n}[f,q])]\le0$ to every bounded nonnegative
$\widetilde{\mathcal F}_{s_1}$-measurable $F$.
Since $\mathcal Z^{t,\tau_n}[f,q]$ is bounded and right-continuous on
$[0,T]$, this bound also holds with $s_1=t$ by decreasing approximation;
before $t$ the process is constant.
This proves \ref{item:weak_passage_supermartingale}.

Finally, we prove \ref{item:weak_passage_martingale}.
Under the additional assumptions,
\eqref{eq:weak_stability_moving_source_lower_limit} gives
\[
    \limsup_{j\to\infty}^{*}(-q_j)\le-q
    \qquad\text{on }K_n.
\]
Applying \ref{item:weak_passage_supermartingale} to
$(-f,-q_j,-q)$ shows that
$-\mathcal Z^{t,\tau_n}[f,q]$ is also a $\mathbb P$-supermartingale.
Hence $\mathcal Z^{t,\tau_n}[f,q]$ is a $\mathbb P$-martingale, proving
\ref{item:weak_passage_martingale}.
\end{proof}

\begin{proposition}
\label{lem:killed_cylinder_payoff_upper_stability}
Fix \(m\ge1\). Let \(T_j,T\ge0\) satisfy \(T_j\to T\), and let
\(\mathbb P_j\Rightarrow\mathbb P\) in \(\mathfrak M\).
Assume that
\[
    \mathbb P(\tau_{\mathrm{exp}}=\infty)=1,
    \qquad
    \mathbb P(\tau_{\mathrm{kill}}=T)=0.
\]
For each \(1\le\ell\le m\), let
\[
    0\le r_{j,\ell}\le T_j,
    \qquad r_{j,\ell}\longrightarrow r_\ell\le T.
\]
Suppose that \(f_j,f\in\USC_b(D^m)\) satisfy
\begin{equation}
\label{eq:killed_cylinder_payoff_upper_limits}
    \sup_{j\ge1}\|f_j\|_\infty<\infty,
    \qquad
    \limsup_{j\to\infty}^{*}f_j\le f
    \quad\text{on }D^m.
\end{equation}
Then
\[
    \limsup_{j\to\infty}
    \mathbb E^{\mathbb P_j}\!\left[
        f_j(X_{r_{j,1}},\ldots,X_{r_{j,m}})
        \mathbb I_{\{T_j<\tau_\infty\}}
    \right]
    \le
    \mathbb E^{\mathbb P}\!\left[
        f(X_{r_1},\ldots,X_{r_m})
        \mathbb I_{\{T<\tau_\infty\}}
    \right],
\]
where the killed cylinder functionals are understood to be zero on
\(\{\tau_\infty\le T_j\}\) and \(\{\tau_\infty\le T\}\),
respectively.
\end{proposition}

\begin{proof}
We verify the pathwise upper-limit condition in
Lemma~\ref{lem:skorokhod_weak_passage}.
Let \(\widetilde\omega_j\to\widetilde\omega\) in the \(J_1\)
topology, with all paths in \(\widetilde\Omega\), and suppose that
\begin{align}\label{eq:excludingcondition}
    \tau_{\mathrm{exp}}(\widetilde\omega)=\infty,
    \qquad
    \tau_{\mathrm{kill}}(\widetilde\omega)\ne T.
\end{align}
If \(T<\tau_\infty(\widetilde\omega)\), the limiting path is
continuous and alive on \([0,T]\).  Evaluation at the convergent
observation times therefore gives
\[
    X_{r_{j,\ell}}(\widetilde\omega_j)
    \longrightarrow X_{r_\ell}(\widetilde\omega),
    \qquad
    X_{T_j}(\widetilde\omega_j)
    \longrightarrow X_T(\widetilde\omega)\in D.
\]
Since \(D\) is open in \(\widehat D\), we have \(T_j<\tau_\infty(\widetilde\omega_j)\) for all sufficiently large \(j\). 
Hence \eqref{eq:killed_cylinder_payoff_upper_limits} and the preceding coordinate convergence yield
\begin{align}\label{eq:limsupinequalityaaa}
\begin{split}
    &\limsup_{j\to\infty}
        f_j\bigl(
            X_{r_{j,1}}(\widetilde\omega_j),\ldots,
            X_{r_{j,m}}(\widetilde\omega_j)
        \bigr)
        \mathbb I_{\{T_j<\tau_\infty(\widetilde\omega_j)\}}
    \\
    &\qquad\le
    f\bigl(
        X_{r_1}(\widetilde\omega),\ldots,
        X_{r_m}(\widetilde\omega)
    \bigr)
    \mathbb I_{\{T<\tau_\infty(\widetilde\omega)\}}.
\end{split}
\end{align}

If \(\tau_\infty(\widetilde\omega)<T\), choose
\(u\in(\tau_\infty(\widetilde\omega),T)\).
The approximating paths must be dead at \(T_j\) eventually.
Indeed, otherwise a subsequence would satisfy
\(u<T_j<\tau_\infty(\widetilde\omega_j)\), so its restrictions to
\([0,u]\) would be continuous.  Restriction at \(u\) preserves
\(J_1\) convergence because the limiting path is continuous there.
Since continuous paths form a closed subset in the \(J_1\) topology,
the limiting restriction would be continuous as well, contradicting
its killing jump at \(\tau_\infty(\widetilde\omega)\).
Thus both killed functionals are eventually zero, so \eqref{eq:limsupinequalityaaa} holds with equality.

The remaining case \(\tau_\infty(\widetilde\omega)=T\) is excluded by \eqref{eq:excludingcondition}.
These assumptions hold \(\mathbb P\)-almost surely.
The conclusion follows from Lemma~\ref{lem:skorokhod_weak_passage}.
\end{proof}

\subsection{Closed and Countably Determined Semigroup Relations}\label{subsec:closed_valuation_relations}

\begin{proposition}
\label{prop:parameterized_closed_valuation_relations}
Let $G\in\mathfrak G_{\mathrm{Lyap}}$ and
$\mathcal T\in\mathfrak V$.
For $a\ge0$, set
\[
    \mathfrak P_a
    :=\{(x,\mathbb P)\in D\times\mathfrak M:
                 \mathbb P\in\mathcal P_{a,x}(G)\}.
\]
For $g\in C_b(D)$ and $a\le r\le s<\infty$, let
$\mathfrak C_a(g,r,s)$ consist of the pairs
$(x,\mathbb P)\in\mathfrak P_a$ for which
\begin{equation}
\label{eq:onestep-supermtgrelation_Ciota}
    \mathbb E^{\mathbb P}
        [F g(X_s)\mathbb I_{\{s<\tau_\infty\}}]
    \le
    \mathbb E^{\mathbb P}
        [F\mathcal T_{s-r}g(X_r)\mathbb I_{\{r<\tau_\infty\}}]
\end{equation}
for every bounded nonnegative
$\widetilde{\mathcal F}_r$-measurable random variable $F$.
Then the following assertions hold.
\begin{enumerate}[label=(\roman*), ref=(\roman*)]
\item\label{prop:closed_valuation_relation_fixed}
Each set $\mathfrak C_a(g,r,s)$ is closed in $\mathfrak P_a$.

\item\label{prop:closed_valuation_relation_graph}
The graph of
\[
    (t,x)\longmapsto
    \mathcal R_{t,x}^{\mathcal T}\cap\mathcal P_{t,x}(G)
\]
is closed in $[0,\infty)\times D\times\mathfrak M$.

\item\label{prop:countable_valuation_relations}
For every $a\ge0$, there is a countable family
\[
    \mathscr I_a=\{(g_i,r_i,s_i):i\ge1\},
    \qquad g_i\in C_b(D),\quad a\le r_i\le s_i<\infty,
\]
such that, for every $x\in D$ and
$\mathbb P\in\mathcal P_{a,x}(G)$,
$\mathbb P\in\mathcal R_{a,x}^{\mathcal T}$ if and only if
\eqref{eq:onestep-supermtgrelation_Ciota} holds with
$(g,r,s)=(g_i,r_i,s_i)$ for every $i\ge1$ and every bounded
nonnegative $\widetilde{\mathcal F}_{r_i}$-measurable $F$.
\end{enumerate}
\end{proposition}

\begin{proof}
We first record two observations.  Suppose
$\mathbb P_j\Rightarrow\mathbb P$, with all these laws belonging to
fibers of $\mathcal P(G)$ with initial states in $D$.
By Proposition~\ref{prop:recovery_aggregation_generator_compatible_characteristics}\textup{(i)} and Lemma~\ref{lem:admissible_nonexplosion} of the main text, their killing times have no deterministic atoms
and there is no finite continuous explosion.
For $u\ge0$, $h\in\USC_b(D)$, and a bounded nonnegative continuous
cylinder function $H$ whose observation times do not exceed $u$,
Proposition~\ref{lem:killed_cylinder_payoff_upper_stability} gives
\begin{equation}
\label{eq:usc_killed_payoff_upper_passage}
    \limsup_{j\to\infty}
    \mathbb E^{\mathbb P_j}
        [Hh(X_u)\mathbb I_{\{u<\tau_\infty\}}]
    \le
    \mathbb E^{\mathbb P}
        [Hh(X_u)\mathbb I_{\{u<\tau_\infty\}}].
\end{equation}
Indeed, the product is a bounded upper semicontinuous function on
the corresponding finite product of $D$.
If $h\in C_b(D)$, applying
\eqref{eq:usc_killed_payoff_upper_passage} to $-h$ shows that
the expectations converge.

Second, for any $(x,\mathbb P)\in\mathfrak P_a$, membership in
$\mathcal R_{a,x}^{\mathcal T}$ is equivalent to
\eqref{eq:onestep-supermtgrelation_Ciota} for every
$g\in C_b(D)$ and $a\le r\le s$.
Necessity follows by taking the terminal horizon to be $s$.
For sufficiency, let $h\in\USC_b(D)$ and choose uniformly bounded
$h_m\in C_b(D)$ with $h_m\searrow h$.
Condition~\textup{(V3)} and bounded convergence extend
\eqref{eq:onestep-supermtgrelation_Ciota} to $h$.
For $R\ge a$, $g\in\USC_b(D)$, and $a\le r\le s\le R$, apply
\eqref{eq:onestep-supermtgrelation_Ciota}, now valid for upper
semicontinuous tests, with $h=\mathcal T_{R-s}g$ in place of its test function.
The semigroup property gives
\[
    \mathbb E^{\mathbb P}[F Y_s^{\mathcal T;R,g}]
    \le
    \mathbb E^{\mathbb P}[F Y_r^{\mathcal T;R,g}]
\]
for every bounded nonnegative
$\widetilde{\mathcal F}_r$-measurable $F$.
Thus $Y^{\mathcal T;R,g}$ is a $\mathbb P$-supermartingale with
respect to $\widetilde{\mathbb F}$ on $[a,R]$ for every $R\ge a$ and
$g\in\USC_b(D)$, which gives $\mathbb P\in\mathcal R_{a,x}^{\mathcal T}$.

For part~\ref{prop:closed_valuation_relation_fixed}, let
$(x_j,\mathbb P_j)\to(x,\mathbb P)$ in $\mathfrak P_a$, with
$(x_j,\mathbb P_j)\in\mathfrak C_a(g,r,s)$.
For bounded nonnegative continuous cylinder multipliers $F$,
the left-hand side of
\eqref{eq:onestep-supermtgrelation_Ciota} converges, while
\eqref{eq:usc_killed_payoff_upper_passage}, applied to
$\mathcal T_{s-r}g\in\USC_b(D)$ at time $r$, bounds the upper limit
of the right-hand side by its value under $\mathbb P$.
Thus \eqref{eq:onestep-supermtgrelation_Ciota} holds under $\mathbb P$
for these multipliers. A monotone-class argument extends it to every bounded nonnegative
$\widetilde{\mathcal F}_r$-measurable $F$, proving closedness.

For part~\ref{prop:closed_valuation_relation_graph}, suppose
$(t_j,x_j,\mathbb P_j)\to(t,x,\mathbb P)$ and
$\mathbb P_j\in\mathcal R_{t_j,x_j}^{\mathcal T}
\cap\mathcal P_{t_j,x_j}(G)$.
The closed-graph assertion of Proposition~\ref{prop:propertiesP_t,x_revised} of the main text gives
$\mathbb P\in\mathcal P_{t,x}(G)$.
Fix $g\in C_b(D)$ and $t\le r<s$.
If $r>t$, then $t_j\le r$ eventually, and the same weak-passage
argument as above proves
\eqref{eq:onestep-supermtgrelation_Ciota} under $\mathbb P$.
If $r=t$, apply \eqref{eq:onestep-supermtgrelation_Ciota} under
$\mathbb P_j$ with $r=t_j$ and $F=1$ to obtain, for all sufficiently large $j$,
\[
    \mathbb E^{\mathbb P_j}
        [g(X_s)\mathbb I_{\{s<\tau_\infty\}}]
    \le \mathcal T_{s-t_j}g(x_j).
\]
The left-hand side converges, and Condition~\textup{(V4)} bounds
the upper limit of the right-hand side by $\mathcal T_{s-t}g(x)$.
Since the prescribed history makes
$\widetilde{\mathcal F}_t$ trivial under $\mathbb P$, this proves
\eqref{eq:onestep-supermtgrelation_Ciota} with $r=t$ for every bounded
nonnegative $\widetilde{\mathcal F}_t$-measurable $F$.
The case $r=s$ is immediate.
The second observation now gives
$\mathbb P\in\mathcal R_{t,x}^{\mathcal T}\cap\mathcal P_{t,x}(G)$,
proving the asserted closedness.

Finally, fix $a\ge0$ and let
\[
    \mathfrak C_a
    :=\{(x,\mathbb P)\in\mathfrak P_a:
                  \mathbb P\in\mathcal R_{a,x}^{\mathcal T}\}.
\]
By the second observation,
$\mathfrak C_a=\bigcap_{g,r,s}\mathfrak C_a(g,r,s)$, where the
intersection is over $g\in C_b(D)$ and $a\le r\le s<\infty$.
Part~\ref{prop:closed_valuation_relation_fixed} shows that the
sets $\mathfrak P_a\setminus\mathfrak C_a(g,r,s)$ form an open
cover of $\mathfrak P_a\setminus\mathfrak C_a$.
As a subspace of the second-countable space $D\times\mathfrak M$,
this complement is second countable and hence Lindel\"of.
Choose a countable subcover, indexed by
$\mathscr I_a=\{(g_i,r_i,s_i):i\ge1\}$, adding repetitions of the
trivial triple $(0,a,a)$ if necessary.
Then
\begin{equation}
\label{eq:countable_intersection_primitive_relations}
    \mathfrak C_a
    =\bigcap_{i\ge1}\mathfrak C_a(g_i,r_i,s_i).
\end{equation}
This identity holds simultaneously for all $x\in D$ and proves
part~\ref{prop:countable_valuation_relations}.
\end{proof}

\section{Approximation for One-Step Realization}
\label{app:stochastic_realization_auxiliary}

This appendix collects the auxiliary approximation results used in the proof of
Lemma~\ref{lem:usc_diagonal_stochastic_approximation} in Appendix~\ref{sec:usc_realization_proofs} of the main article.
It uses the local upper bound from
Section~\ref{sec:valuation_uncertainty}, the support correspondence from Subsection~\ref{subsec:generating_function}, and the Lyapunov condition
in Assumption~\ref{assume:lyapunov} there.
For a compact set \(K\subset\mathbb R\times\mathbb R^d\), write \(C^\infty(K)\) for the restrictions to \(K\) of smooth functions defined on an open neighborhood of \(K\).
We use the jet norm \(\|\cdot\|\), its induced distance
\(d_{\mathfrak J}\), and the corresponding Hausdorff distance
\(d_{H,\mathfrak J}\) from the main text.
We also recall
\[
    \mathfrak K
    :=(-\infty,0]\times\mathbb R^d\times\mathbb S^+(d),
    \qquad
    \kappa_\phi(x):=\frac{\phi(x)}{1+\|J_x^2\phi\|}.
\]

\subsection{Continuous Outer Approximation of a USC Generator}\label{appsubsec:continuous_outer_generator_approximation}

\begin{proposition}
\label{prop:continuous_outer_generator_approximation}
Suppose that \(G\) satisfies conditions~(G1)--(G3) and Assumption~\ref{assume:lyapunov} in the main text.  
Then, for every \(n\ge1\), there exist a Hausdorff-continuous correspondence \(\mathfrak A_n:D\rightrightarrows\mathfrak K\) with nonempty compact convex values and a jointly continuous function \(G_n:D\times\mathfrak J\to\mathbb R\) given by
\begin{align}\label{def:G_n}
    G_n(x,U)
    :=
    \max_{V\in\mathfrak A_n(x)}\ell_V(U),
    \qquad
    (x,U)\in D\times\mathfrak J,
\end{align}
such that the following properties hold.

\begin{enumerate}[label=(\roman*), ref=(\roman*)]
\item
\label{item:continuous_outer_majorization}
For every \(x\in D\),
\(\mathfrak A_G(x)\subseteq\mathfrak A_n(x)\).
Consequently,
\(G(x,U)\le G_n(x,U)\) for all
\((x,U)\in D\times\mathfrak J\).
Moreover, for each \(n\ge1\), \(G_n\) satisfies
Conditions~(G2) and~(G3), is sublinear in the jet
variable \(U\), and is locally uniformly Lipschitz in \(U\).

\item
\label{item:continuous_outer_local_uniform_boundedness}
For every compact \(K\Subset D\),
\[
    \sup_{n\ge1}
    \sup_{\substack{x\in K\\V\in\mathfrak A_n(x)}}
    \|V\|
    <\infty.
\]

\item
\label{item:continuous_outer_common_lyapunov_control}
The approximations may be chosen so that
\begin{equation}
\label{eq:continuous_outer_common_lyapunov}
    G_n\bigl(x,J_x^2\phi\bigr)
    \le
    \bigl(C_\phi+n^{-1}\bigr)\phi(x),
    \qquad x\in D.
\end{equation}

\item
\label{item:continuous_outer_upper_consistency}
One has
\begin{equation}
\label{eq:continuous_outer_generator_upper_consistency}
    \limsup_{n\to\infty}^{*}G_n
    \le G
    \qquad\text{on }D\times\mathfrak J.
\end{equation}

\end{enumerate}
\end{proposition}

\begin{proof}
By Lemma~\ref{lem:support_correspondence_geometry} in the main text,
\(\mathfrak A_G\) has nonempty compact convex values in \(\mathfrak K\), is locally bounded, and has a closed graph.
Fix \(n\ge1\).  
For every \(x_0\in D\), by local boundedness of \(\mathfrak A_G\) and continuity of \(\phi\) and \(J^2\phi\), there exists an open neighborhood \(W_{n,x_0}\Subset D\), of diameter at most \(n^{-1}\), such that
\begin{equation}
\label{eq:outer_approximation_local_lyapunov}
    \ell_V\bigl(J_x^2\phi\bigr)
    \le
    (C_\phi+n^{-1})\phi(x)
\end{equation}
whenever \(x,y\in\overline W_{n,x_0}\) and \(V\in\mathfrak A_G(y)\).  
Indeed, on a compact neighborhood of \(x_0\) the coefficients are bounded, and
\[
    \ell_V\bigl(J_x^2\phi\bigr)
    \le
    \ell_V\bigl(J_y^2\phi\bigr)
    +
    \|V\|
    \bigl\|J_x^2\phi-J_y^2\phi\bigr\|  \le
    C_\phi\phi(y)
    +
    \|V\|
    \bigl\|J_x^2\phi-J_y^2\phi\bigr\|,
\]
so \eqref{eq:outer_approximation_local_lyapunov} follows after shrinking the neighborhood.

For each \(n\ge1\), choose a locally finite precompact open refinement \((U_{n,i})_{i\in I_n}\) of the preceding cover such that \(\operatorname{diam}(U_{n,i})\le n^{-1}\) for every \(i\in I_n\).
Thus, each \(U_{n,i}\Subset D\) is contained in some member of the preceding cover, the family \((U_{n,i})_{i\in I_n}\) covers \(D\), and every point of \(D\) has a neighborhood intersecting only finitely many of the sets
\(U_{n,i}\).
Let \((\chi_{n,i})_{i\in I_n}\) be a continuous partition of unity subordinate to this refinement.
Set
\[
    B_{n,i}
    :=
    \overline{\operatorname{co}}
    \bigcup_{y\in\overline U_{n,i}}
        \mathfrak A_G(y).
\]
Since \(\mathfrak A_G\) is locally bounded and has a closed graph,
\(B_{n,i}\) is a compact convex subset of \(\mathfrak K\).  Define the
weighted Minkowski sum
\[
    \mathfrak A_n(x)
    :=
    \sum_{i\in I_n}\chi_{n,i}(x)B_{n,i}\subseteq\mathfrak K,
    \qquad x\in D.
\]
Since \(\mathfrak A_G\) is locally bounded and \(\overline U_{n,i}\Subset D\), each set \(B_{n,i}\) is bounded.
The local finiteness of the partition \((\chi_{n,i})_{i\in I_n}\) therefore implies that the correspondence \(\mathfrak A_n\) is locally bounded. 
Moreover, the standard Hausdorff-distance estimate for weighted Minkowski sums shows that \(x\mapsto\mathfrak A_n(x)\) is continuous with respect to the Hausdorff metric induced by \(\|\cdot\|\). 
Since \((V,U)\mapsto\ell_V(U)\) is continuous, Berge's maximum theorem implies that the function \(G_n:D\times\mathfrak J\to\mathbb R\) defined by \eqref{def:G_n} is jointly continuous.

First we verify \ref{item:continuous_outer_majorization}.
If \(\chi_{n,i}(x)>0\), then \(x\in U_{n,i}\), and hence \(\mathfrak A_G(x)\subseteq B_{n,i}\). 
Choosing the same \(V\in\mathfrak A_G(x)\) in every active summand gives $\mathfrak A_G(x)\subseteq\mathfrak A_n(x)$.
Taking support functions yields \(G(x,U)\le G_n(x,U)\) for every \((x,U)\in D\times\mathfrak J\).
Since \(\mathfrak A_n(x)\subseteq\mathfrak K\), the function \(G_n\) satisfies Conditions~(G2) and~(G3); moreover, its support-function representation and the local boundedness of \(\mathfrak A_n\) imply that \(G_n\) is sublinear and locally uniformly Lipschitz in the jet variable.

For \ref{item:continuous_outer_local_uniform_boundedness}, fix
\(K\Subset D\) and choose \(\delta>0\) such that
\(K_\delta:=\{y\in D:\operatorname{dist}(y,K)\le\delta\}\Subset D\).
The diameter bound on the refinement yields, whenever \(n^{-1}\le\delta\),
\begin{align}\label{eq:outer_approximation_nearby_convex_hull}
    \mathfrak A_n(x)
    \subseteq
    \overline{\operatorname{co}}
    \bigcup_{\substack{y\in D\\ |y-x|\le n^{-1}}}
        \mathfrak A_G(y)
    \subseteq
    \overline{\operatorname{co}}
    \bigcup_{y\in K_\delta}\mathfrak A_G(y),
    \qquad x\in K.
\end{align}
Local boundedness of \(\mathfrak A_G\) bounds the right-hand side
uniformly in \(n\); the remaining finitely many \(n\) are controlled
by local boundedness of each \(\mathfrak A_n\) and compactness of \(K\).

We now verify \ref{item:continuous_outer_common_lyapunov_control}.
Choose the refinement so that, for every \(i\in I_n\), there exists \(x_{n,i}\in D\) such that \(\overline U_{n,i}\subseteq W_{n,x_{n,i}}\).
If \(\chi_{n,i}(x)>0\), then \(x\in U_{n,i}\subseteq W_{n,x_{n,i}}\), and
\eqref{eq:outer_approximation_local_lyapunov}, together with the linearity
and continuity of \(V\mapsto\ell_V(J_x^2\phi)\), yields
\[
    \ell_V\bigl(J_x^2\phi\bigr)
    \le
    \bigl(C_\phi+n^{-1}\bigr)\phi(x),
    \qquad V\in B_{n,i}.
\]
Consequently, for every
\(V=\sum_{i\in I_n}\chi_{n,i}(x)V_i\in\mathfrak A_n(x)\),
\[
    \ell_V\bigl(J_x^2\phi\bigr)
    \le
    \sum_{i\in I_n}\chi_{n,i}(x)
        \bigl(C_\phi+n^{-1}\bigr)\phi(x)
    =
    \bigl(C_\phi+n^{-1}\bigr)\phi(x).
\]
Taking the maximum over \(V\in\mathfrak A_n(x)\) proves
\eqref{eq:continuous_outer_common_lyapunov}.

Finally, to prove \ref{item:continuous_outer_upper_consistency}, let
\((x_n,U_n)\to(x,U)\).  The first inclusion in
\eqref{eq:outer_approximation_nearby_convex_hull} gives
\[
    G_n(x_n,U_n)
    \le
    \sup_{\substack{y\in D\\ |y-x_n|\le n^{-1}}}G(y,U_n),
\]
and the joint upper semicontinuity of \(G\) proves
\eqref{eq:continuous_outer_generator_upper_consistency}.
This completes the proof.
\end{proof}

\subsection{Localized Residual Control}\label{appsubsec:localized_residual_control}

\begin{lemma}
\label{lem:usc_localized_krylov_bound}
Let \(\upsilon=(c,b,a)\) be continuous and suppose that, on
\([0,R]\times\overline D_m\),
\[
    \|\upsilon(s,x)\|\le K,
    \qquad
    a(s,x)\ge\lambda I_d
\]
for some \(K<\infty\) and \(\lambda>0\).  Then there is
\(N=N(d,R,D_m,K,\lambda)\) such that every
\(\mathbb P\in\widetilde{\mathcal P}_x(L^{\beta^\upsilon})\) satisfies
\begin{equation}
\label{eq:usc_virtual_krylov_bound}
    \mathbb E^{\mathbb P}\!\left[
        \int_0^{R\wedge\tau_m}g(s,X_s)\mathbb I_{\{s<\tau_\infty\}}\,ds
    \right]
    \le
    N\|g\|_{L^{d+1}([0,R]\times D_m)}
\end{equation}
for every nonnegative Borel function \(g\).
\end{lemma}

\begin{proof}
Set \(\rho:=R\wedge\tau_m\).  By
Lemma~\ref{prop:localized_characteristics_linear_gmp}, the stopped
continuous coordinate is an It\^o process whose drift and covariance are bounded
in terms of \(K\), and whose covariance is bounded below by \(\lambda I_d\), before
\(\rho\).  Moreover,
\[
    \int_0^{R\wedge\tau_m}g(s,X_s)\mathbb I_{\{s<\tau_\infty\}}\,ds
    =\int_0^\rho g(s,\overline X_s)\,ds.
\]
Extend \(g\) by zero outside \([0,R]\times D_m\). 
Thus \eqref{eq:usc_virtual_krylov_bound} follows by the parabolic occupation
estimate \cite[Theorem~2.2.4]{supp:krylov1980controlled}.
\end{proof}

\subsection{Smooth Elliptic Approximation}\label{appsubsec:usc_smooth_elliptic_approximation}

\begin{proposition}
\label{prop:usc_smooth_elliptic_approximation}
Let \(G\) be jointly continuous, satisfy (G1)--(G3), and obey
Assumption~\ref{assume:lyapunov} in the main text. Let \(\phi,C_\phi\) be as in that
assumption.  Let
\(w\in\USC_b([0,R+1]\times D)\) be a viscosity subsolution of
\[
    \partial_tw-G(x,J_x^2w)-q(t,x)=0
\]
on \((0,R+1)\times D\), where
\(q\in C([0,R+1]\times D)\).  Fix \(m\ge1\).  Then there are
\[
    \lambda_j>0,
    \qquad
    \underline t_j\searrow0,
    \qquad
    v_j\in C^\infty([\underline t_j,R]\times\overline D_m),
    \qquad
    q_j,e_j\in C([\underline t_j,R]\times\overline D_m),
\]
and continuous sublinear functions
\(G_j:\overline D_m\times\mathfrak J\to\mathbb R\) with the following
properties.

\begin{enumerate}[label=(\roman*), ref=(\roman*)]
\item\label{prop:usc_smooth_elliptic_approximation_1}
For each \(j\ge1\) and \(x\in\overline D_m\), set
\(\mathfrak A_j(x):=\mathfrak A_{G_j}(x)\subseteq\mathfrak K\).
Every covariance coordinate in \(\mathfrak A_j(x)\) is bounded below
by \(\lambda_jI_d\).
Moreover, for every compact set \(K\Subset D_m\),
\begin{equation}
\label{eq:usc_approximating_support_convergence}
    \sup_{x\in K}
    d_{H,\mathfrak J}\bigl(
        \mathfrak A_j(x),
        \mathfrak A_G(x)
    \bigr)
    \longrightarrow0,
\end{equation}

\item\label{prop:usc_smooth_elliptic_approximation_2}
Set
\[
    \mathsf K_j
    :=1+
    \sup_{\substack{x\in\overline D_m\\V\in\mathfrak A_j(x)}}\|V\|<\infty,
    \qquad
    N_j:=1+\max_{1\le\ell\le m}
      N(d,R,D_\ell,\mathsf K_j,\lambda_j),
\]
where \(N(d,R,D_\ell,\mathsf K_j,\lambda_j)\) is the constant in
Lemma~\ref{lem:usc_localized_krylov_bound}.  
Then
\begin{equation}
\label{eq:usc_krylov_small_residual}
    N_j\|e_j\|_{L^{d+1}([\underline t_j,R]\times D_m)}
    \longrightarrow0.
\end{equation}
Moreover, on \([\underline t_j,R]\times D_m\),
\begin{equation}
\label{eq:usc_classical_approximate_subsolution}
    \partial_tv_j
    -G_j(x,J_x^2v_j)
    -q_j
    \le e_j,
    \qquad e_j\ge0.
\end{equation}

\item\label{prop:usc_smooth_elliptic_approximation_3}
For every compact set \(K\Subset D_m\),
\begin{equation}
\label{eq:usc_smooth_source_convergence}
    \sup_{(t,x)\in[\underline t_j,R]\times K}
    \lvert q_j(t,x)-q(t,x)\rvert
    \longrightarrow0.
\end{equation}
Moreover, the approximants satisfy
\begin{align}
\label{eq:usc_approximation_convergence}
\begin{split}
    \limsup_{j\to\infty}^{*}v_j
    \le w&
    \qquad\text{on }[0,R]\times D_m,\\
    v_j(t,x)
    \longrightarrow w(t,x),&
    \qquad (t,x)\in(0,R]\times D_m.
\end{split}
\end{align}
Finally,
\begin{equation}
\label{eq:usc_smooth_approximants_uniform_bound}
    \sup_{j\ge1}
    \lVert v_j\rVert_{L^\infty([\underline t_j,R]\times D_m)}
    \le
    \lVert w\rVert_\infty.
\end{equation}
\end{enumerate}
\end{proposition}

\begin{proof}
\emph{Step 1: Space--time sup-convolution and the nearby outer generator.}
Set
\[
    K_m:=[0,R+1]\times\overline D_{m+2},
    \qquad
    M_w:=1+\|w\|_\infty,
    \qquad
    r_\varepsilon:=2\sqrt{M_w\varepsilon}.
\]
Since \(\overline D_m\Subset D_{m+2}\), we may choose
\(\varepsilon_0>0\) so small that, for every
\(0<\varepsilon\le\varepsilon_0\),
\[
    3r_\varepsilon<R,
    \qquad
    2r_\varepsilon<1,
    \qquad
    \bigl\{
        x\in\mathbb R^d:
        \operatorname{dist}(x,\overline D_m)\le2r_\varepsilon
    \bigr\}
    \Subset D_{m+2}.
\]
For such \(\varepsilon\), define
\[
    w^{\varepsilon,m}(t,x)
    :=
    \max_{(s,y)\in K_m}
    \left\{
        w(s,y)
        -
        \frac{|t-s|^2+|x-y|^2}{2\varepsilon}
    \right\},
    \qquad
    (t,x)\in\mathbb R\times\mathbb R^d.
\]
The maximum is attained because \(K_m\) is compact and \(w\) is upper semicontinuous. 
Moreover, \(w^{\varepsilon,m}\) is \(\varepsilon^{-1}\)-semiconvex on \(\mathbb R\times\mathbb R^d\).
Denote the set of maximizing points by
\[
    \Gamma_{\varepsilon,m}(t,x)
    :=
    \argmax_{(s,y)\in K_m}
    \left\{
        w(s,y)
        -
        \frac{|t-s|^2+|x-y|^2}{2\varepsilon}
    \right\}.
\]
Consider the open cylinder $\mathcal O_{\varepsilon,m}:=(2r_\varepsilon,R+r_\varepsilon)\times\bigl\{x\in\mathbb R^d:\operatorname{dist}(x,\overline D_m)<r_\varepsilon\bigr\}$.
Fix \((t,x)\in\mathcal O_{\varepsilon,m}\) and
\((s,y)\in\Gamma_{\varepsilon,m}(t,x)\).
Since \((t,x)\in K_m\), we have
\[
\begin{aligned}
    \frac{|t-s|^2+|x-y|^2}{2\varepsilon}
    =
    w(s,y)-w^{\varepsilon,m}(t,x)
    \le
    w(s,y)-w(t,x)
    \le 2\|w\|_\infty.
\end{aligned}
\]
Consequently,
\begin{equation}
\label{eq:usc_sup_convolution_maximizer_distance}
    |t-s|^2+|x-y|^2
    \le r_\varepsilon^2.
\end{equation}
In particular, every such maximizing point satisfies $(s,y)\in(0,R+1)\times D_{m+2}$.

For
\(\operatorname{dist}(x,\overline D_m)<r_\varepsilon\), define
\[
    \mathfrak A_G^{\varepsilon,m}(x)
    :=
    \overline{\operatorname{co}}
    \left(
        \bigcup_{\substack{y\in\overline D_{m+2}\\
                           |y-x|\le r_\varepsilon}}
        \mathfrak A_G(y)
    \right), \qquad
    G^{\varepsilon,m}(x,U)
    :=
    \max_{\substack{y\in\overline D_{m+2}\\
                    |y-x|\le r_\varepsilon}}
    G(y,U).
\]
Then we have $\mathfrak A_{G^{\varepsilon,m}}=\mathfrak A_G^{\varepsilon,m}$ and
\begin{equation}
\label{eq:usc_nearby_support_representation}
    G^{\varepsilon,m}(x,U)
    =
    \max_{V\in\mathfrak A_G^{\varepsilon,m}(x)}
    \ell_V(U).
\end{equation}
Thus \(G^{\varepsilon,m}(x,\cdot)\) is sublinear, and \(\mathfrak A_G^{\varepsilon,m}(x)\) is its compact convex support set.
By joint continuity of \(G\), \(G^{\varepsilon,m}\) is also continuous on its domain.
For \((t,x)\in\mathcal O_{\varepsilon,m}\), set
\[
    q^{\varepsilon,m}(t,x)
    :=
    \max_{\substack{(\tau,y)\in K_m\\
                    |\tau-t|\le r_\varepsilon\\
                    |y-x|\le r_\varepsilon}}
    q(\tau,y).
\]
The preceding choice of \(\varepsilon_0\) ensures that the maximizing
set stays inside \((0,R+1)\times D_{m+2}\), so
\(q^{\varepsilon,m}\) is continuous on
\(\mathcal O_{\varepsilon,m}\).

We claim that \(w^{\varepsilon,m}\) is a viscosity subsolution of
\begin{equation}
\label{eq:usc_sup_convolution_outer_equation}
    \partial_tu
    -
    G^{\varepsilon,m}\bigl(x,J_x^2u\bigr)
    -
    q^{\varepsilon,m}
    =
    0
\end{equation}
on \(\mathcal O_{\varepsilon,m}\).
We use the equivalent viscosity formulation with local \(C^{1,2}\)
test functions. Indeed, let \(\psi\in C^{1,2}\) touch \(w^{\varepsilon,m}\) from above
at \((t,x)\in\mathcal O_{\varepsilon,m}\), and choose
\((s,y)\in\Gamma_{\varepsilon,m}(t,x)\).
Define, near \((s,y)\),
\[
    \widetilde\psi(\tau,z)
    :=
    \psi(\tau+t-s,z+x-y)
    +
    w(s,y)-\psi(t,x).
\]
The standard jet-transfer property of the sup-convolution \cite[Lemma~A.5]{supp:crandall1992user} shows that \(\widetilde\psi\) touches \(w\) from above at \((s,y)\).
Writing $(p,H):=(\nabla\psi(t,x),D_x^2\psi(t,x))$, the viscosity-subsolution property of \(w\) yields
\begin{align}\label{eq:viscosity_inequality_psi}
    \partial_t\psi(t,x)
    -
    G\bigl(y,w(s,y),p,H\bigr)
    -
    q(s,y)
    \le0.
\end{align}
By the maximizing identity,
\[
    w(s,y)
    =
    w^{\varepsilon,m}(t,x)
    +
    \frac{|t-s|^2+|x-y|^2}{2\varepsilon}
    \ge
    w^{\varepsilon,m}(t,x)
    =
    \psi(t,x).
\]
Therefore, condition (G3), \eqref{eq:usc_sup_convolution_maximizer_distance} and the definitions of \(G^{\varepsilon,m}\) and \(q^{\varepsilon,m}\) yield
\begin{align}
\label{eq:G,q_inequalities}
    G\bigl(y,w(s,y),p,H\bigr)
    \le
    G^{\varepsilon,m}
    \bigl(x,J_x^2\psi(t,\cdot)\bigr),
    \qquad
    q(s,y)
    \le
    q^{\varepsilon,m}(t,x).
\end{align}
Combining \eqref{eq:viscosity_inequality_psi} with
\eqref{eq:G,q_inequalities} shows that \(\psi\) satisfies the viscosity
subsolution inequality for
\eqref{eq:usc_sup_convolution_outer_equation}.  This proves the claim.

\medskip
\noindent
\emph{Step 2: Addition of a uniform ellipticity floor.}
Fix \(\lambda>0\) and define
\[
    G^{\varepsilon,m,\lambda}(x,z,p,H)
    :=
    G^{\varepsilon,m}(x,z,p,H)
    +
    \frac{\lambda}{2}\operatorname{tr}(H),
    \qquad
     q^{\varepsilon,m,\lambda}
    :=
    q^{\varepsilon,m}
    +
    \frac{d\lambda}{2\varepsilon}.
\]
By \eqref{eq:usc_nearby_support_representation}, the support correspondence of \(G^{\varepsilon,m,\lambda}\) is
\[
    \mathfrak A_G^{\varepsilon,m,\lambda}(x)
    :=
    \left\{
        (c,b,a+\lambda I_d):
        (c,b,a)\in\mathfrak A_G^{\varepsilon,m}(x)
    \right\}.
\]
It is contained in \(\mathfrak K\), and every one of its covariance
coordinates is bounded below by \(\lambda I_d\).

By the regularity theory for semiconvex functions
\cite[Theorem~2.3.1]{supp:cannarsa2004semiconcave}, applied to
\(-w^{\varepsilon,m}\), the function \(w^{\varepsilon,m}\) is locally
Lipschitz and twice differentiable almost everywhere.
At every Alexandrov point of
\(\mathcal O_{\varepsilon,m}\), its classical derivatives satisfy
\[
    D_x^2w^{\varepsilon,m}
    \ge
    -\varepsilon^{-1}I_d.
\]
At such a point, the viscosity inequality from Step~1 gives the
corresponding almost-everywhere differential inequality. Consequently,
\begin{align}\label{eq:almosteverywhere_viscosityinequality}
\begin{split}
    &\partial_tw^{\varepsilon,m}
    -
    G^{\varepsilon,m,\lambda}
        \bigl(x,J_x^2w^{\varepsilon,m}\bigr)
    -
    q^{\varepsilon,m,\lambda}
    \\
    &\quad=
    \partial_tw^{\varepsilon,m}
    -
    G^{\varepsilon,m}
        \bigl(x,J_x^2w^{\varepsilon,m}\bigr)
    -
    q^{\varepsilon,m}
    -
    \frac{\lambda}{2}
        \operatorname{tr}\bigl(D_x^2w^{\varepsilon,m}\bigr)
    -
    \frac{d\lambda}{2\varepsilon}
    \le0
\end{split}
\end{align}
almost everywhere on \(\mathcal O_{\varepsilon,m}\).

\medskip
\noindent
\emph{Step 3: Mollification and vanishing positive residual.}
Set $Q_{\varepsilon,m}:=[3r_\varepsilon,R]\times\overline D_m$.
Let \(\rho\in C_c^\infty(\mathbb R^{d+1})\) be nonnegative, supported in the unit ball, and satisfy \(\int\rho=1\).  
Write $\rho_\delta(\tau,z):=\delta^{-(d+1)}\rho\left(\frac{\tau}{\delta},\frac{z}{\delta}\right)$ and, for \(0<\delta<r_\varepsilon/2\), define $v^{\varepsilon,m,\delta}:=w^{\varepsilon,m}*\rho_\delta$ on a neighborhood of \(Q_{\varepsilon,m}\).

On the other hand, since $w^{\varepsilon,m}$ is $\varepsilon^{-1}$-semiconvex, the distributional spatial Hessian of \(w^{\varepsilon,m}\) has the decomposition
\[
    D_x^2w^{\varepsilon,m}
    =
    H^{\varepsilon,m}\,dt\,dx
    +
    \mu^{\varepsilon,m}_{\mathrm s},
\]
where $H^{\varepsilon,m}=D_x^2w^{\varepsilon,m}$ almost everywhere and \(\mu^{\varepsilon,m}_{\mathrm s}\) is a positive-semidefinite matrix-valued Radon measure.
Therefore,
\begin{equation}
\label{eq:usc_mollified_hessian_decomposition}
    D_x^2v^{\varepsilon,m,\delta}
    =
    H^{\varepsilon,m}*\rho_\delta
    +
    \mu^{\varepsilon,m}_{\mathrm s}*\rho_\delta
    \ge
    H^{\varepsilon,m}*\rho_\delta.
\end{equation}
Define the classical residual on \(Q_{\varepsilon,m}\) by
\[
    \mathcal E^{\varepsilon,m,\lambda,\delta}(t,x)
    :=
    \partial_tv^{\varepsilon,m,\delta}(t,x)
    -
    G^{\varepsilon,m,\lambda}
    \bigl(
        x,
        J_x^2v^{\varepsilon,m,\delta}(t,\cdot)
    \bigr)
    -
    q^{\varepsilon,m,\lambda}(t,x).
\]
Then, at almost every common Lebesgue point of the a.e. derivatives of
\(w^{\varepsilon,m}\) where
\eqref{eq:almosteverywhere_viscosityinequality} holds,
\eqref{eq:usc_mollified_hessian_decomposition}, the Lebesgue
differentiation theorem, and the joint continuity and degenerate
ellipticity of \(G^{\varepsilon,m,\lambda}\) yield
\begin{align}\label{eq:converge_mathcalE}
    \limsup_{\delta\searrow0}
    \mathcal E^{\varepsilon,m,\lambda,\delta}(t,x)
    \le0.
\end{align}

We now prove that, for every \(p\in[1,\infty)\),
\begin{equation}
\label{eq:usc_fixed_parameter_residual_convergence}
    \lim_{\delta\searrow0}
    \left\|
        \bigl(\mathcal E^{\varepsilon,m,\lambda,\delta}\bigr)^+
    \right\|_{L^p(Q_{\varepsilon,m})}
    =0.
\end{equation}
Recall that the closed \(r_\varepsilon/2\)-neighborhood of \(Q_{\varepsilon,m}\) is compactly contained in \(\mathcal O_{\varepsilon,m}\).
Since \(w^{\varepsilon,m}\) is locally Lipschitz, the functions \(v^{\varepsilon,m,\delta}\), \(\partial_tv^{\varepsilon,m,\delta}\), and \(\nabla_xv^{\varepsilon,m,\delta}\) are uniformly bounded for \(0<\delta<r_\varepsilon/2\).
Moreover, mollification preserves semiconvexity, and hence
\[
    D_x^2v^{\varepsilon,m,\delta}
    \ge
    -\varepsilon^{-1}I_d.
\]
By property~\textup{(G2)} of \(G^{\varepsilon,m,\lambda}\),
\[
    G^{\varepsilon,m,\lambda}
    \bigl(
        x,
        v^{\varepsilon,m,\delta},
        \nabla_xv^{\varepsilon,m,\delta},
        D_x^2v^{\varepsilon,m,\delta}
    \bigr)
    \ge
    G^{\varepsilon,m,\lambda}
    \bigl(
        x,
        v^{\varepsilon,m,\delta},
        \nabla_xv^{\varepsilon,m,\delta},
        -\varepsilon^{-1}I_d
    \bigr).
\]
The right-hand side is bounded below uniformly in \(\delta\) on \(Q_{\varepsilon,m}\). 
Consequently, for some constant \(C_{\varepsilon,m,\lambda}\) independent of \(\delta\), we have $0\le(\mathcal E^{\varepsilon,m,\lambda,\delta})^+\le C_{\varepsilon,m,\lambda}$ on \(Q_{\varepsilon,m}\).
Thus, \eqref{eq:converge_mathcalE} and the dominated convergence theorem yield \eqref{eq:usc_fixed_parameter_residual_convergence}.

\medskip
\noindent
\emph{Step 4: Diagonal choice.}
Choose \(\varepsilon_j\searrow0\) such that, with
\(r_j:=r_{\varepsilon_j}\),
\[
    3r_j<R,
    \qquad
    2r_j<1,
    \qquad
    \bigl\{
        x\in\mathbb R^d:
        \operatorname{dist}(x,\overline D_m)\le2r_j
    \bigr\}
    \Subset D_{m+2}
\]
for every \(j\ge1\). 
Set $\underline t_j:=3r_j$ and $\lambda_j:=\varepsilon_j^2$.
Define
\[
    G_j(x,U)
    :=
    G^{\varepsilon_j,m,\lambda_j}(x,U),
    \qquad
    (x,U)\in\overline D_m\times\mathfrak J,
\]
and
\[
    q_j(t,x)
    :=
    q^{\varepsilon_j,m,\lambda_j}(t,x)
    =
    q^{\varepsilon_j,m}(t,x)
    +
    \frac{d\lambda_j}{2\varepsilon_j},
    \qquad
    (t,x)\in[\underline t_j,R]\times\overline D_m.
\]
The support correspondence of \(G_j\) is
\[
    \mathfrak A_j(x)
    =
    \left\{
        (c,b,a+\lambda_jI_d):
        (c,b,a)\in
        \mathfrak A_G^{\varepsilon_j,m}(x)
    \right\}.
\]
Thus \(\mathfrak A_j(x)\subset\mathfrak K\), and every covariance coordinate in \(\mathfrak A_j(x)\) is bounded below by \(\lambda_jI_d\).
Compactness and continuity of these support sets also imply that \(\mathsf K_j\) and \(N_j\), as defined in the statement, are finite.
Fix a compact set \(K\Subset D_m\).
Each correspondence $\mathfrak A_j$ is uniformly Hausdorff continuous on a compact neighborhood of \(K\).
Indeed, by construction, \(\mathfrak A_j(x)\) is obtained by taking the closed convex hull of the sets \(\mathfrak A_G(y)\) with \(|y-x|\le r_{\varepsilon_j}\), and then translating the covariance coordinate by \(\lambda_jI_d\). 
Since \(\mathfrak A_G\) is Hausdorff continuous, it is uniformly Hausdorff continuous on a compact neighborhood of \(K\); hence, by the convexity of \(\mathfrak A_G(x)\),
\[
    \sup_{x\in K}
    d_{H,\mathfrak J}\bigl(
        \mathfrak A_j(x),
        \mathfrak A_G(x)
    \bigr)
    \le
    \omega_K(r_{\varepsilon_j})
    +
    \lambda_j\lVert(0,0,I_d)\rVert
    \longrightarrow0,
\]
where \(\omega_K\) is a corresponding Hausdorff modulus of continuity.
This proves part~\ref{prop:usc_smooth_elliptic_approximation_1}.

For each \(j\), the operator \(G_j\), the source \(q_j\), and hence \(N_j\), are now fixed independently of the mollification parameter.
By \eqref{eq:usc_fixed_parameter_residual_convergence} and the uniform convergence of the mollifications, we may choose \(0<\delta_j<r_j/2\) such that
\begin{align}
\label{eq:uniformerror_delta_j}
    N_j
    \left\|
        \bigl(
            \mathcal E^{\varepsilon_j,m,\lambda_j,\delta_j}
        \bigr)^+
    \right\|_{L^{d+1}(Q_{\varepsilon_j,m})}
    +
    \left\|
        v^{\varepsilon_j,m,\delta_j}
        -
        w^{\varepsilon_j,m}
    \right\|_{L^\infty(Q_{\varepsilon_j,m})}
    \le j^{-1}.
\end{align}
Set $v_j:=v^{\varepsilon_j,m,\delta_j}$ and $e_j:=\bigl(\mathcal E^{\varepsilon_j,m,\lambda_j,\delta_j}\bigr)^+$.
By construction, the resulting objects satisfy part~\ref{prop:usc_smooth_elliptic_approximation_2}.

It remains to verify part~\ref{prop:usc_smooth_elliptic_approximation_3}.
Let \(K\Subset D_m\).
Since \(r_j\to0\), the continuity of \(q\) implies that
\[
    q^{\varepsilon_j,m}\longrightarrow q
    \quad\text{uniformly on }
    [\underline t_j,R]\times K.
\]
Together with \(\varepsilon_j^{-1}\lambda_j\to0\), this proves \eqref{eq:usc_smooth_source_convergence}.
To verify \eqref{eq:usc_approximation_convergence}, let
\[
    (\tau_j,x_j)
    \in
    [\underline t_j,R]\times\overline D_m,
    \qquad
    (\tau_j,x_j)\longrightarrow(t,x)
    \in[0,R]\times D_m,
\]
and choose $(s_j,y_j)\in\Gamma_{\varepsilon_j,m}(\tau_j,x_j)$.
By \eqref{eq:usc_sup_convolution_maximizer_distance}, \((s_j,y_j)\to(t,x)\).
Moreover, by the definition of \(\Gamma_{\varepsilon_j,m}(\tau_j,x_j)\), we have $w^{\varepsilon_j,m}(\tau_j,x_j)\le w(s_j,y_j)$.
Therefore, using \((s_j,y_j)\to(t,x)\), the upper semicontinuity of \(w\), and \eqref{eq:uniformerror_delta_j}, we obtain
\[
    \limsup_{j\to\infty}v_j(\tau_j,x_j)
    =
    \limsup_{j\to\infty}
    w^{\varepsilon_j,m}(\tau_j,x_j)
    \le
    \limsup_{j\to\infty}w(s_j,y_j)
    \le
    w(t,x).
\]
This proves the first assertion in
\eqref{eq:usc_approximation_convergence}.
To prove the second assertion in
\eqref{eq:usc_approximation_convergence}, fix
\((t,x)\in(0,R]\times D_m\).  
For all sufficiently large \(j\), we have \(t\ge\underline t_j\).
Since \(w^{\varepsilon_j,m}(t,x)\ge w(t,x)\), \eqref{eq:uniformerror_delta_j} yields $v_j(t,x)\ge w^{\varepsilon_j,m}(t,x)-j^{-1}\ge w(t,x)-j^{-1}$.
Therefore,
\[
    \liminf_{j\to\infty}v_j(t,x)\ge w(t,x).
\]
On the other hand, applying the first assertion in \eqref{eq:usc_approximation_convergence} to the constant sequence \((\tau_j,x_j)\equiv(t,x)\) yields the reverse upper bound, and hence \(v_j(t,x)\to w(t,x)\).

Finally, the convolution defining $v_j$ on $[\underline t_j,R]\times D_m$ only samples points of $K_m$, where
\[
    -\|w\|_\infty
    \le w^{\varepsilon_j,m}
    \le \|w\|_\infty.
\]
Since the mollifier is nonnegative and has unit mass, we obtain
\[
    \|v_j\|_{L^\infty([\underline t_j,R]\times D_m)}
    \le \|w\|_\infty,
\]
which proves \eqref{eq:usc_smooth_approximants_uniform_bound}.
This completes the proof.
\end{proof}

\subsection{Continuous Coefficient Fields and Realizing Laws}\label{appsubsec:continuous_coefficient_weak_existence}

We first establish weak existence for continuous coefficient fields directly on the killed path space. We then construct near-maximizing fields and their realizing laws.

\begin{lemma}
\label{lem:continuous_locally_bounded_linear_existence}
Let $\upsilon=(c,b,a):[0,\infty)\times D\to\mathfrak K$ be continuous and spatially locally bounded uniformly in time, in the sense that
\[
    \sup_{s\ge0,\,y\in K}\|\upsilon(s,y)\|<\infty
    \qquad\text{for every compact }K\Subset D,
\]
and let $\beta^\upsilon$ be the induced Markovian coefficient field.
Suppose that, for a function $\phi\in C^2(D)$ with $\phi\ge1$ and a constant $C\ge0$,
\begin{equation}
\label{eq:continuous_coefficients_lyapunov}
    m_n^\phi:=\inf_{y\in D\setminus D_n}\phi(y)
    \longrightarrow\infty,
    \qquad\ell_{\upsilon(s,x)}(J_x^2\phi)
    \le C\phi(x),
    \qquad (s,x)\in[0,\infty)\times D.
\end{equation}
Then, for every $x\in D$, $\widetilde{\mathcal P}_x(L^{\beta^\upsilon})\neq\varnothing$.
\end{lemma}

\begin{proof}
Choose $N_0$ with $x\in D_{N_0}$.  For $N\ge N_0$, let
$\chi_N\in C_c^\infty(D)$ satisfy
$0\le\chi_N\le1$, $\chi_N=1$ on $D_N$, and
$\operatorname{supp}\chi_N\Subset D_{N+1}$.  Choose also
$\eta_N\in C_c^\infty(\mathbb R)$ with
$0\le\eta_N\le1$ and $\eta_N=1$ on $[0,N]$, and put
\[
    r_N(s,y):=\eta_N(s)\chi_N(y),
    \qquad
    \upsilon^N(s,y):=r_N(s,y)\upsilon(s,y),
    \quad (s,y)\in[0,\infty)\times D.
\]
Set $\beta^N:=\beta^{\upsilon^N}$ and $k^N:=-c^N$.
The maps $b^N=r_Nb$ and
$\sigma^N:=\sqrt{r_N}\,a^{1/2}$ have compact spatial support in $D$ and
are bounded and continuous.  Their zero extensions to
$[0,\infty)\times\mathbb R^d$ are bounded and continuous, and
$a^N:=\sigma^N(\sigma^N)^\top=r_Na$.  The bounded time-inhomogeneous
existence theorem
\cite[Theorem~6.1.7]{supp:stroock1997multidimensional} therefore gives a weak
solution $Z^N$ starting from $x$ for the coefficients
$(b^N,\sigma^N)$ on a filtered probability space
$(\Omega_N,\mathcal F_N,\mathbb F^N,\mathbf P_N)$.
Define
\[
    \theta_N:=\inf\{s\ge0:Z_s^N\notin D_{N+1}\},
    \qquad
    Y_s^N:=Z_{s\wedge\theta_N}^N.
\]
On $\{\theta_N<\infty\}$, continuity gives
$Z_{\theta_N}^N\in\partial D_{N+1}$, where $b^N$ and $\sigma^N$
vanish at every time.  Thus $Y^N$, kept fixed after $\theta_N$, also
solves the same SDE and remains in the compact set
$\overline D_{N+1}\Subset D$ for all time.

Enlarge the space by a unit exponential random variable $E_N$ independent
of $\mathcal F_\infty^N$, and set
\[
    A_s^N:=\int_0^s k^N(r,Y_r^N)\,dr,
    \qquad
    \zeta_N:=\inf\{s\ge0:A_s^N\ge E_N\}.
\]
The function $k^N$ is nonnegative and bounded, so $A^N$ is continuous
and finite on finite intervals, and $\zeta_N>0$ almost surely.
Let $\mathbb G^N$ be the usual augmentation of the filtration
\[
    \mathcal F_s^N\vee
    \sigma\bigl(\mathbb I_{\{\zeta_N\le r\}}:0\le r\le s\bigr),
    \qquad s\ge0.
\]
Independence of $E_N$ ensures that
$Y^N-x-\int_0^{\cdot}b^N(r,Y_r^N)\,dr$ remains a
$\mathbb G^N$-local martingale.  Moreover, the exponential survival identity gives
\[
    \mathbb I_{\{\zeta_N\le s\}}
    -\int_0^s k^N(r,Y_r^N)\mathbb I_{\{r<\zeta_N\}}\,dr,
    \qquad s\ge0,
\]
as a $\mathbb G^N$-martingale.  Indeed, on survival up to a time $u$,
a bounded $\mathcal G_u^N$-measurable multiplier agrees with an
$\mathcal F_u^N$-measurable one.  Integrating first over $E_N$ reduces
the compensated-increment identity between $u$ and $s$ to
\[
    e^{-A_u^N}-e^{-A_s^N}
    =\int_u^s e^{-A_r^N}k^N(r,Y_r^N)\,dr.
\]

Define $X_s^N:=Y_s^N$ for $s<\zeta_N$ and
$X_s^N:=\triangle$ for $s\ge\zeta_N$, and let $\mathbb P^N$ be its
law on $\widetilde\Omega$.  Its initial state is $x$, it has no
continuous explosion, and its killing time is $\zeta_N$.
For $f\in C_b^\infty(D)$, It\^o's formula and the preceding
compensator identity show that
\[
\begin{aligned}
    f(Y_s^N)\mathbb I_{\{s<\zeta_N\}}-f(x)
    -\int_0^s
        \ell_{\upsilon^N(r,Y_r^N)}(J_{Y_r^N}^2f)
        \mathbb I_{\{r<\zeta_N\}}\,dr,
    \qquad s\ge0,
\end{aligned}
\]
is a $\mathbb G^N$-local martingale.  Since $Y^N$ remains in
$\overline D_{N+1}$, this process is bounded on each finite time
interval and hence is a true martingale there.  Stopping at the
canonical localizations $\tau_n$ and passing to the canonical
filtration therefore gives
\[
    \mathbb P^N\in\widetilde{\mathcal P}_x(L^{\beta^N}).
\]
The cutoffs preserve the Lyapunov inequality, since
\[
    L^{\beta^N}(s,\omega,J_{X_s}^2\phi)
    =r_N(s,X_s)
      L^{\beta^\upsilon}(s,\omega,J_{X_s}^2\phi)
    \le C\phi(X_s).
\]

The family $(\beta^N)_{N\ge N_0}$ is uniformly bounded on each finite horizon before every fixed
$\tau_n$ and satisfies the common Lyapunov estimate
\eqref{eq:continuous_coefficients_lyapunov}.  Hence
Proposition~\ref{prop:compactness_criterion} gives tightness of
$(\mathbb P^N)_N$.  Along a subsequence,
$\mathbb P^N\Rightarrow\mathbb P$.
Fixed-time evaluation at zero is continuous, so
$\mathbb P(X_0=x)=1$.  Fix $T>0$, $n\ge1$, and
$f\in C_b^\infty(D)$, and set
\[
    q_f(s,y):=\ell_{\upsilon(s,y)}(J_y^2f).
\]
For all sufficiently large $N$, one has
$\beta^N=\beta^\upsilon$ on
$[0,T]\times\overline D_n$.  Consequently, under $\mathbb P^N$ the
process $\mathcal Z^{0,\tau_n}[f,q_f]$ is a martingale on $[0,T]$.
Set
\[
    \kappa_s^N:=k^{\beta^N}(s,X),
    \qquad s\ge0.
\]
The compensator identity established in the construction gives
\[
    N_s-\int_0^s\kappa_r^N\mathbb I_{\{r<\tau_\infty\}}\,dr,
    \qquad s\ge0,
\]
as a $\mathbb P^N$-martingale.  The spatial local boundedness of
$\upsilon$, together with $0\le r_N\le1$, gives the uniform local
density bound
\eqref{eq:weak_passage_killing_density_bound}.  The source $q_f$ is
continuous on every compact localization. Applying
Proposition~\ref{prop:weak_passage_martingale_relations}\ref{item:weak_passage_martingale}
with $q_N=q=q_f$ and the preceding $\kappa^N$ shows that
$\mathcal Z^{0,\tau_n}[f,q_f]$ is a $\mathbb P$-martingale on $[0,T]$.
Since $T,n$, and
$f$ are arbitrary, $\mathbb P\in
\widetilde{\mathcal P}_x(L^{\beta^\upsilon})$.
\end{proof}

The following proposition constructs continuous near-maximizing coefficient fields from the smooth approximation and applies the preceding existence result.

\begin{proposition}
\label{prop:usc_outer_compatible_coefficients}
Let \(G,w,q,R\) satisfy the hypotheses of
Proposition~\ref{prop:usc_smooth_elliptic_approximation}, and let
\((\phi,C_\phi)\) be the Lyapunov pair appearing there.
Fix \(m\ge1\), and apply
Proposition~\ref{prop:usc_smooth_elliptic_approximation} with \(m\)
replaced by \(m+1\).
Let $\bigl(\lambda_j,\underline t_j,v_j,q_j,e_j,G_j\bigr)_{j\ge1}$ be the resulting approximation.
Thus,
\[
    v_j\in
    C^\infty\bigl(
        [\underline t_j,R]\times\overline D_{m+1}
    \bigr),
    \qquad
    G_j:
    \overline D_{m+1}\times\mathfrak J
    \longrightarrow\mathbb R.
\]
Then there exist \(\varepsilon_j\searrow0\) and continuous coefficient
maps
\[
    \zeta_j=(c_j,b_j,a_j):
    [0,\infty)\times D
    \longrightarrow\mathfrak K
\]
such that the following properties hold for every \(j\ge1\).

\begin{enumerate}[label=(\roman*), ref=(\roman*)]
\item\label{prop:usc_outer_support_membership} For every \((s,x)\in[0,\infty)\times D\), $\zeta_j(s,x)\in\mathfrak A_G^{\varepsilon_j}(x)$.

\item\label{prop:usc_outer_near_maximizing_jet} For every
\((s,x)\in[0,R-\underline t_j]\times\overline D_m\),
\begin{equation}
\label{eq:usc_outer_near_maximizing_jet}
    \ell_{\zeta_j(s,x)}
    \bigl(J_x^2v_j(R-s,\cdot)\bigr)
    \ge
    G_j\bigl(
        x,J_x^2v_j(R-s,\cdot)
    \bigr)
    -j^{-1}.
\end{equation}

\item\label{prop:usc_outer_covariance_floor}
For every \((s,x)\in[0,R]\times\overline D_m\),
\[
    a_j(s,x)\ge\lambda_jI_d,
    \qquad
    \|\zeta_j(s,x)\|\le\mathsf K_j,
\]
where \(\mathsf K_j\) is defined in
Proposition~\ref{prop:usc_smooth_elliptic_approximation}\ref{prop:usc_smooth_elliptic_approximation_2}
for the approximation on \(\overline D_{m+1}\).

\item\label{prop:usc_outer_linear_realizing_laws}
For every \(x\in D\),
\(\varnothing\neq\widetilde{\mathcal P}_x(L^{\beta^{\zeta_j}})
\subseteq\mathcal P_x(G^{\varepsilon_j})\).
\end{enumerate}
\end{proposition}

\begin{proof}
Choose open sets \(O_m^0,O_m^1\) and a smooth cutoff \(\chi_m\) such
that
\[
    \overline D_m
    \subset
    O_m^0
    \Subset
    O_m^1
    \Subset
    D_{m+1},
    \qquad
    0\le\chi_m\le1,
    \qquad
    \chi_m=1\ \text{on }O_m^0,
    \qquad
    \operatorname{supp}\chi_m\Subset O_m^1.
\]
Write $\mathfrak A_j:=\mathfrak A_{G_j}$.
By \eqref{eq:usc_approximating_support_convergence}, applied with the index \(m+1\),
\begin{equation}
\label{eq:usc_local_support_approximation_error}
    \alpha_j
    :=
    \sup_{x\in\overline O_m^1}
    d_{H,\mathfrak J}\bigl(
        \mathfrak A_j(x),
        \mathfrak A_G(x)
    \bigr)
    \longrightarrow0.
\end{equation}
Set $Q_j:=[0,R-\underline t_j]\times\overline O_m^1$ and $U_j(s,x):=J_x^2v_j(R-s,\cdot)$.
For every \(\xi=(s,x)\in Q_j\), choose \(V_\xi^j\in\mathfrak A_j(x)\) such that $\ell_{V_\xi^j}\bigl(U_j(s,x)\bigr)=G_j\bigl(x,U_j(s,x)\bigr)$.
Since \(U_j\) and \(G_j\) are continuous, there exists a relatively open neighborhood \(\mathcal U_\xi^j\) of \(\xi\) in \(Q_j\) such that
\begin{equation}
\label{eq:usc_local_fixed_coefficient_near_maximizer}
    \ell_{V_\xi^j}\bigl(U_j(r,y)\bigr)
    \ge
    G_j\bigl(y,U_j(r,y)\bigr)-j^{-1},
    \qquad
    (r,y)\in\mathcal U_\xi^j.
\end{equation}
By shrinking \(\mathcal U_\xi^j\), if necessary, the Hausdorff
continuity of \(\mathfrak A_G\) also allows us to assume that
\[
    d_{H,\mathfrak J}\bigl(
        \mathfrak A_G(x),
        \mathfrak A_G(y)
    \bigr)
    \le j^{-1},
    \qquad
    (r,y)\in\mathcal U_\xi^j.
\]
Choose a finite subcover
\[
    Q_j
    \subseteq
    \bigcup_{i=1}^{I_j}\mathcal U_{\xi_i}^j,
    \qquad
    \xi_i=(s_i,x_i),
\]
and let \((\theta_i^j)_{i=1}^{I_j}\) be a continuous partition of
unity on \(Q_j\) subordinate to this cover.  Define
\[
    \widehat\zeta_j(r,y)
    :=
    \sum_{i=1}^{I_j}
    \theta_i^j(r,y)V_{\xi_i}^j,
    \qquad
    (r,y)\in Q_j.
\]
Then \(\widehat\zeta_j\) is continuous and takes values in
\(\mathfrak K\).  Summing
\eqref{eq:usc_local_fixed_coefficient_near_maximizer} against the
partition of unity gives
\begin{equation}
\label{eq:usc_local_field_near_maximizer}
    \ell_{\widehat\zeta_j(r,y)}
    \bigl(U_j(r,y)\bigr)
    \ge
    G_j\bigl(y,U_j(r,y)\bigr)-j^{-1},
    \qquad
    (r,y)\in Q_j.
\end{equation}
Moreover, every coefficient \(V_{\xi_i}^j\) with $\theta_i^j(r,y)>0$ satisfies
\[
    d_{\mathfrak J}\bigl(
        V_{\xi_i}^j,
        \mathfrak A_G(y)
    \bigr)
    \le
    \alpha_j+j^{-1}.
\]
Since \(\mathfrak A_G(y)\) is convex, we have
\begin{equation}
\label{eq:usc_local_field_support_distance}
    d_{\mathfrak J}\bigl(
        \widehat\zeta_j(r,y),
        \mathfrak A_G(y)
    \bigr)
    \le\alpha_j+j^{-1},
    \qquad
    (r,y)\in Q_j.
\end{equation}
In addition, every covariance coordinate of \(V_{\xi_i}^j\) is bounded below by \(\lambda_jI_d\).
Therefore, the covariance coordinate of \(\widehat\zeta_j\) is also bounded below by \(\lambda_jI_d\) throughout \(Q_j\).

The Hausdorff continuity of \(\mathfrak A_G\) implies lower
semicontinuity, and its values are nonempty, closed, and convex.
Hence Michael's continuous selection theorem
\cite{supp:michael2003continuous} provides a continuous map $\zeta^0:D\to\mathfrak K$ such that $\zeta^0(x)\in\mathfrak A_G(x)$.
Define
\[
    \underline\kappa_m
    :=
    \min_{x\in\overline O_m^1}\kappa_\phi(x)>0, \qquad
    \varepsilon_j
    :=
    j^{-1}
    +
    \underline\kappa_m^{-1}
    \sup_{k\ge j}(\alpha_k+k^{-1}).
\]
Then \(\varepsilon_j\searrow0\), and
\begin{equation}
\label{eq:usc_support_error_absorbed_by_outer_radius}
    \alpha_j+j^{-1}
    \le
    \varepsilon_j\underline\kappa_m
    \le
    \varepsilon_j\kappa_\phi(x),
    \qquad
    x\in\overline O_m^1.
\end{equation}

Extend \(\widehat\zeta_j\) in time by setting $\overline\zeta_j(s,x):=\widehat\zeta_j\bigl(s\wedge(R-\underline t_j),x\bigr)$ for $(s,x)\in[0,\infty)\times\overline O_m^1$, and define
\[
    \zeta_j(s,x)
    :=
    \begin{cases}
    \chi_m(x)\overline\zeta_j(s,x)
    +(1-\chi_m(x))\zeta^0(x),
        &x\in O_m^1,\\[1mm]
    \zeta^0(x),
        &x\in D\setminus O_m^1.
    \end{cases}
\]
Since \(\chi_m\) vanishes on a neighborhood of
\(\partial O_m^1\), the map \(\zeta_j\) is continuous on
\([0,\infty)\times D\).

We now verify the asserted properties.
By \eqref{eq:usc_local_field_support_distance} and \eqref{eq:usc_support_error_absorbed_by_outer_radius},
\[
    \overline\zeta_j(s,x)
    \in
    \mathfrak A_G^{\varepsilon_j}(x),
    \qquad
    (s,x)\in[0,\infty)\times\overline O_m^1.
\]
Also,
\(\zeta^0(x)\in\mathfrak A_G(x)
\subseteq\mathfrak A_G^{\varepsilon_j}(x)\).
Since each \(\mathfrak A_G^{\varepsilon_j}(x)\) is convex, the construction proves \ref{prop:usc_outer_support_membership}.
On
\([0,R-\underline t_j]\times\overline D_m\), one has \(\chi_m=1\), so \(\zeta_j=\widehat\zeta_j\).
Thus \eqref{eq:usc_local_field_near_maximizer} gives \ref{prop:usc_outer_near_maximizing_jet}.
The same observation and the constant-in-time extension give
\(a_j(s,x)\ge\lambda_jI_d\) on
\([0,R]\times\overline D_m\).  On this cylinder, \(\zeta_j\) is a
convex combination of coefficients in
\(\bigcup_{y\in\overline D_{m+1}}\mathfrak A_j(y)\), and hence
\(\|\zeta_j\|\le\mathsf K_j-1\).
This proves \ref{prop:usc_outer_covariance_floor} and ensures that the
Krylov constants in
Proposition~\ref{prop:usc_smooth_elliptic_approximation}\ref{prop:usc_smooth_elliptic_approximation_2}
apply to the selected fields on each \(D_\ell\), \(\ell\le m\).

For \ref{prop:usc_outer_linear_realizing_laws}, first note that
\ref{prop:usc_outer_support_membership} and
\(\varepsilon_j\le\varepsilon_1\) place every \(\zeta_j(s,x)\) in
\(\mathfrak A_G^{\varepsilon_1}(x)\).  Local boundedness of this
correspondence gives spatial local boundedness uniformly in \(s\) and
\(j\).  Moreover, Lemma~\ref{lem:monotone_outer_generator_regularization}(ii) of the main text gives
\[
    \ell_{\zeta_j(s,x)}(J_x^2\phi)
    \le G^{\varepsilon_j}(x,J_x^2\phi)
    \le(C_\phi+\varepsilon_j)\phi(x)
    \le(C_\phi+\varepsilon_1)\phi(x).
\]
Thus Lemma~\ref{lem:continuous_locally_bounded_linear_existence} applies
to each continuous field \(\zeta_j\) and gives
\(\widetilde{\mathcal P}_x(L^{\beta^{\zeta_j}})\neq\varnothing\).
Finally, \ref{prop:usc_outer_support_membership} makes
\(\beta^{\zeta_j}\) generator-compatible with \(G^{\varepsilon_j}\),
proving the asserted inclusion.
\end{proof}



\begin{thebibliography}{37}
\providecommand{\natexlab}[1]{#1}
\providecommand{\url}[1]{\texttt{#1}}
\expandafter\ifx\csname urlstyle\endcsname\relax
  \providecommand{\doi}[1]{doi: #1}\else
  \providecommand{\doi}{doi: \begingroup \urlstyle{rm}\Url}\fi

\bibitem[Artzner et~al.(1999)Artzner, Delbaen, Eber, and
  Heath]{artzner1999coherent}
Philippe Artzner, Freddy Delbaen, Jean-Marc Eber, and David Heath.
\newblock Coherent measures of risk.
\newblock \emph{Mathematical finance}, 9\penalty0 (3):\penalty0 203--228, 1999.

\bibitem[Bartl(2020)]{bartl2020conditional}
Daniel Bartl.
\newblock Conditional nonlinear expectations.
\newblock \emph{Stochastic Processes and their Applications}, 130\penalty0
  (2):\penalty0 785--805, 2020.

\bibitem[Bartl et~al.(2019)Bartl, Cheridito, and
  Kupper]{bartl_cheridito_kupper_19}
Daniel Bartl, Patrick Cheridito, and Michael Kupper.
\newblock Robust expected utility maximization with medial limits.
\newblock \emph{J. Math. Anal. Appl.}, 471\penalty0 (1-2):\penalty0 752--775,
  2019.
\newblock ISSN 0022-247X.
\newblock \doi{10.1016/j.jmaa.2018.11.012}.

\bibitem[Bass et~al.(2007)Bass, Burdzy, and Chen]{bass2007pathwise}
Richard~F. Bass, Krzysztof Burdzy, and Zhen-Qing Chen.
\newblock Pathwise uniqueness for a degenerate stochastic differential
  equation.
\newblock \emph{The Annals of Probability}, 35\penalty0 (6):\penalty0
  2385--2418, 2007.
\newblock \doi{10.1214/009117907000000033}.

\bibitem[Blessing et~al.(2025)Blessing, Denk, Kupper, and
  Nendel]{blessing2025gamma}
Jonas Blessing, Robert Denk, Michael Kupper, and Max Nendel.
\newblock Convex monotone semigroups and their generators with respect to
  {$\Gamma$}-convergence.
\newblock \emph{Journal of Functional Analysis}, 288:\penalty0 110841, 2025.
\newblock \doi{10.1016/j.jfa.2025.110841}.
\newblock URL \url{https://arxiv.org/abs/2202.08653}.

\bibitem[Burzoni et~al.(2021)Burzoni, Riedel, and Soner]{burzoni2021viability}
Matteo Burzoni, Frank Riedel, and H.~Mete Soner.
\newblock Viability and arbitrage under {Knightian} uncertainty.
\newblock \emph{Econometrica}, 89\penalty0 (3):\penalty0 1207--1234, 2021.
\newblock \doi{10.3982/ECTA16535}.

\bibitem[Cesaroni(2006)]{cesaroni2006lyapunov}
Annalisa Cesaroni.
\newblock Lyapunov stabilizability of controlled diffusions via a
  superoptimality principle for viscosity solutions.
\newblock \emph{Applied Mathematics and Optimization}, 53:\penalty0 1--29,
  2006.
\newblock \doi{10.1007/s00245-005-0834-1}.
\newblock URL \url{https://arxiv.org/abs/math/0405169}.

\bibitem[Crandall et~al.(1992)Crandall, Ishii, and Lions]{crandall1992user}
Michael~G Crandall, Hitoshi Ishii, and Pierre-Louis Lions.
\newblock User’s guide to viscosity solutions of second order partial
  differential equations.
\newblock \emph{Bulletin of the American mathematical society}, 27\penalty0
  (1):\penalty0 1--67, 1992.

\bibitem[Criens and Kupper(2025)]{criens2025representation}
David Criens and Michael Kupper.
\newblock Representation theorems for convex expectations and semigroups on
  path space.
\newblock \emph{arXiv preprint arXiv:2503.10572}, 2025.

\bibitem[Criens and Niemann(2024)]{criensniemann2024markov}
David Criens and Lars Niemann.
\newblock Markov selections and {Feller} properties of nonlinear diffusions.
\newblock \emph{Stochastic Processes and their Applications}, 173:\penalty0
  104354, 2024.
\newblock \doi{10.1016/j.spa.2024.104354}.

\bibitem[Criens and Niemann(2025{\natexlab{a}})]{criens2025stochastic}
David Criens and Lars Niemann.
\newblock A stochastic representation theorem for sublinear semigroups with
  non-local generators.
\newblock \emph{Electronic Journal of Probability}, 30:\penalty0 1--36,
  2025{\natexlab{a}}.

\bibitem[Criens and Niemann(2025{\natexlab{b}})]{criens_niemann_JEEQ}
David Criens and Lars Niemann.
\newblock Nonlinear semimartingales and {Markov} processes with jumps.
\newblock \emph{J. Evol. Equ.}, 25\penalty0 (1):\penalty0 39,
  2025{\natexlab{b}}.
\newblock ISSN 1424-3199.
\newblock \doi{10.1007/s00028-024-01046-6}.
\newblock Id/No 21.

\bibitem[Delbaen(2006)]{delbaen2006structure}
Freddy Delbaen.
\newblock The structure of {$m$}--stable sets and in particular of the set of
  risk neutral measures.
\newblock In \emph{In memoriam Paul-Andr{\'e} Meyer: S{\'e}minaire de
  Probabilit{\'e}s XXXIX}, pages 215--258. Springer, 2006.

\bibitem[Delbaen et~al.(2010)Delbaen, Peng, and
  Rosazza~Gianin]{Delbaen_Peng_Rosazza_10}
Freddy Delbaen, Shige Peng, and Emanuela Rosazza~Gianin.
\newblock Representation of the penalty term of dynamic concave utilities.
\newblock \emph{Finance Stoch.}, 14\penalty0 (3):\penalty0 449--472, 2010.
\newblock ISSN 0949-2984.
\newblock \doi{10.1007/s00780-009-0119-7}.

\bibitem[Dellacherie and Meyer(1982)]{dellacherie1982probabilities}
Claude Dellacherie and Paul-Andr{\'e} Meyer.
\newblock \emph{Probabilities and Potential. B: Theory of Martingales},
  volume~72 of \emph{North-Holland Mathematics Studies}.
\newblock North-Holland, Amsterdam, 1982.

\bibitem[Denis et~al.(2011)Denis, Hu, and Peng]{denishu2011function}
Laurent Denis, Mingshang Hu, and Shige Peng.
\newblock Function spaces and capacity related to a sublinear expectation:
  Application to {$G$}-{Brownian} motion paths.
\newblock \emph{Potential Analysis}, 34:\penalty0 139--161, 2011.
\newblock URL \url{https://arxiv.org/abs/0802.1240}.

\bibitem[Engelbert and Schmidt(1985)]{engelbert1985solutions}
Hans-J{\"u}rgen Engelbert and Wolfgang Schmidt.
\newblock On solutions of one-dimensional stochastic differential equations
  without drift.
\newblock \emph{Zeitschrift f{\"u}r Wahrscheinlichkeitstheorie und Verwandte
  Gebiete}, 68\penalty0 (3):\penalty0 287--314, 1985.
\newblock \doi{10.1007/BF00532642}.

\bibitem[Fernholz and Karatzas(2011)]{fernholz_karatzas_11}
Daniel Fernholz and Ioannis Karatzas.
\newblock Optimal arbitrage under model uncertainty.
\newblock \emph{The Annals of Applied Probability}, 21\penalty0 (6):\penalty0
  2191--2225, 2011.
\newblock \doi{10.1214/10-AAP755}.

\bibitem[Fleming and Soner(2006)]{Fleming_Soner_06}
Wendell~H. Fleming and H.~Mete Soner.
\newblock \emph{Controlled {Markov} processes and viscosity solutions},
  volume~25 of \emph{Stoch. Model. Appl. Probab.}
\newblock New York, NY: Springer, 2nd ed. edition, 2006.
\newblock ISBN 0-387-26045-5; 978-1-4419-2078-2; 0-387-31071-1.
\newblock \doi{10.1007/0-387-31071-1}.

\bibitem[Girsanov(1962)]{girsanov1962nonuniqueness}
I.~V. Girsanov.
\newblock An example of non-uniqueness of the solution of the stochastic
  equation of {K. Ito}.
\newblock \emph{Theory of Probability and Its Applications}, 7\penalty0
  (3):\penalty0 325--331, 1962.
\newblock \doi{10.1137/1107031}.

\bibitem[Howitt and Warren(2009)]{howitt2009consistent}
Chris Howitt and Jon Warren.
\newblock Consistent families of {Brownian} motions and stochastic flows of
  kernels.
\newblock \emph{The Annals of Probability}, 37\penalty0 (4):\penalty0
  1237--1272, 2009.
\newblock \doi{10.1214/08-AOP431}.

\bibitem[Koike and {\'S}wi{\k{e}}ch(2013)]{koike_swiech2013representation}
Shigeaki Koike and Andrzej {\'S}wi{\k{e}}ch.
\newblock Representation formulas for solutions of {Isaacs} integro-{PDE}.
\newblock \emph{Indiana University Mathematics Journal}, 62\penalty0
  (5):\penalty0 1473--1502, 2013.
\newblock \doi{10.1512/iumj.2013.62.5109}.

\bibitem[K{\"u}hn(2021)]{kuhn2021infinitesimal}
Franziska K{\"u}hn.
\newblock On infinitesimal generators of sublinear {M}arkov semigroups.
\newblock \emph{Osaka Journal of Mathematics}, 58\penalty0 (3):\penalty0
  487--508, 2021.

\bibitem[Kusuoka(2017)]{kusuoka2017supermartingale}
Shigeo Kusuoka.
\newblock On supermartingale problems.
\newblock In Shigeo Kusuoka and Toru Maruyama, editors, \emph{Advances in
  Mathematical Economics}, volume~21 of \emph{Advances in Mathematical
  Economics}, pages 75--97. Springer, Singapore, 2017.
\newblock \doi{10.1007/978-981-10-4145-7_3}.

\bibitem[Le~Jan and Raimond(2004)]{lejan2004flows}
Yves Le~Jan and Olivier Raimond.
\newblock Flows, coalescence and noise.
\newblock \emph{The Annals of Probability}, 32\penalty0 (2):\penalty0
  1247--1315, 2004.
\newblock \doi{10.1214/009117904000000207}.

\bibitem[Lions and Nisio(1982)]{lions_nisio_82}
Pierre-Louis Lions and Makiko Nisio.
\newblock A uniqueness result for the semigroup associated with the {Hamilton}-
  {Jacobi}-{Bellman} operator.
\newblock \emph{Proc. Japan Acad., Ser. A}, 58:\penalty0 273--276, 1982.
\newblock ISSN 0386-2194.
\newblock \doi{10.3792/pjaa.58.273}.

\bibitem[Neufeld and Nutz(2014)]{neufeld2014measurability}
Ariel Neufeld and Marcel Nutz.
\newblock Measurability of semimartingale characteristics with respect to the
  probability law.
\newblock \emph{Stochastic Processes and their Applications}, 124\penalty0
  (11):\penalty0 3819--3845, 2014.
\newblock \doi{10.1016/j.spa.2014.07.006}.

\bibitem[Neufeld and Nutz(2017)]{neufeld_nutz_17}
Ariel Neufeld and Marcel Nutz.
\newblock Nonlinear {L{\'e}vy} processes and their characteristics.
\newblock \emph{Trans. Am. Math. Soc.}, 369\penalty0 (1):\penalty0 69--95,
  2017.
\newblock ISSN 0002-9947.
\newblock \doi{10.1090/tran/6656}.

\bibitem[Nisio(1982)]{nisio_82}
Makiko Nisio.
\newblock Note on uniqueness of semigroup associated with {Bellman} operator.
\newblock Advances in filtering and optimal stochastic control, {Proc}.
  {IFIP}-{WG} 7/1 {Work}. {Conf}., {Cocoyoc}/{Mex}. 1982, {Lect}. {Notes}
  {Contr}. {Inf}. {Sci}. 42, 267-275 (1982)., 1982.

\bibitem[Nisio(2015)]{zbMATH06323249}
Makiko Nisio.
\newblock \emph{Stochastic control theory. {Dynamic} programming principle},
  volume~72 of \emph{Probab. Theory Stoch. Model.}
\newblock Tokyo: Springer, 2nd ed. edition, 2015.
\newblock ISBN 978-4-431-55122-5; 978-4-431-55123-2.
\newblock \doi{10.1007/978-4-431-55123-2}.

\bibitem[Nutz(2013)]{nutz2013random}
Marcel Nutz.
\newblock Random {$G$}-expectations.
\newblock \emph{The Annals of Applied Probability}, 23\penalty0 (5):\penalty0
  1755--1777, 2013.
\newblock \doi{10.1214/12-AAP885}.

\bibitem[Nutz and Van~Handel(2013)]{nutz2013constructing}
Marcel Nutz and Ramon Van~Handel.
\newblock Constructing sublinear expectations on path space.
\newblock \emph{Stochastic processes and their applications}, 123\penalty0
  (8):\penalty0 3100--3121, 2013.

\bibitem[Peng(2019)]{peng2019nonlinear}
Shige Peng.
\newblock \emph{Nonlinear expectations and stochastic calculus under
  uncertainty: with robust {CLT} and {$G$}-{B}rownian motion}, volume~95.
\newblock Springer Nature, 2019.

\bibitem[Salins and Spiliopoulos(2017)]{salins2017markov}
Michael Salins and Konstantinos Spiliopoulos.
\newblock {Markov} processes with spatial delay: Path space characterization,
  occupation time and properties.
\newblock \emph{Stochastics and Dynamics}, 17\penalty0 (6):\penalty0 1750042,
  2017.
\newblock \doi{10.1142/S0219493717500423}.

\bibitem[Schneider(2014)]{schneider2014convex}
Rolf Schneider.
\newblock \emph{Convex Bodies: The Brunn--Minkowski Theory}.
\newblock Cambridge University Press, 2 edition, 2014.

\bibitem[Stroock and Varadhan(1997)]{stroock1997multidimensional}
Daniel~W Stroock and SR~Srinivasa Varadhan.
\newblock \emph{Multidimensional diffusion processes}, volume 233.
\newblock Springer Science \& Business Media, 1997.

\bibitem[{\'S}wi{\k{e}}ch(1996)]{swiech1996another}
Andrzej {\'S}wi{\k{e}}ch.
\newblock Another approach to the existence of value functions of stochastic
  differential games.
\newblock \emph{Journal of Mathematical Analysis and Applications},
  204\penalty0 (3):\penalty0 884--897, 1996.
\newblock \doi{10.1006/jmaa.1996.0474}.

\end{thebibliography}

\begin{thebibliography}{12}
\providecommand{\natexlab}[1]{#1}
\providecommand{\url}[1]{\texttt{#1}}
\expandafter\ifx\csname urlstyle\endcsname\relax
  \providecommand{\doi}[1]{doi: #1}\else
  \providecommand{\doi}{doi: \begingroup \urlstyle{rm}\Url}\fi

\bibitem[Billingsley(2013)]{supp:billingsley2013convergence}
Patrick Billingsley.
\newblock \emph{Convergence of probability measures}.
\newblock John Wiley \& Sons, 2013.

\bibitem[Cannarsa and Sinestrari(2004)]{supp:cannarsa2004semiconcave}
Piermarco Cannarsa and Carlo Sinestrari.
\newblock \emph{Semiconcave Functions, Hamilton--Jacobi Equations, and Optimal
  Control}, volume~58 of \emph{Progress in Nonlinear Differential Equations and
  Their Applications}.
\newblock Birkh\"auser, Boston, 2004.

\bibitem[Crandall et~al.(1992)Crandall, Ishii, and Lions]{supp:crandall1992user}
Michael~G Crandall, Hitoshi Ishii, and Pierre-Louis Lions.
\newblock User’s guide to viscosity solutions of second order partial
  differential equations.
\newblock \emph{Bulletin of the American mathematical society}, 27\penalty0
  (1):\penalty0 1--67, 1992.

\bibitem[Criens(2020)]{supp:criens2020existence}
David Criens.
\newblock On the existence of semimartingales with continuous characteristics.
\newblock \emph{Stochastics}, 92\penalty0 (5):\penalty0 785--813, 2020.

\bibitem[Dzhaparidze and van Zanten(2001)]{supp:dzhaparidze2001bernstein}
K.~Dzhaparidze and J.~H. van Zanten.
\newblock On bernstein-type inequalities for martingales.
\newblock \emph{Stochastic Processes and their Applications}, 93\penalty0
  (1):\penalty0 109--117, 2001.
\newblock \doi{10.1016/S0304-4149(00)00086-7}.

\bibitem[Jacod and Shiryaev(2013)]{supp:jacod2013limit}
Jean Jacod and Albert Shiryaev.
\newblock \emph{Limit theorems for stochastic processes}, volume 288.
\newblock Springer Science \& Business Media, 2013.

\bibitem[Kechris(1995)]{supp:kechris1995classical}
Alexander~S. Kechris.
\newblock \emph{Classical Descriptive Set Theory}, volume 156 of \emph{Graduate
  Texts in Mathematics}.
\newblock Springer-Verlag, New York, 1995.
\newblock \doi{10.1007/978-1-4612-4190-4}.

\bibitem[Krylov(1980)]{supp:krylov1980controlled}
N.~V. Krylov.
\newblock \emph{Controlled Diffusion Processes}, volume~14 of
  \emph{Applications of Mathematics}.
\newblock Springer-Verlag, New York, 1980.
\newblock Translated from the Russian by A. B. Aries.

\bibitem[Lee(2013)]{supp:lee2013introduction}
John~M. Lee.
\newblock \emph{Introduction to Smooth Manifolds}, volume 218 of \emph{Graduate
  Texts in Mathematics}.
\newblock Springer, 2 edition, 2013.

\bibitem[Liu and Neufeld(2019)]{supp:liu2019compactness}
Chong Liu and Ariel Neufeld.
\newblock Compactness criterion for semimartingale laws and semimartingale
  optimal transport.
\newblock \emph{Transactions of the American Mathematical Society},
  372\penalty0 (1):\penalty0 187--231, 2019.
\newblock \doi{10.1090/tran/7663}.

\bibitem[Michael(2003)]{supp:michael2003continuous}
Ernest Michael.
\newblock Continuous selections.
\newblock In \emph{Encyclopedia of General Topology}, pages 107--109. Elsevier,
  2003.

\bibitem[Stroock and Varadhan(1997)]{supp:stroock1997multidimensional}
Daniel~W Stroock and SR~Srinivasa Varadhan.
\newblock \emph{Multidimensional diffusion processes}, volume 233.
\newblock Springer Science \& Business Media, 1997.

\end{thebibliography}
\end{document}